\documentclass[oneside,italian]{amsbook}
\usepackage[latin1]{inputenc}
\usepackage[italian]{babel}
\usepackage{amsmath}
\usepackage{amsfonts}
\usepackage{amssymb}
\usepackage{graphicx}
\newtheorem*{teorema}{Teorema}
\newtheorem{lemma}{Lemma}[section]
\newtheorem{definizione}{Definizione}[section]

\newtheorem{oss}{Osservazione}[section]
\newtheorem{prop}{Proposizione}[section]
\newtheorem{cor}{Corollario}[section]

\makeatletter

\begin{document}
\thispagestyle{empty}
\pagestyle{empty}

\begin{center}
{\large UNIVERSIT\`{A} DEGLI STUDI DI ROMA}\\
{\large ``LA SAPIENZA''}\\
{\normalsize $\;$}\\
{\normalsize FACOLT\`{A} DI SCIENZE MATEMATICHE FISICHE E NATURALI}\\
{\normalsize }{\large $\,$}{\normalsize }\\
{\normalsize TESI DI LAUREA SPECIALISTICA IN FISICA}\\
{\normalsize $\;$}\\
\vskip1cm
{\LARGE {\bf Studio della misura di Sinai - Ruelle - Bowen in un}}\\
{\LARGE {\bf sistema semplice}}
\vskip6cm
{$\;\;\;\;\;\;\;\;\;\;\;$Relatore:\hfill{}$\;\;\;\;\;\;\;\;\;$Candidato:$\;\;\;\;\;\,\,\,\:$\\
Prof. Giovanni Gallavotti\hfill{}Marcello Porta$\;\;\;\;$\hfill{}}\\
\vfill
{\large Anno Accademico 2006 - 2007}
\end{center}
\newpage
{\hfill{}{\it A Viviana, per la sua infinita pazienza.}}

\tableofcontents{}

\chapter{Introduzione}

\section{Il problema di una descrizione microscopica}

\subsection{Equilibrio e non equilibrio: definizioni}

In questa tesi vogliamo discutere il problema della costruzione di una {\em meccanica statistica dei sistemi fuori dall'equilibrio}, ovvero di una teoria che, partendo dalla dinamica microscopica, permetta di predire relazioni tra valori medi di osservabili fisiche in uno stato di non equilibrio.

Innanzitutto, dobbiamo specificare cosa intendiamo per {\em equilibrio}\index{sistema all'equilibrio} e {\em non equilibrio}\index{sistema fuori dall'equilibrio}. Un sistema all'equilibrio \` e un sistema di particelle in uno stato stazionario che evolve risolvendo le equazioni del moto di Hamilton; chiamiamo {\em stato stazionario}\index{stato stazionario} uno stato fisico che non muta nel tempo, ovvero se la termodinamica del sistema \` e descritta da un insieme di osservabili $\mathcal{A}=\left\{A_{i}\right\}$ i valori medi delle $A_{i}$ sono costanti.\footnote{Ad esempio, possiamo immaginare che $\mathcal{A}$ sia formato dall'energia potenziale media $\Phi$, dall'energia cinetica media $T$, dall'energia totale $U$, dalla pressione $P$, dalla densit\` a $\rho$ e dal volume $V$.}

Un sistema fuori dall'equilibrio pu\` o essere tale per due motivi:

\begin{enumerate}
\item\label{punto:1} le particelle evolvono risolvendo equazioni del moto che non sono le equazioni di Hamilton, e/o
\item\label{punto:2} il sistema non si trova in uno stato stazionario.
\end{enumerate}

In seguito, considereremo esclusivamente {\em stati stazionari fuori dall'equilibrio}: in generale, il nostro sistema sar\` a soggetto a delle forze non conservative\index{forze non conservative} - ovvero che non possono essere espresse come il gradiente di un potenziale - e, per permettere l'esistenza di uno stato stazionario, il calore prodotto dovr\` a essere assorbito da un {\em termostato}.

\subsection{Caratteristiche di una teoria microscopica}

 \` E chiaro che, all'equilibrio o fuori dall'equilibrio, lo studio di $10^{23}$ equazioni del moto \` e un problema intrattabile. Ad ogni modo, \` e ragionevole pensare che ai fini del calcolo di ``poche'' osservabili fisiche i dettagli della traiettoria della singola particella siano inessenziali; prima di scrivere una ad una le equazioni del moto su un calcolatore e aspettare il tempo necessario affinch\'e vengano risolte, dobbiamo chiederci se sia davvero necessario avere cos\` i tanta informazione. Quali sono le caratteristiche del moto che hanno una rilevanza da un punto di vista macroscopico, ovvero che {\em giustificano} la termodinamica\index{termodinamica}?
  
Si potrebbe provare a rispondere a questa domanda assumendo le propriet\` a del moto che riteniamo plausibili, e verificando l'accordo tra le nostre predizioni teoriche e la realt\` a fisica. Questo \` e esattamente ci\` o che accade nello studio dei sistemi all'equilibrio, dove l'assunzione sulla dinamica microscopica che permette di giustificare la termodinamica prende il nome di {\em ipotesi ergodica}\index{ipotesi ergodica};  nonostante nelle applicazioni questa ipotesi sia praticamente impossibile da dimostrare, le sue conseguenze fisiche sono largamente verificate. Ci\` o suggerisce che non abbia senso, o perlomeno che sia superfluo, verificare la  stretta  validit\` a dell'ipotesi ergodica ogni volta che si vogliono usare le tecniche della meccanica statistica dell'equilibrio; piuttosto, l'atteggiamento corretto \` e quello di assumere che al fine di calcolare quantit\` a macroscopiche un sistema all'equilibrio si comporti {\em come se} fosse ergodico.

I sistemi fuori dall'equilibrio sono incompatibili con l'ipotesi ergodica. Per derivare la termodinamica partendo dalle leggi del moto \` e necessario fare delle nuove assunzioni sulla dinamica microscopica, e verificare che la teoria che si ottiene sia consistente con i risultati sperimentali e con quanto gi\` a noto all'equilibrio; potremmo dire di aver raggiunto l'obiettivo pi\` u ambizioso se riuscissimo a costruire una teoria che permetta di descrivere la {\em statistica} di un sistema allo stato stazionario indipendentemente dal fatto che si trovi o meno all'equilibrio, e che dunque si riduca all'usuale meccanica statistica dell'equilibrio nel limite di forze non conservative nulle.

Negli ultimi $70$ anni molte propriet\` a di sistemi ``quasi'' all'equilibrio, ovvero sottoposti a forze non conservative con ``piccole'' intensit\` a, sono state efficacemente descritte con le tecniche della meccanica statistica dell'equilibrio\footnote{Vedere \cite{DGM} per una esauriente esposizione di questi risultati.}; chiaramente, c'\` e un'ambiguita in cosa si intende per ``piccolo'', e questa ambiguit\` a non pu\` o che essere risolta caso per caso. In queste situazioni si assume che tra equilibrio e non equilibrio ``le cose cambino poco'' in prima approssimazione, e con questo approccio si riescono a dimostrare risultati notevoli come la {formula di Green - Kubo}\index{formula di Green - Kubo}, meglio noto come {\em teorema di fluttuazione - dissipazione}\index{teorema di fluttuazione-dissipazione}, e le {\em relazioni di reciprocit\` a di Onsager}\index{relazioni di reciprocit\` a di Onsager} (assumendo che la dinamica sia {\em reversibile}); parlando di questi risultati ci riferiremo alla {\em teoria della risposta lineare}\index{teoria della risposta lineare}.

Quindi, un'altra caratteristica richiesta ad una teoria microscopica del non equilibrio \` e che per ``piccole'' forze non conservative implichi i ben noti risultati della teoria della risposta lineare, possibilmente estendendoli agli ordini successivi.

\section{Un esperimento}

Lo studio dei sistemi fuori dall'equilibrio ricevette un forte impulso nel 1993, grazie ai risultati sperimentali di D. J. Evans, E. G. D. Cohen, G. P. Morriss, \cite{ECM93}; gli autori simularono un fluido bidimensionale composto da $N=56$ particelle, che evolve con le equazioni del moto

\begin{equation}
\dot{\underline{q}}_{j}=\frac{1}{m}\underline{p}_{j}+\underline{i}\gamma y_{j}, \qquad \dot{\underline{p}}_{j}=\underline{F}_{j}-\underline{i}\gamma p_{y_{j}}-\alpha\underline{p}_{j},\label{eq:EvCM1}
\end{equation}

dove $\underline{p}_{j}=\left(p_{x_{j}},p_{y_{j}}\right)$, $\gamma=\gamma(\underline{q}_{j})=\frac{\partial u_{x}\left(\underline{q}_{j}\right)}{\partial y}$ se $\underline{u}(\underline{q}_{j})=(u_{x}(\underline{q}_{j}),u_{y}(\underline{q}_{j}))$ \` e la velocit\` a della particella $j$-esima nel punto $\underline{q}_{j}$, $\underline{i}$ \` e il versore nella direzione $x$, $\alpha$ \` e determinato con il {\em principio di minimo sforzo di Gauss}\index{principio di minimo sforzo di Gauss} ({\em cf.} capitolo \ref{cap:noneq}), e $\underline{F}_{j} $ \`e la risultante delle forze d'interazione tra la particella $j$ e le rimanenti $N-1=55$. Le particelle interagiscono mediante un potenziale

\begin{equation}
\Phi(\underline{r})=\left\{\begin{array}{c} 4\epsilon\left[\left(\frac{\sigma}{r}\right)^{12}-\left(\frac{\sigma}{r}\right)^{6}\right]+\epsilon\qquad\mbox{per $r<2^{\frac{1}{6}}\sigma$}\\ 0 \qquad\qquad\qquad\qquad\qquad\mbox{altrimenti}.\end{array}\right.
\end{equation}

Con questo modello Evans, Cohen e Morriss verificarono numericamente che, chiamando $\sigma$ il tasso di produzione di entropia\footnote{Come discuteremo nel capitolo \ref{cap:noneq}, non \` e ovvio chi siano l'entropia e il suo tasso di produzione in un sistema fuori dall'equilibrio; ad ogni modo, nel caso di {\em termostati gaussiani} il tasso di produzione di entropia pu\` o essere identificato con il {\em tasso di contrazione dello spazio delle fasi}\index{tasso di contrazione dello spazio delle fasi}, proporzionale alla quantit\` a $\alpha$ che appare nella (\ref{eq:EvCM1}).}, $\sigma_{\tau}$ la sua media su un tempo $\tau$ e indicando con $Freq(B)$ la frequenza delle occorrenze dell'evento $B$,

\begin{equation}
\frac{1}{\tau P}\log\frac{Freq\left(\sigma_{\tau}=P\right)}{Freq\left(\sigma_{\tau}=-P\right)}\sim1 \label{eq:EvCM2}
\end{equation}

a meno di correzioni tanto pi\` u piccole quanto pi\` u \` e grande $\tau$; l'evento $\left\{\sigma_{\tau}=-P\right\}$ pu\` o essere visto come una ``violazione'' della seconda legge della termodinamica\index{violazione della seconda legge della termodinamica}, vera solo nel limite $\tau\rightarrow\infty$. Nello stesso lavoro Evans, Cohen e Morriss provarono a giustificare la (\ref{eq:EvCM2}) congetturando che la probabilit\` a che il sistema si trovi a percorrere in un tempo $\tau$ il ``segmento di traiettoria nello spazio delle fasi'' $i$ abbia la forma

\begin{equation}
\mu_{i,\tau}=\frac{e^{-\sum_{n}\lambda_{i,n}^{+}\tau}}{Z}, \label{eq:EvCM3}
\end{equation}

dove $\lambda_{i,n}^{+}$ \` e l'$n$-esimo esponente di Lyapunov\index{esponente di Lyapunov} positivo calcolato sul segmento $i$-esimo, e $Z$ \` e un'opportuna normalizzazione.

Come abbiamo accennato poco fa, il parametro $\alpha$ che appare nella (\ref{eq:EvCM1}) \` e determinato con il principio di Gauss; nella sezione \ref{sez:term} del capitolo \ref{cap:noneq} vedremo che questa scelta render\` a {\em reversibili} le equazioni del moto, e inoltre ci permetter\` a di stabilire un'identificazione tra il tasso di produzione di entropia e la somma di tutti gli esponenti di Lyapunov\footnote{Pi\` u precisamente, potremo identificare il tasso di produzione di entropia con il {\em tasso di contrazione dello spazio delle fasi}.}. Grazie a questa corrispondenza e alla reversibilit\` a del moto, assumendo la (\ref{eq:EvCM3}) come misura di probabilit\` a invariante Evans, Cohen e Morriss riuscirono a dedurre la (\ref{eq:EvCM2}), e conclusero interpretando il risultato sperimentale (\ref{eq:EvCM2}) come una  verifica della loro congettura (\ref{eq:EvCM3}).

\section{Il teorema di fluttuazione di Gallavotti-Cohen}

L'argomento di Evans, Cohen e Morriss, per\` o, \` e troppo vago per essere considerato una dimostrazione del risultato (\ref{eq:EvCM2}); data l'importanza della questione, \` e necessario capire {\em perch\'e} la misura di probabilit\` a invariante di un sistema fuori dall'equilibrio dovrebbe avere la forma (\ref{eq:EvCM3}).

Nel 1995 G. Gallavotti e E. G. D. Cohen proposero, {\em cf.} \cite{GC95}, \cite{GC95a}, un punto di vista che, se accettato, permette di dire con rigore matematico quale sia la misura di probabilit\` a invariante di un sistema fuori dall'equilibrio; in questo contesto la formula (\ref{eq:EvCM2}) diventa un teorema, \cite{GdimFT}, \cite{Ge96}, noto come {\em teorema di fluttuazione di Gallavotti - Cohen}. La proposta di Gallavotti e Cohen consiste in un'assunzione sulla dinamica delle particelle che generalizza l'ipotesi ergodica\index{ipotesi ergodica}, perch\'e implica che all'equilibrio lo stato del sistema sia descritto dalla misura di Lebesgue sulla superficie dello spazio delle fasi ad energia costante, ovvero dall'ensemble microcanonico; l'assunzione consiste nel dire che {\em allo scopo di misurare propriet\` a macroscopiche} un sistema macroscopico pu\` o essere pensato appartenente ad una speciale classe di sistemi dinamici caotici, chiamati {\em iperbolici transitivi} o {\em Anosov}, per i quali il valor medio di un'osservabile $\mathcal{O}$ nel futuro pu\` o essere scritto informalmente come:\footnote{Per dare senso alla (\ref{eq:infsrb}) \` e necessaria una {\em discretizzazione dello spazio delle fasi}; come vedremo, per la classe di sistemi ``molto caotici'' che studieremo ci\` o potr\` a essere fatto in modo relativamente facile. }

\begin{equation}
\left\langle\mathcal{O}\right\rangle_{+}=\frac{\sum_{traiettorie\;nello\;spazio\;delle\;fasi}\mathcal{O}[x(0)]e^{-\sum_{k}A_{u}[x(t_{k})]}}{Z}, \label{eq:infsrb}
\end{equation}

dove $x(t)$ \` e lo stato del sistema lungo la traiettoria nello spazio delle fasi al tempo $t$, $A_{u}$ \` e il logaritmo del {\em coefficiente di espansione}\index{coefficiente di espansione}, ovvero la somma di tutti gli esponenti di Lyapunov positivi del sistema di particelle, e $\left\{x(t_{k})\right\}$ \` e un ``campionamento opportuno'' dell'evoluzione del sistema\footnote{Se la dinamica avviene ad intervalli di tempo discreti, come assumeremo sempre in questa tesi, non avremo bisogno del campionamento.} $t\rightarrow x(t)$. Sotto le stesse ipotesi che implicano la (\ref{eq:infsrb}), invertendo la freccia del tempo il valor medio di un'osservabile nel passato pu\`o essere ottenuto sostituendo $-A_{u}[x(t)]$ con il logaritmo del {\em coefficiente di contrazione}\index{coefficience di contrazione} $A_{s}[x(t)]$. Infine, dal momento che la misura \` e invariante possiamo sostituire l'argomento $x(0)$ dell'osservabile $\mathcal{O}$ con $x\left(t'\right)$ e $t'$ generico. L'assunzione che implica la (\ref{eq:infsrb}) \` e chiamata {\em ipotesi caotica}, e verr\` a discussa ampiamente nella sezione \ref{sez:IC} del capitolo \ref{cap:noneq}.

\section{Organizzazione del lavoro di tesi}

Il lavoro di tesi \` e organizzato in 3 capitoli e 4 appendici; ne riassumiamo brevemente il contenuto.

\begin{itemize}
\item Nel \textbf{capitolo \ref{cap:ergo}} presentiamo qualche concetto della teoria matematica che ci permetter\` a di comprendere l'ipotesi caotica, con particolare attenzione ad una semplice applicazione nella quale \` e possibile calcolare esplicitamente i potenziali della misura di probabilit\` a invariante.
\item Nel \textbf{capitolo \ref{cap:noneq}} riprendiamo il problema del non equilibrio accennato in questa introduzione, in particolare discutendo l'identificazione tra tasso di produzione di entropia e tasso di contrazione dello spazio delle fasi, e presentiamo una dimostrazione rigorosa del teorema di fluttuazione di Gallavotti-Cohen. In seguito mostriamo una generalizzazione del teorema, e riassumiamo qualche altro risultato noto fuori dall'equilibrio.
\item Nel \textbf{capitolo \ref{cap:appl}} discutiamo alcune conseguenze del teorema di fluttuazione; il lavoro originale consiste in una generalizzazione della teoria della risposta lineare, ricavando a partire dal teorema di fluttuazione alcune relazioni tra funzioni di correlazione a tutti gli ordini. In particolare, saremo in grado di scrivere il valore di aspettazione di qualunque osservabile $\mathcal{O}$ dispari sotto inversione temporale attraverso uno sviluppo in funzioni di correlazione $\mathcal{O}-\sigma$, se $\sigma$ \` e il tasso di produzione di entropia; come verificheremo, all'ordine pi\` u basso questo risultato implicher\` a la formula di Green - Kubo e le relazioni di reciprocit\` a di Onsager. In seguito, applicheremo le relazioni ottenute per verificare il teorema di fluttuazione in un semplice caso in cui, da alcune simulazioni numeriche, sembra essere vero nonostante non ne siano verificate le ipotesi.
\item Nell'\textbf{appendice \ref{app:eur}} discutiamo una derivazione euristica della misura di probabilit\` a invariante che descrive i sistemi fuori dall'equilibrio nel contesto dell'ipotesi caotica.
\item Nell'\textbf{appendice \ref{app:C}} calcoliamo esplicitamente il funzionale di grandi deviazioni del tasso di produzione di entropia mediato su un tempo finito. Il contenuto di questa appendice \` e originale.
\item Nell'\textbf{appendice \ref{app:D}} discutiamo una tecnica matematica che ci permetter\` a di studiare le propriet\` a di analiticit\` a della misura di probabilit\` a invariante che descrive i sistemi fuori dall'equilibrio nel contesto dell'ipotesi caotica. 
\item Infine, grazie alla tecnica esposta nell'appendice \ref{app:D}, nell'\textbf{appendice \ref{app:B}} discutiamo le propriet\` a di analiticit\` a dei valori medi di osservabili analitiche.
\end{itemize}

\chapter{Cenni di teoria ergodica} \label{cap:ergo}

\section{Sistemi dinamici astratti}

In questa sezione vogliamo presentare una breve introduzione alla teoria astratta dei sistemi dinamici, e ricordare qualche concetto fondamentale di teoria della misura.

\begin{definizione}{{\em (Sistemi dinamici discreti)}}\label{def:1}

Sia $\Omega$ uno spazio metrico compatto\footnote{Un {\em insieme compatto $G\in\Omega$}\index{spazio metrico compatto} \` e un insieme {\em chiuso}\index{insieme chiuso} e {\em limitato}\index{insieme limitato}; un insieme chiuso \` e un insieme che contiene i limiti delle sue successioni, mentre un insieme \` e detto {\em limitato} se pu\` o essere contenuto in una ``palla'' di raggio finito, ovvero se esiste $x\in \Omega$ tale che, indicando con $d$ una {\em distanza} su $\Omega$,  per ogni $g\in G$ $d(x,s)<r$, $r\geq0$.} e separabile\footnote{Un {\em insieme separabile} $G$ \` e un insieme che contiene un sottoinsieme {\em denso} e {\em numerabile}; ci\` o vuol dire che l'insieme $G$ pu\` o essere `` approssimato arbitrariamente bene'' da una famiglia numerabile di suoi sottoinsiemi, ovvero che possono essere messi in corrispondenza uno-a-uno con l'insieme dei numeri interi $\mathbb{Z}$.}, e sia $S$ una mappa continua di $\Omega$ in s\'e stesso. La coppia $(\Omega, S)$ sar\` a chiamata {\em sistema dinamico topologico discreto}\index{sistema dinamico topologico discreto} su $\Omega$. $(\Omega,S)$ sar\` a detto {\em invertibile}\index{sistema dinamico invertibile} se $S$ ha un'inversa continua $S^{-1}$.

Sia $\mu$ una misura di probabilit\` a completa\footnote{Una misura $\mu$ \` e detta {\em completa} se ogni sottoinsieme di un insieme di $\mu$-misura nulla ha $\mu$-misura nulla.}\index{misura di probabilit\` a completa} definita sulla $\sigma$-algebra $\mathcal{B}$\footnote{\label{nota:2}Un'{\em algebra} \` e una collezione di insiemi contenente l'insieme vuoto chiusa rispetto alle operazioni di complemento e di unione finita; se \` e anche chiusa rispetto all'operazione di unione numerabile, la collezione viene detta una {\em $\sigma$-algebra}.}\index{$\sigma$-algebra} dei sottoinsiemi di $\Omega$, dove $N\in\mathcal{B}$ \` e un insieme di $\mu$-misura nulla, e supponiamo che $S$ sia una mappa misurabile\index{funzione misurabile} rispetto a $\mathcal{B}$ al di fuori di $N$\footnote{Una funzione $S$ \` e detta {\em misurabile} rispetto alle $\sigma$-algebre $\mathcal{B}$, $\mathcal{B'}$ se l'immagine inversa $S^{-1}(E)$ di un insieme $E\in\mathcal{B}$ appartiene a $\mathcal{B'}$; in particolare, se $\mathcal{B}=\mathcal{B'}=\mathcal{B}_{\mu}$, dove $B_{\mu}$ \` e la $\sigma$-algebra degli insiemi $\mu$-misurabili, la funzione \` e detta $\mu$-misurabile.}. Se

\begin{equation}
\mu(A)=\mu(S^{-1}A)\qquad\mbox{per ogni $\mathcal{A}\in\mathcal{B}\cap N^{c}$}, \label{eq:erg1}
\end{equation}

dove $N^{c}=\Omega/N$ \` e il complemento di $N$, diremo che $(\Omega,S)$ \` e $\mu$-invariante. La tripla $(\Omega,S,\mu)$ sar\` a chiamata un {\em sistema dinamico discreto mod $0$} definito al di fuori di $N$. Se $N=\emptyset$ la tripla $(\Omega,S,\mu)$ sar\` a detta un {\em sistema dinamico metrico}\index{sistema dinamico metrico}. 
\end{definizione}

La definizione che abbiamo appena dato coinvolge sistemi dinamici che evolvono ad intervalli di tempo discreti; l'analogo nel continuo della definizione \ref{def:1} pu\` o essere formulato nel seguente modo.

\begin{definizione}{{\em (Flussi)}}\label{def:flux}

Sia $\Omega$ uno spazio metrico compatto e separabile. Sia $(S_{t})_{t\in\mathbb{R}}$ o $(S_{t})_{t\in\mathbb{R}_{+}}$ un gruppo o un semigruppo, omomorfo a $\mathbb{R}$ o a $\mathbb{R}_{+}$, di mappe che agiscono con continuit\` a su $\Omega$.\footnote{Questo significa che la funzione $(t,x)\rightarrow S_{t}x$ di $\mathbb{R}\times\Omega$, o di $\mathbb{R}_{+}\times\Omega$, in $\Omega$ \` e continua.}

Le coppie $(\Omega, (S_{t})_{t\in\mathbb{R}})$, $(\Omega, (S_{t})_{t\in\mathbb{R}_{+}})$ saranno chiamate rispettivamente un {\em flusso topologico invertibile}\index{flusso topologico invertibile} o un {\em flusso topologico}\index{flusso topologico} su $\Omega$.

Sia $\mu$ una misura di probabilit\` a completa su una $\sigma$-algebra $\mathcal{B}$ e sia $(S_{t})_{t\in\mathbb{R}}$ un gruppo di mappe $\mu$-misurabili omomorfe a $\mathbb{R}$ e che conservano $\mu$. Supponiamo che la funzione $(t,x)\rightarrow S_{t}x$ definita in $\mathbb{R}\times\Omega$ con valori in $\Omega$ sia $\mu$-misurabile rispetto alla $\sigma$-algebra generata dagli insiemi in\footnote{Con la notazione $\mathcal{B}(\Omega)$ indichiamo la $\sigma$-algebra dei boreliani di $\Omega$, ovvero la $\sigma$-algebra generata dagli insiemi aperti di $\Omega$. Se $\Omega=\mathbb{R}$ \` e possibile dimostrare, {\em cf.} \cite{Si}, che la $\sigma$-algebra dei boreliani pu\` o essere equivalentemente generata dagli insiemi chiusi di $\mathbb{R}$, oppure dagli insiemi semi-aperti di $\mathbb{R}$ ecc.} $\mathcal{B}(\mathbb{R})\times\mathcal{B}$. Allora il flusso $(\Omega, (S_{t})_{t\in\mathbb{R}},\mu)$ sar\` a chiamato un {\em flusso metrico invertibile}\index{flusso metrico invertibile} su $\Omega$.

Se sostituiamo $(S_{t})_{t\in\mathbb{R}}$ con un semigruppo $(S_{t})_{t\in\mathbb{R}_{+}}$, $\mathcal{B}(\mathbb{R}_{+})\times\mathcal{B}$-misurabile, otteniamo la nozione di {\em flusso metrico}\index{flusso metrico}.
\end{definizione}

Nel seguito ci occuperemo sempre di sistemi dinamici discreti, dal momento che un sistema dinamico continuo pu\` o essere ricondotto ad un sistema dinamico discreto mediante la tecnica della {\em sezione di Poincar\'e}\index{sezione di Poincar\'e}, {\em cf.} \cite{Ge96}. Inoltre, da un punto di vista fisico i sistemi dinamici continui sono piuttosto inusuali, perch\'e in un esperimento \` e impossibile controllare la dinamica di un sistema fisico in tempo continuo; quello che si fa, invece, \` e acquisire informazioni sullo stato del sistema a determinati intervalli di tempo, che, a causa delle limitazioni sperimentali, non possono essere arbitrariamente piccoli.

Parlando di sistemi dinamici discreti, qualifiche tipo topologico, metrico ecc. possono essere omesse riferendosi semplicemente a {\em sistemi dinamici astratti}\index{sistema dinamico astratto} $(\Omega,S)$, dove $\Omega$ \` e lo spazio sul quale agisce la mappa $S$. 

Per concludere, introduciamo la nozione di {\em isomorfismo}\index{isomorfismo}, o {\em coniugazione}, tra sistemi dinamici.

\begin{definizione}{{\em (Isomorfismi)}}

Se $(\Omega,S)$ e $(\Omega',S')$ sono due sistemi dinamici astratti diremo che sono {\em isomorfi}\index{sistemi dinamici isomorfi}, o {\em coniugati}\index{sistemi dinamici coniugati}, se esiste una mappa invertibile $I:\Omega\leftrightarrow\Omega'$ tale che\footnote{Con il simbolo ``$\circ$'' indichiamo la {\em composizione} tra mappe, ovvero $(S\circ  I)(x)=S(I(x))$.}

\begin{equation}
I\circ S=S'\circ I.
\end{equation}
 
\end{definizione}

\section{Dinamica simbolica} \label{sez:dinsimb}

L'evoluzione delle traiettorie di un sistema dinamico $(\Omega,S)$ pu\` o essere descritta misurando ad ogni instante di tempo prefissato delle quantit\` a di interesse, e compilando una lista con tutte queste caratteristiche in funzione del tempo; matematicamente, sintetizziamo queste informazioni nella {\em storia}\index{storia di un punto} $\underline{\sigma}(x)$ di un punto $x$. Ovviamente, un punto $x$ identifica univocamente la sua storia, mentre \` e naturale chiedersi se da una storia sia possibile risalire al punto iniziale; approfondiremo questa questione tra poco. 

Proviamo a formulare il concetto di storia di un punto in modo astratto. Considerando una {\em partizione}\index{partizione}\footnote{Una {\em partizione}, o un {\em pavimento}, $\mathcal{P}=\left(P_{1},...,P_{n}\right)$ di $\Omega$ \` e una collezione di sottoinsiemi di $\Omega$ tali che $\bigcup_{i}P_{i}=\Omega$ e $P_{i}\cap P_{j}=\emptyset$ per $i\neq j$. } $\mathcal{P}=\{P_{0},P_{1},...,P_{n}\}$\index{partizione} di $\Omega$, $n\geq 1$, diremo che la sequenza $\underline{\sigma}(x)=\{\sigma_{i}(x)\}_{i\in\mathbb{Z}}\in\{0,...,n\}^{\mathbb{Z}}$ \`e la {\em storia simbolica} di un punto $x$ se

\begin{equation}
S^{i}x\in P_{\sigma_{i}(x)}\qquad\forall i\in\mathbb{Z}; \label{eq:hist1}
\end{equation}

l'insieme $P_{i}$, $i=1,2,3,...,n$, \` e la collezione di punti di $\Omega$ che ha le caratteristiche indicate dall'etichetta $i$, cosicch\'e l'unione $\cup_{i=1}^{n}P_{i}$ \` e l'insieme dei punti che hanno almeno una delle caratteristiche $i=1,2,...,n$ mentre $P_{0}=\Omega\setminus\cup_{i=1}^{n}P_{i}$ \` e l'insieme dei punti che non ne ha nessuna. Richiederemo sempre che gli elementi della partizione $\mathcal{P}$ (gli {\em atomi} della partizione) siano insiemi boreliani in $\mathcal{B}(\Omega)$; partizioni di questo tipo sono chiamate {\em partizioni boreliane}\index{partizione boreliana}.

Possiamo pensare le informazioni acquisite sull'evoluzione del sistema come una successione $\underline{\sigma}\in\{0,...,n\}^{\mathbb{Z}}$, generata prendendo un grande tempo di osservazione $N$ durante il quale il moto visita in successione i sottoinsiemi $P_{\sigma_{i}}$, $i=-N,...,N$; la sequenza $\underline{\sigma}^{N}\in\{0,...,n\}^{[-N,N]}$ che risulter\` a dalle osservazioni avr\` a la propriet\` a

\begin{equation}
\cap_{k=-N}^{N}S^{-k}P_{\sigma_{k}}\neq\emptyset, \label{eq:hist2}
\end{equation}

dal momento che $S^{k}x\in P_{\sigma_{k}}$ $\forall k$ (e quindi l'intersezione nella (\ref{eq:hist2}) conterr\` a almeno un punto).

Per quanto detto, \` e naturale pensare che un moto, o piuttosto la sua storia, abbia la {\em propriet\` a di intersezione finita}, ovvero che

\begin{equation}
\cap_{j\in J}S^{-j}P_{\sigma_{j}}\neq\emptyset\qquad\forall J\subset\mathbb{Z}, |J|<\infty, \label{eq:hist3}
\end{equation}

dove $|J|$ indica il numero di elementi di $J$. Quindi, indentificheremo lo spazio dei moti osservati in $\mathcal{P}$ con l'insieme delle sequenze infinite

\begin{equation}
\widehat{\Omega}=\{\underline{\sigma}|\underline{\sigma}\in\{0,...,n\}^{\mathbb{Z}},\cap_{j\in J}S^{-j}P_{\sigma_{j}}\neq\emptyset\;\forall J\subset\mathbb{Z}, |J|<\infty\}. \label{eq:hist4}
\end{equation}

Inoltre, diremo che una partizione $\mathcal{P}$ \` e {\em $S$-separante}\index{partizione $S$-separante} se la storia $\underline{\sigma}$ indentifica {\em univocamente} il punto $x$; questa propriet\` a equivale all'{\em espansivit\` a}\index{mappa espansiva} di $S$ in $\mathcal{P}$, ovvero al fatto che, per ogni $\underline{\sigma}\in\widehat{\Omega}$,\footnote{Il {\em diametro} di un insieme $G\subset\Omega$, dove $\Omega$ \` e uno spazio metrico, \` e la massima distanza tra due punti di $G$.}

\begin{equation}
\lim_{N\rightarrow\infty}\left(\mbox{diam}\cap_{j=-N}^{N}S^{-j}P_{\sigma_{j}}\right)=0. \label{eq:hist5}
\end{equation}

Riassumiano questi concetti nella seguente definizione.

\begin{definizione}{{\em (Moti simbolici)}}\label{def:motsimb}

Sia $(\Omega,S)$ un sistema dinamico topologico invertibile e sia $\mathcal{P}=\{P_{0},...,P_{n}\}$ una partizione di $\Omega$ in $n+1$, $n\geq 1$, insiemi boreliani.

\begin{enumerate}
\item Consideriamo l'insieme

\begin{equation}
\widehat{\Omega}=\left\{\underline{\sigma}|\underline{\sigma}\in\{0,...,n\}^{\mathbb{Z}},\cap_{j=-N}^{N}S^{-j}P_{\sigma_{j}}\neq\emptyset\;\forall N\right\}; \label{eq:hist6}
\end{equation}

chiamiamo $\widehat{\Omega}$ l'insieme delle {\em $(\mathcal{P},S)$-storie dei moti simbolici generati da $S$ in $\Omega$ e visti da $\mathcal{P}$}. Quando non ci sar\` a ambiguit\` a chiameremo $\widehat{\Omega}$ semplicemente {\em insieme dei moti simbolici}\index{insieme dei moti simbolici}.

Se $x\in\Omega$ chiameremo la $(\mathcal{P},S)$-storia di $x$ l'elemento $\underline{\sigma}(x)$ di $\widehat{\Omega}$ tale che

\begin{equation}
x\in S^{-j}P_{\sigma_{j}(x)}\qquad\forall j\in\mathbb{Z}. \label{eq;hist7}
\end{equation}
 
 \item Se la relazione $\underline{\sigma}(x)=\underline{\sigma}(x')$ implica che $x=x'$ diremo che $\mathcal{P}$ \` e {\em $S$-separante}. Se $\mathcal{P}$ e $S$ verificano la (\ref{eq:hist5}) diremo che $S$ \` e {\em $\mathcal{P}$-espansiva} o che \` e {\em espansiva in $\mathcal{P}$}.
 
 \item La corrispondenza definita dalla mappa $\Sigma: \Omega\rightarrow\widehat{\Omega}$ che associa ad ogni $x\in\Omega$ la sequenza $\underline{\sigma}(x)\in\widehat{\Omega}$ sar\` a chiamata il {\em codice della dinamica simbolica di $S$ rispetto a $\mathcal{P}$}.
 
 \item Infine, considerando gli insiemi $J=(j_{1},...,j_{q})\subset\mathbb{Z}$  e le sequenze $\underline{\sigma}_{J}\in\{0,...,n\}^{J}$ useremo la notazione
 
 \begin{equation}
 P_{\underline{\sigma}_{J}}^{J}\equiv P_{\sigma_{1}...\sigma_{q}}^{j_{1}...j_{q}}=\cap_{j\in J}S^{-j}P_{\sigma_{j}} \label{eq:hist8}
 \end{equation}
 
 per indicare i {\em punti la cui storia in $J$ \` e specificata da $\underline{\sigma}_{J}$}.
\end{enumerate}
\end{definizione}

\begin{oss}
\begin{enumerate}
\item Se $\Omega=\{0,...,n\}^{\mathbb{Z}}$, $S=\tau$ \` e la {\em traslazione}\footnote{L'operatore di traslazione nello spazio delle sequenze ``sposta'' la sequenza verso sinistra, ovvero $\tau\underline{\sigma}_{i}=\underline{\sigma}_{i+1}$.} in $\{0,...,n\}^{\mathbb{Z}}$ e $P_{i}=\{\underline{\sigma}|\sigma_{0}=i\}$, $i=0,1,...,n$, gli insiemi $P_{\underline{\sigma}}^{J}$ saranno indicati anche con $C_{\underline{\sigma}}^{J}$ e saranno chiamati {\em cilindri} di $\{0,...,n\}^{\mathbb{Z}}$ con ``base $J$ e specificazione $\underline{\sigma}$'', ovvero

\begin{equation}
C_{\underline{\sigma}}^{J}=\{\underline{\sigma}'|\underline{\sigma}'\in\{0,...,n\}^{\mathbb{Z}}, \sigma'_{j_{k}}=\sigma_{k}\quad\forall k=1,...,q\}, \label{eq:hist9}
\end{equation}

se $J=\{j_{1},...,j_{q}\}$ e $\underline{\sigma}=(\sigma_{1},...,\sigma_{q})\in\{0,...,n\}^{J}$.

\item In generale $\widehat{\Omega}\supset\Sigma(\Omega)$; questo vuol dire che non tutte lo storie sono permesse.

\end{enumerate}
\end{oss}

Se $(\Omega,S)$ \` e un sistema dinamico astratto e $\mathcal{P}$ \` e una partizione espansiva possiamo ``invertire'' la mappa codice $\Sigma$, ovvero partendo da una storia $\underline{\sigma}$  possiamo identificare il punto che la genera; dato un punto $\underline{\sigma}\in\widehat{\Omega}$, introduciamo l'insieme $\mathcal{X}(\underline{\sigma})\in\Omega$ come

\begin{equation}
\mathcal{X}(\underline{\sigma})=\cap_{j=-\infty}^{\infty}S^{-j}\overline{P}_{\sigma_{j}}, \label{eq:hist10}
\end{equation}

dove $\overline{P}_{j}$ \` e la chiusura di $P_{j}$. L'insieme $\mathcal{X}(\underline{\sigma})$ \` e non vuoto perch\'e $\Omega$ \` e compatto\footnote{Uno spazio \` e compatto se e solo se ogni sua collezione di sottoinsiemi chiusi $\{C_{\alpha}\}_{\alpha\in J}$ con la {\em propriet\` a di intersezione finita} \` e tale che $\cap_{\alpha\in J}C_{\alpha}\neq\emptyset$ (una collezione di insiemi ha la propriet\` a di intersezione finita se per ogni sottoinsieme finito $J'$ di $J$ si ha che $\cap_{\alpha\in J'}C_{\alpha}\neq\emptyset$); nel nostro caso la propriet\` a di intersezione finita \` e garantita dalla definizione di moto simbolico, e l'insieme $\Omega$ \` e compatto per definizione.}.

\begin{prop}{{\em (Codice della dinamica simbolica)}}

Sia $(\Omega,S)$ un sistema dinamico topologico invertibile; sia $\mathcal{P}=\{P_{0},...,P_{n}\}$ una partizione topologica di $\Omega$ sulla quale $S$ \` e espansiva.

\begin{enumerate}
\item L'insieme $\mathcal{X}(\underline{\sigma})$ definito nella (\ref{eq:hist10}) contiene uno ed un solo punto cosicch\'e possiamo definire la mappa $X:\widehat{\Omega}\rightarrow\Omega$ chiamando $X(\underline{\sigma})$ l'unico punto in $\mathcal{X}(\underline{\sigma})$ oppure come, con un leggero abuso di notazione,

\begin{equation}
\underline{\sigma}\rightarrow X(\underline{\sigma})=\cap_{j=-\infty}^{\infty}S^{-j}\overline{P}_{\sigma_{j}}.
\end{equation}

\item $X$ \` e una mappa continua.
\item $\Sigma^{-1}X^{-1}(x)=x$ per ogni $x\in\Omega$.
\item $X$ e $\Sigma$ sono una l'inversa dell'altra se considerate rispettivamente come mappe tra $\Omega$ e $\Sigma\Omega$, $\Sigma\Omega$ e $\Omega$.
\item $S(x)=X(\tau\Sigma(x))$, dove $(\tau\underline{\sigma})_{i}=\sigma_{i+1}$ \` e la traslazione di una unit\` a temporale a sinistra della sequenza $\underline{\sigma}$.
\end{enumerate}
\end{prop}

\begin{proof}
Vedere \cite{ergo}.
\end{proof}

Intuitivamente, ci aspettiamo che moti ``ragionevoli'' passino una frazione di tempo definita in ogni insieme ``ragionevole'' $E\subset\Omega$; tradotto in termini matematici, ci\`o vuol dire che, se $\chi_{E}$ \` e la funzione caratteristica\footnote{La funzione caratteristica $\chi_{E}$ \` e definita come 

\begin{equation}
\chi_{E}(x)=\left\{\begin{array}{c}1\qquad\mbox{se $x\in E$}\\ 0\qquad\mbox{altrimenti}.\end{array}\right.
\end{equation}

} dell'insieme $E$, il limite

\begin{equation}
\lim_{N\rightarrow\infty}N^{-1}\sum_{j=0}^{N-1}\chi_{E}(S^{j}x)=\nu_{x}(E) \label{eq:hist11}
\end{equation}

deve esistere, e con esso la {\em frequenza di visita dell'insieme $E$ da parte moti che hanno origine in $x$}\index{frequenza di visita}. Se $E=P_{j}$ \` e un elemento della partizione $\mathcal{P}=\{P_{1},...,P_{n}\}$ di $\Omega$ il limite (\ref{eq:11}) equivale a 

\begin{equation}
\nu_{x}(P_{j})=\lim_{N\rightarrow\infty}N^{-1}\{\mbox{numero indici $h$ tra $0$ e $N-1$ tali che $\underline{\sigma}_{h}(x)=j$}\}. \label{eq:hist12}
\end{equation}

Analogamente, definiamo la {\em frequenza} di apparizione {\em nel futuro}\index{frequenza di apparizione nel futuro} di un blocco di $p$ simboli $\sigma_{1},...,\sigma_{p}$ in una successione $\underline{\sigma}(x)$ come

\begin{equation}
N^{-1}\{\mbox{numero di indici $h$ tra $0$ e $N-1$ tali che $\underline{\sigma}_{j_{i}+h}(x)=\sigma_{i}, i=1,...p$}\}, \label{eq:hist12b}
\end{equation}

oppure, equivalentemente,

\begin{equation}
p\left(\left.\begin{array}{c}j_{1}...j_{p}\\\sigma_{1}...\sigma_{p}\end{array}\right|\underline{\sigma}(x)\right)=\lim_{N\rightarrow\infty}N^{-1}\sum_{h=0}^{N-1}\chi_{P_{\sigma_{1}...\sigma_{p}}^{j_{1}...j_{p}}}(S^{h}x); \label{eq:hist13}
\end{equation}

il numero definito tra parentesi graffe nella (\ref{eq:hist12b}) \` e chiamato {\em numero di stringhe omologhe\index{stringa omologa} a $\left(\begin{array}{c}J\\\underline{\sigma}\end{array}\right)\equiv\left(\begin{array}{c}j_{1}...j_{p}\\\sigma_{1}...\sigma_{p}\end{array}\right)$ che appaiono in $\sigma(x)$ tra $0$ e $N-1$}, e il valore del limite (\ref{eq:hist13}) sar\` a la {\em frequenza di apparizione della porzione di storia $\left(\begin{array}{c}J\\\underline{\sigma}\end{array}\right)$ in $\underline{\sigma}(x)$}.

Riassumiamo queste nozioni nella seguente definizione.

\begin{definizione}{{\em (Sequenze con frequenze definite)}}

Siano $\widehat{\underline{\sigma}}\in\{0,...,n\}^{\mathbb{Z}}$, $\{j_{1},...,j_{p}\}\in\mathbb{Z}$ e $\sigma_{1},...,\sigma_{p}\in\{0,...,n\}$.

\begin{enumerate}
\item Definiamo il {\em numero di stringhe omologhe a $\left(\begin{array}{c}j_{1}...j_{p}\\\sigma_{1}...\sigma_{p}\end{array}\right)$ che appaiono in $\widehat{\underline{\sigma}}$ tra $0$ e $N-1$ come il numero di indici $h$ tra $0$ e $N-1$ tali che $\widehat{\underline{\sigma}}_{j_{1}+h}=\sigma_{1},...,\widehat{\underline{\sigma}}_{j_{p}+h}=\sigma_{p}$}; indichiamo questo numero con $\mathcal{N}_{N}\left(\left.\begin{array}{c}j_{1}...j_{p}\\\sigma_{1}...\sigma_{p}\end{array}\right|\widehat{\underline{\sigma}}\right)$.
\item Definiamo la {\em frequenza di apparizione in $\widehat{\underline{\sigma}}$ della stringa omologa a $\left(\begin{array}{c}j_{1}...j_{p}\\\sigma_{1}...\sigma_{p}\end{array}\right)$} come il limite, se esiste,

\begin{equation}
\lim_{N\rightarrow\infty}N^{-1}\mathcal{N}_{N}\left(\left.\begin{array}{c}j_{1}...j_{p}\\\sigma_{1}...\sigma_{p}\end{array}\right|\widehat{\underline{\sigma}}\right)=p\left(\left.\begin{array}{c}j_{1}...j_{p}\\\sigma_{1}...\sigma_{p}\end{array}\right|\widehat{\underline{\sigma}}\right). \label{eq:hist14}
\end{equation}
\item Una sequenza $\widehat{\underline{\sigma}}$ sar\` a detta {\em a frequenze definite} se il limite (\ref{eq:hist14}) esiste per ogni $j_{1},...j_{p}\in\mathbb{Z}$, per ogni $\sigma_{1},...\sigma_{p}\in\{0,...,n\}$, e per ogni $p=1,2,...$.
\end{enumerate}
\end{definizione}

Per concludere la sezione, introduciamo il concetto di {\em ergodicit\` a}\index{ergodicit\` a} per sistemi dinamici astratti i cui moti possono essere rappresentati simbolicamente.

\begin{definizione}{{\em (Sequenze ergodiche e sequenze mescolanti)}}

Una sequenza $\underline{\sigma}\in\{0,...,n\}^{\mathbb{Z}}$ \` e detta {\em ergodica} se ha frequenze definite e se 

\begin{eqnarray}
\lim_{N\rightarrow\infty}N^{-1}\sum_{k=0}^{N-1}p\left(\left.\begin{array}{cc}j_{1}...j_{p}&i_{1}+k...i_{q}+k\\\sigma'_{1}...\sigma'_{p}&\sigma''_{1}...\sigma''_{q}\end{array}\right|\underline{\sigma}\right)=\nonumber\\=p\left(\left.\begin{array}{c}j_{1}...j_{p}\\\sigma'_{1}...\sigma'_{p}\end{array}\right|\underline{\sigma}\right)p\left(\left.\begin{array}{c}i_{1}...i_{q}\\\sigma''_{1}...\sigma''_{q}\end{array}\right|\underline{\sigma}\right) \label{eq:hist14a}
\end{eqnarray}

per ogni $p,q=1,2...$, $\sigma'_{1},...,\sigma'_{p}$, $\sigma''_{1},...,\sigma''_{q}\in\{0,...,n\}$, $\{j_{1},...,j_{p}\}$, $\{i_{1},...,i_{q}\}\subset\mathbb{Z}$. Diremo che $\underline{\sigma}$ \`e {\em mescolante} se

\begin{equation}
\lim_{k\rightarrow\infty}p\left(\left.\begin{array}{cc}j_{1}...j_{p}&i_{1}+k...i_{q}+k\\\sigma'_{1}...\sigma'_{p}&\sigma''_{1}...\sigma''_{q}\end{array}\right|\underline{\sigma}\right)=p\left(\left.\begin{array}{c}j_{1}...j_{p}\\\sigma'_{1}...\sigma'_{p}\end{array}\right|\underline{\sigma}\right)p\left(\left.\begin{array}{c}i_{1}...i_{q}\\\sigma''_{1}...\sigma''_{q}\end{array}\right|\underline{\sigma}\right) \label{eq:hist15}
\end{equation}

per tutte le possibili scelte di indici.
\end{definizione}

\begin{oss}
Ovviamente ogni seguenza mescolante \` e anche ergodica.
\end{oss}

Con la prossima proposizione capiremo l'importanza del concetto di ergodicit\` a; vedremo che l'esistenza delle frequenze di visita delle $(\mathcal{P},S)$-storie dei punti di un sistema dinamico \` e un fatto molto generale.

\begin{prop}{{\em (Teorema di Birkhoff)}}\index{teorema di Birkhoff}\label{th:birk}

Sia $(\Omega,S,\mu)$ un sistema dinamico metrico discreto invertibile e sia $\mathcal{B}$ la $\sigma$-algebra sulla quale $\mu$ \` e definita. Se $f$ \` e una funzione $\mu$-misurabile il limite

\begin{equation}
\lim_{N\rightarrow\infty}N^{-1}\sum_{j=0}^{N-1}f(S^{j}x)=\overline{f}(x) \label{eq:birk}
\end{equation}

esiste $\mu$-quasi ovunque, e se $f\in L_{1}(\mu)$ il limite (\ref{eq:birk}) \` e raggiunto in $L_{1}(\mu)$.

 \end{prop}

\begin{cor}\label{cor:birk}

Sotto le ipotesi della proposizione (\ref{th:birk}), il limite (\ref{eq:birk}) pu\` o essere riscritto come

\begin{equation}
\lim_{N\rightarrow\infty}N^{-1}\sum_{j=0}^{N-1}f(S^{j}x)=\int\mu(dy)f(y)\label{eq:condergo}\qquad\mbox{per ogni $f\in L_{1}(\mu)$.}
\end{equation}
\end{cor}

\begin{proof}
Per una dimostrazione di questo teorema e del suo corollario vedere \cite{kh}, \cite{ergo}, \cite{AA}.
\end{proof}

 \begin{oss}
 \item Se $\mathcal{P}=\{P_{1},...,P_{n}\}$ \` e una partizione di $\Omega$ in insiemi misurabili il teorema di Birkhoff implica che, scegliendo come funzione $f$ la funzione indicatrice $\chi_{\cap_{k=1}^{p}S^{-j_{k}}P_{\sigma_{k}}}\equiv \chi_{P_{\sigma_{1}...\sigma_{p}}^{j_{1}...j_{p}}}$,
 
 \begin{equation}
 \lim_{N\rightarrow\infty}N^{-1}\sum_{k=0}^{N-1}\chi_{P_{\sigma_{1}...\sigma_{p}}^{j_{1}...j_{p}}}(S^{k}x)=p\left(\begin{array}{cc}j_{1}...j_{p}\\\sigma_{1}...\sigma_{p}\end{array}\right),
 \end{equation}
 
 a meno di un insieme di misura nulla; ci\` o vuol dire che le $(\mathcal{P},S)$-storie di $\mu$-quasi tutti i punti di $\Omega$ hanno frequenze definite.  Inoltre, grazie al corollario \ref{cor:birk}
 
 \begin{equation}
 \lim_{N\rightarrow\infty}N^{-1}\sum_{k=0}^{N-1}\chi_{P_{\sigma_{1}...\sigma_{p}}^{j_{1}...j_{p}}}(S^{k}x)=\mu\left(\cap_{k=1}^{p}S^{-j_{k}}P_{\sigma_{k}}\right), \label{eq:osserg2}
\end{equation}

ovvero per $\mu$-quasi tutti i dati iniziali la misura di un insieme \` e asintoticamente uguale alla frazione di tempo che i moti trascorrono visitando l'insieme.
\end{oss}

Questo risultati ci permettono di introdurre due definizioni pi\` u naturali di sistema ergodico e sistema mescolante.

\begin{definizione}
Se $(\Omega,S,\mu)$ \` e un sistema dinamico metrico invertibile diremo che \` e {\em ergodico} se per ogni $E,F\in\mathcal{B}$ si ha che

\begin{equation}
\lim_{N\rightarrow\infty}N^{-1}\sum_{j=0}^{N-1}\mu(E\cap S^{-j}F)=\mu(E)\mu(F);\label{eq:defergo}
\end{equation}

invece, diremo che il sistema \` e {\em mescolante} se

\begin{equation}
\lim_{j\rightarrow\infty}\mu(E\cap S^{-j}F)=\mu(E)\mu(F). \label{eq:defmesc}
\end{equation}
\end{definizione}

\begin{oss}
Si pu\` o dimostrare che le condizioni (\ref{eq:defergo}), (\ref{eq:osserg2}) sono {\em equivalenti}, {\em cf.} \cite{ergo}, \cite{AA}, {\em i.e.} che si implicano a vicenda.
\end{oss}

\section{Pavimentazioni markoviane}

Come abbiamo visto nella sezione precedente, l'evoluzione di un sistema dinamico in uno spazio $\Omega$ pu\` o essere schematizzata attraverso l'azione di una traslazione $\tau$ su delle sequenze di simboli $\underline{\sigma}$. Chiaramente, con ci\` o non abbiamo reso pi\` u semplice il problema dello studio dell'evoluzione di un sistema dinamico; tutta la complessit\` a \` e racchiusa ora nel {\em codice}\index{codice} $X(\underline{\sigma})$, che permette di associare ad ogni punto $x\in\Omega$ una $(\mathcal{P},S)$-storia in $\mathcal{P}$. Come vedremo, per una classe di sistemi ``molto caotici'' (che definiremo tra poco) risulter\` a relativamente (e inaspettatamente) facile descrivere le traiettorie mediante dinamica simbolica; in particolare, saremo in grado di costruire esplicitamente la {\em misura di probabilit\` a invariante del sistema}.

Prima di capire come sia possibile construire in modo concreto una rappresentazione simbolica, introduciamo la nozione di {\em compatibilit\` a} tra simboli $\sigma_{i}\in\underline{\sigma}$.

\begin{definizione}{\em{(Matrice di compatibilit\` a)}}\label{def:matrcomp}

Una matrice T $(n+1)\times(n+1)$ con elementi $T_{\sigma\sigma'}$ uguali a $0$ o a $1$ sar\` a chiamata una {\em matrice di compatibilit\` a}. Questa matrice sar\` a detta {\em transitiva} se per ogni coppia $\sigma$, $\sigma'$ esiste un intero $a_{\sigma\sigma'}$ tale che $T_{\sigma\sigma'}^{1+a_{\sigma\sigma'}}>0$ per ogni $\sigma$, $\sigma'$, mentre sar\` a detta {\em mescolante} se esiste $a\geq 0$ tale che $T^{1+a}_{\sigma\sigma'}>0$ per ogni $\sigma$, $\sigma'$; chiameremo $a$ il {\em tempo di mescolamento} di $T$.

Una sequenza $\underline{\sigma}\in\{0,...,n\}^{\mathbb{Z}}$ sar\` a detta {\em $T$-compatibile} o {\em ammissibile} se e solo se ogni coppia di simboli adiacenti che appare in $\underline{\sigma}$ \` e ammissibile, ovvero $T_{\sigma\sigma'}=1$ $\forall i\in\mathbb{Z}$, o equivalentemente se e solo se

\begin{equation}
\prod_{i=-\infty}^{\infty}T_{\sigma_{i}\sigma_{i+1}}=1. \label{eq:comp1} 
\end{equation}
  
  Chiameremo $\{0,...,n\}^{\mathbb{Z}}_{T}\subset\{0,...,n\}^{\mathbb{Z}}$ l'insieme (chiuso e invariante per traslazioni) delle sequenze $\underline{\sigma}$ che verificano la (\ref{eq:comp1}).
  
\end{definizione}

La nozione di compatibilit\` a tiene conto essenzialmente del fatto che non tutte le traiettorie immaginabili sono permesse dall'evoluzione del sistema; l'insieme delle sequenze compatibili equivale all'insieme di tutte le traiettorie in $\Omega$ permesse dalla mappa $S$. La codifica di una traiettoria in una sequenza di simboli \` e resa possibile da una {\em pavimentazione markoviana}\index{pavimentazione markoviana}.

\begin{definizione}{{\em (Pavimentazioni markoviane)}}\label{def:pavmark}

Sia $(\Omega,S)$ un sistema dinamico topologico invertibile. Dato un {\em pavimento} $\mathcal{Q}=\{Q_{1},...,Q_{q}\}$ di $\Omega$, {\em i.e.} una collezioni di insiemi che ricopre $\Omega$ tale che $Q_{i}\cap Q_{j}=\emptyset$ se $i\neq j$, poniamo $T_{\sigma\sigma'}=1$ se\footnote{L'insieme $\mbox{int}(E)\subset E$ corrisponde, indicando con $\partial E$ il {\em bordo} di E, a $E\setminus\partial E$.} $\mbox{{\em int}}(Q_{\sigma})\cap\mbox{{\em int}}(S^{-1}Q_{\sigma'})\neq\emptyset$ e $T_{\sigma\sigma'}=0$ altrimenti, con $\sigma,\sigma'\in\{1,...,q\},q\geq2$. Diremo che $\mathcal{Q}$ \` e {\em markoviana} se le seguenti condizioni sono verificate.

\begin{enumerate}
\item L'insieme

\begin{equation}
\chi(\underline{\sigma})=\bigcap_{k=-\infty}^{+\infty}S^{-k}Q_{\sigma_{k}} \label{eq:pavmark0}
\end{equation}

\` e non vuoto e consiste di un singolo punto $X(\underline{\sigma})$ tale che

\begin{equation}
\prod_{i=-\infty}^{+\infty}T_{\sigma_{i}\sigma_{i+1}}=1. \label{eq:pavmark}
\end{equation}

\item La corrispondenza $\underline{\sigma}\rightarrow X(\underline{\sigma})$ tra $\{1,...,q\}^{\mathbb{Z}}_{T}$ e $\Omega$ \` e H\"older continua, ovvero esistono delle costanti $C,a>0$ tali che

\begin{equation}
d(X(\underline{\sigma}),X(\underline{\sigma}'))\leq Cd(\underline{\sigma};\underline{\sigma}')^{a},\label{eq:pavmark2}
\end{equation}

dove $d$ \` e la {\em distanza}\index{distanza tra sequenze di simboli} tra $\underline{\sigma}$, $\underline{\sigma}'$ in $\{1,...,q\}^{\mathbb{Z}}$ definita come

\begin{equation}
d(\underline{\sigma},\underline{\sigma}')=\exp\left(-\nu(\underline{\sigma},\underline{\sigma}')\right), \label{eq:pavmark3}
\end{equation}

con $\nu(\underline{\sigma},\underline{\sigma}')$ pari al pi\` u grande intero $j$ tale che $\sigma_{i}=\sigma'_{i}$ per $|i|\leq j$.

\item Il numero di sequenze compatibili mappate in un certo punto $x$ pu\` o essere maggiorato con una costante $M<\infty$, ovvero la mappa inversa $X^{-1}$ verifica la propriet\` a $|X^{-1}(x)|\leq M$, per ogni $x\in\Omega$. Il numero $M$ sar\` a chiamato {\em molteplicit\` a del codice}.

\item Ponendo $\partial_{i}=\partial Q_{i}$, $i=1,...,q$ e $\partial=\bigcup_{i=1}^{q}\partial_{i}$ esistono due insiemi chiusi tali che

\begin{equation}
\partial = \partial^{+}\cup\partial^{-},\qquad S\partial^{-}\subset\partial^{-},\qquad S^{-1}\partial^{+}\subset\partial^{+},\label{eq:pavmark4}
\end{equation}

ovvero l'insieme dei bordi degli elementi della partizione pu\` o essere diviso in due sottoinsiemi, che, rispettivamente, vengono contratti ed espansi dall'azione di $S$.

\end{enumerate}

\end{definizione}

\begin{oss}
\begin{enumerate}
\item La continuit\` a di $X$ implica che l'immagine inversa di un insieme boreliano tramite $X^{-1}$ sia ancora un insieme boreliano (perch\'e ogni funzione continua \` e Borel - misurabile, {\em cf.} \cite{kolmo}).
\item $X$ permette di codificare la mappa $S$ in dinamica simbolica, dal momento che $X(\tau\underline{\sigma})=SX(\underline{\sigma})$.
\item La definizione di pavimentazione di Markov che abbiamo appena dato in realt\` a descrive una partizione $S$-separante, {\em cf.} definizione \ref{def:motsimb}; una partizione non $S$-separante pu\` o essere ottenuta eliminando la condizione che l'insieme $\chi(\underline{\sigma})$, {\em cf.} formula (\ref{eq:pavmark0}), consista di un solo punto.
\item \`E chiaro che ogni punto appartenente all'insieme $\Omega\setminus \cup_{i}S^{-i}\delta$ sar\` a l'immagine di una sequenza\footnote{Basta compilare una lista con tutti gli elementi del pavimento visitati dal punto $x$ evoluto con $S^{k}$, con $-\infty\leq k\leq\infty$.}, mentre \` e meno ovvio quale sia il modo di rappresentare i punti che appartengono all'insieme dei bordi di $\mathcal{Q}$. Se $x\in\cup_{i}S^{-i}\partial$ allora esiste $\sigma$ tale che $x\in Q_{\sigma}$ (dal momento che $Q_{\sigma}=\partial Q_{\sigma}\cup \mbox{{\em int}}(Q_{\sigma})=\overline{\mbox{{\em int}}(Q_{\sigma})}$\footnote{L'insieme $\overline{E}$ rappresenta la {\em chiusura}\index{chiusura di un insieme} di $E$, ovvero $E$ completato con i limiti delle successioni in esso contenute.}); poniamo $\sigma=\sigma_{0}$, e deve esistere $\sigma'$ tale che $\mbox{{\em int}}(Q_{\sigma_{0}})\cap\mbox{{\em int}}(S^{-1}Q_{\sigma'})\neq\emptyset$ e $Sx\in Q_{\sigma'}$. Quindi poniamo $\sigma_{1}=\sigma'$, e per costruzione $T_{\sigma_{0}\sigma_{1}}=1$; continuando questo ragionamento si costruisce la sequenza $\underline{\sigma}\in\{1,...,q\}^{\mathbb{Z}}_{T}$ tale che $x\in S^{-k}Q_{\sigma_{k}}$ per ogni $k\in\mathbb{Z}$. Quindi, $x=X(\underline{\sigma})$.
\item\label{oss:isomorf} La mappa codice $X$ pu\` o essere usata per definire una misura di probabilit\` a nello spazio delle sequenze. Infatti, se $m$ \` e una misura di probabilit\` a nello spazio delle sequenze compatibili $\{1,...,q\}^{\mathbb{Z}}_{T}$ la relazione

\begin{equation}
Am(E)=m(X^{-1}E)\qquad\mbox{per ogni $E\in\mathcal{B}(\Omega)$}\label{eq:museq}
\end{equation}

definisce una misura di probabilit\` a $Am$ su $\Omega$. La mappa $A$ trasforma misure invarianti per $S$ in misure invarianti per $\tau$, misure ergodiche rispetto a $S$ in misure ergodiche rispetto a $\tau$ ecc. (per una dimostrazione rimandiamo a \cite{ergo}.)

\end{enumerate}
\end{oss}

Quest'ultima osservazione \` e molto importante; come vedremo, la codifica in simboli della misura di volume di un sistema ``molto caotico'' assumer\` a la forma della misura di probabilit\` a che descrive un modello di spin unidimensionale, interagente attraverso potenziali a molti corpi che decadono esponenzialmente.

\subsection{Pavimentazioni markoviane per sistemi di Anosov}\label{sez:pavan}

Dopo aver introdotto la nozione di pavimentazione di Markov e di codice della dinamica simbolica, cerchiamo adesso di capire come {\em costruire} una pavimentazione di Markov; introduciamo una classe di sistemi dinamici, i {\em sistemi di Anosov}\index{sistema di Anosov} o {\em sistemi iperbolici lisci}\index{sistema iperbolico liscio}, per i quali ci\` o pu\` o essere fatto in modo relativamente facile.

Un semplice esempio di sistema di Anosov \` e il {\em gatto di Arnold}\index{gatto di Arnold}, {\em cf.} \cite{AA}; lo spazio nel quale questo sistema \` e definito \` e il toro $\mathbb{T}^{2}$, ovvero ci\` o che si ottiene identificando i lati opposti del quadrato $[0,2\pi]\times[0,2\pi]$, e l'evoluzione discreta \` e descritta da

\begin{equation}
S\left(\begin{array}{c}\varphi_{1}\\\varphi_{2}\end{array}\right)=\left(\begin{array}{cc}1&1 \\ 1&2\end{array}\right)\left(\begin{array}{c}\varphi_{1}\\\varphi_{2}\end{array}\right)\quad\mbox{mod $2\pi$}. \label{eq:arndcat}
\end{equation}

\` E facile verificare che la matrice $S$ ha due autovalori distinti $\lambda_{+}$, $\lambda_{-}$, $\lambda_{+}=\lambda_{-}^{-1}$, tali che

\begin{eqnarray}
\lambda_{+}&=&\frac{3+\sqrt{5}}{2}\\
\lambda_{-}&=&\frac{3-\sqrt{5}}{2};
\end{eqnarray}

gli autovettori $\underline{v}_{+}$, $\underline{v}_{-}$ corrispondenti a $\lambda_{+}$, $\lambda_{-}$ sono dati da

\begin{eqnarray}
\underline{v}_{+}&=&\left(\begin{array}{c} 1\\\lambda-1\end{array}\right)\\
\underline{v}_{-}&=&\left(\begin{array}{c} 1\\\lambda^{-1}-1\end{array}\right)\qquad\mbox{con $\lambda=\lambda_{-}$}.
\end{eqnarray}

Dal momento che $v_{+}$, $v_{-}$ hanno pendenza che non \` e multiplo razionale di $\pi$, le due rette $W^{u}(\underline{\varphi})$, $W^{s}(\underline{\varphi})$ riempiono densamente il toro $\mathbb{T}^{2}$, {\em cf.} \cite{AA}, e inoltre $W^{u}(\underline{\varphi})$, $W^{s}(\underline{\varphi})$ sono (banalmente) {\em covarianti}, ovvero $SW^{a}(\underline{\varphi})=W^{a}(S\underline{\varphi})$ ($a=u,s$); poich\'e $\lambda_{+}>1$ e $\lambda_{-}<1$, due punti sulla stessa {\em variet\` a stabile}\index{variet\` a stabile} $W^{s}$, o sulla stessa {\em variet\` a instabile}\index{variet\` a instabile} $W^{u}$, si avvicineranno , o allontaneranno, esponenzialmente. Tradotto in termini matematici, per ogni $n>0$ si ha che

\begin{eqnarray}
d(S^{n}\underline{\psi},S^{n}\underline{\psi}')&\leq&\lambda^{n}d(\underline{\psi},\underline{\psi}')\qquad\mbox{per ogni $\underline{\psi},\underline{\psi}'\in W^{s}(\underline{\varphi})$}\\
d(S^{-n}\underline{\psi},S^{-n}\underline{\psi}')&\leq&\lambda^{n}d(\underline{\psi},\underline{\psi}')\qquad\mbox{per ogni $\underline{\psi},\underline{\psi}'\in W^{u}(\underline{\varphi})$},
\end{eqnarray}

se $d$ \` e la distanza misurata sulla variet\` a stabili e instabili. Quindi, poich\' e le variet\` a stabili e instabili sono dense in $\mathbb{T}^{2}$, ogni punto del toro sotto l'azione di $S$  \` e fortemente instabile rispetto a qualunque altro; questa situazione giustifica l'aggettivo {\em iperbolico}.

Ispirati da questo esempio, possiamo dare una definizione generale di {\em sistema di Anosov} (o {\em sistema iperbolico liscio}).

\begin{definizione}{{\em (Sistema iperbolico liscio o sistema di Anosov)}} \label{def:Anosov}

Sia $(\Omega,S)$ un sistema dinamico definito su una variet\` a riemanniana\footnote{Una variet\` a riemanniana \` e una variet\` a differenziabile dove in ogni spazio tangente \` e definito un prodotto interno, in un modo che varia analiticamente da punto a punto; ci\`o ci permette di definire gli angoli, la lunghezza delle curve ecc. In altre parole, una variet\` a riemanniana\index{variet\` a riemanniana} \`  e una variet\` a differenziabile in cui lo spazio tangente in ogni punto \` e uno spazio di Hilbert\index{spazio di Hilbert}.}\index{superficie riemanniana} $\Omega$ compatta e connessa di classe $C^{\infty}$ con $S$ un diffeomorfismo di classe $C^{\infty}$. Supponiamo che il sistema sia {\em topologicamente transitivo}\index{sistema topologicamente transitivo}, ovvero che esista un'orbita densa. Supponiamo inoltre che esista una metrica riemanniana regolare $d$ (la quale pu\` o essere diversa da quella definita su $\Omega$) tale che misurando le lunghezze con $d$ le seguenti propriet\` a sono verificate.
\begin{enumerate}
\item Esistono due superfici $W^{u}(x)$ e $W^{s}(x)$, che supporremo di classe $C^{k}$ con $k>2$, e con piano tangente in $x$ dipendente da $x$ in modo H\"older continuo. Inoltre, $W^{u}(x)$ e $W^{s}(x)$ sono trasverse\footnote{Ci\` o vuol dire che si intersecano con un angolo di incidenza non nullo.} in $x$ e hanno dimensioni positive complementari\footnote{Ovvero $\mbox{dim}(\Omega)=\mbox{dim}(W^{u})+\mbox{dim}(W^{s})$, $\mbox{dim}(W^{u,s})>0$.}.
\item Chiamando $\Sigma_{\gamma}(x)$ la sfera di raggio $\gamma$ centrata in $x$ e ponendo 

\begin{eqnarray}
W^{u}_{\gamma}(x)&=&\{\mbox{parte connessa di $W^{u}(x)\cap\Sigma_{\gamma}(x)$ contenente $x$}\}\nonumber\\
W^{s}_{\gamma}(x)&=&\{\mbox{parte connessa di $W^{s}(x)\cap\Sigma_{\gamma}(x)$ contenente $x$}\},\nonumber
\end{eqnarray}

 esiste $\gamma>0$ tale che

\begin{equation}
SW^{s}_{\gamma}(x)\subset W^{s}_{\gamma}(Sx),\qquad S^{-1}W^{u}_{\gamma}(x)\subset W^{u}_{\gamma}(x)\subset W^{u}_{\gamma}(S^{-1}x). \label{eq:defanos}
\end{equation}
 
 \item Esiste $\lambda<1$ tale che, per ogni $n\geq 0$,
 
 \begin{eqnarray}
 d(S^{n}y,S^{n}z)&\leq&\lambda^{n}d(y,z)\qquad\mbox{per ogni $y,z\in W^{s}_{\gamma}(x)$}\\
d(S^{-n}y,S^{-n}z)&\leq&\lambda^{n}d(y,z)\qquad\mbox{per ogni $y,z\in W^{u}_{\gamma}(x)$},
 \end{eqnarray}
 
 \item\label{propr:int} Esiste $\varepsilon>0$,  $\varepsilon<\gamma$ tale che, se $x,y\in\Omega$ e $d(x,y)<\varepsilon$, l'insieme $W^{u}_{\gamma}(x)\cap W^{s}_{\gamma}(y)$ consiste in un singolo punto $[x,y]$, che dipende in modo continuo da $x$ e $y$.
\end{enumerate}

Sotto queste ipotesi diremo che $(\Omega,S)$ \` e un {\em sistema iperbolico liscio} o {\em sistema di Anosov} e che la metrica riemanniana $d$ \` e {\em adattata}\index{metrica adattata} alla mappa $S$.
\end{definizione}

\begin{oss}
\begin{enumerate}
\item Grazie alla propriet\` a (\ref{propr:int}) la variet\` a stabile $W^{s}$ e la variet\` a instabile $W^{u}$ definiscono un {\em sistema di coordinate locale}\index{sistema di coordinate locale}.
\item La definizione di sistema di Anosov che abbiamo appena dato vale per i sistemi dinamici discreti; ad ogni modo, \` e possibile formulare una definizione analoga per i flussi, e la differenza principale \` e che oltre alla direzione stabile e alla direzione instabile esiste una direzione {\em marginale}, lungo la quale due punti distinti non si avvicinano n\'e si allontanano esponenzialmente. Questa direzione coincide con la traiettoria del punto nello spazio delle fasi.
\end{enumerate}
\end{oss}

Prima di passare ad una {\em dimostrazione costruttiva} dell'esistenza di una pavimentazione markoviana per i sistemi di Anosov, introduciano il concetto di {\em $S$-rettangolo}\index{$S$-rettangolo}.

\begin{definizione}{\em ($S$-rettangoli)}

Un insieme $R\subset\Omega$ \` e un {\em $S$-rettangolo} per il sistema iperbolico $(\Omega,S)$ se

\begin{enumerate}
\item $R=\overline{\mbox{{\em int}}(R)}$,
\item il diametro di $R$ \` e abbastanza piccolo da permettere che se $x,y\in R$ esiste $\gamma$ tale che l'intersezione di $W^{s}_{\gamma}(x)$ e $W^{u}_{\gamma}(y)$ definisce un unico punto $[x,y]\in R$.
\end{enumerate}
\end{definizione}

\begin{oss}
Ad esempio, nel caso di un sistema di Anosov bidimensionale $\left(\mathbb{T}^{2},S\right)$ possiamo ``prolungare'' in modo opportuno due ``lacci'' diretti lungo variet\` a stabile e instabile in modo da ottenere una griglia dove ogni maglia \` e un $S$-rettangolo.
\end{oss}

Siamo pronti per dimostrare il seguente fondamentale risultato.

\begin{prop}{{\em (Esistenza delle pavimentazioni markoviane)}}\label{prop:pavmark}

Se $(\Omega,S)$ \` e un sistema di Anosov e $\delta>0$, allora esiste una pavimentazione markoviana $\mathcal{Q}=(Q_{1},...,Q_{q})$ di $\Omega$ dove $Q_{i}$, $i=1,...,q$, sono $S$-rettangoli con diametro minore di $\delta$.
\end{prop}

\begin{oss}
Per comodit\` a, dimostreremo questa proposizione nel caso $d=2$. L'argomento che presenteremo non \` e il pi\` u facile per questo caso, ma ha il pregio di essere facilmente estendibile a qualsiasi dimensione.
\end{oss}

\begin{proof}
La dimostrazione della proposizione (\ref{prop:pavmark}) risulta essere molto pi\` u semplice di quanto non possa sembrare a prima vista se si rappresenta graficamente ci\` o che si sta descrivendo a parole; per questo motivo, cercheremo sempre di illustrare i punti chiave con dei disegni.

\subsubsection{Descrizione di un rettangolo generico}\label{sez:(A)}

Come abbiamo gi\` a accennato in precedenza, variet\` a stabile e variet\` a instabile definiscono un sistema di coordinate locale. Quindi, se $A$ \` e un $S$-rettangolo e $x\in\mbox{int}(A)$, ponendo $C=W^{u}_{\gamma}(x)\cap A$,  $D=W^{s}_{\gamma}(x)\cap A$ abbiamo che (vedere figura \ref{fig:pav1})

\begin{equation}
A=[C,D]=\cup_{y\in C,z\in D}[y,z]; \label{eq:dimpav}
\end{equation}
   
chiameremo $x$ il {\em centro} di $A$ rispetto alla coppia di assi $C, D$\footnote{Con questa definizione ogni punto interni di $x$ \` e un centro di $A$ rispetto ad opportuni assi.}, e diremo che $C,D$ sono rispettivamente l'{\em asse stabile} e l'{\em asse instabile} di $A$. Se $C,D$ e $C',D'$ sono due coppie di assi per lo stesso rettangolo diremo che sono {\em paralleli}, dal momento che o $C\equiv C'$ oppure $C\cap C'=\emptyset$. Identificheremo con $\partial^{u}_{1}A$, $\partial^{u}_{2}A$ i bordi del rettangolo giacenti sulle variet\` a instabili, e con $\partial^{s}_{1}A$, $\partial^{s}_{2}A$ quelli appartenenti alle variet\` a stabili.

\begin{figure}[htbp]
\centering
\includegraphics[width=0.7\textwidth]{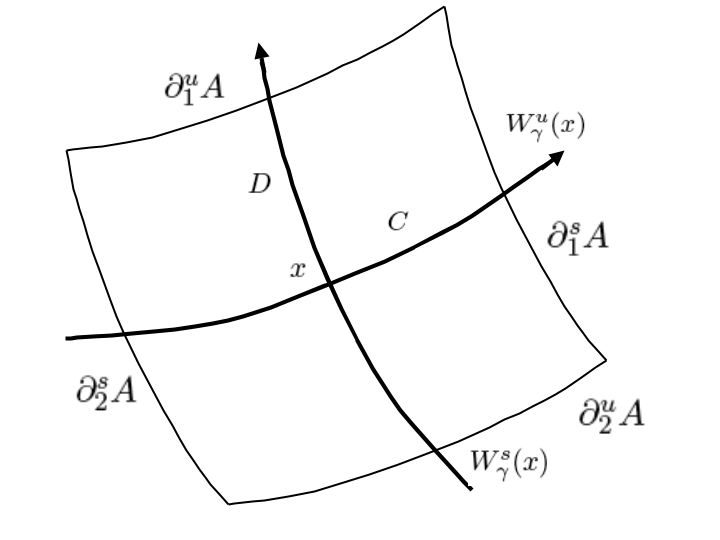}
\caption{Un $S$-rettangolo con assi $C$, $D$ e centro $x$.} \label{fig:pav1}
\end{figure}

Vogliamo dimostrare l'esistenza di un pavimento $\mathcal{Q}$ di $\Omega$ formato da $S$-rettangoli con diametro $<\delta$ arbitrario, tale che per ogni $Q\in\mathcal{Q}$ si abbia

\begin{equation}
S\partial_{\beta}^{s}Q\subset\partial_{\beta'}^{s}Q',\qquad S^{-1}\partial_{\beta}^{u}Q\subset\partial^{u}_{\beta''}Q'', \label{eq:dimpav2}
\end{equation}

con $\beta,\beta',\beta''=1,2$.

Notare che, in realt\` a, la propriet\` a che definisce un pavimento markoviano \` e (vedere definizione \ref{def:pavmark})

\begin{equation}
S\left(\cup_{Q\in\mathcal{Q}}\partial^{s}Q\right)\subset\cup_{Q\in\mathcal{Q}}\partial^{s}Q',\qquad S^{-1}\left(\cup_{Q\in\mathcal{Q}}\right)\partial^{u}Q\subset\cup_{Q\in\mathcal{Q}}\partial^{u}Q'', \label{eq:dimpav3}
\end{equation}

ma \` e facile rendersi conto che la (\ref{eq:dimpav2}), (\ref{eq:dimpav3}) sono equivalenti. Infatti, \` e chiaro che ponendo $\partial^{s}Q=\partial^{s}_{1}Q\cup\partial^{s}_{2}Q$ e $\partial^{u}Q=\partial^{u}_{1}Q\cup\partial^{u}_{2}Q$ la (\ref{eq:dimpav2}) implica la (\ref{eq:dimpav3}); l'implicazione inversa pu\` o essere vista con un semplice ragionamento per assurdo.

Applichiamo la mappa $S$ ad un elemento $Q\in\mathcal{Q}$; per la propriet\` a (\ref{eq:dimpav3}), l'insieme dei bordi stabili $\partial^{s}$ sotto l'azione di $S$ diventa strettamente contenuto in s\'e stesso, e la stessa cosa accade per $\partial^{u}$ sotto l'azione di $S^{-1}$. Assumiamo che la (\ref{eq:dimpav2}) non sia vera; dunque in generale avremo che $S\partial^{s}_{\beta}Q\subset \cup_{Q'\in\mathcal{C}}\partial^{s}Q'$ dove $\mathcal{C}\subset\mathcal{Q}$ \` e una collezioni di insiemi contenuta in $\mathcal{Q}$ con $|\mathcal{C}|>1$. Senza perdita di generalit\` a, fissato $\partial^{s}_{\beta}Q$ consideriamo il caso in cui $\mathcal{C}$ sia composto da due elementi $Q',Q''$ con un bordo instabile in comune, che chiameremo $\partial^{u}_{*}Q'$, e pensiamo che $S\partial^{s}_{\beta}Q\subset \partial^{s}_{\beta}Q'\cup\partial^{s}_{\beta}Q'' $. Ci\` o implica che $\partial^{u}_{*}Q'$ \` e compreso tra $S\partial^{s}_{1}Q$ e $S\partial^{s}_{2}Q$ generando un assurdo, perch\' e $S^{-1}$ manda punti relativamente interni\footnote{Per {\em punti relativamente interni} a $\partial^{u}Q$, o a $\partial^{s}Q$, intendiamo i punti compresi tra $\partial^{u}_{1}Q$ e $\partial^{u}_{2}Q$, o tra $\partial^{s}_{1}Q$ e $\partial^{s}_{2}Q$; il punto $x$ \` e {\em compreso} in $\partial^{u}Q=\partial^{u}_{1}Q\cup\partial^{u}_{2}Q$ se \` e possibile raggiungere $\partial^{u}Q$ percorrendo la variet\` a stabile che passa per $x$, ed \` e compreso in $\partial^{s}Q=\partial^{s}_{1}Q\cup\partial^{s}_{2}Q$ se \` e possibile raggiungere $\partial^{s}Q$ percorrendo la variet\` a instabile passante per $x$.} a $S\partial^{s}Q$ in punti relativamente interni a $\partial^{s}Q$ e per ipotesi $S^{-1}\partial^{u}\subset\partial^{u}$.

\subsubsection{Un problema equivalente}\label{sez:(B)}

In realt\` a, dimostreremo la propriet\` a (\ref{eq:dimpav2}) per la mappa $g=S^{m}$ con $m$ ``grande'' (in un senso che specificheremo in seguito), ovvero troveremo un pavimento $\mathcal{B}$ di $\Omega$ tale che

\begin{equation}
g\cup_{B\in\mathcal{B}}\partial^{s} B\subset\cup_{B\in\mathcal{B}}\partial^{s}B,\qquad g^{-1}\cup_{B\in\mathcal{B}}\partial^{u}B\subset\cup_{B\in\mathcal{B}}\partial^{u}B; \label{eq:dimpav4}
\end{equation}

ma questo problema \` e equivalente a quello di partenza, dal momento che da $\mathcal{B}$ si ottiene un pavimento markoviano per $S$ semplicemente ponendo\footnote{Infatti, con questa scelta di $\mathcal{Q}$ esiste $B'$ tale che

\begin{eqnarray}
S\left(\cup_{Q\in\mathcal{Q}}\partial^{s}Q\right)&=&S\left(\cup_{B_{i_{0}}\in\mathcal{B}}...\cup_{B_{i_{m-1}}\in\mathcal{B}}\partial^{s}\cap_{k=0}^{m-1}S^{k}B_{i_{k}}\right)\nonumber\\&=&S\left(\cup_{B_{i_{0}}\in\mathcal{B}}...\cup_{B_{i_{m-1}}\in\mathcal{B}}\partial^{s}\left(\cap_{k=0}^{m-2}S^{k}B_{i_{k}}\cap S^{m-1}B_{i_{m-1}}\right)\right)\nonumber\\&=&\cup_{B_{i_{0}}\in\mathcal{B}}...\cup_{B_{i_{m-1}}\in\mathcal{B}}\partial^{s}\left(\cap_{k=1}^{m-1}S^{k}B_{i_{k-1}}\cap S^{m}B_{i_{m-1}}\right)\nonumber\\&\subset&\cup_{B_{i_{0}}\in\mathcal{B}}...\cup_{B'\in\mathcal{B}}\partial^{s}\left(\cap_{k=1}^{m-1}S^{k}B_{i_{k-1}}\cap B'\right)\nonumber\\&=&\cup_{B_{i_{0}}\in\mathcal{B}}...\cup_{B_{i_{m-1}}\in\mathcal{B}}\partial^{s}\cap_{k=0}^{m-1}S^{k}B_{i_{k}}=\cup_{Q\in\mathcal{Q}}\partial^{s}Q; \label{eq:dimpav5}
\end{eqnarray}

un ragionamento dello stesso tipo pu\`o essere fatto per la seconda propriet\` a nella (\ref{eq:dimpav3}).
}

\begin{equation}
Q=\cap_{k=0}^{m-1}S^{k}B_{i_{k}}, \label{eq:dimpav5a}
\end{equation}

con $B_{i_{1}},...,B_{i_{k}}\in\mathcal{B}$, per tutte le possibili combinazioni di indici.

\subsubsection{Costruzione di un ricoprimento di rettangoli}\label{sez:(C)}

Introduciamo la {\em distanza di Lebesgue}\index{distanza di Lebesgue} di un ricoprimento $\mathcal{A}^{0}$ come la massima distanza $a$ tale che la sfera di raggio $a$ disegnata prendendo come centro un qualsiasi punto di $\Omega$ \` e strettamente inclusa in un elemento di $\mathcal{A}^{0}$. Nel nostro caso modificheremo questa definizione, sostituendo al posto della sfera $\Sigma_{a}$ le porzioni locali ({\em cf.} definizione \ref{def:Anosov}) delle variet\` a stabili e instabili $W_{\gamma}^{u}(x)$, $W_{\gamma}^{s}(x)$ passanti per $x$; infatti, richiederemo che il nostro ricoprimento $\mathcal{A}^{0}$ ci permetta di trovare per ogni $x\in\Omega$ un $A^{0}_{F(x)}\in\mathcal{A}^{0}$ tale che

\begin{eqnarray}
d(y,W^{u}_{\gamma}(y)\cap\partial^{s}A_{F(x)})&>&a\qquad\mbox{$\forall y\in W^{s}_{\gamma}(x)\cap A^{0}_{F(x)}$}\\
d(z,W^{s}_{\gamma}(z)\cap\partial^{u}A_{F(x)})&>&a\qquad\mbox{$\forall z\in W^{u}_{\gamma}(x)\cap A^{0}_{F(x)}$},
\end{eqnarray}

dove le distanze $d$ sono misurate rispettivamente lungo $W^{u}_{\gamma}(y)$ e $W^{s}_{\gamma}(z)$. In pratica, stiamo escludendo il caso in cui le variet\` a locali $W^{u}_{\gamma}(x)$, $W^{s}_{\gamma}(x)$ si trovino ``troppo'' vicine ai bordi di {\em tutti} gli elementi di $\mathcal{A}^{0}$.

\begin{figure}[htbp]
\centering
\includegraphics[width=0.7\textwidth]{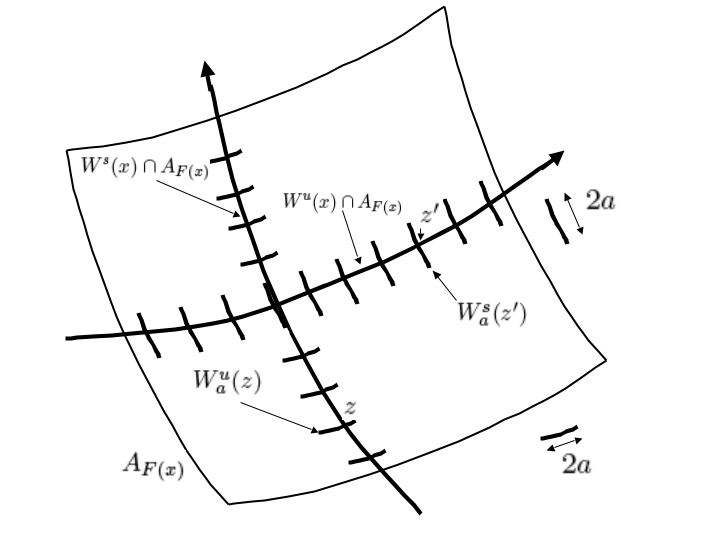}
\caption{Per ogni punto $x\in\Omega$ esiste un $S$-rettangolo $A_{F(x)}$ tale che la variet\` a stabile e la variet\` a instabile passanti per $x$ sono ``abbastanza lontane'' dai bordi di $A_{F(x)}$.} \label{fig:pav2}
\end{figure}

Qui e nel seguito assumeremo sempre che il diametro dei rettangoli sia abbastanza piccolo da poter considerare $W^{u}_{a}(x)$, $W^{s}_{a}(x)$ come dei segmenti dritti di lunghezza $2a$ centrati in $x$; quindi, assumendo che il diametro massimo dei rettangoli di $\mathcal{A}^{0}$ sia $\alpha<\delta/2$ ($\delta$ \` e il diametro massimo dei rettangoli del pavimento di Markov che vogliamo costruire; la ragione della scelta di $\alpha$ sar\` a chiara nel prossimo passo della dimostrazione) fisseremo $0<a<\alpha/2$.

\subsubsection{Raffinamento del ricoprimento $\mathcal{A}^{0}$}\label{sez:(D)}

Adesso, vogliamo dimostrare che partendo dal ricoprimento $\mathcal{A}^{0}$ appena introdotto \` e possibile costruire un nuovo ricoprimento $\mathcal{A}$ che verifica le seguenti propriet\`a: 

\begin{enumerate}
\item per ogni $\beta=1,2$ e $A\in\mathcal{A}$ esistono $\beta',\beta''=1,2$ e $A',A''\in\mathcal{A}$ tali che

\begin{equation}
g\partial^{s}_{\beta}A\subset\partial^{s}_{\beta'}A',\qquad g^{-1}\partial^{u}_{\beta}A\subset\partial^{u}_{\beta''}A''; \label{eq:dimpav3a}
\end{equation}

\item per ogni $z\in\partial^{s}A_{i}$, $y\in\partial^{u}A_{i}$ si ha che

\begin{eqnarray}
W^{s}_{a/2}(gz)&\subset&\partial^{s}A_{j(z)}\nonumber\\ 
W^{u}_{a/2}(g^{-1}y)&\subset&\partial^{u}A_{k(y)}\qquad\mbox{per $j(z)$, $k(y)$ opportuni.} \label{eq:dimpav3b}
\end{eqnarray}
\end{enumerate}

Ci\` o potr\` a essere fatto nel seguente modo. Per ogni $A^{0}_{i}\in\mathcal{A}^{0}$ consideriamo il segmento $C_{i}^{0}$ che collega due punti $\xi\in\partial^{s}_{1}A^{0}$, $\xi'\in\partial^{s}_{2}A^{0}$ giacenti sulla stessa porzione di variet\` a instabile $W^{u}_{\gamma}(\xi)$ (la scelta della coppia di punti \` e indifferente); questi segmenti formano un insieme $\mathcal{C}^{0}=\{C_{j}^{0}\}_{j=1}^{r}$, dove $r$ \` e il numero di elementi di $\mathcal{A}^{0}$. Poich\'e $m$ \` e ``molto grande'' (tra poco spiegheremo in che senso), i ``nuovi'' segmenti $gC_{j}^{0}$, $j=1,...,r$ saranno dei ``fili molto lunghi'', giacenti sulla variet\` a instabile. Scegliamo un ricoprimento $A^{0}_{i_{1}},...,A^{0}_{i_{k_{j}}}$ di $gC_{j}^{0}$ composto da elementi di $\mathcal{A}^{0}$ tali che $gC_{j}^{0}$ li attraversi ad una distanza $\geq a$ dai bordi paralleli, con $a$ pari alla distanza di Lebesgue definita in precedenza. In generale,  $gC_{j}^{0}$ non attraverser\` a completamente alcuni rettangoli che lo ricoprono; ci\`o accadr\` a per i rettangoli $A^{0}_{i_{1}}$, $A^{0}_{i_{k_{j}}}$ che contengono le estremit\` a del ``lungo filo'' $gC^{0}_{j}$, {\em cf.} figura \ref{fig:pav3}. In questo caso, ``allunghiamo'' il filo in modo da far coincidere le sue estremit\` a con un punto appartenente $\partial^{s}_{\beta}A_{i_{0}}$ e uno in $\partial^{s}_{\beta'}A_{i_{k_{j}}}$, ovvero introduciamo un nuovo segmento $C_{j}^{1}$ come\footnote{L'insieme $[C_{i_{h}}^{0},gC_{j}^{0}\cap A_{i_{h}}^{0}]$ corrisponde a $\cup_{x\in\left(gC_{j}^{0}\cap A_{i_{h}}^{0}\right)}[C_{i_{h}},x]$, dove $[C_{i_{h}},x]$ \` e il segmento passante per $x$ parallelo a $C_{i_{h}}$, dunque $[C_{i_{h}}^{0},gC_{j}^{0}\cap A_{i_{h}}^{0}]$ \` e il prolungamento di $gC_{j}^{0}\cap A_{i_{h}}^{0}$ che interseca sia $\partial^{s}_{1}A$ che $\partial^{s}_{2}A$.}

\begin{figure}[htbp]
\centering
\includegraphics[width=0.7\textwidth]{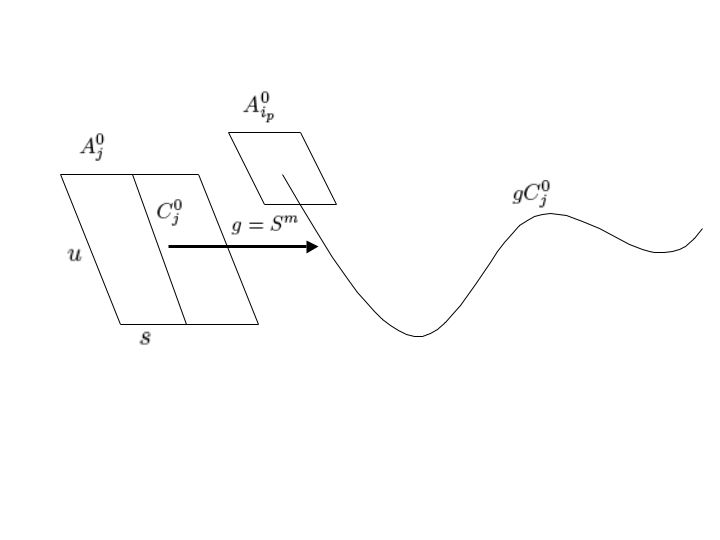}
\caption{Il segmento $C_{j}^{0}$ diventa un ``filo molto lungo'' sotto l'azione di $g=S^{m}$; nella figura, il rettangolo $A_{i_{p}}^{0}$ non \` e completamente attraversato da $gC_{j}^{0}$}. \label{fig:pav3}
\end{figure}

\begin{eqnarray}
C_{j}^{1}=\cup_{h=1}^{k_{j}}g^{-1}\left( \mbox{\{continuazione di $gC_{j}^{j}$ che attraversa tutti} \right.\nonumber\\ \left.\mbox{i rettangoli $A_{i_{h}}^{0}$\}} \cap A_{i_{h}}^{0}\right)=\cup_{h=1}^{k_{j}}g^{-1}\left([C_{i_{h}}^{0},gC_{j}^{0}\cap A_{i_{h}}^{0}]\right), \label{eq:dimpav6}
\end{eqnarray}  

e sostituiamo il vecchio $A_{j}^{0}$ con $A_{j}^{1}=[C_{j}^{1},D_{j}^{0}]$ ($C_{j}^{1}$ \` e leggermente ``pi\` u lungo'' di $C_{j}^{0}$, dunque definiamo un nuovo $A_{j}^{1}$ che lo contiene).

Adesso, applichiamo $g$ al segmento $C_{j}^{1}$ definito nella (\ref{eq:dimpav6}), e ripetiamo il discorso per questo ``nuovo filo''; iterando $n$ volte si ha che

\begin{equation}
C_{j}^{n}=\cup_{h=1}^{k_{j}}g^{-1}\left([C_{i_{h}}^{n-1},gC_{j}^{n-1}\cap A_{i_{h}}^{n-1}]\right), \label{eq:dimpav7}
\end{equation}

e resta da discutere se questa procedura ha senso nel limite $n\rightarrow\infty$, {\em i.e.} se il filo converge ad un ``filo asintotico''.

Ci\` o pu\` o essere fatto con un ragionamento per induzione. Al termine del primo passo, quello che porta da $C_{j}^{0}$ a $C_{j}^{1}$, il ``nuovo filo'' $C_{j}^{1}$ potr\` a essere espresso come $C_{j}^{1}=q_{1}\cup C_{j}^{0}\cup \overline{q}_{1}$\footnote{Non \` e detto che il $C^{0}_{j}$ che compare il questa definizione di $C_{j}^{1}$ sia lo stesso $C_{j}^{0}$ di partenza; in generale sar\` a un segmento parallelo al ``vecchio'' $C_{j}^{0}$, con la stessa lunghezza.}, dove $q_{1}$, $\overline{q}_{1}$ sono i due ``segmentini'' vengono aggiunti alle estremit\` a di $C_{j}^{0}$ per permettere a $gC_{j}^{1}$ di attraversare completamente tutti i rettangoli del ricoprimento $\{A^{0}_{i_{h}}\}$, con\footnote{$\alpha$ \` e il massimo diametro dei rettangoli $A^{0}\in\mathcal{A}$, e il fattore $\lambda^{m}$ \` e dovuto all'applicazione di $g^{-1}$.} $q_{1},\overline{q}_{1}\leq \alpha\lambda^{m}$; definendo $A_{j}^{n}\equiv [C_{j}^{n}, D_{j}^{0}]$, e assumendo che al passo $n$ si abbia $|q_{n}|\leq \left(\lambda^{m}\right)^{n}\alpha$ ({\em idem} per $\overline{q}_{n}$),  al passo $n+1$ il ``filo'' $gC_{j}^{n}$, che attraversava completamente i rettangoli $A^{n-1}_{i_{h}}$, dovr\` a essere ``prolungato'' ad ogni estremit\` a con dei segmentini $q_{n}$, $\overline{q}_{n}$, ovvero $C_{j}^{n+1}=g^{-1}q_{n}\cup C_{j}^{n}\cup g^{-1}\overline{q}_{n}\equiv q_{n+1}\cup C_{j}^{n}\cup\overline{q}_{n+1}$, con $|q_{n+1}|, |\overline{q}_{n+1}|\leq \left(\lambda^{m}\right)^{n+1}\alpha$.

Grazie a questo argomento \` e chiaro che il limite $C_{j}=\cup_{n=0}^{\infty}C_{j}^{n}$ avr\` a lunghezza finita, dal momento che $|C_{j}|\leq |C_{j}^{0}|+\sum_{n=1}^{\infty}2\alpha\lambda^{mn}= |C_{j}^{0}|+\frac{2\alpha\lambda^{m}}{1-\lambda^{m}}<\infty$, e scegliamo $m$ in modo tale che 

\begin{equation}
\frac{2\alpha\lambda^{m}}{1-\lambda^{m}}<a/2<\alpha. \label{eq:dimpav7a} 
\end{equation}

Adesso, costruiamo i segmenti $\{D_{j}\}_{j=1}^{r}$ partendo da $\{D_{j}^{0}\}_{j=1}^{r}$ con un procedimento del tutto simile a quello usato per $\{C_{j}\}_{j=1}^{r}$, sostituendo $g$ con $g^{-1}$; il risultato \` e che possiamo definire un nuovo ricoprimento $\mathcal{A}$ con $A_{j}=[C_{j},D_{j}]$, $\mbox{diam}(A_{j})<2\alpha=\delta$, tale che

\begin{equation}
g[C_{j},z]\supset [C_{j'},gz], \qquad g^{-1}[z,D_{j}]\supset [g^{-1}z,D_{j''}], \label{eq: dimpav7}
\end{equation}

e inoltre

\begin{eqnarray}
A_{j'}&\supset& W^{s}_{\frac{a}{2}}(y)\qquad\forall y\in [C_{j'},gz]\\
A_{j''}&\supset& W^{u}_{\frac{a}{2}}(z)\qquad\forall z\in [g^{-1}z, D_{j''}]. \label{eq;dimpav8}
\end{eqnarray}

Queste propriet\` a coincidono con le (\ref{eq:dimpav3a}), (\ref{eq:dimpav3b}), dunque abbiamo costruito il ricoprimento con le caratteristiche richieste; i rettangoli che formano questo ricoprimento, per\` o, in generale avranno punti interni in comunque. Nella prossimo punto della dimostrazione faremo vedere come da questo ricoprimento sia possibile estrarre una pavimentazione markoviana.

\subsubsection{Costruzione della pavimentazione di Markov}\label{sez:(E)}

Per concludere la dimostrazione, facciamo vedere come da $\mathcal{A}$ sia possibile ottenere un pavimento di Markov $\mathcal{B}$, formato da $S$-rettangoli $B_{i}$ con $\mbox{diam}(B_{i})<2\alpha=\delta$.

Partendo da $\mathcal{A}$, possiamo ottenere un {\em raffinamento} di $\mathcal{A}$ intersecando elementi di $\mathcal{A}$ con i complementari di altri. Pi\`u precisamente, per ogni $\mathbf{r}\subset\{1,...,r\}$, possiamo definire $P_{\mathbf{r}}=\left(\cap_{i\in\mathbf{r}}A_{i}\right)\cap\left(\cup_{i\notin\mathbf{r}}\Omega\setminus A_{i}\right)$\footnote{Quindi, se $\mathbf{r}\neq\mathbf{r}'$ si ha che $P_{\mathbf{r}}\cap P_{\mathbf{r}'}=\emptyset$.}; per costruzione si ha che (come per $\mathcal{A}$, i bordi di $P_{\mathbf{r}}$ appartengono alle variet\` a stabili e instabili)

\begin{equation}
g\partial^{s}P\subset\cup_{P'\in P}\partial^{s}P', \qquad g^{-1}\partial^{u}P\subset\cup_{P'\in P}\partial^{u}P'
\end{equation}

\begin{figure}[htbp]
\centering
\includegraphics[width=0.7\textwidth]{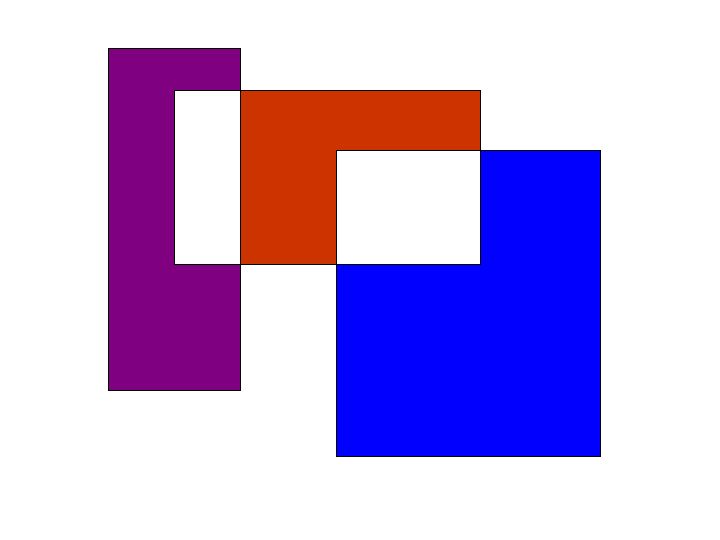}
\caption{Non tutti gli insiemi generati dalle intersezioni di elementi di $\mathcal{A}$ con i complementari di altri sono $S$-rettangoli.}. \label{fig:pav4}
\end{figure}
 
 Questo pavimento, per\` o, non \`e formato da $S$-rettangoli, perch\' e non \` e vero che se $x,y\in P$ allora $[x,y]\in P$ (vedere figura \ref{fig:pav4}). Per ottenere un pavimento formato da quadrati, basta ``continuare un po''' i lati degli elementi di $A_{i}\in\mathcal{A}$ fino a quando il prolungamento del lato non incontra un lato del ``tipo opposto'' (ovvero un lato stabile per un prolungamento instabile, e viceversa) (vedere figura \ref{fig:pav5}); chiaramente, dal momento che il diametro dei rettangoli contenuti in $\mathcal{A}$ \` e $<2\alpha$, la lunghezza del prolungamento sar\` a $<2\alpha$. Pi\` u precisamente, ci\`o vuol dire sostituire $\partial^{s}_{\beta}A$, $\partial^{u}_{\beta}A$ con $\tilde{\partial}^{u}_{\beta}A$, $\tilde{\partial}^{u}_{\beta}A$ definiti come
 
 \begin{equation}
 \tilde{\partial}^{u}_{\beta}A=\cup_{x\in\partial^{s}_{\beta}A}W^{s}_{\gamma}(x), \qquad\tilde{\partial}^{u}_{\beta}A=\cup_{x\in\partial^{u}_{\beta}A}W^{u}_{\gamma}(x), \label{eq:dimpav8}
 \end{equation}
 
 \begin{figure}[htbp]
\centering
\includegraphics[width=0.7\textwidth]{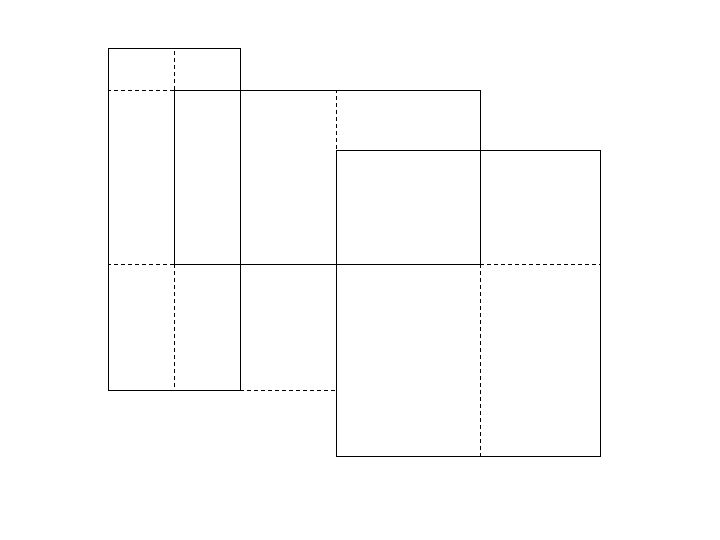}
\caption{Prolungando i lati degli $S$-rettangoli $A_{i}\in\mathcal{A}$ in modo opportuno si ottiene un nuovo pavimento $\mathcal{B}$ formato da $S$-rettangoli che non si intersecano tra di loro.}. \label{fig:pav5}
\end{figure}
 
 ovvero continuiamo $\partial^{s}_{\beta}A$, $\partial^{u}_{\beta}A$ ad ogni estremit\` a aggiungengo un pezzo della variet\` a dello stesso tipo di lunghezza $\gamma$. Poich\'e $\gamma>2\alpha$, gli insiemi definiti nella (\ref{eq:dimpav8}) andranno oltre il punto dove intersecano per la prima volta un lato di tipo ``opposto'', dunque dobbiamo considerare solo la parte di $\tilde{\partial}^{u}_{\beta}A$, $\tilde{\partial}^{s}_{\beta}A$ che non supera il punto prefissato; con questi ``nuovi lati'', gli insiemi $P_{\mathbf{r}}$, che chiameremo $B\in\mathcal{B}$, saranno degli $S$-rettangoli con diametro $<2\alpha=\delta$. L'ultima cosa che rimane da fare \` e verificare che gli insiemi $B\in\mathcal{B}$ rispettano le propriet\` a del pavimento di Markov (\ref{eq:dimpav4}); per fare ci\` o, notiamo che per ogni punto $z$ appartenente al bordo di $B$, ad esempio al bordo stabile $\partial^{s}B$, esiste un $x\in\partial^{s}A_{G(z)}$, con $G(z)$ opportuno, tale che
 
 \begin{equation}
 z\in W^{s}_{2\alpha}(x),\label{eq:dimpav9}
 \end{equation}
 
 e lo stesso vale se $z\in\partial^{u}B$, sostituendo $W^{s}_{2\alpha}(x)$ con $W^{u}_{2\alpha}(y)$ per un certo $y\in\partial^{u}A_{H(z)}$. Grazie alle propriet\` a (\ref{eq:dimpav3b}) abbiamo che, per $j(x)$ opportuno,
 
 \begin{equation}
 gx\in\partial^{s}A_{j(x)}\quad\mbox{e}\quad W^{s}_{a/2}(gx)\subset\partial^{s}A_{j(x)}; \label{eq:dimpav10}
 \end{equation}
 
poich\' e il diametro di $gW^{s}_{2\alpha}(x)$ \` e minore di $2\lambda^{m}\alpha<a/2$ (per la scelta di $m$, vedere eq. (\ref{eq:dimpav7a})), abbiamo che $gW^{s}_{2\alpha}(x)\subset W^{s}_{a/2}(gx)$. Dunque, grazie alla (\ref{eq:dimpav10}) e sapendo che $\cup_{A\in\mathcal{A}}\partial^{s}A\subset\cup_{B\in\mathcal{B}}\partial^{s}B$ (lo stesso vale sostituendo $s$ con $u$),
 
 \begin{equation}
 g\partial^{s}_{\beta}B\subset\cup_{B'\in\mathcal{B}}\partial^{s}B'; \label{eq:dimpav11}
 \end{equation}
 
allo stesso modo \` e possibile dimostrare che

 \begin{equation}
 g^{-1}\partial^{u}_{\beta}B\subset\cup_{B'\in\mathcal{B}}\partial^{u}B'. \label{eq;dimpav11a}
 \end{equation}
 
 \subsubsection{H\"older continuit\` a del codice}
 
 \` E facile vedere che la corrispondenza che associa ad ogni sequenza $\underline{\sigma}$ un punto $x$ \` e una funzione H\"older continua, nel senso della definizione (\ref{eq:pavmark}). Infatti, assumiamo che due punti $x_{1}=X(\underline{\sigma})$, $x_{2}=X(\underline{\sigma}')$ abbiano la stessa storia $\sigma_{1},...,\sigma_{n}$ fino al tempo $n>0$, e che al tempo $n+1$ $\sigma_{n+1}\neq\sigma'_{n+1}$; ci\` o vuol dire che chiamando $\underline{v}_{12}(x_{1})$ il vettore che congiunge i due punti
 
 \begin{equation}
 \left|S^{n}\underline{v}_{12}(x_{1})\right|\leq \delta. \label{eq:dimh}
 \end{equation}
 
 Possiamo scrivere il vettore $\underline{v}_{12}$ come
 
 \begin{equation}
 \underline{v}_{12}(x_{1}) = c_{u}(x_{1})w_{u}(x_{1})+c_{s}(x_{1})w_{s}(x_{1}), \label{eq:dimpavh}
 \end{equation}
  
  dove $w_{u}(x_{1})$, $w_{s}(x_{1})$ sono i due vettori normalizzati tangenti alla variet\` a instabile e alla variet\` a stabile nel punto $x_{1}$; grazie alla (\ref{eq:dimpavh}) la (\ref{eq:dimh}) diventa, assumendo per semplicit\` a che $S$ agisca linearmente sui vettori e che\footnote{\` E possibile infittire abbastanza la partizione in modo da escludere l'eventualit\` a che sotto l'azione di una iterata $S$ due punti giacenti sulla stessa variet\` a instabile e nello stesso $S$-rettangolo si spostino in un $S$-rettangolo comune.} $c_{u}\neq 0$,
  
  \begin{equation}
  \left|c_{u}(x_{1})S^{n}w_{u}(x_{1})+c_{s}(x_{1})S^{n}w_{s}(x_{1})\right|\leq \delta. \label{eq:dimh2}
  \end{equation}
  
  Quindi, per la definizione di sistema di Anosov \ref{def:Anosov} la (\ref{eq:dimh2}) implica che
  
  \begin{equation}
  |c_{u}(x_{1})|\lambda^{-n}d(x_{1},x_{2})\left(1-\lambda^{2n}\frac{c_{s}(x_{1})}{c_{u}(x_{1})}\right)\leq \delta d(x_{1},x_{2})
  \end{equation}
  
  ovvero, infittendo opportunamente\footnote{In particolare, dobbiamo infittire la partizione evitando che si formino dei rettangoli troppo ``lunghi e stretti''.} la partizione in modo da avere $\lambda^{2n}\frac{c_{s}(x_{1})}{c_{u}(x_{1})}\neq 1$, esiste una costante $C>0$ tale che
  
  \begin{equation}
  d(x_{1},x_{2}) \leq \lambda^{n}C; \label{eq:dimh3}
  \end{equation}
  
 infine, ponendo $c=-\log\lambda$, possiamo riscrivere la (\ref{eq:dimh3}) come
  
  \begin{equation}
  d(x_{1},x_{2})=d\left(X(\underline{\sigma},X(\underline{\sigma}'))\right)\leq Ce^{-cn}= Cd(\underline{\sigma},\underline{\sigma}').
  \end{equation}
  
Inoltre, scegliendo una partizione abbastanza fitta l'intersezione tra un elemento di $\mathcal{P}$ e l'immagine sotto $S$ di qualunque altro sar\` a un insieme connesso\footnote{Stiamo escludendo l'ipotesi che sotto l'azione di $S$ un rettangolo si allunghi abbastanza da poter intersecare due volte lo stesso rettangolo.}, dunque in questa situazione $X^{-1}$ non pu\` o avere molteplicit\` a infinita; ci\` o conclude la dimostrazione.

\end{proof}

\section{Codifica della misura di volume}

In questa sezione mostreremo che, per un sistema di Anosov\index{sistema di Anosov}, la codifica in simboli dell'evoluzione di un dato iniziale permetter\` a di rappresentare la misura di volume su $\Omega$ come la misura di probabilit\` a di un modello di spin unidimensionale con potenziali a decadimento esponenziale. A tal fine, introduciamo la nozione di {\em probabilit\` a condizionata}\index{probabilit\`a  condizionata}.

\begin{definizione}{{\em (Probabilit\` a condizionata)}}\label{def:probcond}

Sia $m$ una misura di probabilit\` a definita sui sottoinsiemi boreliani dello spazio delle sequenze con $q$ simboli $\{1,...,q\}^{\mathbb{Z}}$, e siano $J=\{j_{1},...,j_{s}\}\subset\mathbb{Z}$, $\underline{\sigma}_{J}=(\sigma_{J_{1}}....\sigma_{J_{s}})\in\{1,...,q\}^{J}$.

Consideriamo le $\sigma$-algebre $\mathcal{B}(J^{c})$ generate dai cilindri $C_{\underline{\sigma}'}^{J'}$, con base $J'\subset J^{c}=\mathbb{Z}\setminus J$; dato $\underline{\sigma}_{J}$ possiamo definire su $\mathcal{B}(J^{c})$ le misure

\begin{equation}
E\rightarrow m'(E)\equiv m(E\cap C^{J}_{\underline{\sigma}_{J}}),\qquad E\rightarrow \overline{m}(E)\equiv m(E), \label{eq:vol01}
\end{equation}

per $E\in\mathcal{B}(J^{c})$. Dunque la misura $m'$ \` e {\em assolutamente continua}\footnote{Due misure $\nu$, $\mu$ definite in $\Omega$ sono assolutamente continue se e solo se $\nu(A)=0\iff\mu(A)=0$; il teorema di Radon Nikodym afferma che per due misure assolutamente continue \` e possibile trovare una funzione misurabile $f$ tale che

\begin{equation}
\nu(A)=\int_{A}fd\mu. \label{eq:RN}
\end{equation}

} rispetto a $\overline{m}$, ed \` e possibile definire la {\em derivata di Radon-Nikodym}\index{derivata di Radon-Nikodym} di $m'$ rispetto a $m$

\begin{equation}
\frac{dm'}{d\overline{m}}(\underline{\sigma}')\equiv m(\underline{\sigma}_{J}|\underline{\sigma}'_{J^{c}}); \label{eq:vol02}
\end{equation}

l'espressione (\ref{eq:vol02}) \` e chiamata la {\em probabilit\` a condizionata di $\underline{\sigma}_{J}$ rispetto a $\underline{\sigma}'_{J}$}. 

\end{definizione}

\begin{oss}
La probabilit\` a condizionata $m(\underline{\sigma}_{J}|\underline{\sigma}_{J^{c}}')$ deve verificare l'ovvia {\em condizione di compatibilit\`a} 

\begin{equation}
\sum_{\underline{\sigma}_{J}}m(\underline{\sigma}_{J}|\underline{\sigma}_{J^{c}}')=1; \label{eq:osscomp}
\end{equation}

quindi, grazie alla (\ref{eq:osscomp}), per sapere chi \` e $m(\underline{\sigma}_{J}|\underline{\sigma}'_{J^{c}})$ \` e sufficiente conoscere i rapporti

\begin{equation}
\frac{m(\underline{\sigma}'_{J}|\underline{\sigma}_{J^{c}})}{m(\underline{\sigma}''_{J^{c}}|\underline{\sigma}_{J^{c}})},\qquad\forall\; \underline{\sigma}'_{J},\underline{\sigma}''_{J},\underline{\sigma}_{J^{c}}. \label{eq:03}
\end{equation}

\end{oss} 

Come abbiamo gi\` a detto introducendo i sistemi di Anosov\index{sistema di Anosov}, ha senso considerare la mappa $S$ come una mappa $S:W^{u}_{\gamma}(x)\rightarrow W^{u}_{\gamma}(Sx)$, oppure come $S:W^{s}_{\gamma}(x)\rightarrow W^{s}_{\gamma}(Sx)$; chiamiamo {\em coefficiente locale di espansione} e {\em coefficiente locale di contrazione} il modulo del determinante della matrice jacobiana di $S$ rispettivamente lungo la variet\` a instabile e stabile. Nel seguito considereremo per semplicit\` a solo il caso bidimensionale.  

\begin{definizione}{{\em (Coefficiente di espansione locale e coefficiente di contrazione locale)}}\label{def:contr}
Sia $(\Omega,S)$ un sistema di Anosov bidimensionale, e sia $\mathcal{Q}=\{Q_{1},...,Q_{q}\}$ un pavimento markoviano\index{pavimento markoviano}. Sia $T$ la sua matrice di compatibilit\` a\index{matrice di compatibilit\` a}, e $X$ il codice della dinamica\index{codice della dinamica simbolica} $X:\{1,...,q\}_{T}^{\mathbb{Z}}\rightarrow\Omega$. Definiamo il {\em coefficiente di espansione locale $\lambda_{u}(x)$}\index{coefficiente di espansione locale} e il {\em coefficiente di contrazione locale $\lambda_{s}(x)$}\index{coefficiente di contrazione locale} come

\begin{equation}
\lambda_{u}^{-1}(x)\equiv \left|\frac{dS^{-1}}{d\xi}(x)\right|_{u},\qquad\lambda_{s}(x)=\left|\frac{dS}{d\xi}(x)\right|_{s}, \label{eq:vol1}
\end{equation}

dove con i pedici $u,s$ indichiamo che la derivata di $S$ \` e calcolata sulla variet\` a instabile o sulla variet\` a stabile, e $\xi$ rappresenta la parametrizzazione della variet\` a; inoltre, definiamo il {\em tasso di contrazione locale}\index{tasso di contrazione locale} e il {\em tasso di espansione locale}\index{tasso di espansione locale} come

\begin{equation}
A_{u}(x)=\log \lambda_{u}(x),\qquad A_{s}(\underline{\sigma})=-\log\lambda_{s}(x). \label{eq:vol2}
\end{equation}

\end{definizione}

\begin{oss}
\begin{enumerate}

\item La definizione di sistema iperbolico liscio implica che $\lambda_{u}(x)$, $\lambda_{s}(x)$, e dunque anche i logaritmi (\ref{eq:vol2}), siano funzioni H\"older continue di $x$.
\item Nel caso del gatto di Arnold\index{gatto di Arnold} i coefficienti di espansione e di contrazione coincidono con l'autovalore instabile $\lambda_{+}$ e con l'autovalore stabile $\lambda_{-}$.
\item\label{oss:potexp} Come vedremo tra poco, partendo da $A_{u}(X(\underline{\sigma}))$, $A_{s}(X(\underline{\sigma}))$ potremo definire i potenziali a molti corpi del modello di Ising unidimensionale associato alla misura di volume. Inoltre, grazie alla H\"older continuit\` a di $\lambda_{u}(x)$, $\lambda_{s}(x)$ e del codice $X$, $A_{u}$ e $A_{s}$ sono funzioni H\"older continue nello spazio delle sequenze; questo fatto implicher\` a  il decadimento esponenziale dei potenziali associati alla misura di volume.
\item Nel seguito, quando scriveremo $A_{u}(\underline{\sigma})$, $A_{s}(\underline{\sigma})$ ci riferiremo in realt\` a, con un leggero abuso di notazione, a $A_{u}(X(\underline{\sigma}))$, $A_{s}(X(\underline{\sigma}))$.
\end{enumerate}
\end{oss}

Con l'osservazione \ref{oss:potexp} abbiamo anticipato il seguente risultato.

\begin{prop}{{\em (Sviluppo in potenziali di funzioni H\"older continue)}}\label{prop:svilpot}

Le funzioni $A_{u}(\underline{\sigma})$, $A_{s}(\underline{\sigma})$ introdotte nella definizione \ref{def:contr}, possono essere scritte come\footnote{In generale ci\` o sar\` a vero per qualunque funzione H\"older continua.}:

\begin{equation}
A_{a}=cost+\sum_{n=0}^{\infty}\Phi_{2n+1}^{a}(\sigma_{-n},...,\sigma_{n}),\quad a=u,s,
\end{equation}

dove $\left\{\Phi_{2n+1}^{a}\right\}$ decadono esponenzialmente, ovvero esistono $C,k>0$ tali che

\begin{equation}
|\Phi_{2n+1}^{a}(\sigma_{-n},...,\sigma_{n})|\leq Ce^{-kn}. \label{eq:vol3}
\end{equation}

\end{prop}

\begin{oss}
Le funzioni $\Phi^{a}_{2n+1}$ sono dette {\em cilindriche}\index{funzione cilindrica}, dal momento che dipendono da un numero finito di simboli $\sigma_{-n},...,\sigma_{n}$.
\end{oss}

\begin{proof}
La dimostrazione di questa proposizione \` e molto semplice, e si basa su un argomento piuttosto ricorrente nella teoria dei sistemi di Anosov. Consideriamo le sequenze $\underline{\sigma}^{0},\underline{\sigma}\in\{1,...,q\}^{\mathbb{Z}}_{T}$, e sia $z$ il {\em tempo di mescolamento}\index{tempo di mescolamento} della matrice di compatibilit\` a $T$, ovvero tale che $T^{z}_{\sigma\sigma'}>0$ per ogni $\sigma,\sigma'\in\{1,...,q\}$; quindi, per ogni $\sigma,\sigma'$ esiste almeno una stringa $\vartheta(\sigma,\sigma')$ formata da $z$ simboli compatibili tale che la successione $\sigma\vartheta(\sigma,\sigma')\sigma'$ \` e compatibile.

Costruiamo le sequenze $\underline{\sigma}^{n}$ come, indicando con $\underline{\sigma}_{[a,b]}$ la frazione della stringa $\underline{\sigma}$ compresa tra gli indici $a,b$ (inclusi),

\begin{equation}
\underline{\sigma}^{n}\equiv \underline{\sigma}_{[-\infty,-n-z-1]}\vartheta(\sigma_{-n-z-1},\sigma^{0}_{-n})\underline{\sigma}^{0}_{[-n,n]}\vartheta(\sigma^{0}_{n},\sigma_{n+z+1})\underline{\sigma}_{[n+z+1,+\infty]}; \label{eq:vol4}
\end{equation}

ponendo $A_{a}(\underline{\sigma}^{0})=cost$ possiamo scrivere la funzione $A_{a}(\underline{\sigma})$ come

\begin{equation}
A_{a}(\underline{\sigma})=cost + \sum_{n=0}^{\infty}\left(A_{a}(\underline{\sigma}^{n+1})-A_{a}(\underline{\sigma}^{n})\right)\equiv cost + \sum_{n=0}^{\infty}\Phi^{a}_{2n+1}(\sigma_{-n},...,\sigma_{n}).
\end{equation}

Infine, grazie alla H\"older continuit\` a di $A_{a}(\underline{\sigma})$ esistono due costanti $C,k>0$ tali che

\begin{equation}
\left|\Phi^{a}_{2n+1}\right| = \left|A_{a}(\underline{\sigma}^{n+1})-A_{a}(\underline{\sigma}^{n})\right|\leq Ce^{-kn}, \label{eq:vol5}
\end{equation}

e ci\` o conclude la dimostrazione.

\end{proof}

Abbiamo ora gli strumenti per capire cosa diventa la misura di volume di un sistema di Anosov se ``codificata'' in dinamica simbolica.

\begin{prop}{{\em (Codifica della misura di volume)}}

Sia $(\Omega,S)$ un sistema di Anosov bidimensionale, e $\mathcal{Q}=\{Q_{1},...,Q_{q}\}$ un pavimento markoviano con matrice di compatibilit\` a $T$ e codice $X$. Sia $\mu^{0}$ la misura di volume in $\Omega$ associata alla metrica definita in $\Omega$, e $m_{0}$ la sua misura ``immagine'' nello spazio delle sequenze $\{1,...,q\}^{\mathbb{Z}}_{T}$ tramite $X$ ({\em cf.} osservazione \ref{oss:isomorf} nella definizione \ref{def:pavmark}). I rapporti di probabilit\` a condizionate calcolate con la misura $m_{0}$ verificano

\begin{eqnarray}
\frac{m_{0}(\sigma'_{-n},...,\sigma'_{n}|\sigma_{j},|j|>n)}{m_{0}(\sigma''_{-n},...,\sigma''_{n}|\sigma_{j},|j|>n)}=\frac{\sin\varphi(X(\underline{\sigma}'))}{\sin\varphi(X(\underline{\sigma}''))}\times\nonumber\\\exp\left(-\sum_{k=1}^{\infty}A_{u}(\tau^{k}\underline{\sigma}')-A_{u}(\tau^{k}\underline{\sigma}'')+A_{s}(\tau^{-k}\underline{\sigma}')-A_{s}(\tau^{-k}\underline{\sigma}'')\right), \label{eq:vol6}
\end{eqnarray}

dove $\varphi(x)$ \` e l'angolo tra la variet\` a stabile e la variet\` a instabile nel punto $x$, e $\underline{\sigma}'$, $\underline{\sigma}''$ sono sequenze in $\{1,...,q\}^{\mathbb{Z}}_{T}$ i cui valori con indice $j$ coincidono per $|j|>n$.

\end{prop}

\begin{oss}\label{oss:vol}
Grazie alla proposizione \ref{prop:svilpot} possiamo riscrivere la (\ref{eq:vol6}) in modo che assuma una forma reminiscente della teoria dei modelli di spin. Definiamo le funzioni $(\Phi^{a}_{X})_{X\subset\mathbb{Z}}$ nel seguente modo:

\begin{eqnarray}
\Phi^{a}_{X}(\underline{\sigma}_{X})&=&0\quad\mbox{a meno che $X=(l,l+1,...,l+2n)$}\\
\Phi^{a}_{X}(\underline{\sigma}_{X})&=&\Phi^{a}_{2n+1}(\underline{\sigma}_{X})\quad\mbox{se $X=(l,l+1,...,l+2n)$},
\end{eqnarray}

per $\underline{\sigma}_{X}\in\{1,...,q\}^{X}$ e qualche $l,n\in\mathbb{Z}$; con questa definizione si ha che

\begin{equation}
A_{a}(\underline{\sigma})= \sum_{X\;centrato\;in\;0}\Phi_{X}^{a}(\underline{\sigma}_{X}),\quad a=u,s. \label{eq:vol7.0}
\end{equation} 

Quindi la (\ref{eq:vol6}) pu\` o essere riscritta come, se $a(X)=u$ per $\mbox{{\em centro}}(X)\geq 0$ e $a(X)=s$ per $\mbox{{\em centro}}(X)<0$,

\begin{equation}
\frac{m_{0}(\sigma'_{[-n,n]}|\sigma_{j},|j|>n)}{m_{0}(\sigma''_{[-n,n]}|\sigma_{j},|j|>n)}=\exp{\left(-\sum_{X\cap[-n,n]\neq\emptyset}\left[\Phi^{a(X)}_{X}(\underline{\sigma}'_{X})-\Phi^{a(X)}_{X}(\underline{\sigma}''_{X})\right]+\mathcal{B}(\underline{\sigma})\right)}, \label{eq:vol7}
\end{equation}

dove il termine $\mathcal{B}$ \` e dovuto al rapporto dei seni e al fatto che nella (\ref{eq:vol6}) la somma su $k$ parte da $1$; quindi, poich\' e $\mathcal{B}(\underline{\sigma})$ \` e H\"older continuo pu\` o essere sviluppato a sua volta in potenziali a decadimento esponenziale centrati in $0$. La (\ref{eq:vol7}) \` e a tutti gli effetti il rapporto delle probabilit\` a di due configurazioni di {\em spin} $\underline{\sigma}'$, $\underline{\sigma}''$ di un modello di Ising unidimensionale\index{modello di Ising unidimensionale} con {\em condizioni al bordo fissate} al di fuori di un volume $[-n,n]$, e interagente attraverso potenziali a molti corpi diversi se il centro dell'intervallo reticolare $X$ si trova a destra o a sinistra dell'origine $0$.

Notare che nella (\ref{eq:vol7}), come nella (\ref{eq:vol6}), le somme all'esponente sono assolutamente convergenti perch\'e le quantit\` a sommate sono calcolate in configurazioni $\underline{\sigma}$, $\underline{\sigma}'$ diverse solo in un intervallo finito $[-n,n]$, e le funzioni sono H\"older continue.
\end{oss}

\begin{proof}
Come per l'esistenza delle partizioni di Markov, anche in questo caso la dimostrazione della proposizione diventa notevolmente pi\` u semplice se visualizzata graficamente; per questo motivo, cercheremo sempre di illustrare i punti chiave con delle figure. Otterremo la (\ref{eq:vol6}) nel caso $n=0$; il caso $n\geq 0$ si studia con un argomento del tutto analogo a quello che stiamo per presentare.

Consideriamo due sequenze $\underline{\sigma}',\underline{\sigma}''\in \{1,...,q\}^{\mathbb{Z}}_{T}$ con $\sigma'_{j}=\sigma''_{j}=\sigma_{j}$ per $|j|>0$, e poniamo $X(\underline{\sigma})=x$, $X(\underline{\sigma}')=y$; assumiamo per semplicit\` a che $x,y\notin\cup_{i=1}^{q}\partial Q$; \` e chiaro che i due punti $x,y$ appartengono alla stesse variet\` a stabili e instabili, dal momento che hanno storie uguali nel passato e nel futuro. Per il teorema di Doob\index{teorema di Doob}, {\em cf.} \cite{ergo}, si ha che, $\mu$-quasi ovunque,

\begin{equation}
\frac{m_{0}(\sigma'_{0}|\sigma_{j},|j|>0)}{m_{0}(\sigma''_{0}|\sigma_{j},|j|>0)}=\lim_{N\rightarrow\infty}\frac{m_{0}\left(C^{-N...-1\quad0\quad1...N}_{\sigma_{-N}...\sigma_{-1}\sigma'_{0}\sigma_{1}...\sigma_{N}}\right)}{m_{0}\left(C^{-N...-1\quad0\quad1...N}_{\sigma_{-N}...\sigma_{-1}\sigma''_{0}\sigma_{1}...\sigma_{N}}\right)}; \label{eq:vol8}
\end{equation}

quindi, per dimostrare la (\ref{eq:vol6}) dobbiamo semplicemente calcolare il limite nella parte destra della (\ref{eq:vol8}).

L'area dell'$S$-rettangolo $C^{-N...N}_{\sigma_{-N}...\sigma'_{0}...\sigma_{N}}$, e dunque la sua misura di Lebesgue $m_{0}$, pu\` o essere approssimata con

\begin{equation}
A\left(C^{-N...N}_{\sigma_{-N}...\sigma'_{0}...\sigma_{N}}\right)=|\delta_{u}^{x}|\cdot |\delta_{s}^{x}|\sin\varphi(x)+o(\lambda^{2N}), \label{eq:vol9}
\end{equation}

dove $\delta_{u}^{x}$, $\delta_{s}^{x}$ sono vettori tangenti alla variet\` a instabile e alla variet\` a stabile applicati in un vertice $x'$ (non importa quale) dell'$S$-rettangolo che congiungono $x'$ ai punti $x''$, $x'''$ (vedere figure \ref{fig:vol1}, \ref{fig:vol2}),  e $\varphi(x)$ \` e il seno dell'angolo compreso tra $\delta_{u}^{x}$, $\delta_{s}^{x}$; nella formula (\ref{eq:vol9}) abbiamo sostituito $\varphi(x')$ con $\varphi(x)$ perch\'e\footnote{Ci\` o \` e dovuto al fatto che $W^{u}(x)$, $W^{s}(x)$ dipendono in modo H\"older continuo da $x$.} $|\sin\varphi(x)-\sin\varphi(x')|\leq C\lambda^{cN}$, con $c,C>0$.

 \begin{figure}[htbp]
\centering
\includegraphics[width=0.7\textwidth]{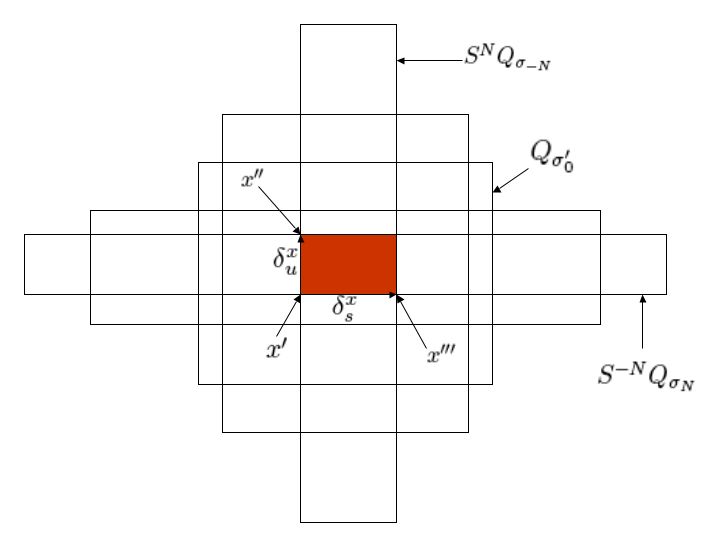}
\caption{Rappresentazione schematica del cilindro $C_{\sigma_{-N}...\sigma_{0}...\sigma_{N}}^{-N...N}$. In questo disegno abbiamo approssimato il cilindro con un rettangolo, e questo \` e possibile a meno di $o(\lambda^{N})$; vedere figura \ref{fig:vol2} per una rappresentazione pi\` u realistica della situazione.} \label{fig:vol1}
\end{figure}

 \begin{figure}[htbp]
\centering
\includegraphics[width=0.7\textwidth]{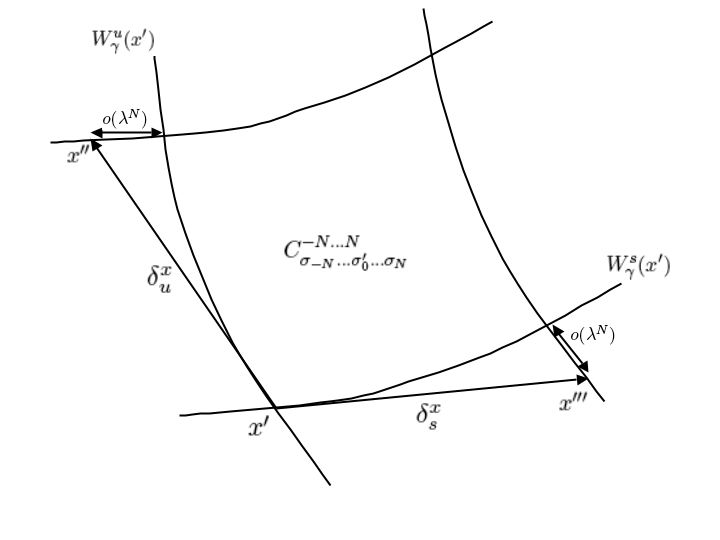}
\caption{} \label{fig:vol2}
\end{figure}

Lo stesso argomento pu\` o essere ripetuto per il cilindro $C^{-N...N}_{\sigma_{-N}...\sigma''_{0}...\sigma_{N}}$, e quindi il limite nella parte destra della (\ref{eq:vol8}) coincide con 

\begin{equation}
\frac{\sin\varphi(x)}{\sin\varphi(y)}\lim_{N\rightarrow\infty}\frac{|\delta_{u}^{x}|\cdot|\delta_{s}^{x}|}{|\delta_{u}^{y}|\cdot|\delta_{s}^{y}|}; \label{eq:vol10}
\end{equation}

possiamo riscrivere i rapporti nella (\ref{eq:vol10}) come, poich\'e $S^{-k}\circ S^{k}=1$,

\begin{eqnarray}
\frac{|\delta_{u}^{x}|}{|\delta_{u}^{y}|}&=&\frac{\prod_{j=1}^{k}\lambda_{u}^{-1}(S^{j}x')|S^{k}\delta_{u}^{x}|}{\prod_{j=1}^{k}\lambda_{u}^{-1}(S^{j}y')|S^{k}\delta_{u}^{y}|} \label{eq:vol11}\\
\frac{|\delta_{s}^{x}|}{|\delta_{s}^{y}|}&=&\frac{\prod_{j=1}^{k}\lambda_{s}(S^{-j}x')|S^{-k}\delta_{s}^{x}|}{\prod_{j=1}^{k}\lambda_{s}(S^{-j}y')|S^{-k}\delta_{s}^{y}|}, \label{eq:vol12}
\end{eqnarray}

e poich\'e $k$ \` e arbitrario poniamo $k=\omega N$ con $\omega\in(0,1)$. Per $N\rightarrow\infty$ i punti $x'$, $y'$ tendono a $x$, $y$, i quali si trovano sulla stessa variet\` a stabile e instabile dal momento che $x=X(\underline{\sigma}')$, $y=X(\underline{\sigma}'')$ con $\sigma'_{i}=\sigma''_{i}$ se $|i|>0$. Quindi, grazie alla H\"older continuit\` a di $\lambda_{u}(x)$, $\lambda_{s}(x)$ i rapporti di prodotti nelle (\ref{eq:vol11}), (\ref{eq:vol12}) tenderanno ad un limite finito; l'ultima cosa che rimane da discutere \` e il comportamento di $\frac{|S^{\omega N}\delta_{u}^{x}|}{|S^{\omega N}\delta_{u}^{y}|}$, $\frac{|S^{-\omega N}\delta_{s}^{x}|}{|S^{-\omega N}\delta_{s}^{y}|}$ per $N\rightarrow\infty$. I due vettori $S^{\omega N}\delta_{u}^{x}$, $S^{\omega N}\delta_{u}^{y}$ diventano ``molto vicini'' per $N$ grandi, nel senso che la distanza tra $x$ e $y$ \` e $O(\lambda^{\omega N})$ (inizialmente \` e $O(1)$), e inoltre $|S^{\omega N}\delta_{u}^{x,y}|=O(\lambda^{(1-\omega)N})$; quindi, poich\'e l'angolo $\varphi(x)$ H\"older continuo in $x$,

\begin{equation}
\lim_{N\rightarrow\infty}\frac{|S^{\omega N}\delta_{u}^{x}|}{|S^{\omega N}\delta_{u}^{y}|}=1. \label{eq:vol13}
\end{equation}

 \begin{figure}[htbp]
\centering
\includegraphics[width=0.7\textwidth]{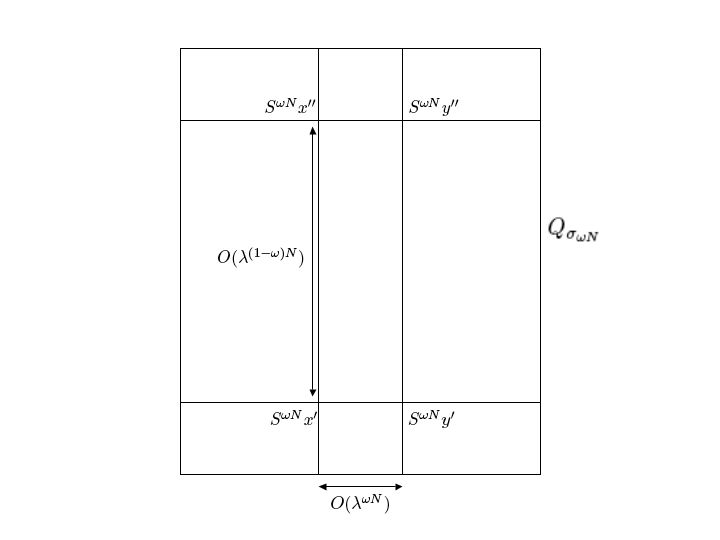}
\caption{Avvicinamento dei segmenti $S^{\omega N}\delta_{u}^{x}$, $S^{\omega N}\delta_{u}^{y}$.} \label{fig:vol3}
\end{figure}

Un argomento dello stesso tipo pu\` o essere ripetuto per $\delta_{s}^{x}$, $\delta_{s}^{y}$, quindi il risultato finale \` e che

\begin{equation}
\lim_{N\rightarrow\infty}\frac{\delta_{u}^{x}}{\delta_{u}^{y}}=\prod_{k=1}^{\infty}\frac{\lambda_{u}^{-1}(S^{k}x)}{\lambda_{u}^{-1}(S^{k}y)},\qquad\lim_{N\rightarrow\infty}\frac{\delta_{s}^{x}}{\delta_{s}^{y}}=\prod_{k=1}^{\infty}\frac{\lambda_{s}(S^{-k}x)}{\lambda_{s}(S^{-k}y)}, \label{eq:vol14}
\end{equation}

e la (\ref{eq:vol6}) \` e dimostrata ponendo $A_{u}(\underline{\sigma})=\log\lambda_{u}(X(\underline{\sigma}))$, $A_{s}(\underline{\sigma})=-\log\lambda_{s}(X(\underline{\sigma}))$.

\end{proof}

Quindi, come abbiamo gi\` a notato con l'osservazione \ref{oss:vol}, possiamo immaginare la misura di volume codificata in dinamica simbolica come la misura di probabilit\` a di un modello di Ising unidimensionale interagente attraverso potenziali a molti corpi {\em diversi} a seconda che ci si trovi a destra o a sinistra dal centro del volume nel quale \` e confinato il modello. \` E abbastanza naturale chiedersi se nel limite in cui ci si sposta {\em infinitamente lontano} verso destra o verso sinistra, {\em i.e.} nel passato o nel futuro del sistema di Anosov, la misura di probabilit\` a del corrispondente modello di spin sia determinata da potenziali d'interazione di {\em un'unico tipo}; nel caso ci\` o fosse vero, la misura non potrebbe che coincidere con la codifica in simboli della misura di probabilit\` a invariante del sistema dinamico.

Rispondere a questa domanda equivale a risolvere il problema dell'esistenza del {\em limite termodinamico} per un modello di Ising unidimensionale con potenziali a decadimento esponenziale; ma dalla teoria dei sistemi di spin sappiamo che questo limite esiste e che {\em non dipende dalle condizioni al bordo}, dunque nel nostro caso dai potenziali localizzati nel passato e nell'origine. Riprenderemo il problema in termini pi\` u precisi nella prossima sezione.

\section{Una misura invariante per i sistemi di Anosov}\label{cap:gibbs}

\subsection{Misure di Gibbs}

Prima di tornare a discutere la questione dell'esistenza di una misura di probabilit\` a invariante per un sistema di Anosov, diamo una definizione generale di {\em misura di Gibbs}; come vedremo, sia la misura invariante di un sistema di Anosov che la codifica in simboli della misura di volume apparterranno a questa classe di misure di probabilit\` a.  Prima di fare ci\` o, specifichiamo le condizioni ``fisiche'' che devono soddisfare i potenziali a molti corpi di una misura di Gibbs.

\begin{definizione}{{\em (Potenziali per la dinamica simbolica)}}\label{def:potsimb}

Sia $T$ una matrice di compatibilit\`a\index{matrice di compatibilit\` a} mescolante per le sequenze in $\{0,...,n\}^{\mathbb{Z}}$. Chiameremo $B$ lo spazio delle sequenze $\Phi=\{\Phi_{X}\}_{X\subset\mathbb{Z}}$ parametrizzate dai sottoinsiemi finiti di $\mathbb{Z}$, formato dalle funzioni

\begin{equation}
\Phi_{X}:\{0,...,n\}_{T}^{X}\rightarrow\mathbb{R}
\end{equation}

che, ponendo $\|\Phi_{X}\|=\max_{\underline{\sigma}\in\{0,...,n\}_{T}^{X}}|\Phi_{X}(\underline{\sigma})|$, verificano le seguenti propriet\` a.

\begin{enumerate}
\item $\Phi$ \` e invariante per traslazioni su $\mathbb{Z}$, ovvero $\forall\;\sigma_{1},...,\sigma_{p}$

\begin{equation}
\Phi_{X}(\sigma_{1},...,\sigma_{p})=\Phi_{\tau X}(\sigma_{1},...,\sigma_{p})
\end{equation}

se $X=(\xi_{1},...,\xi_{p})$ e $\tau X=(\xi_{1}+1,...,\xi_{p}+1)$.

\item $\Phi$ \` e ``sommabile'':

\begin{equation}
\|\Phi\|\equiv\sum_{X\ni 0}\frac{\|\Phi_{X}\|}{|X|}<+\infty.
\end{equation}

\end{enumerate}

Diremo che $B$ \` e lo spazio dei {\em potenziali} su $\{0,...,n\}^{\mathbb{Z}}_{T}$. La funzione

\begin{equation}
A_{\Phi}(\underline{\sigma})=\sum_{X\ni 0}\frac{\Phi_{X}(\underline{\sigma}_{X})}{|X|}
\end{equation}

sar\` a chiamata {\em energia potenziale per sito} oppure {\em funzione energia} associata a $\Phi$.

\end{definizione}

\begin{oss}
Possiamo prendere in prestito dalla meccanica statistica classica un semplice esempio di funzione energia, considerando ad esempio un sistema di spin interagenti su un reticolo bidimensionale con energia potenziale totale data da

\begin{equation}
\sum_{(x,y)\in\mathbb{Z}^{2}}\sigma_{x}\sigma_{y}\Phi(x,y)=\sum_{x\in\mathbb{Z}}\sum_{y\in\mathbb{Z}, y\neq x}\frac{\sigma_{x}\sigma_{y}\Phi(x,y)}{2};
\end{equation}

l'argomento della prima sommatoria \` e l'energia potenziale per sito, e $\sigma_{x}\sigma_{y}\Phi(x,y)\equiv \Phi_{x,y}(\sigma,\sigma')$ \` e un {\em potenziale a due corpi}.
\end{oss}

In questo contesto, possiamo dare una definizione generale di {\em misura}, o {\em stato}, {\em di Gibbs}.

\begin{definizione}{{\em (Misura di Gibbs)}}

La misura $m$ \` e una misura di Gibbs su $\{0,...,n\}^{\mathbb{Z}}_{T}$ con potenziali $\Phi=\left\{\Phi_{X}\right\}_{X\in\mathbb{Z}}\in B$ se \` e una misura di probabilit\` a sui boreliani di $\{0,...,n\}^{\mathbb{Z}}_{T}$ le cui probabilit\` a condizionate verificano $m$-quasi ovunque la propriet\` a

\begin{equation}
\frac{m(\underline{\sigma}_{\Lambda}'|\underline{\sigma}_{\Lambda^{c}})}{m(\underline{\sigma}_{\Lambda}''|\underline{\sigma}_{\Lambda^{c}})}=\exp\left(-\sum_{k=-\infty}^{\infty}\left[A_{\Phi}(\tau^{k}\underline{\sigma}')-A_{\Phi}(\tau^{k}\underline{\sigma}'')\right]\right), \label{eq:gibbs1}
\end{equation}

per ogni intervallo connesso $\Lambda\subset\mathbb{Z}$ tale che $|\Lambda|>a(T)$, se $a(T)$ \` e il tempo di mescolamento di $T$ ({\em cf.} definizione \ref{def:matrcomp}), e per $\underline{\sigma}'=(\underline{\sigma}'_{\Lambda}\underline{\sigma}_{\Lambda^{c}}), \underline{\sigma}''=(\underline{\sigma}''_{\Lambda}\underline{\sigma}_{\Lambda^{c}})\in\{0,...,n\}^{\mathbb{Z}}_{T}$.

La propriet\` a (\ref{eq:gibbs1}), insieme alla condizione di compatibilit\` a $\sum_{\underline{\sigma}_{\Lambda}}m(\underline{\sigma}_{\Lambda}|\underline{\sigma}_{\Lambda^{c}})=1$, implica che

\begin{equation}
m(\underline{\sigma}_{\Lambda}|\underline{\sigma}_{\Lambda^{c}})=\frac{e^{-\sum_{R\cap\Lambda\neq\emptyset}\Phi_{R}(\underline{\sigma}_{R})}}{\sum_{\underline{\sigma}'_{\Lambda}}e^{-\sum_{R\cap\Lambda\neq\emptyset}\Phi_{R}(\underline{\sigma}'_{R})}},\label{eq:gibbs2}
\end{equation}

interpretando la (\ref{eq:gibbs2}) come zero se $\underline{\sigma}\notin\{0,...,n\}^{\mathbb{Z}}_{T}$.

\end{definizione}

\begin{oss}
\begin{enumerate}
\item Il paragone del risultato (\ref{eq:vol6}) ottenuto per la misura di volume con la (\ref{eq:gibbs1}) ci permette di dire con cognizione di causa che la misura di volume \` e uno stato di Gibbs.
\item In altri approcci le misure di Gibbs vengono introdotte diversamente, ad esempio come soluzione di un problema variazionale; in quei casi la propriet\` a (\ref{eq:gibbs1}), che abbiamo assunto nel nostro caso essere la propriet\` a che definisce una misura di Gibbs, diventa un teorema. La (\ref{eq:gibbs1}) forma un insieme di equazioni per le probabilit\` a condizionate che vengono chiamate {\em equazioni di Dobrushin-Lanford-Ruelle}, o pi\` u semplicemente {\em equazioni DLR}\index{equazioni DLR}.
\end{enumerate}
\end{oss}

\subsection{La misura SRB}\label{sez:srb}

 Concludiamo la nostra introduzione alla teoria dei sistemi di Anosov discutendo il problema dell'esistenza di misure di probabilit\` a invarianti, anche dette {\em statistiche del moto}, per questa classe di sistemi dinamici; in generale $\det S\neq 1$, dunque la misura di volume non sar\` a invariante sotto l'azione dell'evoluzione temporale. Ciononostante, attribuiremo una particolare importanza alle statistiche del moto i cui dati iniziali siano scelti a caso con una misura di probabilit\` a assolutamente continua rispetto al volume.
 
 \begin{definizione}{{\em (Misura SRB)}}
 
 Sia $(\Omega,S)$ un sistema di Anosov. Se $\mu_{0}$ \` e una misura di probabilit\` a assolutamente continua rispetto alla misura di volume su $\Omega$ e per tutte le funzioni continue $F$ definite in $\Omega$ il limite $\lim_{N\rightarrow\infty}N^{-1}\sum_{j=0}^{N-1}F(S^{j}x)$ esiste $\mu_{0}$-quasi ovunque e pu\` o essere scritto come
 
 \begin{equation}
 \lim_{N\rightarrow\infty}N^{-1}\sum_{j=0}^{N-1}F(S^{j}x)=\int_{\Omega}\mu(dy)F(y),
 \end{equation}
 
 allora $\mu$ \` e chiamata la {\em misura SRB}\index{misura SRB} di $(\Omega,S)$.
 \end{definizione} 

\begin{oss}
La propriet\` a di invarianza sotto $S$ della misura SRB discende banalmente dalla sua definizione, infatti

\begin{eqnarray}
\int_{\Omega}\mu(dy)F(Sy)&=&\lim_{N\rightarrow\infty}N^{-1}\sum_{j=0}^{N-1}F(S^{j+1}x)\nonumber\\&=&\lim_{N\rightarrow\infty}N^{-1}\sum_{j=0}^{N-1}F(S^{j}x)=\int_{\Omega}\mu(dy)F(y), \label{eq:srb1}
\end{eqnarray}

se $F$ \` e limitata, e ci\` o \` e implicato dalla continuit\` a di $F$ e dalla compattezza di $\Omega$.

\end{oss}

\begin{prop}{{\em (Esistenza della misura SRB)}}\label{prop:srb2}

Sia $(\Omega,S)$ un sistema di Anosov. Esistono due misure di probabilit\` a $\mu_{+}$, $\mu_{-}$ tali che, per ogni $F\in C^{\alpha}(\Omega)$, $\alpha>0$,

\begin{equation}
\lim_{N\rightarrow\infty}N^{-1}\sum_{j=0}^{N-1}F(S^{\pm j}x)=\int\mu_{\pm}(dy)F(y) \label{eq:srb2}
\end{equation}

quasi ovunque in $x$ rispetto alla misura di volume $\mu_{0}$ su $\Omega$.

\end{prop}

\begin{oss}
\begin{enumerate}
\item Notare che, a differenza del teorema di Birkhoff, {\em cf.} proposizione \ref{th:birk}, la misura di probabilit\` a con la quale vengono scelti i dati iniziali {\em non} \` e invariante.
\item In generale, la {\em misura SRB nel futuro} $\mu_{+}$ e la {\em misura SRB nel passato} $\mu_{-}$ sono {\em diverse}, ovvero per una generica funzione H\"older continua $F$ $\mu_{+}(F)\neq\mu_{-}(F)$.
\end{enumerate}
\end{oss} 

\begin{proof}
Per una dimostrazione di questo teorema rimandiamo a \cite{ergo}; la dimostrazione non \` e particolarmente difficile, ma si basa su alcuni risultati che, dato il carattere introduttivo di questi primi due capitoli, non abbiamo discusso.
\end{proof}

Dalla dimostrazione della proposizione \ref{prop:srb2} scopriamo che la misura SRB \` e uno stato di Gibbs determinato dai potenziali a molti corpi generati dallo sviluppo del tasso di espansione $A_{u}(\underline{\sigma})$, per la statistica nel futuro, o da $A_{s}(\underline{\sigma})$, per la statistica nel passato; ci\` o equivale a dire che, ad esempio nel futuro, le probabilit\` a condizionate soddisfano le equazioni DLR

\begin{equation}
\frac{m(\underline{\sigma}'_{\Lambda}|\underline{\sigma}_{\Lambda^{c}})}{m(\underline{\sigma}''_{\Lambda}|\underline{\sigma}_{\Lambda^{c}})}=\exp{\left(\sum_{k=-\infty}^{\infty}A_{u}(\tau^{k}\underline{\sigma}'')-A_{u}(\tau^{k}\underline{\sigma}')\right)}, \label{eq:srb3}
\end{equation}

dove $\underline{\sigma}'=(\underline{\sigma}'_{\Lambda}\underline{\sigma}_{\Lambda^{c}})$, $\underline{\sigma}''=(\underline{\sigma}''_{\Lambda}\underline{\sigma}_{\Lambda^{c}})$. Questo fatto pu\` o essere capito facilmente in modo euristico, {\em cf.} appendice \ref{app:eur}.

\` E piuttosto curioso notare come il problema apparentemente senza speranza della statistica dei moti di un sistema estremamente caotico, in un certo senso il {\em paradigma} di un sistema caotico, possa essere risolto attraverso lo studio di un modello di Ising unidimensionale, uno dei pi\` u semplici modelli noti in meccanica statistica; in particolare, nelle appendici \ref{app:D}, \ref{app:B} dimostreremo alcune propriet\` a di analiticit\` a della misura SRB attraverso una tecnica ben nota in meccanica statistica dei sistemi di spin.

\section{Un'applicazione: il gatto di Arnold perturbato}\label{sez:cat}

Come abbiamo gi\` a visto nella sezione \ref{sez:pavan}, {\em gatto di Arnold}\index{gatto di Arnold} \` e un semplice esempio di sistema di Anosov; l'evoluzione di questo sistema avviene sul toro $\mathbb{T}^{2}$ con la mappa

\begin{equation}
S=\left(\begin{array}{cc} 1 & 1 \\ 1 & 2 \end{array}\right). \label{eq:cat0}
\end{equation}

La misura SRB di questo sistema di Anosov \` e banalmente la misura di volume, dal momento che $|\det S|=1$. Come vedremo, perturbando ``poco'' la mappa $S$ il sistema perturbato sar\` a ancora un sistema di Anosov; in particolare, potremo calcolare esplicitamente i potenziali della nuova misura SRB. Tutta la nostra analisi sar\` a resa possibile dal seguente teorema, dovuto ad Anosov.

\begin{prop}{{\em (Teorema di stabilit\` a strutturale)}}\index{teorema di stabilit\` a strutturale}

Siano $(\Omega,S)$ e $(\Omega,S')$ due sistemi di Anosov. Se le mappe $S$, $S'$ sono abbastanza vicine insieme alle loro derivate prime allora esiste un omeomorfismo $H:\Omega\leftrightarrow\Omega$ tale che

\begin{equation}
S\circ H=H\circ S'. \label{eq:cat1}
\end{equation}

\end{prop}

\begin{oss}
In altre parole, ci\` o vuol dire che \` e possibile mappare una dinamica nell'altra attraverso un ``cambiamento di coordinate'', {\em cf. figura \ref{fig:cat0}}, generalmente non differenziabile.
\end{oss}

Tra poco costruiremo esplicitamente $H$ nel caso del gatto di Arnold perturbato, e dimostreremo che $H$ \` e analitica nella perturbazione e solo H\"older continuo nel punto. Per una ``idea'' della dimostrazione generale di questo teorema, ovvero che non si basi su un caso specifico, rimandiamo a \cite{AA}. Nel nostro caso, seguiremo l'analisi svolta in \cite{ergo}, \cite{BFG03}.

\begin{figure}[htbp]
\centering
\includegraphics[width=0.7\textwidth]{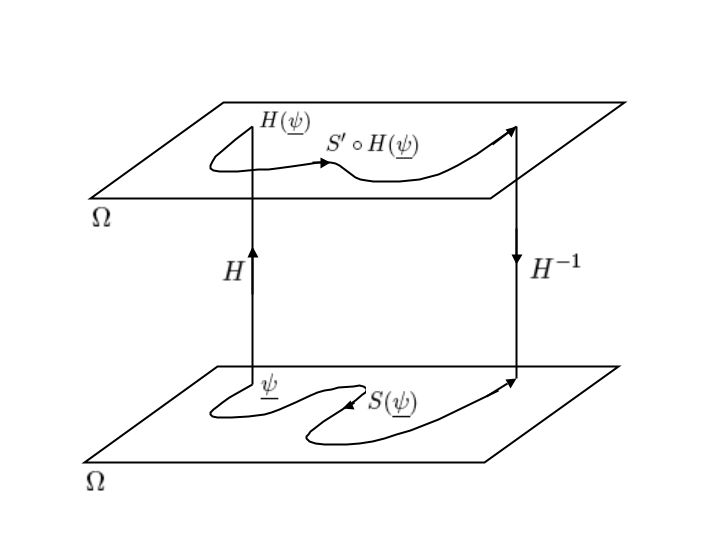}
\caption{Rappresentazione schematica dell'azione della coniugazione $H$ (``diagramma commutativo''\index{diagramma commutativo}). Per semplicit\` a, nella figura abbiamo disegnato una traiettoria continua.} \label{fig:cat0}
\end{figure}

\subsection{Costruzione della coniugazione $H$}

Ci proponiamo di calcolare esplicitamente la coniugazione\index{coniugazione} $H$ partendo dalla propriet\` a (\ref{eq:cat1}), nel caso in cui l'evoluzione del sistema perturbato sia descritta da

\begin{equation}
S_{\varepsilon}\underline{\psi}=S_{0}\underline{\psi}+\varepsilon\underline{f}(\underline{\psi}) \label{eq:cat2}
\end{equation}

dove $S_{0}$ \` e data dalla (\ref{eq:cat0}) e $\underline{f}(\underline{\psi})$ \` e un polinomio trigonometrico\index{polinomio trigonometrico}, ovvero $\underline{f}(\underline{\psi})=\sum_{\underline{\nu}\in\mathbb{Z}^{2},|\nu|\leq N}e^{i\underline{\nu}\cdot\underline{\psi}}\underline{f}_{\underline{\nu}}$; assumeremo che $f_{\underline{\nu}}(\underline{\psi})$ sia una funzione analitica in $\underline{\psi}$. Per fare ci\` o, cercheremo una soluzione dell'equazione

\begin{equation}
H\circ S=S_{\varepsilon}\circ H
\end{equation}

che sia esprimibile come sviluppo in serie di potenze in $\varepsilon$, ovvero $H(\underline{\psi})=\underline{\psi}+\sum_{k\geq 1}\varepsilon^{k}\underline{h}(\underline{\psi})$, e dimostreremo che la serie \`e convergente per $|\varepsilon|<\varepsilon_{0}>0$; quindi, in questo intervallo la coniugazione $H$ sar\` a una funzione analitica di $\varepsilon$. Inoltre, 
con la stessa tecnica otterremo anche l'H\"older continuit\` a in $\underline{\psi}$ con esponente $\beta$ per $\varepsilon<\varepsilon(\beta)\leq\varepsilon(0)=\varepsilon_{0}$. Scriviamo la (\ref{eq:cat2}) come

\begin{equation}
\underline{h}\circ S_{0}=S_{0}\circ \underline{h}+\varepsilon\underline{f}\circ (Id+\underline{h}); \label{eq:cat3}
\end{equation}

al primo ordine in $\varepsilon$ la (\ref{eq:cat3}) diventa

\begin{equation}
S_{0}\circ \underline{h}^{(1)}+\underline{f}=\underline{h}^{(1)}\circ S_{0}. \label{eq:cat4}
\end{equation}

Chiamiamo $\underline{v}_{+}$, $\underline{v}_{-}$ i due autovettori normalizzati di $S_{0}$ corrispondenti rispettivamente all'autovalore $\lambda_{+}$ e all'autovalore $\lambda_{-}$, e definiamo $\lambda\equiv\lambda_{+}^{-1}=\lambda_{-}$; poich\'e

\begin{eqnarray}
\underline{f}(\underline{\psi})&=&f_{+}(\underline{\psi})\underline{v}_{+}+f_{-}(\underline{\psi}) \underline{v}_{-} \nonumber\\
\underline{h}(\underline{\psi})&=&A_{u}(\underline{\psi})\underline{v}_{+}+h_{-}(\underline{\psi})\underline{v}_{-}, \label{eq:cat5}
\end{eqnarray}

le due componenti della (\ref{eq:cat4}) nelle direzioni $\underline{v}_{+}$, $\underline{v}_{-}$ sono

\begin{eqnarray}
\lambda_{+}A_{u}^{(1)}(\underline{\psi})-A_{u}^{(1)}(S_{0}\underline{\psi})&=&-f_{+}(\underline{\psi})\\
\lambda_{-}h_{-}^{(1)}(\underline{\psi})-h_{-}^{(1)}(S_{0}\underline{\psi})&=&-f_{-}(\underline{\psi}). \label{eq:cat6}
\end{eqnarray}

Possiamo risolvere le equazioni (\ref{eq:cat6}) iterativamente, e il risultato \` e

\begin{equation}
h_{\alpha}^{(1)}(\underline{\psi})=-\alpha\sum_{p\in\mathbb{Z}_{\alpha}}\lambda_{\alpha}^{-|p+1|\alpha}f_{\alpha}(S_{0}^{p}\underline{\psi})\qquad\alpha=\pm, \label{eq:cat7}
\end{equation}

dove $\mathbb{Z}_{+}=[0,\infty)\cap\mathbb{Z}$ e $\mathbb{Z}_{-}=(-\infty,0)\cap\mathbb{Z}$. In generale, all'ordine $k$-esimo si avr\` a che

\begin{eqnarray}
h_{\alpha}^{(k)}(\underline{\psi})=\sum_{s=0}^{\infty}\frac{1}{s!}\sum_{\begin{array}{c} k_{1}+...+k_{s}=k-1,k_{i}\geq 1\nonumber\\\alpha_{1},...,\alpha_{s}=\pm \end{array}}\sum_{p\in\mathbb{Z}_{\alpha}}\alpha\lambda_{\alpha}^{-|p+1|\alpha}\\\times\left(\prod_{j=1}^{s}\underline{v}_{\alpha_{j}}\cdot\partial_{\underline{\varphi}}\right)f_{\alpha}(S_{0}^{p}\underline{\psi})\left(\prod_{j=1}^{s}h_{\alpha_{j}}^{(k_{j})}(S_{0}^{p}\underline{\psi})\right), \label{eq:cat8}
\end{eqnarray}

dunque \` e chiaro che tutti gli ordini dello sviluppo di  $\underline{h}$ possono essere costruiti in funzione di $\underline{h}^{(1)}$; inoltre, l'espressione (\ref{eq:cat8}) pu\` o essere rappresentata in termine di {\em alberi}, {\em cf.} figura \ref{fig:cat1}.

\begin{figure}[htbp]
\centering
\includegraphics[width=0.7\textwidth]{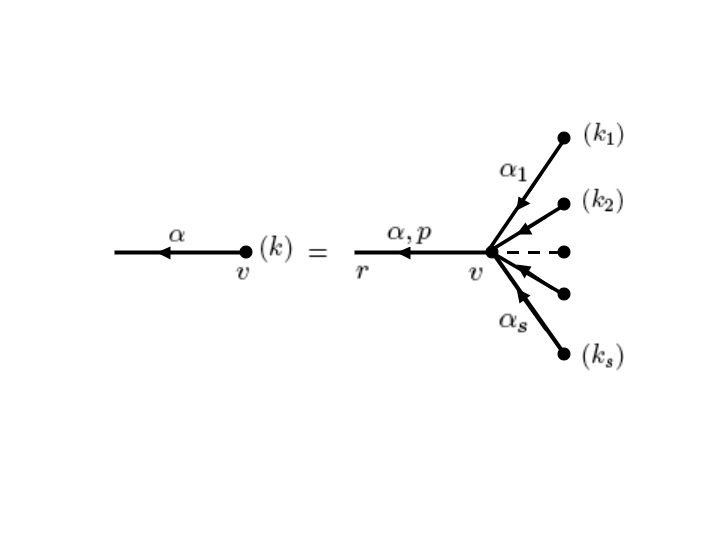}
\caption{Rappresentazione grafica della (\ref{eq:cat8}).} \label{fig:cat1}
\end{figure}

La parte sinistra della figura \ref{fig:cat1} corrisponde a $h_{\alpha}^{(k)}(\underline{\psi})$, i rami della parte destra sono $h^{(k_{1})}_{\alpha_{1}}$, $h^{(k_{2})}_{\alpha_{2}}$ ecc., mentre il vertice $v$ equivale a 

\begin{equation}
\frac{1}{s!}\alpha\lambda_{\alpha}^{-|p+1|\alpha}\left(\prod_{j=1}^{s}\underline{v}_{\alpha_{j}}\cdot\partial_{\underline{\varphi}}\right)f_{\alpha}(S_{0}^{p}\underline{\psi}); \label{eq:cat9}
\end{equation}

 le somme su $k_{i}$, $\alpha$, $p$, $s$ sono sottintese. Un albero\index{albero} $\vartheta$ con $k$ nodi porta su un generico ramo $l$ una coppia di indici $\alpha_{l}$, $p_{l}$ con $p_{l}\in\mathbb{Z}$ e $\alpha_{l}\in\{-,+\}$, e sui nodi $v$ due simboli $\alpha_{v}$, $p_{v}$ con $\alpha_{v}=\alpha_{l_{v}}$ e $p_{v}\in\mathbb{Z}_{\alpha}$ tali che
 
 \begin{equation}
 p(v)\equiv p_{l_{v}}=\sum_{w\succeq v}p_{w}, \label{eq:cat9a}
 \end{equation}

dove la somma \`e sui nodi che seguono $v$ (l'ordinamento \` e specificato dalle frecce), e $l_{v}$ denota il ramo $v'v$ uscente dal nodo $v$; ad ogni albero assegniamo un valore

\begin{equation}
\mbox{Val}(\vartheta)=\prod_{v\in V(\vartheta)}\frac{\alpha}{s_{v}!}\lambda_{\alpha_{v}}^{-|p_{v}+1|\alpha_{v}}\left(\prod_{j=1}^{s_{v}}\partial_{\alpha_{v_{j}}}\right)f_{\alpha_{v}}(S_{0}^{p(v)}\underline{\psi}), \label{eq:cat10}
\end{equation} 

dove $\partial_{\alpha}\equiv \underline{v}_{\alpha}\cdot\partial_{\underline{\varphi}}$, $V(\vartheta)$ \` e l'insieme di nodi in $\vartheta$, e $v_{1},...,v_{s_{v}}$ sono gli $s_{v}$ nodi che precedono\footnote{Nel caso in cui $v$ sia un nodo iniziale, {\em i.e.} in cui corrisponda ad una ``foglia'' dell'albero, nella (\ref{eq:cat10}) non c'\` e nessuna derivata} $v$. Se con $\Theta_{k,\alpha}$ indichiamo l'insieme di tutti gli alberi con $k$ nodi e con etichetta $\alpha$ associata alla {\em radice} dell'albero, ovvero al ramo uscente dall'ultimo vertice, allora la quantit\` a $h_{\alpha}(\underline{\psi})$ pu\` o essere rappresentata come

\begin{equation}
h_{\alpha}(\underline{\psi})=\sum_{k=1}^{\infty}\varepsilon^{k}\sum_{\vartheta\in\Theta_{k,\alpha}}\mbox{Val}(\vartheta), \label{eq:cat11}
\end{equation} 

e l'ultima cosa che rimane da fare \` e stimare il raggio di convergenza della serie di potenze. Per fare ci\` o useremo la {\em stima di Cauchy}\index{stima di Cauchy}, grazie alla quale possiamo maggiorare le derivate nella (\ref{eq:cat10}) come, chiamando $G$ il massimo di $f_{\alpha}$ nella regione $\mathcal{D}=\left\{\psi_{i}:\left|\mbox{Im}\psi_{i}\right|<r_{0},i=1,2\right\}$ per un certo $r_{0}>0$ e $m_{+}$, $m_{-}$ le molteplicit\` a delle derivate $\partial_{+}$, $\partial_{-}$

\begin{equation}
\prod_{j=1}^{s_{v}}\partial_{\alpha_{v_{j}}}f_{\alpha_{v}}(S_{0}^{p(v)}\underline{\psi})\leq G\frac{m_{+}!m_{-}!}{r_{0}^{s_{v}}}\leq G\frac{s_{v}!}{r_{0}^{s_{v}}}, \label{eq:cauchy}
\end{equation}

dal momento che $m_{1}+m_{2}=s_{v}$; grazie alla (\ref{eq:cauchy}) abbiamo che\footnote{Per ottenere l'ultima uguaglianza nella (\ref{eq:cat0011}) abbiamo usato che $\sum_{v\in\Theta_{k,\alpha}}s_{k}=k-1$.}

\begin{equation}
\left|\mbox{Val}(\vartheta)\right|\leq \prod_{v\in V(\vartheta)}\lambda_{\alpha_{v}}^{-|p_{v}+1|\alpha_{v}}\frac{G}{r_{0}^{s_{v}}}= \frac{G}{r_{0}^{k-1}}\prod_{v\in V(\vartheta)}\lambda_{\alpha_{v}}^{-|p_{v}+1|\alpha_{v}}, \label{eq:cat0011}
\end{equation}

dunque poich\'e il numero di possibili alberi con $k$ nodi \` e $\leq 2^{2k}$,

\begin{equation}
\left|\sum_{\vartheta\in\Theta_{k,\alpha}}\mbox{Val}(\vartheta)\right|\leq 2^{k}2^{2k}\frac{G^{k}}{r_{0}^{k-1}}\left(\frac{1}{1-\lambda}\right)^{k}. \label{eq:cat011}
\end{equation} 

Quindi, grazie alle (\ref{eq:cat11}), (\ref{eq:cat011}) il raggio di convergenza di $h_{\alpha}$ \` e dato da

\begin{equation}
\varepsilon_{0}=\left[2^{3}\frac{G}{r_{0}}\frac{1}{1-\lambda}\right]^{-1}.
\end{equation}

L'H\"older continuit\` a di $H$ pu\` o essere dimostrata semplicemente notando che\footnote{Per ottenere la (\ref{eq:cauchy2}) abbiamo scritto la differenza nella parte sinistra come l'integrale di una derivata, abbiamo stimato la derivata con la formula di Cauchy e abbiamo usato il fatto che la massima distanza tra due punti sul toro $\mathbb{T}^{2}$ \` e $\sqrt{2}\pi$.}, per $0\leq \beta<1$,

\begin{eqnarray}
\left|\prod_{j=1}^{s_{v}}\partial_{\alpha_{v_{j}}}f_{\alpha_{v}}(S_{0}^{p(v)}\underline{\psi})-\prod_{j=1}^{s_{v}}\partial_{\alpha_{v_{j}}}f_{\alpha_{v}}(S_{0}^{p(v)}\underline{\psi'})\right|\leq \frac{G(s_{v}+1)!}{r_{0}^{s_{v}+1}}\left(2\pi^{2}\right)^{\frac{1-\beta}{2}}\lambda^{-\beta p}\left|\underline{\psi}-\underline{\psi}'\right|\label{eq:cauchy2};\nonumber\\
\end{eqnarray}

infatti, la (\ref{eq:cauchy2}) implica che

\begin{equation}
\frac{\left|\underline{h}^{(k)}_{\alpha}(\underline{\psi})-\underline{h}^{(k)}_{\alpha}(\underline{\psi}')\right|}{\left|\underline{\psi}-\underline{\psi}'\right|^{\beta}}\leq 2^{2k}2^{k}\left(\frac{1}{1-\lambda^{1-\beta}}\right)\left(2\pi^{2}\right)^{\frac{1-\beta}{2}}\frac{G^{k}}{r_{0}^{k}}(2k-1).
\end{equation}

Quindi, per $|\varepsilon|<\varepsilon(\beta)$, con

\begin{equation}
\varepsilon(\beta)=\left[2^{3}\frac{G}{r_{0}}\frac{1}{1-\lambda^{1-\beta}}\right]^{-1},
\end{equation}

la coniugazione $H$ \` e H\"older continua in $\underline{\psi}$ con esponente $\beta$ e analitica in $\varepsilon$, dal momento che $\varepsilon(\beta)<\varepsilon(0)=\varepsilon_{0}$. Notare che $H$, in generale, {\em non} \` e lipschitziana, perch\'e $\lim_{\beta\rightarrow 1}\varepsilon(\beta)=0$.

Infine, la mappa $H$ \` e un omeomorfismo. Infatti, se $\psi_{1}\neq\psi_{2}$ allora $H(\psi_{1})\neq H(\psi_{2})$, perch\'e ci\` o implicherebbe per ogni $k$ $S_{\varepsilon}^{k}(H(\psi_{1}))=S_{\varepsilon}^{k}(H(\psi_{2}))\Rightarrow H(S_{0}^{k}\psi_{1})=H(S_{0}^{k}\psi_{2})$; ma ci\` o \` e assurdo, dal momento che $S_{0}$ \` e mescolante e $H$ \` e H\"older continua. Inoltre per ogni $\varphi$ esiste $\psi$ tale che $H(\psi)=\varphi$; questo \` e facile da visualizzare sul toro $\mathbb{T}^{1}$, dove se $\psi$ ``gira'' attorno al cerchio anche il punto $H(\psi)$ lo far\` a, e dunque per ogni fissato $\psi$ esiste $\varphi$ e una funzione continua $k(\varphi)$ tali che $\psi=\varphi+k(\varphi)$. Per $d$ generico l'argomento \` e lo stesso, solo ripetuto $d$ volte.

\subsection{Costruzione della misura SRB}

Il prossimo passo consiste nel riuscire a costruire esplicitamente i potenziali di Gibbs della misura SRB del sistema perturbato; per fare ci\` o, abbiamo bisogno di calcolare esplicitamente il coefficiente di espansione, se siamo interessati alla statistica nel futuro, oppure il coefficiente di contrazione, se vogliamo descrivere la statistica del sistema nel passato. Ingenuamente, potremmo pensare di parametrizzare la variet\` a stabile e la variet\` a instabile del sistema perturbato attraverso la coniugazione $H$, ovvero con delle equazioni tipo

\begin{equation}
\underline{\varphi}(t)=H(\underline{\psi}+t\underline{v}_{\alpha}),\qquad\alpha=\pm, \label{eq:cat16}
\end{equation}

dal momento che, per $\alpha=\pm$, $t\rightarrow\underline{\psi}+t\underline{v}_{\alpha}$ \` e la variet\` a instabile o stabile di $S_{0}$. Questo tentativo per\` o \` e destinato a rivelarsi infruttuoso, perch\'e la parametrizzazione definita tramite $H$ non pu\` o che essere in generale solo H\"older continua in $t$, e dunque inutile ai fini della costruzione dei coefficienti di contrazione e di espansione, {\em cf.} definizione \ref{def:contr}.

Chiamando $\widetilde{\underline{w}}_{+}(\underline{\psi})$, $\widetilde{\underline{w}}_{-}(\underline{\psi})$ i vettori tangenti alla variet\` a instabile e alla variet\` a stabile del sistema perturbato nel punto $\underline{\psi}$, questi vettori devono soddisfare l'equazione

\begin{equation}
DS_{\varepsilon}(\underline{\psi})\widetilde{\underline{w}}_{\pm}(\underline{\psi})=\widetilde{\lambda}_{\pm}(\underline{\psi})\widetilde{\underline{w}}_{\pm}(S_{\varepsilon}(\underline{\psi})), \label{eq:cat016}
\end{equation}

dove $\widetilde{\lambda}_{\pm}$ \` e in generale una funzione H\"older continua di $\underline{\psi}$, ed \` e ci\` o di cui avremmo bisogno per costruire i potenziali della misura SRB.

Per\` o, per delle ragioni che saranno chiare tra poco, saremo in grado di costruire esplicitamente queste quantit\` a calcolate nel punto $H(\underline{\psi})$ e {\em non} in $\underline{\psi}$; l'equazione (\ref{eq:cat016}) diventa, sostituendo $\underline{\psi}$ con $H(\underline{\psi})$ e definendo $\underline{w}_{\pm}\equiv \widetilde{\underline{w}}_{\pm}\circ H$, $\lambda_{\pm}=\widetilde{\lambda}_{\pm}\circ H$,

\begin{equation}
DS_{\varepsilon}(H(\underline{\psi}))\underline{w}_{\pm}(\underline{\psi})=\lambda_{\pm}(\underline{\psi})\underline{w}_{\pm}(S_{0}\underline{\psi}). \label{eq:cat017}
\end{equation}

Come nel caso della coniugazione $H$, possiamo provare a risolvere la (\ref{eq:cat017}) perturbativamente, ovvero cercando una soluzione del tipo $\lambda_{\pm}(\underline{\psi})=\lambda_{\pm}+\delta\lambda_{\pm}(\underline{\psi})$, $\underline{w}_{\pm}(\underline{\psi})=\underline{v}_{\pm}+\delta\underline{w}_{\pm}(\underline{\psi})$ con

\begin{eqnarray}
\delta\lambda_{\pm}(\underline{\psi})&\equiv&\gamma_{\pm}(\underline{\psi})=\sum_{k\geq 1}\gamma^{(1)}_{\pm}(\underline{\psi})\\
\delta\underline{w}_{\pm}(\underline{\psi})&\equiv&k_{\mp}(\underline{\psi})\underline{v}_{\mp}= \sum_{k\geq 1}k^{(k)}_{\mp}(\underline{\psi})\underline{v}_{\mp},
\end{eqnarray}

$\gamma_{\pm}^{(k)},k_{\mp}^{(k)}=O\left(\varepsilon^{k}\right)$. Scrivendo la (\ref{eq:cat017}) nella base degli autovettori del gatto di Arnold $\underline{v}_{\pm}$, ponendo $DS_{\varepsilon}(\underline{\psi})=S_{0}+\varepsilon D\underline{f}(\underline{\psi})$, la (\ref{eq:cat017}) implica che:

\begin{eqnarray}
\gamma_{\alpha}(\underline{\psi})&=&\varepsilon\partial_{\alpha}f_{\alpha}(\underline{\varphi})+\varepsilon k_{\beta}(\underline{\psi})\partial_{\beta}f_{\alpha}(\underline{\varphi})\nonumber\\
k_{\alpha}(\underline{\psi})&+&\lambda^{2}k_{\alpha}(S_{0}^{\alpha}\underline{\psi})=\nonumber\\&&-\lambda\left(\varepsilon\partial_{\beta}f_{\alpha}(\underline{\varphi}^{\alpha})+\varepsilon k_{\alpha}(\underline{\psi}^{\alpha})\partial_{\alpha}f_{\alpha}(\underline{\varphi}^{\alpha})-k_{\alpha}(S_{0}\underline{\psi}^{\alpha})\gamma_{\beta}(\underline{\psi}^{\alpha})\right), \label{eq:cat26}
\end{eqnarray}

dove $\alpha=\pm$ e $\beta=-\alpha$. Le equazioni (\ref{eq:cat26}) possono essere risolte ad ogni ordine in $\varepsilon$ in funzione di tutti gli ordini precedenti, e dunque alla fine del primo ordine, dato da\footnote{Notare che se al posto del gatto di Arnold (\ref{eq:cat0}) consideriamo la sua ``radice quadrata''

\begin{equation}
\left(\begin{array}{cc} 1 & 1 \\ 1 & 0 \end{array}\right)
\end{equation}

i fattori di convergenza $\lambda_{+}^{-(2n+1)}$ saranno sostituiti da $(-1)^{n}\lambda_{+}^{-(2n+1)}$, dal momento che in questo caso $\lambda_{-}=-\lambda_{+}^{-1}$.

}

\begin{eqnarray}
\gamma_{\alpha}^{(1)}&=&\partial_{\alpha}f_{\alpha}(\underline{\psi})\\
k_{+}^{(1)}(\underline{\psi})&=&-\sum_{n=0}^{\infty}\lambda_{+}^{-(2n+1)}\partial_{-}f_{+}(S_{0}^{n}\underline{\psi})\\
k_{-}^{(1)}(\underline{\psi})&=&\sum_{n=0}^{\infty}\lambda_{+}^{-(2n+1)}\partial_{+}f_{-}(S_{0}^{-(n+1)}\underline{\psi}). \label{eq:cat270}
\end{eqnarray}

 Con una tecnica analoga a quella usata nel caso della coniugazione $H$ si pu\` o dimostrare che esistono $\varepsilon'_{0}, \varepsilon'(\beta)$ tali che gli sviluppi in $\varepsilon$ di $\gamma_{\pm}$, $k_{\pm}$ convergono per $|\varepsilon|<\varepsilon'_{0}$, e che le soluzioni $\gamma_{\pm}$, $k_{\pm}$ sono H\"older continue con esponente $\beta$ in $|\varepsilon|<\varepsilon'(\beta)<\varepsilon'(0)=\varepsilon'_{0}$, {\em cf.} \cite{ergo}, \cite{BFG03}.
 
 Grazie questi sviluppi perturbativi siamo ora in grado di dire esplicitamente chi sono i coefficienti di espansione e di contrazione  e i vettori tangenti alle variet\` a stabili e instabili {\em nel punto $H(\underline{\psi})$}. Se non avessimo composto sin dall'inizio con la coniugazione $H$ avremmo trovato la (\ref{eq:cat26}) con $S_{0}$ sostituito da $S_{\varepsilon}$; dunque, per dimostrare l'analiticit\` a in $\varepsilon$ avremmo dovuto prima verificare la differenziabilit\` a in $\underline{\psi}$, e viceversa. Ad ogni modo, il fatto di essere in grado di calcolare le quantit\` a necessarie per la costruzione della misura SRB solo nel punto $H(\underline{\psi})$ non rappresenta un problema, perch\' e se $\mathcal{P}_{0}$ \` e una partizione di Markov per il sistema imperturbato $(\Omega,S_{0})$ possiamo introdurre una partizione di Markov $\mathcal{P}_{\varepsilon}$ per $(\Omega,S_{\varepsilon})$ come
 
 \begin{equation}
 \mathcal{P}_{\varepsilon}=H(\mathcal{P}_{0}); \label{eq:cat027}
 \end{equation}
 
 dunque, $H(\underline{\psi})$ pu\` o essere rappresentato con la stessa storia $\underline{\sigma}$ di $\underline{\psi}$, solo codificata con un {\em nuovo codice} $X_{\varepsilon}=H\circ X$. Quindi, possiamo scrivere i valori medi SRB come delle somme sulle storie simboliche del sistema imperturbato codificate con il nuovo codice $X_{\varepsilon}$, pesando ogni addendo con un fattore
 
 \begin{equation}
 \frac{e^{-\sum_{k=-\Lambda/2}^{\Lambda/2-1}A_{u}(\tau^{k}\underline{\sigma})}}{Z_{\Lambda}},
 \end{equation}
 
 dove $A_{u}\equiv A_{u}\circ X_{\varepsilon}$ \` e dato da
 
 \begin{equation}
 A_{u}(\tau\underline{\sigma})=\log\left(\left[\lambda_{+}+\gamma_{+}(X_{\varepsilon}(\underline{\sigma}))\right]\frac{\left|\underline{w}_{+}(X_{\varepsilon}(\tau\underline{\sigma}))\right|}{\left|\underline{w}_{+}(X_{\varepsilon}(\underline{\sigma}))\right|}\right). \label{eq:cat0027}
 \end{equation}
 
 Riassumiamo il contenuto di questa sezione nella seguente proposizione.
 
\begin{prop}{{\em (Potenziali del gatto di Arnold perturbato)}}

Data una partizione di Markov $\mathcal{P}_{0}=\{P_{1},...,P_{n}\}$ in $\mathbb{T}^{2}$ per $S_{0}$, sia $\underline{\sigma}$ la rappresentazione simbolica rispetto a $\mathcal{P}_{\varepsilon}=H(\mathcal{P}_{0})$ di un punto $\underline{\varphi}=X_{\varepsilon}(\underline{\sigma})$. Esiste $\varepsilon'(\beta)$ tale che il coefficiente di espansione $\lambda_{u}(\underline{\sigma})$ di $S_{\varepsilon}(\underline{\sigma})$ lungo la variet\` a instabile di $\underline{\varphi}$ \` e analitico in $\varepsilon$ nel disco $|\varepsilon|<\varepsilon'(\beta)$, e come funzione di $\underline{\sigma}$ \` e H\"older continua con esponente $\beta$ e modulo $C(\beta)$; lo stesso vale per il tasso di espansione

 \begin{equation}
 A_{u}(\tau\underline{\sigma})=\log\left(\left[\lambda_{+}+\gamma_{+}(X_{\varepsilon}(\underline{\sigma}))\right]\frac{\left|\underline{w}_{+}(X_{\varepsilon}(\tau\underline{\sigma}))\right|}{\left|\underline{w}_{+}(X_{\varepsilon}(\underline{\sigma}))\right|}\right), \label{eq:cat27}
 \end{equation}

che genera (nel senso discusso nella proposizione \ref{prop:svilpot}) i potenziali $\Phi_{[-n,n]}(\underline{\sigma}_{[-n,n]})$ dello stato di Gibbs associato alla misura SRB.

\end{prop}

\begin{oss}
\begin{enumerate}
\item Dal momento che tutte le funzioni che appaiono nella (\ref{eq:cat27}) sono H\"older continue i potenziali $\Phi_{X}(\underline{\sigma}_{X})$, con $X=[-n,n]$, decadono esponenzialmente, {\em i.e.} esistono $C,c>0$ tali che, se $\underline{\sigma}_{[-n,n]}$ corrisponde ad una sequenza $\underline{\sigma}$ ``troncata'' nell'intervallo $[-n,n]$ e continuata con opportune condizioni al bordo,

\begin{equation}
|\Phi_{[-n,n]}(\underline{\sigma}_{[-n,n]})|=|A_{u}(\underline{\sigma}_{[-(n+1),(n+1)]})-A_{u}(\underline{\sigma}_{[-n,n]})|\leq Ce^{-cn}; \label{eq:cat28}
\end{equation}

nella teoria dei sistemi di spin unidimensionali, un potenziale di questo tipo viene chiamato un {\em potenziale di Fisher}\index{potenziale di Fisher}.

\item Notare che, nel limite termodinamico, il rapporto $\frac{\left|\underline{w}_{+}(X_{\varepsilon}(\tau\underline{\sigma}))\right|}{\left|\underline{w}_{+}(X_{\varepsilon}(\underline{\sigma}))\right|}$ {\em non modifica lo stato di Gibbs}\index{stato di Gibbs}, perch\'e sia il numeratore che il denominatore sono funzioni H\"older continue; dunque 

\begin{equation}
\sum_{k=-\Lambda/2}^{\Lambda/2-1}\left(|\underline{w}_{+}(X_{\varepsilon}(\tau^{k+1}\underline{\psi}))|-|\underline{w}_{+}(X_{\varepsilon}(\tau^{k}\underline{\sigma})|\right)=O(1), \label{eq:cat29}
\end{equation}

mentre $\sum_{k=-\Lambda/2}^{\Lambda/2-1}A_{u}(\tau^{k}\underline{\sigma})=O(\Lambda)$, perci\` o la (\ref{eq:cat29}) rappresenta un {\em termine di bordo}\index{termine di bordo}. Ci\` o vuol dire che lo stato di Gibbs del sistema \` e {\em indipendente} dalla particolare metrica scelta per misurare i moduli dei vettori, a patto che si tratti di scelte ``ragionevoli'', ovvero che gli elementi di matrice del tensore metrico siano funzioni H\"older continue.

\end{enumerate}
\end{oss}

\chapter{Non equilibrio}\label{cap:noneq}

Il contatto tra la teoria matematica dei sistemi di Anosov che abbiamo appena discusso e la meccanica statistica dei sistemi fuori dall'equilibrio \` e contenuto nell'{\em ipotesi caotica}, che presenteremo nella sezione \ref{sez:IC}; prima di fare ci\` o per\` o, discutiamo brevemente il significato dell'{\em ipotesi ergodica} e le sue conseguenze nella descrizione statistica dei sistemi all'equilibrio, \cite{trattgal}. 

\section{L'ipotesi ergodica}

La {\em meccanica statistica dell'equilibrio} ha come scopo la deduzione delle propriet\` a macroscopiche di insiemi di particelle che evolvono seguendo una {\em dinamica hamiltoniana}; ci\` o vuol dire che se $\mathcal{H}$ \` e la funzione hamiltoniana e $\underline{p}=\left(p_{1},...,p_{3N}\right)$, $\underline{q}=\left(q_{1},...,q_{3N}\right)$ sono rispettivamente gli impulsi e le coordinate spaziali di $N$ particelle, allora il sistema evolver\` a risolvendo le equazioni differenziali

\begin{equation}
\dot{p}_{i}=-\frac{\partial\mathcal{H}}{\partial q_{i}},\qquad \dot{q}_{i}=\frac{\partial\mathcal{H}}{\partial p_{i}}. \label{eq:meceq1}
\end{equation}

Nella realt\` a fisica l'impulso e la posizione non possono essere misurati {\em simultaneamente} con precisione arbitraria; una misura ``molto precisa'' dell'impulso implicher\` a una misura ``poco precisa'' della posizione, nel senso che se indichiamo con $\delta p$, $\delta q$ le incertezze su impulso e posizione allora la massima precisione raggiungibile sar\` a limitata dalla condizione

\begin{equation}
\delta p\delta q =h, \label{eq:meceq2}
\end{equation}

dove $h$ \` e un'opportuna costante. In generale avremo che

\begin{equation}
\delta p\delta q\geq h, \label{eq:meceq2b}
\end{equation}

e, se $h$ coincide con la costante di Planck, l'affermazione (\ref{eq:meceq2b}) \` e il contenuto del {\em principio di indeterminazione di Heisenberg}.

Quindi, possiamo considerare lo spazio delle fasi come decomposto in tante piccole cellette numerate $\Delta_{i}$, con $i=1,...,\mathcal{N}$, ognuna di volume $h$; ogni celletta contiene un {\em continuo} di possibili configurazioni microscopiche $(\underline{p},\underline{q})$, ma {\em tutte} le configurazioni contenute all'interno di ogni $\Delta_{i}$ si riferiscono allo stesso {\em stato fisico}, nel senso che le propriet\` a macroscopiche del sistema non cambiano in modo fisicamente apprezzabile se ci si sposta all'interno di una celletta. Nonostante il numero di possibili configurazioni microscopiche sia {\em infinito}, lo stato {\em fisico} della totalit\` a delle particelle pu\` o variare solo in un numero {\em finito} di modi; ad ogni modo, per un sistema di $10^{23}$ particelle (e anche per sistemi molto pi\` u piccoli) questo numero \` e enormemente pi\` u grande di qualunque scala astronomica.

Supponiamo che l'hamiltoniana totale del sistema non dipenda esplicitamente dal tempo, ovvero $\frac{\partial}{\partial t}\mathcal{H}=0$; ci\` o implica che l'energia totale sia uguale ad una costante, che chiameremo $E$. Dunque, possiamo pensare che la dinamica hamiltoniana (\ref{eq:meceq1}) induca una mappa $S$ sulla superficie dello spazio delle fasi ad energia costante la cui azione consista in una {\em permutazione} delle cellette $\Delta_{i}$; chiaramente, questa mappa rifletter\` a tutte le propriet\` a della dinamica hamiltoniana, in particolare, per il teorema di Liouville, {\em il volume delle cellette non varia sotto $S$}.

Nonostante questa schematizzazione, il problema della deduzione di propriet\` a macroscopiche di un sistema - e dunque della sua {\em termodinamica} - a partire dalla dinamica hamiltoniana rimane ancora inaccessibile; per riuscire nel nostro scopo abbiamo bisogno di sapere, o di assumere, qualcosa in pi\` u sulla natura della dinamica microscopica. L'idea  di Boltzmann, tradotta nel nostro linguaggio, \` e la seguente.

\begin{quote}
Ipotesi ergodica (IE): {\em l'azione della trasformazione di evoluzione $S$, vista come permutazione di celle dello spazio delle fasi sulla superficie con energia costante, \` e una permutazione a ciclo unico;

\begin{equation}
S\Delta_{k}=\Delta_{k+1}\qquad k=1,2,...,\mathcal{N}
\end{equation}

se le celle sono opportunamente numerate (e $\Delta_{\mathcal{N}+1}\equiv\Delta_{1}$}).
\end{quote}

L'ipotesi ergodica rende {\em il pi\` u semplice possibile} l'azione di $S$ sulla superficie dello spazio delle fasi ad energia costante\footnote{Chiaramente, questa ipotesi \` e manifestamente falsa nel caso in cui il sistema ammetta altre costanti del moto oltre all'energia totale; ci\` o implica che la dinamica, vista come permutazione di celle, non pu\` o essere un ciclo unico. In questi casi l'ipotesi va riformulata considerando l'azione di $S$ come permutazione delle cellette sulla superficie dello spazio delle fasi determinata da $\mathcal{H}=E$ e da tutti gli integrali primi della dinamica.}; come vedremo in modo euristico, la conseguenza di questa ipotesi \` e che le medie temporali possono essere calcolate come {\em integrali di volume} sullo spazio delle fasi.

Chiamiamo $\mathcal{C}$ l'insieme di cellette $\Delta$ che appartengono alla superficie ad energia costante, $\mathcal{N}$ il loro numero totale, e definiamo la misura di probabilit\` a $\mu$ come

\begin{equation}
\mu(\Delta)=\left\{\begin{array}{c} 1/\mathcal{N}\quad \mbox{se $\Delta\in\mathcal{C}$}\\ 0 \qquad \mbox{altrimenti} \end{array}\right.;
\end{equation}

chiaramente $\mu$ \` e {\em stazionaria}, ovvero $\mu(S\Delta)=\mu(\Delta)$, dal momento che $S$ permuta celle dello stesso ciclo. Dunque, assumendo (IE) per il teorema di Birkhoff ({\em cf.} proposizione \ref{th:birk}) abbiamo che, scegliendo un ``dato iniziale'' $\overline{\Delta}$,

\begin{equation}
\lim_{T\rightarrow\infty}\frac{1}{T}\sum_{j=0}^{T-1}\mathcal{O}\left(S^{k}\overline{\Delta}\right)=\sum_{k}\mathcal{O}\left(\Delta_{k}\right)\mu(\Delta_{k}); \label{eq:meceq5}
\end{equation}

inoltre, dal momento che $h$ \` e piccolo possiamo approssimare la somma sulle celle con un integrale,

\begin{equation}
\lim_{T\rightarrow\infty}\frac{1}{T}\sum_{j=0}^{T-1}\mathcal{O}\left(S^{k}\overline{\Delta}\right) = \frac{\int_{E=cost}\mathcal{O}\left(\underline{p},\underline{q}\right)d^{3N}pd^{3N}q}{Z}, \label{eq:meceq6}
\end{equation}

dove $Z=\int_{E=cost}d^{3N}pd^{3N}q$ \` e la normalizzazione. Quindi, l'ipotesi ergodica attribuisce una particolare importanza alle misure di probabilit\` a {\em assolutamente continue} rispetto al volume, ovvero del tipo $\mu\left(d^{3N}pd^{3N}q\right)=f(\underline{p},\underline{q})d^{3N}pd^{3N}q$; l'integrale (\ref{eq:meceq6}), ad esempio, pu\` o essere riscritto come un integrale su tutto lo spazio delle fasi $\Omega$ con una misura $\mu\left(d^{3N}pd^{3N}q\right)=\delta\left(\mathcal{H}=E\right)d^{3N}pd^{3N}q$, se $\delta$ \` e la delta di Dirac. 

In generale, per un sistema fisico l'integrale nella (\ref{eq:meceq6}) \` e estremamente difficile da calcolare; ad ogni modo, ci\` o che \` e davvero fisicamente rilevante sono le {\em relazioni} tra valori medi di osservabili macroscopiche piuttosto che i valori medi in s\'e, e molto spesso queste relazioni possono essere dimostrate partendo dalla (\ref{eq:meceq6}) senza dover necessariamente calcolare gli integrali. Quindi, accettando l'ipotesi ergodica la termodinamica di un sistema con un numero arbitrario di gradi di libert\` a diventa il riflesso della dinamica hamiltoniana, o meglio di poche sue propriet\` a.

\section{L'ipotesi caotica} \label{sez:IC}

I {\em sistemi fuori dall'equilibrio} rappresentano una classe di sistemi di particelle per i quali l'ipotesi ergodica {\em non} fornisce le risposte giuste. Come abbiamo discusso nell'introduzione, questi sistemi sono sottoposti a delle forze non conservative, ovvero {\em dissipano calore}. In particolare, poich\' e la dinamica non \` e hamiltoniana la divergenza delle equazioni del moto sar\` a in generale diversa da zero; ci\` o rende il teorema di Liouville inapplicabile, e con esso anche l'ipotesi ergodica. Per avere una comprensione microscopica della fenomenologia di un sistema fuori dall'equilibrio sembra essere indispensabile fare delle {\em nuove assunzioni} sulla dinamica delle particelle, necessariamente diverse dall'ipotesi ergodica, e verificare che queste ipotesi abbiano delle conseguenze osservabili.

Il fatto che il volume di una celletta $\Delta$ non sia conservato dall'evoluzione temporale implica che se il sistema \` e confinato in uno spazio delle fasi limitato la {\em contrazione dello spazio} delle fasi non potr\` a che essere in media minore o uguale di zero; nel caso in cui sia strettamente minore di zero il sistema, nel limite di tempi infiniti, evolver\` a  in un insieme di {\em misura di volume nulla}. Dunque, ammesso che esista la misura di probabilit\` a stazionaria che descrive lo stato del sistema {\em non pu\` o} essere assolutamente continua rispetto al volume. 

Nel capitolo \ref{cap:ergo} abbiamo visto che i sistemi di Anosov ({\em cf.} definizione \ref{def:Anosov}) ammettono una misura di probabilit\` a invariante - chiamata misura SRB - che ha dei potenziali {\em diversi} da quelli della misura di volume, dunque in generale \` e {\em singolare}\footnote{Due misure $\mu$, $\nu$ definite in $\Omega$ sono {\em singolari} se esistono due insiemi disgiunti $A,B\in\Omega$ con $A\cup B=\Omega$ tali che $\mu$ \`e zero su tutti i sottoinsiemi misurabili di $B$ mentre $\nu$ \` e zero su tutti i sottoinsiemi misurabili di $A$.} rispetto a quest'ultima; ci chiediamo se questi sistemi dinamici, oltre ad essere interessanti da un punto di vista matematico, abbiano qualcosa a che vedere con la Fisica. A prima vista non sembrerebbe: un sistema reale, ad esempio un gas composto da $10^{23}$ particelle, \` e probabilmente caotico {\em in qualche senso}, ma \` e decisamente troppo restrittivo pensare che la dinamica sia {\em strettamente} assimilabile a quella di un sistema di Anosov.

Anche l'ipotesi ergodica per\`o, se presa alla lettera, \` e praticamente impossibile da verificare nelle applicazioni, e gli unici sistemi per i quali pu\`o essere dimostrata rigorosamente sono estremamente ``semplici'', tanto da non avere nulla a che vedere con la realt\` a fisica studiata dagli sperimentatori. Ciononostante, le {\em conseguenze} dell'ipotesi ergodica, ovvero le relazioni tra valori medi di osservabili che formano la {\em termodinamica}, sono largamente verificate. Ci\`o suggerisce che non abbia senso, o perlomeno che sia superfluo, verificare la {\em stretta} validit\` a dell'ipotesi ergodica in ogni sistema fisico all'equilibrio che si vuole studiare con le tecniche della meccanica statistica; la cosa che ha senso fare, invece, \` e accettare che {\em al fine di calcolare quantit\` a macroscopiche} il sistema si comporti {\em come se fosse}  ergodico.

Il nostro punto di vista sar\` a esattamente lo stesso; {\em assumeremo} che su scala macroscopica il fatto che il sistema non sia strettamente un sistema di Anosov {\em non abbia nessuna importanza}. Riassumiamo queste idee con la seguente {\em ipotesi caotica}, \cite{GC95a}, \cite{GC95}, \cite{trattgal}, \cite{conv}.

\begin{quote}
Ipotesi caotica (IC): \em{Allo scopo di misurare propriet\` a macroscopiche, un sistema di particelle in uno stato stazionario si comporta come un sistema di Anosov reversibile.}
\end{quote}

In (IC), per sistema {\em reversibile}\index{sistema reversibile} intendiamo un sistema dinamico $(\Omega,S)$ per il quale esiste una mappa (differenziabile) $I$ - un' {\em inversione temporale}\index{inversione temporale} - tale che

\begin{equation}
I \circ S=S^{-1} \circ I\qquad\mbox{con $I^{2}=\pm 1$}; \label{eq;IC2}
\end{equation}

vedremo in seguito che la reversibilit\` a giocher\` a un ruolo importante per la deduzione di conseguenze osservabili dell'ipotesi caotica.

L'ipotesi caotica \` e un'assunzione che prescinde dal fatto che il sistema si trovi o meno all'equilibrio. Un sistema di Anosov {\em conservativo}\index{sistema di Anosov conservativo}, ovvero per il quale la contrazione dello spazio delle fasi \` e identicamente nulla, ammette come misura di probabilit\` a invariante la misura di volume; dunque, l'ipotesi caotica all'equilibrio {\em implica} l'ipotesi ergodica.

La validit\` a di (IC) pu\` o essere decretata solo ed esclusivamente partendo da risultati sperimentali, e per avere qualcosa da verificare sperimentalmente dobbiamo dimostrare che (IC) abbia delle conseguenze osservabili. Il {\em teorema di fluttuazione di Gallavotti - Cohen} \` e un esempio (per il momento uno dei pochi) di conseguenza osservabile dell'ipotesi caotica; come abbiamo accennato nell'introduzione, questo teorema stabilisce una propriet\` a di simmetria sulle fluttuazioni del tasso di produzione di entropia. Ma, per il momento, non \` e chiaro cosa siano l'entropia e il suo tasso di produzione in un sistema fuori dall'equilibrio; discuteremo questo punto nella prossima sezione.

\section{Il problema della definizione di entropia}\label{sez:problentr}

In generale, un sistema fuori dall'equilibrio in uno stato stazionario dissipa calore ad un regime costante; dunque il sistema {\em produce} e {\em cede} entropia, nel senso che il valor medio del {\em tasso di produzione di entropia} \` e strettamente positivo e indipendente dal tempo. Quindi, l'entropia totale del sistema {\em non pu\` o} essere una quantit\` a ben definita.

Nelle prossime due sezioni discuteremo l'identificazione del tasso di produzione di entropia con il {\em tasso di contrazione dello spazio delle fasi}; ci baseremo prima su un particolare modello di sistema fuori dall'equilibrio molto importante nelle applicazioni, \cite{ECM93}, \cite{BGGtest}, \cite{therm}, e in seguito faremo vedere che, se la misura di volume ammette un {\em limite debole} per tempi infiniti, questa identificazione \` e del tutto generale, \cite{An}, \cite{trattgal}.

\subsection{Un modello} \label{sez:term}

Indichiamo le posizioni delle particelle all'interno di un contenitore $C_{0}$ con $\underline{x}=(\underline{x}_{1},...,\underline{x}_{N})$, chiamiamo $V(\underline{x}_{1},...,\underline{x}_{N})=V(\underline{x})$ l'energia potenziale d'interazione e $\underline{F}_{i}(\underline{x};\underline{G})$, $i=1,...,N$, le forze esterne non conservative, dove $\underline{G}=(G_{1},...,G_{s})$ sono i parametri che regolano le intensit\` a delle $\underline{F}_{i}$; dunque, $\underline{F}(\underline{x};0)\equiv 0$. Come abbiamo accennato in precedenza, per evitare che l'energia aumenti indefinitamente, ovvero per garantire l'esistenza di uno stato stazionario, \` e necessario che il sistema sia posto a contatto con un termostato; quindi, se $\left\{-\underline{\vartheta}_{i}\right\}$ sono le forze dovute alla presenza del termostato, $i=1,...,N$, le equazioni del moto delle particelle all'interno di $C_{0}$ possono essere scritte come, indicando con $m_{i}$ le masse delle $N$ particelle,

\begin{equation}
m_{i}\ddot{\underline{x}}_{i}=-\partial_{\underline{x}_{i}}V(\underline{x})+\underline{F}_{i}(\underline{x};\underline{G})-\underline{\vartheta}_{i},\qquad\mbox{$i=1,...,N$}.
\end{equation}

Si pu\` o pensare ad un tipo di termostato come ad un insieme di {\em serbatoi infiniti}, ad esempio contenenti un gas, messi a contatto con il sistema $\Sigma$ e asintoticamente in equilibrio termico;  ogni serbatoio, che chiameremo $C_{a}$, dove $a=1,2....,m$, sar\` a caratterizzato da una temperatura $T_{a}=cost$. Questi termostati, per\` o, sono ovviamente intrattabili da un punto di vista sperimantale, {\em i.e.} numerico. Dunque, possiamo immaginare un'altra classe di termostati pensando che il sistema $C_{0}$ sia a contatto con un numero $m$ di serbatoi {\em finiti} $C_{a}$, e che le forze non conservative agiscano solo su $C_{0}$; sui sistemi $\{\Sigma_{a}\}$, contenuti nei $\{C_{a}\}$, faremo agire delle forze $\{\vartheta_{a}\}$ in modo da imporre il {\em vincolo anolonomo}\index{vincolo anolonomo}

\begin{equation}
K_{a}=\sum_{j}\frac{1}{2}m_{a}\dot{\underline{x}}_{j}^{a}=cost. \label{eq:nequ1} 
\end{equation} 

Definiamo la temperatura del sistema $\Sigma_{a}$ come

\begin{equation}
T_{a}\equiv\frac{2K_{a}}{3N_{a}k_{B}}, \label{eq:nequ2}
\end{equation}

dove $k_{B}$ \` e la costante di Boltzmann e $N_{a}$ \` e il numero di particelle contenute in $C_{a}$; accetando la definizione ($\ref{eq:nequ2}$), le forze $\{\vartheta_{a}\}$ mantengono i sistemi $\Sigma_{a}$ a temperatura costante. Notare che, a rigore, {\em ogni} modello di termostato finito\index{termostato finito} deve essere considerato ``non fisico'', perch\' e in generale \` e impossibile separare il sistema dal resto dell'universo; ma d'altro canto non ci rimane scelta, se vogliamo simulare l'evoluzione di questo modello su un calcolatore. Ad ogni modo, \` e un'opinione diffusa che, nel limite termodinamico, il modo in cui venga estratto calore dal sistema sia {\em inessenziale}, ovvero che lo stato stazionario sia {\em indipendente} dal tipo di termostato utilizzato; il fatto che lo stato fisico di un sistema sia indipendente dal tipo di termostato scelto darebbe luogo ad una ``equivalenza degli ensemble''\index{equivalenza degli ensemble} per la meccanica statistica fuori dall'equilibrio, \cite{equi}.

Se le particelle contenute in $C_{0}$ interagiscono mediante un potenziale $\varphi(\underline{q}_{i}-\underline{q}_{j})$, dove $\underline{q}_{i}$, $\underline{q}_{j}$ rappresentano le coordinate spaziali di due particelle, l'energia potenziale d'interazione del sistema $\Sigma_{0}$ sar\` a data da

\begin{equation}
V_{0}(\underline{q})=\sum_{i<j}\varphi(\underline{q}_{i}-\underline{q}_{j}); \label{eq:neq7}
\end{equation} 

analogamente, indicando le coordinate delle particelle contenute in $C_{a}$ con $\underline{x}_{j}^{a}$, $j=1,...,N_{a}$, si ha che 

\begin{equation}
V_{a}(\underline{x})=\sum_{i<j}^{N_{a}}\varphi_{a}(\underline{x}_{i}^{a}-\underline{x}_{j}^{a}), \label{eq:neq7b}
\end{equation}

mentre l'energia potenziale dovuta all'interazione tra le particelle all'interno di $C_{a}$ e quelle contenute in $C_{0}$ potr\` a essere espressa come

\begin{equation}
W_{a}(\underline{q},\underline{x}^{a})=\sum_{i=1}^{N}\sum_{j<i}w_{a}(\underline{q}_{i}-\underline{x}_{j}^{a}), \label{eq:neq7c}
\end{equation}
 
dove $w_{a}(\underline{q}_{i}-\underline{x}_{j}^{a})$ rappresenta l'interazione tra due particelle contenute rispettivamente in $C_{a}$ e in $C_{0}$ nelle posizioni $\underline{x}_{j}^{a}$, $\underline{q}_{j}$. Infine, ipotizziamo che i sistemi $\{\Sigma_{a}\}$ non interagiscano direttamente tra di loro, ma solo attraverso $C_{0}$. 

Grazie alle (\ref{eq:neq7}), (\ref{eq:neq7b}), (\ref{eq:neq7c}), le equazioni del moto del sistema possono essere scritte come (assumiamo per semplicit\` a che sia in $\Sigma$ che nei $\{\Sigma_{a}\}$ le masse siano tutte uguali):

\begin{eqnarray}
m\ddot{\underline{q}}_{j}&=&-\partial_{\underline{q}_{j}}\left(V_{0}(\underline{q})+\sum_{a=1}^{m}W_{a}(\underline{q},\underline{x}^{a})\right)+\underline{F}_{j}(\underline{q};\underline{G}) \label{eq:neq10a}\\
m_{a}\ddot{\underline{x}}_{j}^{a}&=&-\partial_{\underline{x}_{j}^{a}}\left(V_{a}(\underline{x}^{a})+W_{a}(\underline{q},\underline{x}^{a})\right)-\underline{\vartheta}_{j}^{a}, \label{eq:neq10}
\end{eqnarray}

e, per realizzare la condizione (\ref{eq:nequ2}), imponiamo che le forze $\{\underline{\vartheta}_{j}\}$ verifichino il {\em principio di minimo sforzo di Gauss}\index{principio di Gauss}, \cite{trattgal}. Definiamo la forza $\underline{B}_{i}^{a}$ e lo {\em sforzo} $Z_{a}$\index{sforzo} relativo alle accelerazioni $\{\underline{a}_{i}^{a}\}$ come

\begin{eqnarray}
\underline{B}_{i}^{a}&\equiv&-\partial_{\dot{\underline{x}}_{i}}\left(V_{a}(\underline{x})+W_{a}(\underline{q},\underline{x}^{a})\right)\\ 
Z_{a}&\equiv&\sum_{i}\frac{(\underline{B}_{i}^{a}-m_{a}\underline{a}_{i}^{a})^{2}}{m_{a}} \label{eq:nequ3};\label{eq:nequ4}
\end{eqnarray}

si pu\` o dimostrare che, {\em cf.} \cite{trattgal}, la forza da aggiungere alle $\{\underline{B}_{i}^{a}\}$ che minimizza lo sforzo $Z_{a}$ tra tutte le possibili accelerazioni $\{\underline{a}_{i}^{a}\}$ compatibili con un vincolo $\psi_{a}(\dot{\underline{x}},\underline{x})=0$ \` e data da

\begin{equation}
\underline{\vartheta}_{i}^{a}(\dot{\underline{x}},\underline{x})=\partial_{\dot{\underline{x}}_{i}}\psi_{a}(\dot{\underline{x}},\underline{x})\cdot\left(\frac{\sum_{j}\dot{\underline{x}}_{j}\cdot\partial_{\underline{x}_{j}}\psi_{a}(\dot{\underline{x}},\underline{x})+\frac{1}{m_{a}}\underline{B}_{j}^{a}\cdot\partial_{\dot{\underline{x}}_{j}}\psi_{a}(\dot{\underline{x}},\underline{x})}{\sum_{j}\frac{1}{m_{a}}\left(\partial_{\dot{\underline{x}}_{j}}\psi_{a}(\dot{\underline{x}},\underline{x})\right)^{2}}\right). 
\label{eq:neq2}
\end{equation} 

Quindi, se $\psi_{a}=K_{a}$ la (\ref{eq:neq2}) diventa

\begin{equation}
\underline{\vartheta}_{j}^{a}=\frac{L_{a}-\dot{V}_{a}}{3N_{a}k_{B}T_{a}}\dot{\underline{x}}_{j}^{a}\equiv\alpha^{a}\dot{\underline{x}}_{j}^{a}, \label{eq:neq11}
\end{equation}

e $L_{a}$, $\dot{V}_{a}$ rappresentano, rispettivamente, il lavoro esercitato dalle particelle di $\Sigma$ su quelle appartenenti a $\Sigma_{a}$ e una derivata totale, ovvero

\begin{eqnarray}
L_{a}&=&-\sum_{j=1}^{N_{a}}\partial_{\underline{x}_{j}^{a}}W_{a}(\underline{q},\underline{x}^{a})\cdot\dot{\underline{x}}_{j}^{a} \label{eq:neq12}\\
\dot{V}_{a}&=&\sum_{j=1}^{N_{a}}\partial_{\underline{x}_{j}^{a}}V_{a}(\underline{x}^{a})\cdot\dot{\underline{x}_{j}^{a}}. \label{eq:neq13}
\end{eqnarray}

La {\em quantit\`a di calore prodotta per unit\` a di tempo\index{calore prodotto per unit\` a di tempo}} $\dot{Q}$ da un sistema allo stato stazionario \` e naturalmente definita come il lavoro che le forze del termostato $\{\underline{\vartheta}_{i}\}$ esercitano sul sistema $\Sigma$ per unit\` a di tempo, ovvero

\begin{equation}
\dot{Q}\equiv\sum_{i}\underline{\vartheta}_{i}\cdot\dot{\underline{x}}_{i}; \label{eq:neq3}
\end{equation}

dal momento che

\begin{equation}
\underline{\vartheta}=\sum_{a}\underline{\vartheta}_{a}(\dot{\underline{x}},\underline{x}), \label{eq:neq4}
\end{equation}

la quantit\` a $\dot{Q}$ \`e data da

\begin{equation}
\dot{Q}=\sum_{a}\dot{Q}_{a}. \label{eq:neq5}
\end{equation}

Introduciamo la contrazione dello spazio delle fasi\index{contrazione dello spazio delle fasi} dovuta al termostato $a$-esimo $\Sigma_{a}(\dot{\underline{x}},\underline{x})$ come

\begin{equation}
\sigma_{a}(\dot{\underline{x}},\underline{x})\equiv\sum_{j}\partial_{\dot{\underline{x}}_{j}}\cdot\underline{\vartheta}^{a}_{j}(\dot{\underline{x}},\underline{x}); \label{eq:neq6}
\end{equation}

come vedremo, \` e ragionevole identificare questa quantit\` a con il {\em tasso di produzione di entropia\index{tasso di produzione di entropia}} nell'$a$-esimo termostato. La divergenza (\ref{eq:neq6}) $\sigma_{a}\equiv(3N_{a}-1)\alpha^{a}$ (nella (\ref{eq:neq6}), $\underline{\vartheta}_{\alpha}$ \` e data dalla (\ref{eq:neq11})) pu\` o essere scritta esplicitamente come

\begin{equation}
\sigma_{a}\equiv\left[\frac{L_{a}}{k_{B}T_{a}}-\frac{\dot{V}_{a}}{k_{B}T_{a}}\right]\frac{(3N_{a}-1)}{3N_{a}}, \label{eq:neq14}
\end{equation}

dove $L_{a}$ corrisponde al calore ceduto per unit\` a di tempo da $\Sigma$ a $\Sigma_{a}$, e $\dot{V_{a}}$ \` e una derivata temporale totale; assumendo che $V_{a}$ sia limitata, $\dot{V}$ non contribuisce in valor medio\footnote{Infatti,

\begin{equation}
\lim_{T\rightarrow\infty}\frac{1}{T}\int_{0}^{T} dt\dot{V}=\lim_{T\rightarrow\infty}\frac{V_{T}-V_{0}}{T}=0,
\end{equation}

se $V$ \` e limitata. Ad ogni modo, \` e possibile trattare anche il caso in cui $V$ non \` e limitata, {\em cf.} \cite{GZG2}.
}.

Quindi, ponendo $L_{a}=\dot{Q}_{a}$ possiamo definire una nuova contrazione dello spazio delle fasi come

\begin{equation}
\widehat{\sigma}(\dot{\underline{x}},\underline{x})=\sum_{a=1}^{m}\frac{\dot{Q}_{a}}{k_{B}T_{a}}, \label{eq:neq15}
\end{equation} 

dove stiamo trascurando di sommare una derivata totale e dei termini $O\left(\frac{1}{N_{a}}\right)$. Chiaramente, l'espressione nella (\ref{eq:neq15}) \` e {\em diversa} dalla vera contrazione dello spazio delle fasi $\sigma$; ma, per $N_{a}$ ``grandi'', in valor medio la differenza \` e trascurabile.

 Per questa classe di termostati la (\ref{eq:neq15}) giustifica l'identificazione tra contrazione dello spazio delle fasi ed entropia prodotta per unit\` a di tempo\index{entropia prodotta per unit\` a di tempo}. Notare che {\em non} stiamo cercando in nessun modo di definire l'entropia per un sistema dinamico  fuori dall'equilibrio: la (\ref{eq:neq15}) \` e uguale all'entropia ceduta per unit\` a di tempo al termostato, che \` e all'equilibrio.

Oltre a mettere in luce una relazione tra contrazione dello spazio delle fasi e tasso di produzione di entropia, i termostati gaussiani {\em isocinetici}\index{termostati gaussiani} rendono le equazioni del moto (\ref{eq:neq10a}), (\ref{eq:neq10}) invarianti sotto l'{\em inversione temporale}\index{inversione temporale} $\dot{\underline{x}}\rightarrow-\dot{\underline{x}}$, dal momento che $\underline{\vartheta}^{a}_{j}=\alpha^{a}\dot{\underline{x}}_{j}$ \` e {\em pari} sotto questa trasformazione, e sia le forze $\underline{F}_{j}$ che i potenziali d'interazione sono puramente posizionali; inoltre, per le definizioni (\ref{eq:neq15}), (\ref{eq:neq12}), la quantit\` a (\ref{eq:neq15}) {\em cambia segno} sotto l'azione dell'{\em inversione temporale}, naturalmente identificata come l'inversione dei segni delle velocit\` a di tutte le particelle. Come vedremo in seguito, la reversibilit\` a giocher\` a un ruolo importante nella deduzione di conseguenze osservabili dell'ipotesi caotica.
 
\subsection{Un'interpretazione generale} \label{sez:interp}

Se il sistema tende ad uno stato stazionario l'identificazione tra tasso di contrazione dello spazio delle fasi e tasso di produzione di entropia pu\` o essere dedotta in modo del tutto generale, {\em i.e.} senza senza riferirsi ad un modello particolare, \cite{An}. Consideriamo le equazioni del moto

\begin{equation}
\dot{x}=f(x),
\end{equation}

le quali inducono una mappa $S_{t}$ sullo spazio delle fasi, e immaginiamo che al tempo $t=0$ lo stato del sistema sia descritto da una misura $\mu_{0}(dx)$ assolutamente continua rispetto al volume\index{misura assolutamente continua rispetto al volume}, ovvero della forma $\mu_{0}(dx)=\rho_{0}(x)dx$; al tempo $t$ la misura evolver\` a in $\mu_{t}(dx)=\rho_{t}(x)dx$ 

\begin{equation}
\rho_{t}(x)=\left|\frac{\partial S_{-t}(x)}{\partial x}\right|\rho_{0}(x), \label{eq:ent}
\end{equation}

dove $J_{t}(x)\equiv\left|\frac{\partial S_{-t}(x)}{\partial x}\right|$ \` e il determinante jacobiano in $x$ della mappa $x\rightarrow S_{-t}x$. Definiamo l'entropia\index{entropia} di Gibbs al tempo $t$ come

\begin{equation}
\mathcal{E}(t)\equiv-\int\rho_{t}(x)\log\rho_{t}(x)dx, \label{eq:en1}
\end{equation}

che diventa, per la (\ref{eq:ent}),

\begin{eqnarray}
\mathcal{E}&=&-\int\rho_{0}(S_{-t}x)\left|\frac{\partial S_{-t}(x)}{\partial x}\right|\log\left(\rho_{0}(S_{-t}x)\left|\frac{\partial S_{-t}(x)}{\partial x}\right|\right)dx\nonumber\\&=&-\int\rho_{0}(S_{-t}x)\left|\frac{\partial S_{-t}(x)}{\partial x}\right|\log\left(\rho_{0}(S_{-t}x)\right)dx\nonumber\\&&-\int\rho_{0}(S_{-t}x)\left|\frac{\partial S_{-t}(x)}{\partial x}\right|\log\left(\left|\frac{\partial S_{-t}(x)}{\partial x}\right|\right)dx; \label{eq:en2}
\end{eqnarray}

il primo termine nella parte destra della (\ref{eq:en2}) non contribuisce a $\dot{\mathcal{E}}$ perch\'e ponendo $S_{-t}x=y$ diventa uguale alla costante $-\int\rho_{0}(y)\log\rho_{0}(y)dy$, mentre il secondo pu\`o essere riscritto come, con lo stesso cambio di variabile,  

\begin{eqnarray}
-\int\rho_{0}(y)\log\left|\frac{\partial S_{-t}(S_{t}y)}{\partial (S_{t}y)}\right|dy=\int\rho_{0}(y)\log\left|\frac{\partial S_{t}y}{\partial y}\right|, \label{eq:en3}
\end{eqnarray}

dove l'uguaglianza nella (\ref{eq:en3}) discende dall'identit\` a

\begin{equation}
\frac{\partial S_{-t}(S_{t}x)}{\partial S_{t}x}\frac{\partial S_{t}x}{\partial x}=1. \label{eq:en4}
\end{equation}

Dalle equazioni del moto $\dot{x}=f(x)$ si ottiene che, se $x=S_{t}y$ e $\sigma(x)=-\mbox{div} f(x)$,

\begin{equation}
\frac{d}{dt}\left|\frac{\partial S_{t}y}{\partial y}\right|=-\left|\frac{\partial S_{t}y}{\partial y}\right|\sigma(S_{t}y), \label{eq:en5}
\end{equation}

e il tasso di creazione di entropia\index{tasso di creazione di entropia} $\dot{\mathcal{E}}$ diventa, grazie alla (\ref{eq:en5}),

\begin{eqnarray}
\dot{\mathcal{E}}&=&-\int\rho_{0}(y)\sigma(S_{t}y)dy\nonumber\\&=&-\int\rho_{0}(S_{-t}z)\left|\frac{\partial S_{-t}z}{\partial z}\right|\sigma(z)dz\nonumber\\&=&-\int\mu_{t}(dz)\sigma(z)\rightarrow_{t\rightarrow\infty}-\int\mu(dz)\sigma(z), \label{eq:en6}
\end{eqnarray}

ammettendo che $\mu_{t}(dx)$ converga debolmente a $\mu(dx)$.

Quindi, se il sistema tende ad uno stato stazionario \index{stato stazionario} descritto da una misura di probabilit\` a invariante $\mu$ (che accettando l'ipotesi caotica coincide con la misura SRB nel futuro $\mu_{+}$), interpretiamo la (\ref{eq:en6}) come l'entropia creata dal sistema e ceduta al termostato per unit\` a di tempo {\em nello} stato stazionario. Infine, \` e un risultato noto, {\em cf.} \cite{thH}, che per un sistema che verifica l'ipotesi caotica $\left\langle\sigma\right\rangle\geq 0$; in particolare, se la misura stazionaria non \` e assolutamente continua rispetto al volume si ha che $\left\langle\sigma\right\rangle>0$.

\section{Il teorema di fluttuazione: introduzione}

Accettando l'ipotesi caotica, il teorema di fluttuazione di Gallavotti - Cohen afferma che, chiamando $\sigma(x)$ il tasso di contrazione del volume dello spazio delle fasi e $\langle\sigma\rangle_{+}$ il suo valor medio con la misura SRB, introducendo la quantit\` a $\varepsilon_{\tau}$ come

\begin{equation}
\varepsilon_{\tau}(x)=\frac{1}{\langle\sigma\rangle_{+}\tau}\int_{-\frac{\tau}{2}}^{\frac{\tau}{2}}\sigma(x(t))dt,
\end{equation}

le fluttuazioni di $\varepsilon_{\tau}$ verificano la propriet\` a

\begin{equation}
\lim_{\tau\rightarrow\infty}\frac{1}{\tau\langle\sigma\rangle_{+}p}\log\frac{\mu_{+}\left(\varepsilon_{\tau}=p\right)}{\mu_{+}\left(\varepsilon_{\tau}=-p\right)}=1,\label{eq:FT0b}
\end{equation}

se $\mu_{+}\left(\varepsilon_{\tau}=p\right)$ \` e la probabilit\` a SRB dell'evento $\left\{x:\varepsilon_{\tau}(x)=p\right\}$. Nella prossima sezione vogliamo presentare una dimostrazione rigorosa del risultato (\ref{eq:FT0b}) nel caso di un sistema di Anosov discreto $(\Omega,S)$; il caso nel continuo pu\` o essere ricondotto al discreto mediante la tecnica della sezione di Poincar\'e\index{sezione di Poincar\'e}, {\em cf.} \cite{Ge96}.

\section{Dimostrazione del teorema di fluttuazione} \label{sez:dimFT}

In questa sezione riportiamo una dimostrazione, tratta da \cite{GdimFT}, della (\ref{eq:FT0b}) nel caso di sistemi dinamici discreti $(\Omega,S)$. Grazie alle tecniche esposte nel capitolo \ref{cap:ergo}, vedremo che il problema sar\` a ridotto allo studio di un modello di Ising unidimensionale\index{modello di Ising unidimensionale} con potenziali a decadimento esponenziale.

\subsection{Notazioni} 

Prima di iniziare, ricordiamo qualche definizione. Se $(\Omega,S)$ \` e un sistema di Anosov, chiamiamo $W^{u}_{\gamma}(x)$, $W^{s}_{\gamma}(x)$ le porzioni connesse di variet\` a instabile e variet\` a stabile passanti per il punto $x$ e contenute all'interno di una sfera di raggio $\gamma$ ({\em cf.} definizione \ref{def:Anosov}), $\lambda(x)$ il determinante jacobiano della mappa $S:\Omega\rightarrow\Omega$, e indichiamo con $\lambda_{n}(x)$ la quantit\` a 

\begin{equation}
\lambda_{n}(x)=\prod_{j=-\frac{n}{2}}^{\frac{n}{2}-1}\lambda\left(S^{j}x\right); \label{eq:detit}
\end{equation}

infine, con $\lambda_{u}$, $\lambda_{s}$ ci riferiremo ai determinanti jacobiani di $S$ vista rispettivamente come mappa sulla variet\` a instabile e sulla variet\` a stabile e, analogamente alla (\ref{eq:detit}),

\begin{eqnarray}
\lambda_{u,n}(x)&=&\prod_{j=-\frac{n}{2}}^{\frac{n}{2}-1}\lambda_{u}\left(S^{j}x\right)\\
\lambda_{s,n}(x)&=&\prod_{j=-\frac{n}{2}}^{\frac{n}{2}-1}\lambda_{s}\left(S^{j}x\right).
\end{eqnarray} 

Dunque, $\lambda_{n}$ pu\` o essere scritto come
 
\begin{equation}
\lambda_{n}(x)=\lambda_{u,n}(x)\lambda_{s,n}(x)\chi_{n}(x), \label{eq:dimFT0}
\end{equation}

dove $\chi_{n}(x)=\frac{\sin\alpha\left(S^{\frac{n}{2}-1}x\right)}{\sin\alpha\left(S^{-\frac{n}{2}}x\right)}$ \` e il rapporto dei seni degli angoli tra $W^{u}_{\gamma}$ e $W^{s}_{\gamma}$ nei punti $S^{\frac{n}{2}-1}x$ e $S^{-\frac{n}{2}}x$.

\subsection{La misura SRB approssimata}\label{sez:srbappr}
 
Introduciamo la misura $\mu_{T,\tau}$ come, per $T\geq\frac{\tau}{2}$,

\begin{eqnarray}
\mu_{T,\tau}(F)&\equiv&\int_{\Omega}\mu_{T,\tau}(dx)F(x)\nonumber\\&=&\frac{\sum_{\underline{\sigma}_{T}\in\{1,...,n\}^{[-T,T]}_{C}}F(X(\underline{\sigma}_{T}))\lambda_{u,\tau}^{-1}(X(\underline{\sigma}_{T}))}{\sum_{\underline{\sigma}_{T}\in\{1,...,n\}^{[-T,T]}_{C}}\lambda_{u,\tau}^{-1}(X(\underline{\sigma}_{T}))}\\
&=&\frac{\sum_{\underline{\sigma}_{T}\in\{1,...,n\}^{[-T,T]}_{C}}F(X(\underline{\sigma}_{T}))e^{-A_{u,\tau}(\underline{\sigma}_{T})}}{\sum_{\underline{\sigma}_{T}\in\{1,...,n\}^{[-T,T]}_{C}}e^{-A_{u,\tau}(\underline{\sigma}_{T})}}, \label{eq:FT03}
\end{eqnarray}

dove $\underline{\sigma}_{T}$ \` e una sequenza $\sigma_{-T},...,\sigma_{T}$ con $\sigma_{i}\in\{1,...,n\}$, $\{1,...,n\}^{[-T,T]}_{C}$ \` e l'insieme delle sequenze compatibili\footnote{{\em cf.} definizione \ref{def:matrcomp}; per evitare confusione con il ``volume'' $T$ della misura $\mu_{T,\tau}$ indichiamo l'insieme delle sequenze di $n$ simboli compatibi in un insieme $A$ con $\{1,...,n\}^{A}_{C}$ e {\em non} con $\{1,...,n\}^{A}_{T}$.} in $[-T,T]$, $X$ \` e il codice\index{codice della dinamica simbolica} della dinamica simbolica e $A_{u,\tau}=\sum_{j=-\tau/2}^{\tau/2-1}A_{u}\circ\theta^{j}=\sum_{j=-\tau/2}^{\tau/2-1}\log\lambda_{u}\circ X\circ\theta^{j}$, se $\theta$ \` e l'operatore di traslazione verso sinistra nello spazio delle sequenze\footnote{In questa sezione chiamiamo $\theta$ e non $\tau$ l'operatore di traslazione  verso sinistra nello spazio delle sequenze per non creare confusione con il tempo $\tau$.}; le funzioni $F(X(\underline{\sigma}_{T}))$, $\lambda_{u}(X(\underline{\sigma}_{T}))$ vanno pensate calcolate nella configurazione che si ottiene ``prolungando'' in modo compatibile le sequenze finite $\underline{\sigma}_{T}\equiv \underline{\sigma}_{[-T,T]}$ all'infinito. In questo modo abbiamo {\em raffinato} la partizione di Markov di $(\Omega,S)$, {\em cf.} definizione \ref{def:pavmark}, ovvero abbiamo diviso $\Omega$ in piccoli rettangoli identificati dalle sequenze $\underline{\sigma}_{T}$, e assegniamo ad $F$, $\lambda_{u}$ un valore per ogni rettangolo scegliendo un punto in ognuno di essi; ci\` o pu\` o essere fatto in modo completamente costruttivo, ovvero senza ricorrere all'{\em assioma della scelta}, dal momento che il numero $\mathcal{N}$ di rettangoli \` e finito, $\mathcal{N}\leq n^{2T}$. Inoltre, se $T$ \` e molto grande questa scelta influenzer\` a ``poco'' il risultato del valor medio perch\'e $F$, $\lambda_{u}$, $X$ sono funzioni H\"older continue. Sempre grazie alla H\"older continuit\` a, nel linguaggio della meccanica statistica dei sistemi di spin la scelta della continuazione delle sequenze $\underline{\sigma}_{T}$ equivale a fissare le {\em condizioni al bordo} di un modello di Ising unidimensionale con potenziali a decadimento esponenziale; \` e noto che questi modelli non presentano {\em transizioni di fase}, {\em cf.} appendice $D$, ovvero non limite $T\rightarrow\infty$ lo stato di Gibbs {\em non \` e sensibile alla scelta delle condizioni al bordo}. Quindi, possiamo definire il valor medio SRB nel futuro di qualunque osservabile H\"older continua con il {\em limite termodinamico}

\begin{equation}
\lim_{T\geq\frac{\tau}{2},\tau\rightarrow\infty}\mu_{T,\tau}(F)=\mu_{+}(F);
\end{equation}

ci riferiremo alla misura $\mu_{T,\tau}$, per $T,\tau$ fissati con $T\geq\frac{\tau}{2}$, parlando di una {\em misura SRB ``approssimata''}\index{misura SRB ``approssimata''}. 

\subsection{Sistemi dissipativi reversibili}\label{sez:disrev}

Assumiamo che il nostro sistema sia {\em dissipativo}\index{sistema dissipativo}, ovvero che

\begin{eqnarray}
\int_{\Omega}\mu_{+}(dx)\log\lambda^{-1}(x)&=&\left\langle\sigma\right\rangle_{+}>0,
\end{eqnarray}

e che esista\footnote{\` E importante verificare che sistemi di Anosov dissipativi reversibili\index{sistema di Anosov dissipativo reversibile} esistano; ci\` o pu\` o essere fatto ad esempio notando che se il sistema $(\Omega,S)$ \` e dissipativo allora anche il sistema $(\Omega',S')$ con $\Omega=\Omega\times\Omega$ e $S'(x,y)=\left(Sx,S^{-1}y\right)$ lo sar\` a, e $I\circ S=I\circ S^{-1}$ se $I(x,y)=(y,x)$.} un'isometria differenziabile $I$ con $I^{2}=\pm 1$ tale che

\begin{equation}
I\circ S=S^{-1}\circ I; \label{eq:FT04}
\end{equation}

con la (\ref{eq:FT04}) stiamo assumendo che la dinamica sia {\em reversibile}. Considerando per semplicit\` a il caso bidimensionale, chiamando $w_{u}(y)$, $w_{s}(y)$ due vettori normalizzati tangenti alla variet\` a instabile e alla variet\` a stabile nel punto $y$, possiamo definire $\lambda_{u}$ e $\lambda_{s}$ mediante le relazioni\footnote{{\em cf.} definizione \ref{def:contr}.} (indicando con $DS$ il {\em differenziale} di $S$)

\begin{eqnarray}
DS^{-1}(x)w_{u}(x)&=&\lambda_{u}^{-1}(x)w_{u}(S^{-1}x) \label{eq:detit1}\\
DS(x)w_{s}(x)&=&\lambda_{s}(x)w_{s}(Sx); \label{eq:detit2}
\end{eqnarray}

poich\'e\footnote{Questa propriet\` a \` e una diretta conseguenza delle definizioni di $W^{u}(x)$, $W^{s}(x)$, {\em cf.} definizione \ref{def:Anosov}; infatti, se $d$ \` e una distanza in $\Omega$,

\begin{eqnarray}
W^{u}(x)&\equiv&\left\{y\in\Omega\left|d\left(S^{-n}x,S^{-n}y\right)\leq \lambda^{n}d(x,y)\right.\right\}\\
W^{s}(x)&\equiv&\left\{y\in\Omega\left|d\left(S^{n}x,S^{n}y\right)\leq \lambda^{n}d(x,y)\right.\right\},
\end{eqnarray}

e poich\'e $I$ \` e un'isometria, {\em i.e.} non altera le distanze, ponedo $y'=Iy$ e assumendo che $I^{2}=1$,

\begin{eqnarray}
W^{u}(Ix)&=&\left\{y\in\Omega\left|d\left(S^{-n}Ix,S^{-n}y\right)\leq \lambda^{n}d(Ix,y)\right.\right\}\nonumber\\
&=&\left\{Iy'\in\Omega\left|d\left(IS^{n}x,IS^{n}y'\right)\leq \lambda^{n}d(Ix,Iy')\right.\right\}\nonumber\\
&=& IW^{s}(x).
\end{eqnarray}

} $W^{u}\left(Ix\right)=IW^{s}(x)$ e quindi

\begin{equation}
w_{u}\left(Ix\right)=DI(x)w_{s}(x), \label{eq:FT004}
\end{equation}

grazie alla (\ref{eq:FT04}),

\begin{eqnarray}
DS^{-1}(Ix)w_{u}\left(Ix\right)&=&DS^{-1}(Ix)DI(x)w_{s}(x)\nonumber\\&=&D(S^{-1}\circ I)(x)w_{s}(x)\nonumber\\&=&DI\left(Sx\right)DS(x)w_{s}(x). \label{eq:FT0004}
\end{eqnarray}

Dunque, per la (\ref{eq:FT0004}) la (\ref{eq:detit1}) calcolata in $Ix$ diventa, moltiplicando la parte destra e la parte sinistra per $\left[DI(Sx)\right]^{-1}$,

\begin{equation}
DS(x)w_{s}(x)=\lambda_{u}^{-1}(Ix)w_{s}(Sx),
\end{equation}

e questa equazione confrontata con la (\ref{eq:detit2}) implica che

\begin{equation}
\lambda_{u}^{-1}\circ I = \lambda_{s}. \label{eq:consrev}
\end{equation}

Nel caso $d$-dimensionale il ragionamento \` e analogo, con l'unica differenza che i vettori $w_{u}$, $w_{s}$ saranno sostituiti da due sottospazi con dimensioni $d_{u}$, $d_{s}$, con $d_{u}+d_{s}=d$. 

\subsection{Dimostrazione del teorema}\label{sez:dimoFT}

Introduciamo la media adimensionale su un tempo $\tau$ del tasso di contrazione dello spazio delle fasi\index{tasso di contrazione dello spazio delle fasi} come

\begin{equation}
\varepsilon_{\tau}(x)=\frac{1}{\left\langle\sigma\right\rangle_{+}\tau}\sum_{j=-\tau/2}^{\tau/2-1}\log\lambda^{-1}(S^{j}x)=\frac{1}{\left\langle\sigma\right\rangle_{+}\tau}\log\lambda_{\tau}^{-1}(x); \label{eq:dimFT1}
\end{equation}

per quanto discusso nella sezione \ref{sez:problentr}, la quantit\` a definita nella (\ref{eq:dimFT1}) pu\` o essere interpretata come la media adimensionale del tasso di produzione di entropia. Se $\mu_{0}$ \` e una misura di probabilit\` a assolutamente continua rispetto al volume, per $\mu_{0}$-quasi tutti i punti di $\Omega$ si ha che

\begin{equation}
\left\langle\varepsilon_{\tau}\right\rangle_{+}=\lim_{T\rightarrow\infty}\frac{1}{T}\sum_{j=0}^{T-1}\varepsilon_{\tau}(S^{j}x)\equiv\int_{\Omega}\mu_{+}(dx)\varepsilon_{\tau}(y)=1 \label{eq:dimFT2}
\end{equation}

 e
 
 \begin{equation}
 \lim_{\tau\rightarrow\infty}\varepsilon_{\tau}(x)=0 \label{eq:dimFT3}
 \end{equation}
 
dal momento che nel limite $\tau\rightarrow\infty$ la quantit\` a $\varepsilon_{\tau}(x)$ non dipende pi\` u dal punto ({\em cf.} proposizione \ref{prop:srb2}) e, per le (\ref{eq:dimFT0}), (\ref{eq:consrev}), 
 
 \begin{equation}
 \lim_{\tau\rightarrow\infty}\varepsilon_{\tau}(x)=-\lim_{\tau\rightarrow\infty}\varepsilon_{\tau}(Ix); \label{eq:dimFT013}
 \end{equation}
 
 nella (\ref{eq:dimFT013}) l'argomento del limite nella parte destra differisce da quello nella parte sinistra per una quantit\` a che pu\` o essere maggiorata con $\frac{B_{2}}{\tau}$, se $B_{2}$ \` e una costante positiva ({\em cf.} formule (\ref{eq:dimFT0}), (\ref{eq:dimFT1})). Siamo pronti per dimostrare il seguente risultato.
 
 \begin{teorema}[di fluttuazione]\index{teorema di fluttuazione}
 Esiste $p^{*}\geq 1$ tale che la misura SRB $\mu_{+}$ verifica
 
 \begin{equation}
 p-\delta\leq\lim_{\tau\rightarrow\infty}\frac{1}{\tau\left\langle\sigma\right\rangle_{+}}\log\frac{\mu_{+}(\varepsilon_{\tau}(x)\in[p-\delta,p+\delta])}{\mu_{+}(\varepsilon_{\tau}(x)\in-[p-\delta,p+\delta])}\leq p+\delta \label{eq:GCFT}
 \end{equation}
 
 $\forall p$, $|p|<p^{*}$.
 
 \end{teorema}
 
 \begin{proof}
 
Se $X$ \` e il codice della dinamica simbolica\index{codice della dinamica simbolica} ({\em cf.} capitolo \ref{cap:ergo}, proposizione \ref{prop:pavmark}), la funzione $\varepsilon_{\tau}(x)$ pu\` o essere convertita in una funzione sulle configurazioni di spin\footnote{Con un abuso di notazione indicheremo con lo stesso simbolo le funzioni $\varepsilon_{\tau}$, $\varepsilon_{\tau}\circ X$.},
 
 \begin{equation}
\varepsilon_{\tau}(\underline{\sigma})=\frac{1}{\tau}\sum_{k=-\tau/2}^{\tau/2-1}L(\theta^{k}\underline{\sigma}), \label{eq:dimFT4a}
 \end{equation}
 
 dove $\theta$ \` e l'operatore di traslazione\index{operatore di traslazione} verso sinistra nello spazio delle sequenze compatibili\index{sequenza compatibile} e $L(\underline{\sigma})=\frac{1}{\left\langle\sigma\right\rangle_{+}}\log\lambda^{-1}(X(\underline{\sigma}))$. Poich\'e $L$ \` e H\"older continua pu\` o essere espansa in potenziali a decadimento esponenziale, {\em cf.} proposizione \ref{prop:svilpot},
 
 \begin{equation}
 L(\underline{\sigma})=\sum_{n\geq 0}l_{[-n,n]}\left(\underline{\sigma}_{[-n,n]}\right);
 \end{equation}
 
 dunque, chiamando $X=(x,x+1,...,x+2m)$ un generico intervallo reticolare connesso composto da $2m+1$ simboli e $\overline{X}$ il suo centro, riscriviamo la (\ref{eq:dimFT4}) come
 
 \begin{equation}
 \varepsilon_{\tau}(\underline{\sigma})=\frac{1}{\tau}\sum_{X:\overline{X}\in[-\frac{\tau}{2},\frac{\tau}{2}-1]}l_{X}\left(\underline{\sigma}_{X}\right).
 \end{equation} 
 
Le seguenti proposizioni sono equivalenti, {\em cf.} \cite{Ge96}, \cite{RuStat}.
 
 \begin{prop}\label{prop:zeta}
 
 Esiste $p^{*}\geq 1$ tale che il limite
 
 \begin{equation}
  \frac{1}{\tau}\log\mu_{+}(\varepsilon_{\tau}(\underline{\sigma})\in[p-\delta,p+\delta])\rightarrow_{\tau\rightarrow\infty}\max_{s\in[p-\delta,p+\delta]}-\zeta(s) \label{eq:dimFT4}
 \end{equation}
 
 esiste e la funzione $\zeta(s)$ \` e analitica e strettamente convessa per $p\in(-p^{*},p^{*})$; inoltre, la differenza tra la parte destra e la parte sinistra della (\ref{eq:dimFT4}) tende a zero maggiorata da $\frac{D}{\tau}$, per un'opportuna costante $D>0$.
 
 \end{prop}
 
 \begin{prop}\label{prop:lambda}
 
 La funzione $\lambda(\beta)$ definita come
 
 \begin{equation}
 \lambda(\beta)=\lim_{\tau\rightarrow\infty}\frac{1}{\tau}\log \left\langle e^{\beta\tau\langle\sigma\rangle_{+}\left(\varepsilon_{\tau}-1\right)}\right\rangle_{+},
 \end{equation}
 
 ovvero
 
 \begin{equation}
 \lambda(\beta)=\max_{s}\left[\beta\langle\sigma\rangle_{+}(s-1)-\zeta(s)\right],
 \end{equation}
 
\` e analitica in $\beta\in(-\infty,\infty)$ e {\em asintoticamente lineare}, {\em i.e.} esiste $p^{*}\geq 1$ tale che

\begin{equation}
\lim_{\beta\rightarrow\infty}\beta^{-1}\lambda(\beta)=\langle\sigma\rangle_{+}\left(p^{*}-1\right).
\end{equation}
 
 \end{prop}
 
 \begin{oss}
 \begin{enumerate}
 \item Il valore $p^{*}$ \` e il supremo dei valori possibili di $\varepsilon_{\tau}(x)$ nel limite $\tau\rightarrow\infty$, ovvero $p^{*}=\sup_{x\in\Omega}\limsup_{\tau}\varepsilon_{\tau}(x)$, \cite{GdimFT}.
 \item Quindi, per $|p|<p^{*}$ il limite nella (\ref{eq:GCFT}) esiste; al di fuori di questo intervallo nulla pu\` o essere detto.
 \end{enumerate} 
 \end{oss}
 
Rimandiamo all'appendice \ref{app:D} per una dimostrazione dell'analiticit\` a di $\lambda(\beta)$. Introduciamo la funzione $\tilde{\varepsilon}_{\tau}$ come
 
 \begin{equation}
 \tilde{\varepsilon}_{\tau}(\underline{\sigma})=\frac{1}{\tau}\sum_{X\subseteq [-\frac{\tau}{2},\frac{\tau}{2}-1]}\overline{l}_{X}\left(\underline{\sigma}_{X}\right), \label{eq:dimFT40}
 \end{equation}
 
 dove con $\overline{l}_{X}\left(\underline{\sigma}_{X}\right)$ indichiamo ci\` o che rimane sottraendo ai potenziali $l_{X}\left(\underline{\sigma}_{X}\right)$ la parte dovuta al rapporto dei seni, {\em cf.} formula (\ref{eq:dimFT0}). La (\ref{eq:dimFT40}) ``approssimer\` a'' $\varepsilon_{\tau}$ nel senso che esiste $C>0$ tale che
 
 \begin{equation}
 \left|\varepsilon_{\tau}-\tilde{\varepsilon}_{\tau}\right|\leq \frac{C}{\tau};
 \end{equation}
 
 notare che ora $\widetilde\varepsilon\circ I=-\widetilde\varepsilon$. Quindi, poich\'e
 
 \begin{equation}
 \left\{\underline{\sigma}:\tilde{\varepsilon}_{\tau}(\underline{\sigma})\in I_{p,\delta-\frac{C}{\tau}}\right\}\subseteq \left\{\underline{\sigma}:\varepsilon_{\tau}(\underline{\sigma})\in I_{p,\delta}\right\}\subseteq\left\{\underline{\sigma}:\tilde{\varepsilon}_{\tau}(\underline{\sigma})\in I_{p,\delta+\frac{C}{\tau}}\right\},
 \end{equation}
 
 avremo
 
 \begin{equation}
 \mu_{+}\left(\tilde{\varepsilon}_{\tau}\in I_{p,\delta-\frac{C}{\tau}}\right)\leq \mu_{+}\left(\varepsilon_{\tau}\in I_{p,\delta}\right)\leq \mu_{+}\left(\tilde{\varepsilon}_{\tau}\in I_{p,\delta+\frac{C}{\tau}}\right); \label{eq:dimFT04}
 \end{equation}

un discorso analogo pu\` o essere fatto per la misura approssimata $\mu_{T,\tau}$. Per dimostrare il teorema sar\` a sufficiente provare che, se $I_{p,\delta}=[p-\delta,p+\delta]$,

\begin{equation}
\frac{1}{\left\langle\sigma\right\rangle_{+}\tau}\log\frac{\mu_{+}(\varepsilon_{\tau}(\underline{\sigma})\in I_{p,\delta\mp\eta(\tau)})}{\mu_{+}(\varepsilon_{\tau}(\underline{\sigma})\in I_{-p,\delta\pm\eta(\tau)})}\left\{
\begin{array}{c}
<p+\delta+\eta'(\tau) \\ >p-\delta-\eta'(\tau)
\end{array},\right.
\end{equation}

con $\eta'(\tau),\,\eta(\tau)\rightarrow_{\tau\rightarrow\infty}0$. Per fare ci\` o useremo il fatto che le probabilit\` a calcolate con la misura SRB\index{misura SRB} possono essere stimate con la misura $\mu_{T,\tau}$\index{misura SRB ``approssimata''}, se $T\geq \frac{\tau}{2}$; in seguito considereremo il caso $T=\frac{\tau}{2}$. 

Per quanto detto in precedenza, {\em cf.} formula (\ref{eq:FT03}) e discussione successiva, la probabilit\` a $SRB$ di una configurazione di spin $\underline{\sigma}_{\tau}=\sigma_{-\frac{\tau}{2}}...\sigma_{\frac{\tau}{2}}$ \` e data da, se $\Phi_{X}$ sono i potenziali che si ottengono sviluppando $A_{u}$ nel senso della proposizione \ref{prop:svilpot},

 \begin{equation}
 \mu_{+}(\sigma_{-\frac{\tau}{2}}...\sigma_{\frac{\tau}{2}})=\lim_{|\Gamma|\rightarrow\infty}\frac{\sum_{\underline{\sigma}}^{*}e^{-\sum_{X\subset\Gamma}\Phi_{X}(\underline{\sigma}_{X})}}{\sum_{\underline{\sigma}}e^{-\sum_{X\subset\Gamma}\Phi_{X}(\underline{\sigma}_{X})}}, \label{eq:dimFT13}
 \end{equation}  
 
 dove l'asterisco sulla sommatoria indica che stiamo fissando la configurazione $\sigma_{-\frac{\tau}{2}}...\sigma_{\frac{\tau}{2}}$, e sommando sulle altre compatibili. La (\ref{eq:dimFT13}) pu\` o essere riscritta come
 
 \begin{equation}
  \mu_{+}(\sigma_{-\frac{\tau}{2}}...\sigma_{\frac{\tau}{2}})=\frac{e^{-\sum_{X\subseteq[-\frac{\tau}{2},\frac{\tau}{2}]}\Phi_{X}(\underline{\sigma}_{X})}\sum^{*}_{\underline{\sigma}}e^{-\sum_{X\nsubseteq[-\frac{\tau}{2},\frac{\tau}{2}]}\Phi_{X}(\underline{\sigma}_{X})}}{\sum_{\sigma_{[-\frac{\tau}{2},\frac{\tau}{2}]}}e^{-\sum_{X\subseteq[-\frac{\tau}{2},\frac{\tau}{2}]}\Phi_{X}(\underline{\sigma}_{X})}\sum_{\underline{\sigma}}^{*}e^{-\sum_{X\nsubseteq[-\frac{\tau}{2},\frac{\tau}{2}]}\Phi_{X}(\underline{\sigma}_{X})}} \label{eq:dimFT14}
 \end{equation}
 
con, se $a$ \` e il tempo di mescolamento della matrice di compatibilit\` a ({\em cf.} definizione \ref{def:matrcomp}),
 
 \begin{eqnarray}
 \sum_{X\nsubseteq[-\frac{\tau}{2},\frac{\tau}{2}]}\Phi_{X}(\underline{\sigma}_{X})&=&\sum_{X\nsubseteq[-\frac{\tau}{2},\frac{\tau}{2}],X\cap[-\frac{\tau}{2}-a,\frac{\tau}{2}+a]=\emptyset}\Phi_{X}(\underline{\sigma}_{X})\nonumber\\&&+\sum_{X\nsubseteq[-\frac{\tau}{2},\frac{\tau}{2}],X\cap[-\frac{\tau}{2}-a,\frac{\tau}{2}+a]\neq\emptyset}\Phi_{X}(\underline{\sigma}_{X}); \label{eq:dimFT014}
 \end{eqnarray}
 
poich\' e $A_{u}$ \` e h\"older continua esistono $b,b'>0$ tali che ({\em cf.} proposizione \ref{prop:svilpot}) $\left|\Phi_{X}\left(\underline{\sigma}_{X}\right)\right|\leq be^{-(|X|-1)b'}$, dunque
 
 \begin{equation}
 \left|\sum_{X\nsubseteq[-\frac{\tau}{2},\frac{\tau}{2}], X\cap[-\frac{\tau}{2}-a,\frac{\tau}{2}+a]\neq\emptyset}\Phi_{X}(\underline{\sigma}_{X})\right|\leq2a\sup_{\xi}\sum_{X\ni\xi}\left|\Phi_{X}\left(\underline{\sigma}_{X}\right)\right|\equiv\log B. \label{eq:dimFT0140}
 \end{equation}
 
Grazie alla (\ref{eq:dimFT0140}) possiamo maggiorare la (\ref{eq:dimFT14}) come, ricordando che ogni ``spin'' pu\` o assumere $n$ valori,
 
 \begin{equation}
 (\ref{eq:dimFT14})\leq \frac{e^{-\sum_{X\subseteq[-\frac{\tau}{2},\frac{\tau}{2}]}\Phi_{X}(\underline{\sigma}_{X})}n^{2a}\sum^{**}e^{-\sum_{X\cap[-\frac{\tau}{2}-a,\frac{\tau}{2}+a]=\emptyset}\Phi_{X}(\underline{\sigma}_{X})}}{\sum_{\sigma_{[-\frac{\tau}{2},\frac{\tau}{2}]}}e^{-\sum_{X\subseteq[-\frac{\tau}{2},\frac{\tau}{2}]}\Phi_{X}(\underline{\sigma}_{X})}\sum_{\underline{\sigma}}^{**}e^{-\sum_{X\nsubseteq[-\frac{\tau}{2},\frac{\tau}{2}]}\Phi_{X}(\underline{\sigma}_{X})}}B^{2},
 \end{equation}
 
 dove con il doppio asterisco indichiamo che stiamo sommando sulle sequenze $\underline{\sigma}_{\left[-\infty,-\frac{\tau}{2}-a\right)\cup\left(\frac{\tau}{2}+a,\infty\right]}$, indipendenti da $\underline{\sigma}_{\tau}$; quindi, i potenziali relativi ad intervalli reticolari che non intersecano $\left[-\frac{\tau}{2}-a,\frac{\tau}{2}+a\right]$ fattorizzano tra numeratore e denominatore. Per la minorazione il discorso \` e analogo a quello che abbiamo appena fatto, quindi si ottiene che
 
 \begin{equation}
  \mu_{+}(\sigma_{-\frac{\tau}{2}}...\sigma_{\frac{\tau}{2}})\left\{
\begin{array}{c}
\leq\frac{e^{-\sum_{X\subset [-\frac{\tau}{2},\frac{\tau}{2}]}\Phi_{X}(\underline{\sigma}_{X})}}{\sum_{\sigma_{[-\frac{\tau}{2},\frac{\tau}{2}]}}e^{-\sum_{X\subset[-\frac{\tau}{2},\frac{\tau}{2}]}\Phi_{X}(\underline{\sigma}_{X})}}B^{2}n^{2a} \\ \geq\frac{e^{-\sum_{X\subset [-\frac{\tau}{2},\frac{\tau}{2}]}\Phi_{X}(\underline{\sigma}_{X})}}{\sum_{\sigma_{[-\frac{\tau}{2},\frac{\tau}{2}]}}e^{-\sum_{X\subset[-\frac{\tau}{2},\frac{\tau}{2}]}\Phi_{X}(\underline{\sigma}_{X})}}B^{-2}n^{-2a}.
\end{array}\right. \label{eq:dimFT15}
 \end{equation}
 
Dalle (\ref{eq:dimFT04}), (\ref{eq:dimFT15}), dal momento che $\tilde{\varepsilon}_{\tau}$ dipende da un numero finito di simboli, \` e chiaro che possiamo trovare una costante $B_{1}$ tale che, per ogni $T\geq\frac{\tau}{2}$,

 \begin{eqnarray}
\mu_{+}\left(\varepsilon_{\tau}\in I_{p,\delta}\right)&\leq&\mu_{+}\left(\tilde{\varepsilon}_{\tau}\in I_{p,\delta+\frac{C}{\tau}}\right)\nonumber\\ &\leq& B_{1}\mu_{T,\tau}\left(\tilde{\varepsilon}_{\tau}\in I_{p,\delta+\frac{C}{\tau}}\right)\leq B_{1}\mu_{T,\tau}\left(\overline{\varepsilon}_{\tau}\in I_{p,\delta+2\frac{C}{\tau}}\right), \label{eq:dimFT16}
 \end{eqnarray}

dove $\overline{\varepsilon}_{\tau}$ \` e ci\` o che si ottiene sottraendo a $\varepsilon_{\tau}$ la parte relativa al rapporto dei seni; in questo modo, $\overline{\varepsilon}_{\tau}=\frac{1}{\tau\langle\sigma\rangle_{+}}\lambda_{u,\tau}^{-1}\lambda_{s,\tau}^{-1}=-\overline{\varepsilon}_{\tau}\circ I$. Una stima analoga pu\` o essere fatta per la minorazione, dunque

 \begin{equation}
 \frac{1}{\tau\left\langle\sigma\right\rangle_{+}}\log\frac{\mu_{+}\left(\varepsilon_{\tau}\in I_{p,\delta\mp\frac{2C}{\tau}}\right)}{\mu_{+}\left(\varepsilon_{\tau}\in I_{-p,\delta\pm\frac{2C}{\tau}}\right)}\left\{
\begin{array}{c}
\leq\frac{\log B_{1}^{2}}{\tau\left\langle\sigma\right\rangle_{+}}+\frac{1}{\tau\left\langle\sigma\right\rangle_{+}}\log\frac{\mu_{T,\tau}(\overline{\varepsilon}_{\tau}\in I_{p,\delta})}{\mu_{T,\tau}(\overline{\varepsilon}_{\tau}\in I_{-p,\delta})} \\ \geq-\frac{\log B_{1}^{2}}{\tau\left\langle\sigma\right\rangle_{+}}+\frac{1}{\tau\left\langle\sigma\right\rangle_{+}}\log\frac{\mu_{T,\tau}(\overline{\varepsilon}_{\tau}\in I_{p,\delta})}{\mu_{T,\tau}(\overline{\varepsilon}_{\tau}\in I_{-p,\delta})}
\end{array}\right.,
 \end{equation}
 
 e otterremo la (\ref{eq:GCFT}) mostrando che
 
 \begin{equation}
 \frac{1}{\tau\left\langle\sigma\right\rangle_{+}}\log\frac{\mu_{T,\tau}(\overline{\varepsilon}_{\tau}\in I_{p,\delta})}{\mu_{T,\tau}(\overline{\varepsilon}_{\tau}\in I_{-p,\delta})}
 \left\{
\begin{array}{c}
\leq p+\delta \\ \geq p-\delta
\end{array}\right.. \label{eq:dimFT17}
 \end{equation}
  
 Il rapporto di probabilit\` a nella (\ref{eq:dimFT17}) pu\` o essere scritto esplicitamente come
  
  \begin{equation}
  \frac{\sum_{\overline{\varepsilon}_{\tau}\in I_{p,\delta}}\lambda^{-1}_{u,\tau}(X(\underline{\sigma}_{T}))}{\sum_{\overline{\varepsilon}_{\tau}\in I_{-p,\delta}}\lambda^{-1}_{u,\tau}(X(\underline{\sigma}_{T}))}=\frac{\sum_{\overline{\varepsilon}_{\tau}\in I_{p,\delta}}\lambda^{-1}_{u,\tau}(X(\underline{\sigma}_{T}))}{\sum_{\overline{\varepsilon}_{\tau}\in I_{p,\delta}}\lambda^{-1}_{u,\tau}(IX(\underline{\sigma}_{T}))}=\frac{\sum_{\overline{\varepsilon}_{\tau}\in I_{p,\delta}}\lambda^{-1}_{u,\tau}(X(\underline{\sigma}_{T}))}{\sum_{\overline{\varepsilon}_{\tau}\in I_{p,\delta}}\lambda_{s,\tau}(X(\underline{\sigma}_{T}))}, \label{eq:dimFT18}
  \end{equation}
  
  e
  
  \begin{equation}
 \frac{\sum_{\overline{\varepsilon}_{\tau}\in I_{p,\delta}}\lambda^{-1}_{u,\tau}(X(\underline{\sigma}_{T}))}{\sum_{\overline{\varepsilon}_{\tau}\in I_{p,\delta}}\lambda_{s,\tau}(X(\underline{\sigma}_{T}))}\left\{
\begin{array}{c}
\leq\max_{\underline{\sigma}_{T}}\lambda_{u,\tau}^{-1}(X(\underline{\sigma}_{T}))\lambda_{s,\tau}^{-1}(X(\underline{\sigma}_{T})) \\ \geq\min_{\underline{\sigma}_{T}}\lambda_{u,\tau}^{-1}(X(\underline{\sigma}_{T}))\lambda_{s,\tau}^{-1}(X(\underline{\sigma}_{T}))
\end{array}\right.; \label{eq:dimFT19}
  \end{equation}
  
poich\' e il massimo e il minimo nella (\ref{eq:dimFT19}) sono calcolati rispettando la condizione $$\frac{1}{\left\langle\sigma\right\rangle_{+}\tau}\log\lambda_{u,\tau}^{-1}(X(\underline{\sigma}_{T}))\lambda_{s,\tau}^{-1}(X(\underline{\sigma}_{T}))\in I_{p,\delta}\;,$$ la (\ref{eq:dimFT17}) \` e immediata; ci\` o conclude la dimostrazione del teorema.
 
 \end{proof}
 
 \section{La reversibilit\` a condizionata} \label{sez:revcond}
 
 Consideriamo un'osservabile $F$ con una parit\` a ben definita\index{parit\` a di un'osservabile} sotto l'azione dell'inversione temporale $I$, ovvero
 
 \begin{equation}
 F(Ix)=\eta_{F}F(x) \label{eq:cond1}
 \end{equation} 
 
 con $\eta_{F}=\pm1$, e introduciamo una funzione continua $\varphi(t)$. Vogliamo dimostrare che, nel contesto dell'ipotesi caotica\index{ipotesi caotica} e considerando sistemi dissipativi\index{sistema dissipativo},
 
 \begin{equation}
 \frac{1}{\tau p\left\langle\sigma\right\rangle_{+}}\log\frac{\mu_{+}\left(\left\{F(S^{k}x)\right\}_{k=-\frac{\tau}{2}}^{\frac{\tau}{2}}\in U_{\gamma},\varepsilon_{\tau}\in I_{p,\delta}\right)}{\mu_{+}\left(\left\{F(S^{k}x)\right\}_{k=-\frac{\tau}{2}}^{\frac{\tau}{2}}\in \eta_{F} U_{\gamma},\varepsilon_{\tau}\in I_{-p,\delta}\right)}\rightarrow_{\tau\rightarrow\infty}1, \label{eq:cond2}
 \end{equation}
 
 dove $U_{\gamma}$, $\eta_{F}U_{\gamma}$ sono dei ``tubi'' di larghezza $\gamma$ attorno a $\varphi(k)$, $\eta_{F}\varphi(-k)$. In quel che segue nulla cambier\` a se al posto di $F$ consideriamo un insieme di osservabili $\{F_{1},...,F_{n}\}$ con parit\` a definita, ognuna con una traiettoria $\varphi_{j}$; per comodit\` a di notazione ne considereremo una sola.
  
 La dimostrazione della (\ref{eq:cond2}) \` e praticamente identica alla dimostrazione del teorema di fluttuazione riportata nella sezione \ref{sez:dimFT}; l'unica differenza \` e che la (\ref{eq:dimFT18}) \` e sostituita da
 
 \begin{eqnarray}
  \frac{\sum_{\overline{\varepsilon}_{\tau}\in I_{p,\delta},F(S^{k}X(\underline{\sigma}_{T}))\in U_{\gamma}}\lambda^{-1}_{u,\tau}(X(\underline{\sigma}_{T}))}{\sum_{\overline{\varepsilon}_{\tau}\in I_{-p,\delta},F(S^{k}X(\underline{\sigma}_{T})\in\eta_{F}U_{\gamma}}\lambda^{-1}_{u,\tau}(X(\underline{\sigma}_{T}))}\nonumber\\=\frac{\sum_{\overline{\varepsilon}_{\tau}\in I_{p,\delta},F(S^{k}X(\underline{\sigma}_{T}))\in U_{\gamma}}\lambda^{-1}_{u,\tau}(X(\underline{\sigma}_{T}))}{\sum_{\overline{\varepsilon}_{\tau}\in I_{p,\delta},F(S^{k}X(\underline{\sigma}_{T}))\in U_{\gamma}}\lambda^{-1}_{u,\tau}(IX(\underline{\sigma}_{T}))}\nonumber\\=\frac{\sum_{\overline{\varepsilon}_{\tau}\in I_{p,\delta},F(S^{k}X(\underline{\sigma}_{T}))\in U_{\gamma}}\lambda^{-1}_{u,\tau}(X(\underline{\sigma}_{T}))}{\sum_{\overline{\varepsilon}_{\tau}\in I_{p,\delta},F(S^{k}X(\underline{\sigma}_{T}))\in U_{\gamma}}\lambda_{s,\tau}(X(\underline{\sigma}_{T}))}, \label{eq:cond3}
 \end{eqnarray}
 
 e da qui in poi tutto procede come gi\` a visto in precedenza. 
 
 Possiamo riscrivere il teorema di fluttuazione (\ref{eq:cond2}) nella forma approssimata
 
 \begin{equation}
 \frac{\mu_{+}\left(\left\{F(S^{k}x)\right\}_{k=-\frac{\tau}{2}}^{\frac{\tau}{2}}\in U_{\gamma},\varepsilon_{\tau}\in I_{p,\delta}\right)}{\mu_{+}\left(\left\{F(S^{k}x)\right\}_{k=-\frac{\tau}{2}}^{\frac{\tau}{2}}\in \eta_{F} U_{\gamma},\varepsilon_{\tau}\in I_{-p,\delta}\right)}\sim e^{\tau p\left\langle\sigma\right\rangle_{+}}\sim\frac{\mu_{+}(\varepsilon_{\tau}\in I_{p,\delta})}{\mu_{+}(\varepsilon_{\tau}\in I_{-p,\delta})},
 \end{equation} 
 
 ovvero
 
 \begin{equation}
 \frac{\mu_{+}\left(\left.\left\{F(S^{k}x)\right\}_{k=-\frac{\tau}{2}}^{\frac{\tau}{2}}\in U_{\gamma}\right|\varepsilon_{\tau}\in I_{p,\delta}\right)}{\mu_{+}\left(\left.\left\{F(S^{k}x)\right\}_{k=-\frac{\tau}{2}}^{\frac{\tau}{2}}\in \eta_{F} U_{\gamma}\right|\varepsilon_{\tau}\in I_{-p,\delta}\right)}\sim 1; \label{eq:cond03}
 \end{equation}
  
 la formula (\ref{eq:cond03}) ha un interessante significato fisico.
 
  Pensiamo di osservare la quantit\` a $\varepsilon_{\tau}$ su un intervallo temporale $T\gg\tau\gg 1$, in modo da averne $\frac{T}{\tau}$ diverse misure; questi valori saranno molto spesso vicini ad $1$ (dal momento che $1$ \` e il valor medio di $\varepsilon_{\tau}$), e i valori lontani da $1$ occorreranno in modo essenzialmente casuale. La frequenza con la quale si manifestano i valori $-p$ \` e legata alla frequenza di $p$ dal teorema di fluttuazione (\ref{eq:GCFT}); ad esempio, la frequenza con la quale il valore $p=-1$ si manifester\` a rispetto alla frequenza del valore $p=1$ \` e proporzionale a $e^{-\tau\left\langle\sigma\right\rangle_{+}}$. Ci\` o vuol dire che nella nostra raccolta di dati sperimentali per $\sim e^{\tau\left\langle\sigma\right\rangle_{+}}$ valori $p=1$ avremo un valore $p=-1$, e in corrispondenza di questo dato, per il teorema (\ref{eq:cond2}), la probabilit\` a che un'osservabile fisica segua il cammino invertito nel tempo sar\` a asintoticamente uguale alla probabilit\` a che segua il cammino ``corretto'' con $p=1$; per questa ragione ha senso dire che in un sistema che verifica l'ipotesi caotica la freccia temporale \` e {\em intermittente}\index{intermittenza della freccia temporale}, \cite{fluttrev}.
 
\section{Altri risultati noti fuori dall'equilibrio}

Per concludere il capitolo, in questa sezione riassumiamo brevemente qualche altro risultato noto fuori dall'equilibrio. In realt\` a, per ``non equilibrio'' in questi casi si intender\` a una cosa diversa da quella che abbiamo discusso finora; il primo risultato che presenteremo, noto come {\em identit\` a di Evans - Searles} o {\em teorema di fluttuazione transiente}, descrive il comportamento di sistemi fuori dall'equilibrio in uno stato non stazionario, mentre gli altri due, la {\em formula di Jarzynski} e i {\em teoremi di fluttuazione a molti indici} di Bochov e Kuzovlev, trattano sistemi che vengono portati fuori dall'equilibrio attraverso trasformazioni hamiltoniane, {\em i.e.} senza dissipazione.

Nonostante la loro generalit\` a dunque, questi risultati (a parte il primo), non avranno la pretesa di descrivere il non equilibrio come lo abbiamo inteso finora, ovvero come stato stazionario di un sistema sottoposto a forze non conservative; anche se a prima vista potrebbero esserci delle similitudini con il teorema di fluttuazione di Gallavotti - Cohen, e bene tenere in mente che la situazione fisica che si sta prendendo in considerazione \` e completamente diversa.

\subsection{L'identit\` a di Evans e Searles}

In un articolo successivo (\cite{ES94}) al lavoro ``pioneristico'' \cite{ECM93}, Evans e Searles dimostrarono che, assumendo semplicemente che il sistema sia reversibile\index{sistema reversibile}\footnote{Affermando che un sistema \` e {\em reversibile} Evans e Searles intendono dire che se le equazioni del moto ammettono come soluzione la traiettoria $\gamma_{\Gamma(t_{0})\rightarrow \Gamma(t_{1})}$ allora ammetteranno anche la soluzione invertita nel tempo $\gamma_{\Gamma(t_{1})\rightarrow \Gamma(t_{0})}$, dove $\Gamma(t)$ \` e lo stato del sistema al tempo $t$, specificato da tutti i valori di posizione e impulso delle particelle.}, chiamando $E_{\pm p}$ l'insieme delle condizioni iniziali delle traiettorie lungo le quali il volume dello spazio delle fasi si contrae di $e^{\mp p\sigma_{+}\tau}$ in un tempo $\tau$, con $\sigma_{+}$ pari al valore medio di $\sigma$, allora

\begin{equation}
\frac{\mu_{L}(E_{p})}{\mu_{L}(E_{-p})}=e^{p\sigma_{+}\tau}, \label{eq:IES}
\end{equation}

dove $\mu_{L}$ \` e la misura di Liouville\index{misura di Liouville} nello spazio delle fasi. 

La dimostrazione della (\ref{eq:IES}) pu\` o essere riassunta nel modo seguente:

\begin{eqnarray}
\frac{\mu_{L}(E_{p})}{\mu_{L}(E_{-p})}&=&\frac{\mu_{L}(E_{p})}{\mu_{L}(IS_{\tau}E_{p})}\nonumber\\&=&\frac{\mu_{L}(E_{p})}{\mu_{L}(S_{\tau}E_{p})}\nonumber\\&=&\frac{\mu_{L}(E_{p})}{\mu_{L}(E_{p})e^{-\tau p\sigma_{+}}}\nonumber\\&=&e^{\tau p\sigma_{+}}, \label{eq:dimIES}
\end{eqnarray}

dove \` e stato usato il fatto che $\mu_{L}(IA)=\mu_{L}(A)$ ($|\det DI|=1$).

Le differenze principali con il teorema di fluttuazione di Gallavotti-Cohen sono le seguenti:

\begin{itemize}
\item la misura \` e assolutamente continua rispetto al volume, al contrario della misura SRB la quale, nel contesto dell'ipotesi caotica, per un sistema dissipativo \` e concentrata in un insieme di misura di volume nulla;
\item la (\ref{eq:IES}) non ha termini di errore, ovvero \` e un'identit\` a esatta a tutti i tempi $\tau$.
\end{itemize}

Al contrario di quanto sostennero Evans e Searles in seguito, {\em non} \` e ovvio che la relazione (\ref{eq:IES}) implichi il teorema di fluttuazione di Gallavotti-Cohen; secondo Evans e Searles, l'argomento per cui ci\` o dovrebbe essere possibile \` e il seguente. Scriviamo la (\ref{eq:IES}) applicando l'operatore $S_{T}$ alla la misura di Liouville $\mu_{L}$,
\begin{equation}
\frac{S_{T}(\mu_{L})(E_{p})}{S_{T}(\mu_{L})(E_{-p})}=\frac{\mu_{L}(S^{-T}E_{p})}{\mu_{L}(S^{-T}E_{-p})}=e^{\tau(p\sigma_{+}+O\left(\frac{T}{\tau}\right))}, \label{IES2}
\end{equation}

e effettuiamo il limite $\tau\rightarrow\infty$ {\em prima} del limite $T\rightarrow\infty$, in modo che il termine di errore all'esponente vada a zero; infine, nel limite $T\rightarrow\infty$ la misura $S_{T}(\mu_{L})$ dovrebbe tendere alla misura stazionaria del sistema fuori dall'equilibrio, ovvero alla misura SRB accettando l'ipotesi caotica.
  
 La critica pi\` u spontanea a questo argomento \` e che non \` e detto che sia possibile invertire i limiti in questo modo (nel teorema di Gallavotti-Cohen si effettua {\em prima} il limite che definisce la SRB e {\em poi} quello su $\tau$), e sicuramente non lo \` e senza fare prima delle assunzioni particolari. Da un punto di vista fisico si potrebbe obiettare che dovrebbe essere sufficiente prendere $T$ ``molto grande'', in modo tale che $S_{T}(\mu_{L})$ sia ``vicina'' alla misura SRB; ma questa affermazione \` e decisamente poco ovvia, dal momento che la misura SRB \` e {\em singolare} rispetto alla misura di volume.
 
 Il fatto che il teorema di fluttuazione e l'identit\` a di Evans e Searles\index{identit\` a di Evans e Searles} siano due cose ben distinte \` e ulteriormente evidenziato dalla presenza di esempi di semplici sistemi per i quali la (\ref{eq:IES}) \` e valida, ma non il teorema di fluttuazione, \cite{twoth}.
 
 \subsection{La formula di Jarzynski}
 
 La formula di Jarzynsky\index{formula di Jarzynski}, \cite{jar}, permette di calcolare la differenza di energia libera $\Delta F$ tra due stati all'equilibrio collegati da una generica trasformazione hamiltoniana (un {\em protocollo}), tramite successive misure del lavoro $W$ effettuato dal sistema sull'ambiente esterno al termine della trasformazione; il risultato di Jarzynski \` e il seguente:
 
 \begin{equation}
 e^{-\beta\Delta F}=\left\langle e^{-\beta W}\right\rangle, \label{eq:jar}
 \end{equation}
 
con $\beta=\frac{1}{k_{B}T}$, dove $T$ \` e la temperatura dei termostati (che supponiamo costante) e il valore di aspettazione $\left\langle...\right\rangle$ \` e calcolato con la misura $e^{-\beta H_{2}}dpdq$, se con $H_{1}$, $H_{2}$ indichiamo l'hamiltoniana del sistema rispettivamente allo stato iniziale e allo stato finale.

Per una dimostrazione della (\ref{eq:jar}) rimandiamo alla letteratura originale, \cite{jar}. Le ipotesi essenziali affinch\' e la (\ref{eq:jar}) sia vera sono due.

\begin{itemize}
\item La trasformazione che collega i due stati deve essere hamiltoniana, ovvero se immaginiamo che il processo sia governato da un parametro $\lambda(t)$ ad ogni tempo l'hamiltoniana del sistema sar\` a data da $H_{\lambda(t)}$, con $H_{\lambda(0)=H_{1}}$ e $H_{\lambda(t_{finale})}=H_{2}$.
\item L'accoppiamento tra il sistema e il termostato deve essere {\em debole}:  ci\`o vuol dire che se con $G_{\lambda(t)}$ indichiamo l'hamiltoniana totale di sistema e termostato, {\em i.e.}

\begin{equation}
G_{\lambda(t)}(\underline{q},\underline{x})=H_{\lambda(t)}(\underline{q})+\tilde{H}(\underline{x})+h(\underline{q},\underline{x}),
\end{equation}

dove $\tilde{H}(\underline{x})$ rappresenta l'hamiltoniana del termostato e $h(\underline{q},\underline{x})$ l'interazione tra termostato e sistema, si deve avere che $|h(\underline{q},\underline{x})|\ll|H(\underline{q})|$ $\forall \underline{q}, \underline{x}$; in questo modo, l'interazione $h(\underline{q},\underline{x})$ viene trascurata\footnote{Non \` e importante che $|h|\ll|\tilde{H}|$, dal momento che alla fine $\Delta F$ non dipender\` a dal termostato.} .
\end{itemize}

La quantit\` a $W$ che compare nella (\ref{eq:jar}) merita una discussione; in realt\` a questa quantit\` a \` e identificata con la variazione di energia {\em totale} di sistema e termostato, ovvero

\begin{equation}
W=G_{\lambda(t_{fin})}-G_{\lambda(0)},
\end{equation}

 dal momento che sistema e termostato sono considerati {\em isolati} dal mondo esterno. Quindi in un esperimento, per avere le ``risposte giuste'' dalla (\ref{eq:jar}), bisogna essere in grado di controllare completamente gli scambi di energia che avvengono con l'ambiente.
 
\subsection{Due precursori} 
 
Prima di passare a studiare le conseguenze del teorema di fluttuazione vogliamo concludere questa breve rassegna di risultati in meccanica statistica fuori dall'equilibrio discutendo un lavoro di G. N. Bochov e Y. E. Kuzovlev del 1981, {\em cf.} \cite{BK81a}, \cite{BK81b}, che, ben $12$  anni prima delle simulazioni numeriche di Evans, Cohen e Morriss, conteneva gi\` a alcuni degli elementi delle teorie sviluppate successivamente.

Consideriamo un sistema classico che evolve con hamiltoniana $H_{0}=H_{0}(\underline{q},\underline{p})$, dove $\underline{q}=(q_{1},...,q_{3N})$, $\underline{p}=(p_{1},...,p_{3N})$ sono le coordinate e gli impulsi di $N$ particelle; indichiamo un insieme di $m$ osservabili macroscopiche con $Q(\underline{q},\underline{p})=\{Q_{\alpha}(\underline{q},\underline{p})\}$, dove $\alpha=1,...m$. Supponiamo che il sistema sia all'equilibrio, e che dunque possa essere descritto con la misura canonica

\begin{equation}
\rho_{0}(\underline{q},\underline{p})=\frac{e^{-\beta H_{0}(\underline{q},\underline{p})}}{\int_{\Omega} d\underline{p}d\underline{q}e^{-\beta H_{0}(\underline{q},\underline{p})}}, \label{eq:BK1}
\end{equation}

dove con $\Omega$ indichiamo lo spazio delle fasi. 

Al tempo $t=0$ accendiamo una perturbazione, che si manifesta con un cambiamento di hamiltoniana

\begin{equation}
H(\underline{q},\underline{p},t)=H_{0}(\underline{q},\underline{p})-h(\underline{q},\underline{p},t), \label{eq:BK2}
\end{equation}

con

\begin{equation}
h(\underline{q},\underline{p},t)=X_{\alpha}(t)Q_{\alpha}(\underline{q},\underline{p})\equiv X(t)Q, \label{eq:BK3}
\end{equation}

dove gli indici ripetuti sono sommati; notare che, in questo modo, l'evoluzione del sistema rimane hamiltoniana, e quindi l'elemento di volume dello spazio delle fasi pu\` o contrarsi, ma solo mantendendo il suo volume costante.

Ipotizziamo che il nostro insieme di osservabili $Q=\{Q_{\alpha}\}$ abbia parit\` a definita sotto inversione temporale\footnote{L'inversione temporale \` e naturalmente identificata con  il cambiamento di segno delle velocit\` a delle particelle.}\index{ inversione temporale}, ovvero

\begin{equation}
Q_{\alpha}(\underline{q},-\underline{p})=\varepsilon_{\alpha}Q(\underline{q},\underline{p}),\qquad\varepsilon_{\alpha}=\pm 1
\end{equation}

e inoltre che

\begin{equation}
H_{0}(\underline{q},-\underline{p})=H_{0}(\underline{q},\underline{p});
\end{equation}

nel seguito, con $\left\langle...\right\rangle_{X(\theta)}$ indichiamo i valori medi calcolati con la misura canonica relativa all'hamiltoniana (\ref{eq:BK2}), dove $h(\underline{p},\underline{q},t)$ dato dalla (\ref{eq:BK3}), e la perturbazione $X(\theta)$ \` e stata accesa al tempo $\theta$. Definiamo $I_{\alpha}$, $E$ come

\begin{eqnarray}
 I_{\alpha}(\underline{q}(t),\underline{p}(t))&=&\frac{d}{dt}Q(\underline{q}(t),\underline{p}(t))\\
 E&\equiv&\int_{0}^{t}X(\tau)I(\tau)d\tau\nonumber\\&=&\int_{0}^{t}X(\tau)\dot{Q}(\tau)d\tau\nonumber\\&=&\int_{0}^{t}\frac{d}{d\tau}H_{0}(\underline{q}(\tau),\underline{p}(\tau))d\tau; \label{eq:BK4}
\end{eqnarray}

introducendo una {\em funzione di prova} $u(t)$ tale che $\lim_{t\rightarrow\pm\infty}u(t)=0$, e considerando\footnote{Stranamente, Bochov e Kuzovlev nel loro articolo \cite{BK81a} non discutono affatto il problema della convergenza degli integrali per $t\rightarrow\infty$; effettuano i limiti senza fare assunzioni speciali su $X(t)$ e su $u(t)$.} delle perturbazioni $X(t)$ caratterizzate da $\lim_{t\rightarrow\pm\infty}X(t)=0$, nel loro lavoro \cite{BK81a} gli autori dimostrano che 

\begin{equation}
\left\langle e^{\int_{-\infty}^{\infty}u(\tau)Q(\tau)d\tau}e^{-\beta E}\right\rangle_{X(\theta)}=\left\langle e^{\int_{-\infty}^{\infty}u(-\tau)\varepsilon Q(\tau)}\right\rangle_{\varepsilon X(-\theta)}. \label{eq:BK5}
\end{equation}

Notare che $X(t)I(t)$ \` e l'incremento di energia interna per unit\` a di tempo dovuto alle forze esterne (conservative), dunque la quantit\` a $E$ che appare nella (\ref{eq:BK5}) rappresenta l'incremento totale di energia interna nel sistema. Introduciamo il funzionale $D_{\theta}[u(\tau);X(\tau)]$ come

\begin{equation}
D_{\theta}[u(\tau);X(\tau)]\equiv\ln\left\langle e^{\int_{-\infty}^{\infty}u(\tau)I(\tau)d\tau}\right\rangle_{X(\theta)}; \label{eq:BK6}
\end{equation}

\` e possibile dimostrare che, partendo  dalla (\ref{eq:BK5}),

\begin{equation}
D_{\theta}[u(\tau)-\beta X(\tau);X(\tau)]=D_{-\theta}[-\varepsilon u(-\tau);\varepsilon X(-\tau)]. \label{eq:BK7}
\end{equation}

Grazie alla (\ref{eq:BK5}), o equivalentemente alla (\ref{eq:BK7}), Bochov e Kuzovlev riescono ad ottenere delle relazioni simili alla generalizzazione del teorema di fluttuazione che abbiamo discusso nella sezione \ref{sez:revcond}, {\em i.e.}

\begin{eqnarray}
P_{\theta}[\{Q_{\alpha}(\tau)\};X(\theta)]e^{-\beta E}&=&P_{\theta}[\left\{\varepsilon_{\alpha} Q_{\alpha}(-\tau)\right\};\varepsilon X(-\theta)] \label{eq:BK8}\\
P_{\theta}[\{I_{\alpha}(\tau)\};X(\theta)]e^{-\beta E}&=&P_{\theta}[\left\{\varepsilon_{\alpha} I_{\alpha}(-\tau)\right\};\varepsilon X(-\theta)], \label{eq:BK9}
\end{eqnarray}

dove $P_{\theta}[A(\tau);X(\tau)]$ \` e la probabilit\` a di $A$ al tempo $\tau$ quando il sistema \` e sottoposto ad una perturbazione $X(\tau)$ accesa al tempo $\theta$. Una conseguenza diretta delle (\ref{eq:BK8}), (\ref{eq:BK9}) \`e che

\begin{equation}
\left\langle e^{-\beta E}\right\rangle_{X(\theta)}=1, \label{eq:BK10}
\end{equation}

quindi, poich\'e per la convessit\` a dell'esponenziale $\left\langle e^{-x}\right\rangle\geq e^{-\left\langle x\right\rangle}$, dal momento che $\beta> 0$ la (\ref{eq:BK10}) implica la condizione

\begin{equation}
\left\langle E\right\rangle_{X(\theta)}\geq 0. \label{eq:BK11}
\end{equation}

Infine, \` e possibile dimostrare che la (\ref{eq:BK7}) implica delle relazioni tra funzioni di correlazione a tutti gli ordini nella perturbazione. Indichiamo con $\left\langle...\right\rangle_{X(\theta)}^{T}$ le {\em aspettazioni troncate}\index{aspettazione troncata} rispetto alla misura di probabilit\` a canonica con hamiltioniana (\ref{eq:BK2}); se $\frac{\delta}{\delta u(t)}$ \` e la derivata funzionale rispetto a $u(t)$ allora l'{\em aspettazione troncata} di ordine $n$ \` e definita come

\begin{equation}
\left\langle I_{\alpha_{1}}(t_{1})...I_{\alpha_{n}}(t_{n})\right\rangle_{X(\theta)}^{T}=\frac{\delta^{n}}{\delta u_{\alpha_{1}}(t_{1})...\delta u_{\alpha_{n}}(t_{n})}D_{\theta}[u(\tau);X(\tau)], \label{eq:BK12}
\end{equation}

{\em i.e.} $D_{\theta}$ \` e il {\em funzionale generatore delle aspettazioni troncate}\index{funzionale generatore delle aspettazioni troncate} di $I(t)$. Grazie alle (\ref{eq:BK7}), (\ref{eq:BK12}) Bochov e Kuzovlev dimostrano che

\begin{eqnarray}
\prod_{i=1}^{n}(-\varepsilon_{\alpha_{i}})\left\langle I_{\alpha_{1}}(t_{1})...I_{\alpha_{n}}(t_{n})\right\rangle_{\varepsilon X(-\theta)}^{T}\label{eq:BK13}\nonumber\\=\sum_{k=0}^{\infty}\frac{(-\beta)^{k}}{k!}\int_{-\infty}^{\infty}\left\langle I_{\alpha_{1}}(t_{1})...I_{\alpha_{n}}(t_{n})I_{\beta_{1}}(\tau_{1})...I_{\beta_{k}}(\tau_{k})\right\rangle_{X(\theta)}^{T}\prod_{j=1}^{k}X_{\beta_{j}}(\tau_{j})d\tau_{1}...d\tau_{k};\nonumber\\ \label{eq:BK120}
\end{eqnarray}

considerando il caso $n=1$, $t_{1}=t$, la causalit\` a implica

\begin{equation}
\left\langle I(-t)\right\rangle_{\varepsilon X(-\theta)}=\left\langle I(-t)\right\rangle_{0}=-\frac{d}{dt}\left\langle Q(-t)\right\rangle_{0}=0\qquad\mbox{per $t>\theta$},\label{eq:BK14}
\end{equation}

e sempre per la causalit\` a nel caso $n=1,t_{1}=t$ gli integrali nella parte destra della (\ref{eq:BK120}) sono identicamente nulli se $\tau_{i}>t$. Con queste considerazioni la (\ref{eq:BK120}) si riduce a 

\begin{equation}
\left\langle I_{\alpha}(t)\right\rangle_{X(\theta)}=\sum_{k=1}^{\infty}\frac{(-1)^{k-1}\beta^{k}}{k!}\int_{-\infty}^{t}\left\langle I_{\alpha}(t)I_{\beta_{1}}(\tau_{1})...I_{\beta_{k}}(\tau_{k})\right\rangle_{X(\theta)}^{T}\prod_{j=1}^{k}X_{\beta_{j}}(\tau_{k})d\tau_{1}...d\tau_{k}, \label{eq:BK15}
\end{equation}

ovvero \` e possibile determinare completamente la {\em risposta non lineare del sistema}. Al primo ordine in $X$ abbiamo che

\begin{equation}
\left\langle I_{\alpha}(t)\right\rangle_{X(\theta)}=\beta\int_{-\infty}^{t}\left\langle I_{\alpha}(t)I_{\gamma}(\tau)\right\rangle_{0}X_{\gamma}(\tau)d\tau, \label{eq:BK16}
\end{equation}

ovvero l'usuale teorema di fluttuazione-dissipazione\index{teorema di fluttuazione-dissipazione}.

Il resto del lavoro \` e dedicato alla discussione di applicazioni delle relazioni (\ref{eq:BK13}) nel caso di modelli particolari, e al tentativo di ricavare equazioni simili per sistemi inizialmente in uno stato stazionario fuori dall'equilibrio; purtroppo per\`o, la discussione diventa molto oscura.

Ad ogni modo, i risultati da loro ottenuti ricordano molto nella forma quelli ottenuti a partire dalla met\` a degli anni '90; in particolare, vedremo che relazioni simili alle (\ref{eq:BK13}) (chiamate da Bochov e Kuzovlev {\em teoremi di fluttuazione a molti indici}\index{teorema di fluttuazione a molti indici}) potranno essere dedotte dal teorema di fluttuazione di Gallavotti-Cohen, {\em cf.} capitolo \ref{cap:appl}, sezioni \ref{sez:GK}, \ref{sez:ONS}.

\chapter{Applicazioni del teorema di fluttuazione} \label{cap:appl}

In questo capitolo mostreremo che, accettando l'ipotesi caotica, il teorema di fluttuazione di Gallavotti-Cohen implica delle relazioni tra funzioni di correlazione che includono la formula di Green-Kubo\index{formula di Green-Kubo} e le relazioni di reciprocit\` a di Onsager, e in generale estendono la prima al di fuori della teoria della risposta lineare.

\section{Risposta non lineare} \label{sez:GK}

Consideriamo un sistema di Anosov\footnote{In generale, possiamo pensare che $(\Omega,S)$ sia il sistema dinamico discreto che si ottiene osservando ad intervalli di tempo opportunamente prefissati un sistema dinamico che evolve in tempo continuo, {\em cf.} \cite{Ge96}, {\em i.e.} attraverso la tecnica della sezione di Poincar\'e.} $(\Omega,S)$, e definiamo $\varepsilon_{\tau}$ come la media adimensionale su un tempo $\tau$ del tasso di contrazione del volume dello spazio delle fasi $\sigma(x)$, ovvero

\begin{equation}
\varepsilon_{\tau}(x)=\frac{1}{\tau\left\langle\sigma\right\rangle_{+}}\sum_{j=-\tau/2}^{\tau/2-1}\sigma(S^{j}x);
\end{equation}

il teorema di fluttuazione afferma che\footnote{Per comodit\` a di notazione, nel seguito con $\mu_{+}(\varepsilon_{\tau}=\pm p)$ indicheremo sempre $\mu_{+}(\varepsilon_{\tau}\in I_{p,\delta})$, se $I_{p,\delta}$ \` e un intervallo $[p-\delta,p+\delta]$ con $\delta>0$ arbitrariamente piccolo.}

\begin{equation}
\frac{1}{\tau\left\langle\sigma\right\rangle_{+}p}\log\frac{\mu_{+}(\varepsilon_{\tau}=p)}{\mu_{+}(\varepsilon_{\tau}=-p)}\rightarrow_{\tau\rightarrow\infty}1\qquad\mbox{se $|p|<p^{*}$,} \label{eq:1}
\end{equation}

dove $\mu_{+}$ \` e la misura $SRB$ nel futuro, e $\mu_{+}(\varepsilon_{\tau}=p)$, $\left\langle\sigma\right\rangle_{+}$ sono rispettivamente la probabilit\` a dell'evento $\{x:\varepsilon(x)=p\}$ e il valor medio di $\sigma(x)$ calcolati con la misura SRB. Nel seguito studieremo esclusivamente sistemi {\em dissipativi}, ovvero per i quali $\left\langle\sigma\right\rangle_{+} > 0$; diremo che il sistema \` e dissipativo se $\underline{G}\neq\underline{0}$, dove $\underline{G}$ sono i parametri che regolano le intensit\` a delle forze non conservative. In generale, assumeremo che $\sigma(x)$ sia analitica in $x$ e in $\underline{G}$, e che dunque abbia la forma, sottintendendo una somma sugli indici ripetuti,

\begin{equation}
\sigma(x)=\sum_{i=1}^{s}G_{i}J_{i}^{(0)}(x)+\sum_{k\geq 2}G_{i_{1}}...G_{i_{k}}\sigma_{i_{1}...i_{k}}(x). \label{eq:1b}
\end{equation}

\subsection{Conseguenze del teorema di fluttuazione}

Introduciamo la funzione $\lambda(\beta)$ come

\begin{equation}
\lambda(\beta)=\lim_{\tau\rightarrow\infty}\frac{1}{\tau}\log\left\langle e^{\beta\tau(\varepsilon_{\tau}(x)-1)\left\langle\sigma\right\rangle_{+}}\right\rangle_{+}; \label{eq:2}
\end{equation}

si pu\` o dimostrare che $\lambda(\beta)$ \` e analitica in $\beta$ e che \` e possibile portare le derivate in $\beta$ dentro al valor medio SRB,  {\em cf.} appendici \ref{app:D}, \ref{app:B}. Dunque, definendo la quantit\` a $C_{k}$ come

\begin{equation}
C_{k}\equiv\left.\frac{\partial^{k}}{\partial\beta^{k}}\lambda(\beta)\right|_{\beta=0}\equiv\lim_{\tau\rightarrow\infty}\left\langle\tau^{k-1}\sigma_{+}^{k}\left[\varepsilon_{\tau}(x)-1\right]^{k}\right\rangle_{+}^{T}
\end{equation}

se $\sigma_{+}\equiv\left\langle\sigma\right\rangle_{+}$, e dove con $\left\langle...\right\rangle_{+}^{T}$ indichiamo l'aspettazione troncata, {\em i.e.} 

\begin{equation}
\left\langle\tau^{k-1}\sigma_{+}^{k}\left[\varepsilon_{\tau}(x)-1\right]^{k}\right\rangle_{+}^{T}\equiv\partial_{\beta}^{k}\lambda_{\tau}(\beta),\qquad\lambda_{\tau}(\beta)\equiv \frac{1}{\tau}\left\langle e^{\beta\tau(\varepsilon_{\tau}(x)-1)\left\langle\sigma\right\rangle_{+}}\right\rangle_{+},
\end{equation}

 lo sviluppo in serie di Taylor della (\ref{eq:2}) sar\` a dato da 

\begin{equation}
\lambda(\beta)=\sum_{k\geq 2}\frac{C_{k}\beta^{k}}{k!}. \label{eq:2b}
\end{equation}

Il seguente lemma sar\` a ci\` o di cui avremo bisogno per studiare la risposta non lineare.

\begin{lemma}
Il teorema di fluttuazione (\ref{eq:1}) implica che

\begin{equation}
\lambda(\beta)=\lambda(-1-\beta)-2\left\langle\sigma\right\rangle_{+}\beta-\left\langle\sigma\right\rangle_{+}. \label{eq:rel1}
\end{equation}

\end{lemma}

\begin{proof}
Se non effettuiamo il limite $\tau\rightarrow\infty$ il teorema di fluttuazione diventa, ponendo $\pi_{\tau}(p)\equiv\mu_{+}(\varepsilon_{\tau}=p)$,

\begin{equation}
\pi_{\tau}(p)=\pi_{\tau}(-p)e^{\tau\left\langle\sigma\right\rangle_{+}p}e^{O(1)}; \label{eq:nnlin3b}
\end{equation}

quindi, grazie alla (\ref{eq:nnlin3b}),

\begin{eqnarray}
\lambda(\beta)&=&\lim_{\tau\rightarrow\infty}\frac{1}{\tau}\log\int e^{\tau\left\langle\sigma\right\rangle_{+}(p-1)\beta}\pi_{\tau}(p)dp\nonumber\\&=&\lim_{\tau\rightarrow\infty}\frac{1}{\tau}\log\int e^{\tau\left\langle\sigma\right\rangle_{+}(p-1)\beta+\tau p\left\langle\sigma\right\rangle_{+}}\pi_{\tau}(-p)dp\nonumber\\&=&\lim_{\tau\rightarrow\infty}\frac{1}{\tau}\log\int e^{\tau\left\langle\sigma\right\rangle_{+}(p-1)(-\beta-1)}\pi_{\tau}(p)dp-2\beta\left\langle\sigma\right\rangle_{+}-\left\langle\sigma\right\rangle_{+}\nonumber\\&=&\lambda(-1-\beta)-2\beta\left\langle\sigma\right\rangle_{+}-\left\langle\sigma\right\rangle_{+},
\end{eqnarray}

e il lemma \` e dimostrato.
\end{proof}

Ponendo $\beta=0$ la (\ref{eq:rel1}) diventa

\begin{equation}
\left\langle\sigma\right\rangle_{+}=\lambda(-1)=\sum_{k\geq2}\frac{(-1)^{k}C_{k}}{k!}; \label{eq:rel2}
\end{equation}

come vedremo, a partire dalla (\ref{eq:rel2}) al secondo ordine, ovvero $\left\langle\sigma\right\rangle_{+}^{(2)}=\frac{C_{2}^{(2)}}{2}$, si pu\` o dedurre la formula di Green-Kubo, risultato ben noto in teoria della risposta lineare. Concludiamo la sezione dimostrando una generalizzazione della (\ref{eq:rel2}).

\begin{lemma}
Il teorema di fluttuazione per il funzionale generatore delle aspettazioni troncate $\lambda(\beta)$ (\ref{eq:rel1}) implica che

\begin{equation}
C_{n}=\sum_{k\geq0}\frac{(-1)^{k+n}C_{k+n}}{k!}\qquad\mbox{$n\geq2$}. \label{eq:rel3}
\end{equation}

\end{lemma}

\begin{proof}
Scriviamo la relazione (\ref{eq:rel1}) usando lo sviluppo di $\lambda(\beta)$ (\ref{eq:2b}),

\begin{equation}
\sum_{k\geq2}\frac{\beta^{k}C_{k}}{k!}=\sum_{k\geq2}\frac{(-1)^{k}(\beta+1)^{k}C_{k}}{k!}-2\left\langle\sigma\right\rangle_{+}\beta-\left\langle\sigma\right\rangle_{+}; \label{eq:dimrel1}
\end{equation}

poich\'e

\begin{equation}
(1+\beta)^{k}=\sum_{j=0}^{k}\left(\begin{array}{c} k \\ j\end{array}\right)\beta^{k-j} \label{eq:dimrel2}
\end{equation}

la (\ref{eq:dimrel1}) diventa

\begin{equation}
\sum_{k\geq2}\frac{\beta^{k}C_{k}}{k!}=\sum_{k\geq2}\frac{(-1)^{k}C_{k}}{k!}\sum_{j=0}^{k}\left(\begin{array}{c} k \\ j\end{array}\right)\beta^{k-j}-2\left\langle\sigma\right\rangle_{+}\beta-\left\langle\sigma\right\rangle_{+}, \label{eq:dimrel3}
\end{equation}

e dunque, per $n\geq2$,

\begin{eqnarray}
\frac{\partial^{n}}{\partial\beta^{n}}\lambda(\beta)|_{\beta=0}=C_{n}&=&\sum_{k\geq n}(-1)^{k}\frac{C_{k}}{k!}\left(\begin{array}{c} k \\ k-n\end{array}\right)n!\\&=&\sum_{k\geq n}\frac{(-1)^{k}C_{k}}{(k-n)!}\\&=&\sum_{k\geq0}\frac{(-1)^{k+n}C_{k+n}}{k!}. \label{eq:dimrel4}
\end{eqnarray}

\end{proof}

\subsection{Conseguenze del teorema di fluttuazione generalizzato}

Adesso, consideriamo una generica osservabile $\mathcal{O}$ antisimmetrica sotto l'azione dell'inversione temporale $I$, ovvero tale che $\mathcal{O}\circ I = -\mathcal{O}$, e chiamiamo $\eta_{\tau}$ la sua media temporale adimensionale su un tempo $\tau$, ovvero

\begin{equation}
\eta_{\tau}=\frac{1}{\tau\langle\mathcal{O}\rangle_{+}}\sum_{j=-\frac{\tau}{2}}^{\frac{\tau}{2}-1}\mathcal{O}\circ S^{j}.
\end{equation}

Se $\pi_{\tau}(p,q)\equiv \mu_{+}(\varepsilon_{\tau}=p,\eta_{\tau}=q)$ \` e la probabilit\` a congiunta delle variabili $p,q$, con una dimostrazione praticamente identica a quella del  teorema di fluttuazione (\ref{eq:1}) ({\em cf.} capitolo \ref{cap:noneq}, sezione \ref{sez:revcond}) si ottiene che

\begin{equation}
\lim_{\tau\rightarrow\infty}\frac{1}{\tau\left\langle\sigma\right\rangle_{+}p}\log\frac{\pi_{\tau}(p,q)}{\pi_{\tau}(-p,-q)}=1. \label{eq:GK9}
\end{equation}

Analogamente al caso precedente, ponendo $\sigma_{+}\equiv\langle\sigma\rangle_{+}$, $\mathcal{O}_{+}\equiv\langle\mathcal{O}\rangle_{+}$, introduciamo il funzionale generatore delle aspettazioni troncate delle variabili $\tau\sigma_{+}\left[\varepsilon_{\tau}-1\right]$, $\tau\mathcal{O}_{+}\left[\eta_{\tau}-1\right]$ come

\begin{equation}
\lambda(\beta_{1},\beta_{2})=\lim_{\tau\rightarrow\infty}\frac{1}{\tau}\log\left\langle e^{\tau\left(\beta_{1}\left(\varepsilon_{\tau}(x)-1\right)\left\langle\sigma\right\rangle_{+}+\beta_{2}\left(\eta_{\tau}(x)-1\right)\left\langle\mathcal{O}\right\rangle_{+}\right)}\right\rangle_{+} \label{eq:GK10}.
\end{equation}

Dal momento che $\lambda(\beta_{1},\beta_{2})$ \` e una funzione analitica\footnote{Come per il caso precedente, questa affermazione pu\` o essere dimostrata grazie alla tecnica illustrata nell'appendice \ref{app:D}.} per $\beta_{1},\beta_{2}\in(-\infty,\infty)$ pu\` o essere scritta come, sommando gli indici ripetuti,

\begin{equation}
\lambda(\beta_{1},\beta_{2})=\sum_{k=2}^{\infty}\frac{C_{\alpha_{1}...\alpha_{k}}}{k!}\beta_{\alpha_{1}}...\beta_{\alpha_{k}},\qquad\alpha_{i}=1,2; \label{eq:GKb1}
\end{equation}

se $n_{1}$, $n_{2}$, $n_{1}+n_{2}=k$, sono rispettivamente il numero di $\alpha_{i}=1$, $\alpha_{j}=2$ in $\left\{\alpha_{1},...,\alpha_{k}\right\}$, i coefficienti $C_{\alpha_{1}...\alpha_{k}}$ sono dati da

\begin{equation}
C_{\alpha_{1}...\alpha_{k}}\equiv \left.\partial_{\beta_{\alpha_{1}}}...\partial_{\beta_{\alpha_{k}}}\lambda(\beta_{1},\beta_{2})\right|_{\beta_{1}=\beta_{2}=0}\equiv \lim_{\tau\rightarrow\infty}\left\langle\tau^{k-1}\sigma_{+}^{n_{1}}\left[\varepsilon_{\tau}-1\right]^{n_{1}}\mathcal{O}_{+}^{n_{2}}\left[\eta_{\tau}-1\right]^{n_{2}}\right\rangle_{+}^{T},
\end{equation}

dove

\begin{eqnarray}
\left\langle\tau^{k-1}\sigma_{+}^{n_{1}}\left[\varepsilon_{\tau}-1\right]^{n_{1}}\mathcal{O}_{+}^{n_{2}}\left[\eta_{\tau}-1\right]^{n_{2}}\right\rangle_{+}^{T}&\equiv&\partial_{\beta_{1}}^{n_{1}}\partial_{\beta_{2}}^{n_{2}}\lambda_{\tau}(\beta_{1},\beta_{2})\nonumber\\&\equiv&\partial_{\beta_{1}}^{n_{1}}\partial_{\beta_{2}}^{n_{2}}\frac{1}{\tau}\left\langle e^{\tau\left(\beta_{1}\left(\varepsilon_{\tau}(x)-1\right)\left\langle\sigma\right\rangle_{+}+\beta_{2}\left(\eta_{\tau}(x)-1\right)\left\langle\mathcal{O}\right\rangle_{+}\right)}\right\rangle_{+}.
\end{eqnarray}

In modo del tutto simile al caso precedente, dimostreremo che il teorema di fluttuazione (\ref{eq:GK9}) implica una relazione per il funzionale generatore delle aspettazioni troncate $\lambda(\beta_{1},\beta_{2})$.

\begin{lemma}\label{lem:rel2}
Il teorema di fluttuazione (\ref{eq:GK9}) implica che

\begin{equation}
\lambda(\beta_{1},\beta_{2})=\lambda(-1-\beta_{1},-\beta_{2})-2\beta_{1}\left\langle\sigma\right\rangle_{+}-\left\langle\sigma\right\rangle_{+}-2\beta_{2}\left\langle \mathcal{O}\right\rangle_{+}. \label{eq:relaz1}
\end{equation}

\end{lemma}

\begin{proof}
Come per la (\ref{eq:rel1}), questa relazione si dimostra partendo da

\begin{equation}
\pi_{\tau}(p,q)=\pi_{\tau}(-p,-q)e^{\tau\left\langle\sigma\right\rangle_{+} p}e^{O(1)} \label{eq:dimrelaz1};
\end{equation}

infatti, grazie alla (\ref{eq:dimrelaz1}) si ha che

\begin{eqnarray}
\lambda(\beta_{1},\beta_{2})&=&\lim_{\tau\rightarrow\infty}\frac{1}{\tau}\log\int e^{\tau\left(\beta_{1}(p-1)\left\langle\sigma\right\rangle_{+}+\beta_{2}(q-1)\left\langle \mathcal{O}\right\rangle_{+}\right)}\pi_{\tau}(p,q)dpdq\nonumber\\&=&\lim_{\tau\rightarrow\infty}\frac{1}{\tau}\log\int e^{\tau\left(\beta_{1}(p-1)\left\langle\sigma\right\rangle_{+}+\beta_{2}(q-1)\left\langle \mathcal{O}\right\rangle_{+}\right)+\tau p\left\langle\sigma\right\rangle_{+}}\nonumber\\&&\times\pi_{\tau}(-p,-q)dpdq\nonumber\\&=&\lim_{\tau\rightarrow\infty}\frac{1}{\tau}\log\int e^{\tau\left(\beta_{1}(-p-1)\left\langle\sigma\right\rangle_{+}+\beta_{2}(-q-1)\left\langle \mathcal{O}\right\rangle_{+}\right)-\tau p\left\langle\sigma\right\rangle_{+}}\nonumber\\&&\times\pi_{\tau}(p,q)dpdq\nonumber\\&=&\lim_{\tau\rightarrow\infty}\frac{1}{\tau}\log\int e^{\tau\left((-\beta_{1}-1)(p-1)\left\langle\sigma\right\rangle_{+}-\beta_{2}(q-1)\left\langle \mathcal{O}\right\rangle_{+}\right)}\nonumber\\&&\times e^{\tau\left(-2\beta_{1}\left\langle\sigma\right\rangle_{+}-2\beta_{2}\left\langle \mathcal{O}\right\rangle_{+}-\left\langle\sigma\right\rangle_{+}\right)}\pi_{\tau}(p,q)dpdq\nonumber\\&=&\lambda(-1-\beta_{1},-\beta_{2})-2\left\langle\sigma\right\rangle_{+}\beta_{1}-2\left\langle \mathcal{O}\right\rangle_{+}\beta_{2}-\left\langle\sigma\right\rangle_{+}\label{eq:dimrelaz2}.
\end{eqnarray}

\end{proof}

Ponendo $\beta_{1}=-\frac{1}{2}$, $\beta_{2}=-\frac{1}{2}$ nella (\ref{eq:relaz1}) si ottiene che

\begin{equation}
\left\langle \mathcal{O}\right\rangle_{+}=\lambda\left(-\frac{1}{2},-\frac{1}{2}\right)-\lambda\left(-\frac{1}{2},\frac{1}{2}\right), \label{eq:GK15}
\end{equation}

ovvero, sostituendo nella (\ref{eq:GK15}) $\lambda(\beta_{1},\beta_{2})$ con il suo sviluppo (\ref{eq:GKb1}),

\begin{eqnarray}
\left\langle \mathcal{O}\right\rangle_{+}&=&\sum_{k\geq 2}\sum_{\alpha_{1}...\alpha_{k}}\frac{C_{\alpha_{1}...\alpha_{k}}}{k!}\left(-\frac{1}{2}\right)^{k}\nonumber\\&&-\sum_{k\geq 2}\sum_{\alpha_{1}...\alpha_{k}}\frac{C_{\alpha_{1}...\alpha_{k}}}{k!}\frac{(-1)^{\sum_{i=1}^{k}2-\alpha_{i}}}{2^{k}}\nonumber\\
&=&\sum_{k\geq 2}\left(-\frac{1}{2}\right)^{k}\frac{1}{k!}\sum_{\alpha_{1}...\alpha_{k}}C_{\alpha_{1}...\alpha_{k}}\left(1-(-1)^{\sum_{i=1}^{k}1-\alpha_{i}}\right)\label{eq:GK16}\\
&=&\sum_{k\geq 2}\frac{(-1)^{k}}{2^{k-1}}\sum_{l=0}^{\left\|\frac{k-1}{2}\right\|}\frac{C_{\underline{1}_{k-(2l+1)}\underline{2}_{(2l+1)}}}{(2l+1)!(k-(2l+1))!},\label{eq:GK016}
\end{eqnarray}

se con $\|a\|$ indichiamo la {\em parte intera} di $a$, e \

\begin{equation}
\underline{1}_{k}=\underbrace{1,1,1,...}_{k\; \scriptstyle{volte}}. 
\end{equation}

Dunque, mentre all'equilibrio il valor medio di un'osservabile antisimmetrica sotto $I$ \` e identicamente zero\footnote{Infatti, poich\'e $|\det DI|=1$ ($\det D I$ \` e il determinante della matrice jacobiana di $I$), indicando con $\left\langle\mathcal{O}\right\rangle_{0}$ la media di $\mathcal{O}$ calcolata con la misura di Liouville,

\begin{equation}
\left\langle\mathcal{O}\right\rangle_{0}=-\left\langle\mathcal{O}\circ I\right\rangle_{0}=-\left\langle\mathcal{O}\right\rangle_{0}\Rightarrow\left\langle\mathcal{O}\right\rangle_{0}\equiv 0.
\end{equation}

}, fuori dall'equilibrio diventa una somma di quantit\` a $O\left(G^{k}\right)$ per $k\geq 1$; in particolare, $\left\langle\mathcal{O}\right\rangle_{+}$ pu\` o essere espresso come somma di aspettazioni troncate con un numero {\em dispari} di\footnote{Ricordiamo che nella notazione $C_{\alpha_{1}...\alpha_{j}}$ $\alpha_{i}=1$ indica la presenza di $\sigma_{+}\left[\varepsilon_{\tau}-1\right]$, mentre $\alpha_{i}=2$ corrisponde a $\mathcal{O}_{+}\left[\eta_{\tau}-1\right]$.} $\alpha_{i}=2$. Questo fatto non \` e sorprendente, perch\'e la propriet\` a $\mathcal{O}\circ I=-\mathcal{O}$ implica che una qualunque rappresentazione di $\left\langle\mathcal{O}\right\rangle_{+}$ debba necessariamente cambiare segno sostituendo $\mathcal{O}$ con $\mathcal{O}\circ I$; ci\` o non potrebbe accadere se nello sviluppo ci fossero dei termini con un numero pari di $\alpha_{i}=2$. Nella prossima sezione vedremo che la (\ref{eq:GK016}) all'ordine pi\` u basso diventer\` a la formula di Green-Kubo; dunque, accettando l'ipotesi caotica la (\ref{eq:GK016}) pu\` o essere vista come una generalizzazione\index{generalizzazione della formula di Green-Kubo} di quest'ultimo risultato.

Notare che ponendo $\mathcal{O}=\sigma$, la (\ref{eq:GK016}) si riduce alla (\ref{eq:rel1}), poich\'e in questo caso $C_{\alpha_{1}...\alpha_{k}}\equiv C_{k}$ e, nella (\ref{eq:GK16}),

\begin{equation}
\sum_{\alpha_{1}...\alpha_{k}}\left(1-(-1)^{\sum_{i=1}^{k}1-\alpha_{i}}\right)=2^{k}-(-1)^{k}\sum_{\alpha_{1}=1}^{2}(-1)^{\alpha_{1}}...\sum_{\alpha_{k}=1}^{2}(-1)^{\alpha_{k}}=2^{k}.
\end{equation}

Infine, per concludere la sezione dimostriamo il seguente risultato.

\begin{lemma}
L'identit\` a (\ref{eq:relaz1}) implica che

\begin{equation}
C_{\underbrace{\scriptstyle{1,1,1,....}}_{l\: volte},\underbrace{\scriptstyle{2,2,2,....}}_{n-l\:volte}}=\sum_{k\geq 0}\frac{(-1)^{n+k}}{k!}C_{\underbrace{\scriptstyle{1,1,1....}}_{k+l\:volte},\underbrace{\scriptstyle{2,2,2....}}_{n-l\:volte}},\qquad n\geq 2. \label{eq:GK18}
\end{equation}

\end{lemma}

\begin{proof}
La (\ref{eq:relaz1}) diventa, sviluppando $\lambda(\beta_{1},\beta_{2})$ aspettazioni troncate,

\begin{eqnarray}
\sum_{k\geq2}\frac{\prod_{i=1}^{k}\beta_{\alpha_{i}}}{k!}C_{\alpha_{1}....\alpha_{k}}&=&\sum_{k\geq 2}\frac{(-1)^{k}\prod_{i=1}^{k}\tilde{\beta}_{\alpha_{i}}}{k!}C_{\alpha_{1}....\alpha_{k}}-2\beta_{1}\left\langle\sigma\right\rangle_{+}\nonumber\\&&-2\beta_{2}\left\langle\mathcal{O}\right\rangle_{+}-\left\langle\sigma\right\rangle_{+}, \label{eq:GK19}
\end{eqnarray}

dove $\tilde{\beta}_{1}=(1+\beta_{1})$ e $\tilde{\beta_{2}}=\beta_{2}$; poich\' e

\begin{equation}
\sum_{\alpha_{1}...\alpha_{k}}\prod_{i=1}^{k}\tilde{\beta}_{\alpha_{i}}C_{\alpha_{1}...\alpha_{k}}=\sum_{j=1}^{k}\left(\begin{array}{c} k\\j \end{array}\right)(1+\beta_{1})^{j}\beta_{2}^{k-j}C_{\underbrace{\scriptstyle{1,1...}}_{j\:volte},\underbrace{\scriptstyle{2,2...}}_{k-j\:volte}}, \label{eq:GK20}
\end{equation}

la (\ref{eq:GK19}) pu\` o essere scritta come

\begin{eqnarray}
\sum_{k\geq2}\frac{\prod_{i=1}^{k}\beta_{\alpha_{i}}}{k!}C_{\alpha_{1}....\alpha_{k}}&=&\sum_{k\geq 2}\frac{(-1)^{k}}{k!}\sum_{j=1}^{k}\left(\begin{array}{c} k\\j \end{array}\right)(1+\beta_{1})^{j}\beta_{2}^{k-j}C_{\underbrace{\scriptstyle{1,1...}}_{j\:volte},\underbrace{\scriptstyle{2,2...}}_{k-j\:volte}}\nonumber\\&&-2\beta_{1}\left\langle\sigma\right\rangle_{+}-2\beta_{2}\left\langle \mathcal{O}\right\rangle_{+}-\left\langle\sigma\right\rangle_{+}\nonumber\\&=&\sum_{k\geq 2}\frac{(-1)^{k}}{k!}\sum_{j=1}^{k}\sum_{i=1}^{j}\left(\begin{array}{c} k\\j \end{array}\right)\left(\begin{array}{c} j\\i \end{array}\right)\beta_{1}^{i}\beta_{2}^{k-j}C_{\underline{1}_{j},\underline{2}_{k-j}}\nonumber\\&&-2\beta_{1}\left\langle\sigma\right\rangle_{+}-2\beta_{2}\left\langle \mathcal{O}\right\rangle_{+}-\left\langle\sigma\right\rangle_{+}, \label{eq:GK20a}
\end{eqnarray}

dove $\underline{1}_{j}=(\overbrace{1,1,1....}^{j\:volte})$, $\underline{2}_{k-j}=(\overbrace{2,2,2...}^{k-j\:volte})$.

Quindi, per $n\geq2$,

\begin{eqnarray}
\overbrace{\partial_{\beta_{1}}\partial_{\beta_{1}}...}^{l\:volte}\overbrace{\partial_{\beta_{2}}\partial_{\beta_{2}}...}^{n-l\;volte}\sum_{k\geq2}\frac{\prod_{i=1}^{k}\beta_{\alpha_{i}}}{k!}C_{\alpha_{1}....\alpha_{k}}=C_{\underbrace{\scriptstyle{1,1,1...}}_{l\:volte},\underbrace{\scriptstyle{2,2,2...}}_{n-l\:volte}}\nonumber\\=\sum_{k\geq n}\frac{(-1)^{k}}{k!}\left(\begin{array}{c} k-n+l\\l \end{array}\right)\left(\begin{array}{c} k\\k-n+l \end{array}\right)l!(n-l)!C_{\underbrace{\scriptstyle{1,1,1...}}_{k-n+l\:volte},\underbrace{\scriptstyle{2,2,2...}}_{n-l\:volte}}\nonumber\\=\sum_{k\geq 0}\frac{(-1)^{k+n}}{k!}C_{\underbrace{\scriptstyle{1,1,1...}}_{k+l\:volte},\underbrace{\scriptstyle{2,2,2...}}_{n-l\:volte}}.\label{eq:GK21}
\end{eqnarray}

\end{proof}

Ponendo $\mathcal{O}=\sigma$, la (\ref{eq:GK21}) si riduce alla (\ref{eq:rel3}).

\section{Risposta lineare}\label{sez:ONS}

Affinch\'e la teoria sia consistente \` e necessario verificare che i risultati appena ottenuti implichino all'ordine pi\` u basso l'insieme di relazioni tra funzioni di correlazione che formano la teoria della risposta lineare, \cite{DGM}; come vedremo tra poco, dal teorema di fluttuazione sar\` a possibile ricavare due risultati notevoli come la formula di Green - Kubo e le relazioni di reciprocit\` a di Onsager. Il collegamento tra la teoria della risposta lineare e il teorema di fluttuazione \` e stato notato per la prima volta in \cite{BGGtest}, ed approfondito nei lavori successivi \cite{G96Ons}, \cite{G96Ons2}.

Per la (\ref{eq:rel2}), il teorema di fluttuazione all'ordine $G^{2}$ diventa

\begin{equation}
\left\langle\sigma\right\rangle_{+}^{(2)}=\frac{C_{2}^{(2)}}{2}; \label{eq:GK1}
\end{equation}

vogliamo mostrare che la relazione (\ref{eq:GK1}) contiene la formula di Green-Kubo\index{formula di Green-Kubo}, risultato meglio noto come {\em teorema di fluttuazione-dissipazione\index{teorema di fluttuazione-dissipazione}}.

Poich\'e $\sigma(x)=\sum_{i=1}^{s}G_{i}J_{i}^{(0)}(x)+O(G^{2})$ la (\ref{eq:GK1}) pu\` o essere riscritta come

\begin{equation}
\partial_{G_{i}}\partial_{G_{j}}\left.\left\langle\sigma\right\rangle_{+}\right|_{\underline{G}=\underline{0}}=\sum_{k=-\infty}^{\infty}\left.\left(\left\langle J_{i}^{(0)}(S^{k}\cdot)J_{j}^{(0)}(\cdot)\right\rangle_{+}-\left\langle J_{i}^{(0)}\right\rangle_{+}\left\langle J_{j}^{(0)}\right\rangle_{+}\right)\right|_{\underline{G}=\underline{0}}, \label{eq:GK1a}
\end{equation}

dove\footnote{Questa espressione \` e chiaramente formale, dal momento che stiamo scambiando l'integrale con la derivata, e non \` e ben chiaro cosa sia la derivata della misura. Rimandiamo all'appendice \ref{app:B} per una discussione rigorosa di questa questione. Notare che se il sistema imperturbato \` e un sistema di Anosov per $\underline{G}$ abbastanza piccoli il teorema di stabilit\` a strutturale, {\em cf.} \cite{AA}, afferma che anche il sistema perturbato lo sar\` a; in particolare, con una tecnica analoga a quella esposta nella sezione \ref{sez:cat} del capitolo \ref{cap:ergo}, si pu\` o dimostrare che i potenziali della misura SRB sono funzioni analitiche in $\underline{G}$.}

\begin{eqnarray}
\partial_{G_{i}}\partial_{G_{j}}\left.\left\langle\sigma\right\rangle_{+}\right|_{\underline{G}=\underline{0}}&=&\partial_{G_{i}}\partial_{G_{j}}\left.\left(\int\sigma(x)\mu_{+}(dx)\right)\right|_{\underline{G}=\underline{0}}\nonumber\\&=&\left.\int\partial_{G_{i}}\partial_{G_{j}}\sigma(x)\mu_{+}(dx)\right|_{\underline{G}=\underline{0}}\nonumber\\&&+\left.\left(\int\partial_{G_{i}}\sigma(x)\partial_{G_{j}}\mu_{+}(dx)+(i\leftrightarrow j)\right)\right|_{\underline{G}=\underline{0}}\nonumber\\&&+\left.\int\sigma(x)\partial_{G_{i}}\partial_{G_{j}}\mu_{+}(dx)\right|_{\underline{G}=\underline{0}}.\label{eq:GK2}
\end{eqnarray}

Il primo addendo della (\ref{eq:GK2}) \` e zero perch\'e $\sigma=-\sigma\circ I$, $|\mbox{det}\partial I|=1$ e la misura $\mu_{+}$ per $\underline{G}=\underline{0}$ diventa la misura di volume; poich\'e $\left.\sigma(x)\right|_{\underline{G}=\underline{0}}=0$ anche il terzo termine \` e nullo,  dunque

\begin{equation}
\partial_{G_{i}}\partial_{G_{j}}\left.\left\langle\sigma\right\rangle_{+}\right|_{\underline{G}=\underline{0}}=\left.\left(\partial_{G_{j}}\left\langle J_{i}^{(0)}\right\rangle_{+}+\partial_{G_{i}}\left\langle J_{j}^{(0)}\right\rangle_{+}\right)\right|_{\underline{G}=\underline{0}}. \label{eq:GK02}
\end{equation}

Definiamo la {\em corrente}\index{corrente} $J_{i}(x)$ e il {\em coefficiente di trasporto}\index{coefficiente di trasporto} $L_{ij}$ come

\begin{eqnarray}
J_{i}(x)&=&\partial_{G_{i}}\sigma(x) \label{eq:GK4}\\
L_{ij}&=&\left.\partial_{G_{j}}\left\langle J_{i}(x)\right\rangle_{+}\right|_{\underline{G}=\underline{0}}; \label{eq:GK5}
\end{eqnarray}

\` e chiaro che $J_{i}^{(0)}(x)=\left.J_{i}(x)\right|_{\underline{G}=\underline{0}}$, e poich\'e $J_{i}(x)\equiv\sum_{k\geq 0}J_{i}^{(k)}$ con $J_{i}^{(k)}=O\left(G^{k}\right)$, $J_{i}^{(k)}(Ix)=-J_{i}^{(k)}(x)$,

\begin{eqnarray}
L_{ij}=\partial_{G_{j}}\left.\left\langle J_{i}\right\rangle_{+}\right|_{\underline{G}=\underline{0}}&=&\partial_{G_{j}}\left.\left\langle J_{i}^{(0)}\right\rangle_{+}\right|_{\underline{G}=\underline{0}}\nonumber\\&&+\sum_{k\geq 1}\left.\partial_{G_{j}}\left\langle J_{i}^{(k)}\right\rangle_{+}\right|_{\underline{G}=\underline{0}}=\partial_{G_{j}}\left.\left\langle J_{i}^{(0)}\right\rangle_{+}\right|_{\underline{G}=\underline{0}}.\label{eq:GK6}
\end{eqnarray}

Grazie alle (\ref{eq:GK02}), (\ref{eq:GK6}) la (\ref{eq:GK1a}) diventa (moltiplicando per $\frac{1}{2}$ entrambi i membri)

\begin{equation}
\frac{L_{ij}+L_{ji}}{2}=\frac{1}{2}\sum_{k=-\infty}^{\infty}\left.\left(\left\langle J_{i}^{(0)}(S^{k}\cdot)J_{j}^{(0)}(\cdot)\right\rangle_{+}-\left\langle J_{i}^{(0)}\right\rangle_{+}\left\langle J_{j}^{(0)}\right\rangle_{+}\right)\right|_{\underline{G}=\underline{0}}, \label{eq:GK7}
\end{equation}

che, almeno per $i=j$, \` e la formula di Green-Kubo. La (\ref{eq:GK016}) al secondo ordine diventa

\begin{equation}
\left\langle G_{j}\partial_{G_{j}}\sigma\right\rangle_{+}^{(2)}=\frac{C_{12}^{(2)}}{2},
\end{equation}

che pu\` o essere riscritta esplicitamente come

\begin{eqnarray}
\left.\partial_{G_{i}}\left\langle\partial_{G_{j}}\sigma\right\rangle_{+}\right|_{\underline{G}=\underline{0}}=L_{ij}&=&\frac{1}{2}\sum_{k=-\infty}^{\infty}\left.\left(\left\langle J_{i}^{(0)}(S^{k}\cdot)J_{j}^{(0)}(\cdot)\right\rangle_{+}-\left\langle J_{i}^{(0)}\right\rangle_{+}\left\langle J_{j}^{(0)}\right\rangle_{+}\right)\right|_{\underline{G}=\underline{0}}\nonumber\\&=&\frac{1}{2}\sum_{k=-\infty}^{\infty}\left(\left\langle J_{i}^{(0)}(S^{k}\cdot)J_{j}^{(0)}(\cdot)\right\rangle_{0}-\left\langle J_{i}^{(0)}\right\rangle_{0}\left\langle J_{j}^{(0)}\right\rangle_{0}\right), \label{eq:GK13}
\end{eqnarray}

indicando con $\left\langle...\right\rangle_{0}$ la misura invariante all'equilibrio, {\em i.e.} la misura di Liouville. La (\ref{eq:GK13}) \` e il teorema di fluttuazione - dissipazione, e poich\'e la parte destra \` e simmetrica sotto lo scambio $i\leftrightarrow j$ con la (\ref{eq:GK13}) abbiamo ottenuto anche le relazioni di reciprocit\` a di Onsager 

\begin{equation}
L_{ij}=L_{ji}. \label{eq:recipons}
\end{equation}

Quindi, se il sistema ammette un'inversione temporale $I$ il teorema di fluttuazione all'ordine pi\` u basso, il secondo, equivale alla formula di Green - Kubo.

\section{Una verifica del teorema di fluttuazione}\label{sez:test}

Per concludere, applicheremo le relazioni ottenute nelle sezione \ref{sez:GK} di questo capitolo per verificare il teorema di fluttuazione in un caso semplice. Consideriamo un sistema dinamico definito sul toro $\mathbb{T}^{2}$, ottenuto identificando i lati del quadrato $[0,2\pi]\times [0,2\pi]$, che evolve ad intervalli di tempo discreti con la mappa $S$ data da

\begin{eqnarray}
S=\left(
\begin{array}{cc}
1 & 1 \\ 1 & 2
\end{array}
\right); \label{eq:10b}
\end{eqnarray}

questo semplice sistema dinamico \` e chiamato {\em gatto di Arnold}\index{gatto di Arnold}, ed \` e un tipico esempio di sistema di Anosov\index{sistema di Anosov}, {\em cf.} capitolo \ref{cap:ergo}. 

Perturbiamo $S$ in modo da avere che l'evoluzione sul toro sia descritta da

\begin{equation}
S_{\varepsilon}\underline{\psi} = S\underline{\psi}+\varepsilon\underline{f}(\underline{\psi})\qquad\mbox{mod $2\pi$}; \label{eq:11}
\end{equation}

il teorema di Anosov afferma che il sistema perturbato, per $\varepsilon$ sufficientemente piccoli, \` e ancora un sistema di Anosov dal momento che per $\varepsilon<\varepsilon_{0}$ esiste un omeomorfismo $H$ (un ``cambiamento di coordinate'') tale che

\begin{equation}
S_{\varepsilon}\circ H=H\circ S, \label{eq:12}
\end{equation}

{\em cf.} capitolo \ref{cap:ergo}. In generale, questo omeomorfismo \` e analitico in $\varepsilon$ in un disco di raggio $\varepsilon_{0}$ (dove pu\` o essere espresso come $H=1+\sum_{k\geq1}h^{(k)}$, con $h^{(k)}=O(\varepsilon^{k})$), e solo H\"older continuo in $\underline{\varphi}$.

Come abbiamo visto nel capitolo \ref{cap:noneq}, affinch\'e il teorema di fluttuazione sia vero \` e necessario che il sistema $(\Omega,S)$ sia {\em reversibile}, ovvero che esista $I$ tale che

\begin{equation}
I\circ S=S^{-1} \circ I,\label{eq:0120}
\end{equation}

e la mappa $I$ deve essere {\em differenziabile};  indicando con

\begin{equation}
\varepsilon_{\tau}=\frac{1}{\tau}\sum_{j=-\tau/2}^{\tau/2-1}\sigma\circ S^{j}
\end{equation}

la media su un tempo $\tau$ del tasso di contrazione del volume dello spazio delle fasi $\sigma$, l'esistenza di un'inversione temporale differenziabile implica che, {\em cf.} capitolo \ref{cap:noneq}, sezione \ref{sez:dimFT}, se $A_{u}$, $A_{s}$ sono rispettivamente il logaritmo del coefficiente di espansione e di contrazione, {\em cf.} formula (\ref{eq:consrev}) e definizione \ref{def:contr},

\begin{equation}
A_{u}\circ I = -A_{s},  \label{eq:excat1}
\end{equation}

ovvero

\begin{equation}
\varepsilon_{\tau}\circ I=-\varepsilon_{\tau}+O\left(\frac{1}{\tau}\right).
\end{equation}

Nel caso del gatto di Arnold imperturbato l'inversione temporale esiste ed \` e data da

\begin{equation}
I_{0}=\left(
\begin{array}{cc}
-1 & 0 \\ -1 & 1
\end{array}
\right), \label{eq:12b}
\end{equation}

ed \` e (banalmente) differenziabile; ma in questo caso il teorema di fluttuazione diventa {\em vuoto}, dal momento che la contrazione dello spazio delle fasi $\sigma=-\log\left|\mbox{det}DS\right|$ \` e identicamente uguale a zero. Nel caso della mappa perturbata $S_{\varepsilon}$, invece, questa quantit\` a \` e una funzione non banale di $\varepsilon$, $x$; inoltre, per $\varepsilon$ abbastanza piccoli la perturbazione non rompe la simmetria per inversione temporale (\ref{eq:0120}), dal momento che possiamo definire una mappa $I_{\varepsilon}$ con la propriet\` a $I_{\varepsilon}\circ S_{\varepsilon}=S_{\varepsilon}\circ I_{\varepsilon}$ ponendo

\begin{equation}
I_{\varepsilon}=H\circ I_{0}\circ H^{-1}.
\end{equation} 

Ma come per la coniugazione\index{coniugazione} $H$, $I_{\varepsilon}(\underline{\psi})$ sar\` a in generale analitica in $\varepsilon$ per $|\varepsilon|\leq\varepsilon_{0}$ e solo H\"older continua in $\underline{\psi}$, dunque l'esistenza di $I_{\varepsilon}$ non basta per soddisfare le ipotesi del teorema di fluttuazione. Inoltre, con un calcolo esplicito al primo ordine in $\varepsilon$ abbiamo che, {\em cf.} capitolo \ref{cap:ergo}, sezione \ref{sez:cat},

\begin{eqnarray}
A_{u}^{(1)}(\underline{\psi})&=&\lambda_{+}^{-1}\partial_{+}f_{+}(\underline{\psi})\\
A_{s}^{(1)}(\underline{\psi})&=&-\lambda_{-}^{-1}\partial_{-}f_{-}(\underline{\psi}),
\end{eqnarray}  

quindi\footnote{All'argomento di $A^{(1)}$ compare la mappa $I_{0}$ e non $I_{\varepsilon}$ perch\'e nel sistema perturbato i coefficienti di espansione e di contrazione sono definiti nel punto $H(\underline{\psi})$, {\em cf.} capitolo \ref{cap:ergo}, sezione \ref{sez:cat}; in realt\` a, la funzione $A_{u}(\underline{\psi})$ corrisponde a $\widetilde{A}_{u}(H(\underline{\psi}))$, se $\widetilde{A}_{u}$ \` e il logaritmo del coefficiente di contrazione definito nella (\ref{eq:cat016}).} $A^{(1)}_{u}\circ I_{0}\neq -A^{(1)}_{s}$. Perci\`o, ci aspettiamo che nel caso del gatto di Arnold perturbato il teorema di fluttuazione sia {\em falso}.

\subsection{Risultati numerici}

In realt\` a, dai risultati di alcune simulazioni numeriche, figure \ref{fig:sim1}, \ref{fig:sim2}, non \` e affatto chiaro che in questo caso il teorema di fluttuazione sia violato, dal momento che l'andamento teorico (retta verde nelle figure \ref{fig:sim1}, \ref{fig:sim2}) \` e consistente con la dispersione sperimentale; riportiamo i dettagli delle simulazioni.

\begin{itemize}
\item Le dispersioni di punti delle figure \ref{fig:sim1}, \ref{fig:sim2} sono state ottenute effettuando per $N=100$ volte un numero $T=10^{8}$ di iterazioni della mappa perturbata rispettivamente con 

\begin{equation}
\underline{f}_{1}(\underline{\psi})=\left(\begin{array}{c} \sin(\psi_{1}) \\ 0 \end{array}\right)\qquad\underline{f}_{2}(\underline{\psi})=\left(\begin{array}{c} \sin(\psi_{1})+\sin(2\psi_{1}) \\ 0 \end{array}\right), 
\end{equation}

scegliendo ogni volta delle condizioni iniziali casuali con un generatore di numeri random. 
\item In ogni punto l'errore coincide con la deviazione standard della media di $100$ misure.
\item Il tempo $\tau$ in $\varepsilon_{\tau}$ \` e pari a $100$.
\item La media $\left\langle\sigma\right\rangle_{T}$ equivale a $\frac{1}{T}\sum_{j=0}^{T-1}\sigma\circ S_{\varepsilon}^{j}$ con $T=10^{8}$.
\end{itemize} 

\begin{figure}[htbp]
\centering
\includegraphics[width=0.7\textwidth]{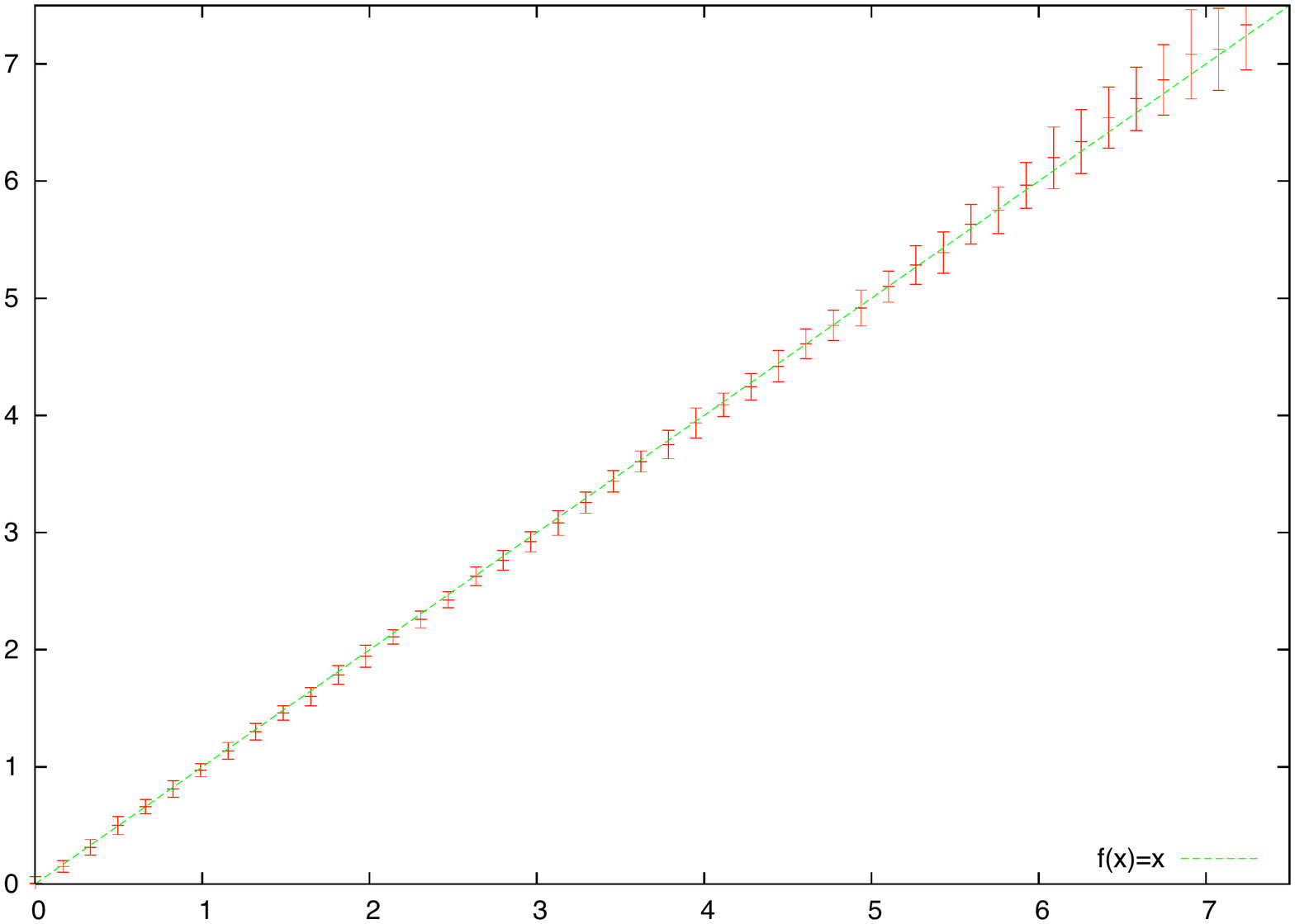}
\caption{Andamento $\frac{1}{\tau\left\langle\sigma\right\rangle_{T}}\log\frac{\mbox{Num. occorrenze di} \{\varepsilon_{\tau}=p\}}{\mbox{Num. occorrenze di} \{\varepsilon_{\tau}=-p\}}\;vs\;p$ per $T=10^{8},\tau=100,\varepsilon=0.05$, perturbando con una singola armonica; ogni punto \` e il risultato della media di $N=100$ simulazioni.} \label{fig:sim1}
\end{figure}

\begin{figure}[htbp]
\centering
\includegraphics[width=0.7\textwidth]{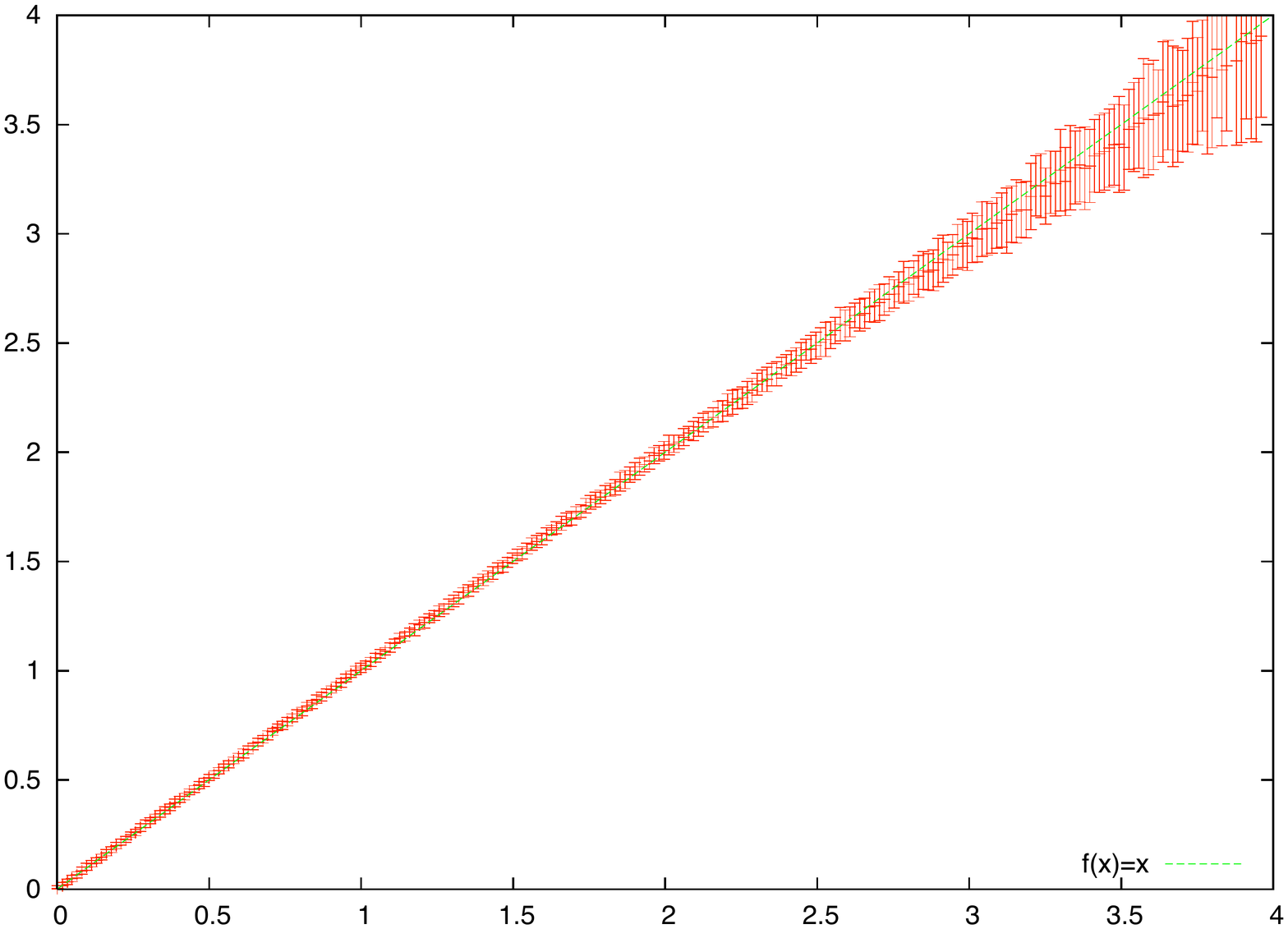}
\caption{Andamento $\frac{1}{\tau\left\langle\sigma\right\rangle_{T}}\log\frac{\mbox{Numero occorrenze di} \{\varepsilon_{\tau}=p\}}{\mbox{Numero occorrenze di} \{\varepsilon_{\tau}=-p\}}\;vs\;p$ per $T=10^{8},\tau=100$ e $\varepsilon=0.05$, perturbando con due armoniche; ogni punto \` e il risultato della media di $N=100$ simulazioni.} \label{fig:sim2}
\end{figure}

Ispirati da questi risultati numerici, poich\'e i valori medi di funzioni analitiche in $\underline{\psi}$ e $\varepsilon$ sono analitici in $\varepsilon$, {\em cf.} appendice \ref{app:B}, proviamo a verificare ordine per ordine il teorema di fluttuazione mediante le relazioni tra funzioni di correlazione ricavate nella sezione \ref{sez:GK} di questo capitolo.

\subsection{Primo caso: perturbazione composta da una singola armonica}

Perturbiamo il gatto di Arnold con

\begin{equation}
\underline{f}(\underline{\psi})=\left(
\begin{array}{c}
\sin(\psi_{1}) \\ 0
\end{array}
\right); \label{eq:16}
\end{equation}

con questa perturbazione, il tasso di contrazione dello spazio delle fasi nel punto $\underline{\varphi}=H(\underline{\psi})\in\mathbb{T}^{2}$ \` e dato da, se con $|DS_{\varepsilon}|$ indichiamo il modulo del determinante della matrice jacobiana di $S_{\varepsilon}$,

\begin{equation}
\sigma(\underline{\psi})\equiv -\log\left|DS_{\varepsilon}\right| = -\log(1+2\varepsilon\cos(\psi_{1})). \label{eq:17}
\end{equation} 

Sia $\sigma$ che la coniugazione $H$ e il tasso di espansione $A_{u}$ sono funzioni di $\varepsilon$ analitiche in $|\varepsilon|<\varepsilon_{0}$, dunque esprimibili come

\begin{eqnarray}
\sigma(\underline{\psi})&=&\sum_{k\geq 1}\sigma^{(k)}(\underline{\psi})\\
H(\underline{\psi})&=&\underline{\psi}+\sum_{k\geq 1}\underline{h}^{(k)}(\underline{\psi})\\
A_{u}(\underline{\psi})&=& \log\left(\lambda_{+}\right)+\sum_{k\geq 1}A_{u}^{(k)}(\underline{\psi}),
\end{eqnarray}

con $\sigma^{(k)},\underline{h}^{(k)},A_{u}^{(k)}=O(\varepsilon^{k})$. In quel che segue dovremo calcolare esplicitamente le derivate di valori medi relativi alla misura SRB; la derivata coinvolger\` a l'osservabile contenuto nel valor medio e la misura SRB, intesa come stato di Gibbs di un modello di Ising unidimensionale. Quindi diremo che, {\em cf.} appendice \ref{app:B},

\begin{eqnarray}
\partial_{\varepsilon}\left\langle\mathcal{O}\right\rangle_{+}&=&\left\langle\partial_{\varepsilon}\mathcal{O}\right\rangle_{+}\nonumber\\&&-\sum_{k=-\infty}^{\infty}\left(\left\langle\mathcal{O}\partial_{\varepsilon}A_{u}\circ S^{k}\right\rangle_{+}-\left\langle\mathcal{O}\right\rangle_{+}\left\langle\partial_{\varepsilon}A_{u}\circ S^{k}\right\rangle_{+}\right),\label{eq:0170}
\end{eqnarray}

e la sommatoria converge perch\'e la misura SRB \` e mescolante, {\em cf.} \cite{ergo}.

\subsubsection{Secondo ordine}

Il teorema di fluttuazione all'ordine $\varepsilon^{2}$ equivale alla relazione\footnote{Come abbiamo visto nella sezione \ref{sez:ONS}, la relazione (\ref{eq:excat2}) contiene la formula di Green-Kubo.}

\begin{equation}
\left\langle\sigma\right\rangle_{+}^{(2)}=\frac{C_{2}^{(2)}}{2};\label{eq:excat2}
\end{equation}

 le quantit\` a $\left\langle\sigma\right\rangle_{+}^{(2)}$, $C_{2}^{(2)}$ che vi compaiono possono essere scritte esplicitamente come

\begin{eqnarray}
\left\langle\sigma\right\rangle_{+}^{(2)}&=& \left\langle\sigma^{(1)}U^{(1)}\right\rangle_{0}+\left\langle\sigma^{(2)}+\partial_{\alpha}\sigma^{(1)}h_{\alpha}^{(1)}\right\rangle_{0}\label{eq:18}\\
C_{2}^{(2)}&=&\sum_{k=-\infty}^{\infty}\left\langle(\sigma\circ S^{k}-\left\langle\sigma\right\rangle_{+})\left(\sigma-\left\langle\sigma\right\rangle_{+}\right)\right\rangle_{+}^{(2)}\nonumber\\
&=& \sum_{k=-\infty}^{\infty}\left\langle\sigma^{(1)}\circ S^{k}\sigma^{(1)}\right\rangle_{0}\label{eq:19},
\end{eqnarray}

dove con $\left\langle...\right\rangle_{0}$ indichiamo che i valori di aspettazione sono calcolati con la misura di Lebesgue\index{misura di Lebesgue} $\mu_{0}$, $U(\cdot)=-\sum_{k=-\infty}^{\infty}A_{u}\circ S^{k}$, e gli indici ripetuti $\alpha=1,2$ sono sommati. Le quantit\` a $\sigma^{(1)}$, $\sigma^{(2)}$, $A_{u}^{(1)}$, $h_{\alpha}^{(1)}$ sono date da ({\em cf.} capitolo \ref{cap:ergo}, sezione \ref{sez:cat})

\begin{eqnarray}
A_{u}^{(1)}(\underline{\psi})&=&\lambda_{+}^{-1}\partial_{+}f_{+}(\underline{\psi})=\frac{\lambda_{+}^{-1}}{\lambda_{+}+1}\cos(\psi_{1})\\
h_{+}^{(1)}(\underline{\psi})&=&-\sum_{p=0}^{\infty}\lambda_{+}^{-(p+1)}f_{+}(S^{p}\underline{\psi})\nonumber\\&=&-\sum_{p=0}^{\infty}\frac{\lambda_{+}^{-(p+1)}}{(\lambda_{+}+1)^{\frac{1}{2}}}\sin(S^{p}\underline{\psi}\cdot  e_{1})
\end{eqnarray}

\begin{eqnarray}
h_{-}^{(1)}(\underline{\psi})&=&\sum_{p=1}^{\infty}\lambda_{-}^{p-1}f_{-}(S^{-p}\underline{\psi})\nonumber\\&=&\sum_{p=1}^{\infty}\frac{\lambda_{-}^{p-1}}{(\lambda_{-}+1)^{\frac{1}{2}}}\sin(S^{-p}\underline{\psi}\cdot  e_{1})\\
\sigma^{(1)}(\underline{\psi})&=&-2\cos(\psi_{1})\\
\sigma^{(2)}(\underline{\psi})&=&2\cos^{2}(\psi_{1}).
\end{eqnarray}

\` E facile vedere che i valori di aspettazione nelle (\ref{eq:18}), (\ref{eq:19}) si riducono essenzialmente a degli integrali del prodotto di due funzioni trigonometriche, calcolate in $\varphi_{1}$ e in $(S^{k}\underline{\varphi})\cdot e_{1}=S^{k}_{11}\varphi_{1}+S^{k}_{12}\varphi_{2}$ ($e_{1}^{T}=(1\quad0)$) con $k$ arbitrario; l'unico contributo diverso da zero \` e quello dovuto a $k=0$, dal momento che $S^{k=0}_{11}=1$ e $S^{k}_{11}$ \` e una funzione strettamente crescente di $|k|$\footnote{Per $k\geq0$ ci\` o \` e ovvio, dal momento che $S$ ha elementi di matrice strettamente positivi. Per $k<0$ l'affermazione \` e vera perch\' e $S^{k}_{11}=S^{-k}_{22}$, e per $k$ negativi $S^{-k}_{22}$ \` e una funzione strettamente crescente di $|k|$.\label{nota:1}}. Un semplice calcolo mostra che

\begin{eqnarray}
\left\langle\sigma\right\rangle_{+}^{(2)}&=&\varepsilon^{2}\\
C_{2}^{(2)}&=&2\varepsilon^{2},
\end{eqnarray}

 dunque la relazione $\left\langle\sigma\right\rangle_{+}^{(2)}=\frac{C_{2}^{(2)}}{2}$ \` e verificata (e con essa il teorema di fluttuazione al secondo ordine in $\varepsilon$).
 
 \subsubsection{Terzo ordine} 
 
 Al terzo ordine, le relazioni da verificare sono
 
 \begin{eqnarray}
 \left\langle\sigma\right\rangle_{+}^{(3)}&=&\frac{C_{2}^{(3)}}{2}\\
 C_{3}^{(3)}&=&0,
 \end{eqnarray}

dove

\begin{eqnarray}
\left\langle\sigma\right\rangle_{+}^{(3)}&=&\left\langle\sigma^{(3)}+\partial_{\alpha}\sigma^{(2)}h_{\alpha}^{(1)}+\partial_{\alpha}\sigma^{(1)}h_{\alpha}^{(2)}+\frac{1}{2}\partial_{\alpha\beta}\sigma^{(1)}h_{\alpha}^{(1)}h_{\beta}^{(1)}\right\rangle_{0}\nonumber\\&&+\left\langle\left(\sigma^{(2)}+\partial_{\alpha}\sigma^{(1)}h_{\alpha}^{(1)}\right)U^{(1)}\right\rangle_{0}+\left\langle\sigma^{(1)}U^{(2)}\right\rangle_{0}\nonumber\\&&+\frac{1}{2}\left\langle\sigma^{(1)}U^{(1)}U^{(1)}\right\rangle_{0}\label{eq:20}
\end{eqnarray}

\begin{eqnarray}
C_{2}^{(3)}&=&\sum_{k=-\infty}^{\infty}\left\langle(\sigma\circ S^{k}-\left\langle\sigma\right\rangle_{+})\left(\sigma-\left\langle\sigma\right\rangle_{+}\right)\right\rangle_{+}^{(3)}\nonumber\\&=&\sum_{k=-\infty}^{\infty}2\left\langle\sigma^{(1)}\circ S^{k}\left(\sigma^{(2)}+2\partial_{\alpha}\sigma^{(1)}h_{\alpha}^{(1)}-\left\langle\sigma\right\rangle_{+}^{(2)}\right)\right\rangle_{0}\nonumber\\&&+\sum_{k=-\infty}^{\infty}\left\langle\sigma^{(1)}\circ S^{k}\sigma^{(1)}U^{(1)}\right\rangle_{0}\label{eq:21}\\
C_{3}^{(3)}&=&\sum_{k,j=-\infty}^{\infty}\left\langle(\sigma\circ S^{k}-\left\langle\sigma\right\rangle_{+})(\sigma\circ S^{j}-\left\langle\sigma\right\rangle_{+})(\sigma-\left\langle\sigma\right\rangle_{+})\right\rangle_{+}^{(3)}\nonumber\\&=&\sum_{k,j=-\infty}^{\infty}\left\langle\sigma^{(1)}\circ S^{k}\sigma^{(1)}\circ S^{j}\sigma^{(1)}\right\rangle_{0}\label{eq:22}.
\end{eqnarray}

Oltre alle quantit\` a gi\` a discusse al secondo ordine, nelle (\ref{eq:20})-(\ref{eq:22}) appaiono anche ({\em cf.} capitolo \ref{cap:ergo}, sezione \ref{sez:cat})

\begin{eqnarray}
A_{u}^{(2)}(\underline{\psi})&=&-\sum_{p=0}^{\infty}\lambda_{+}^{-(p+1)}\partial_{\alpha}f_{+}(S^{p}\underline{\psi})h_{\alpha}^{(1)}(S^{p}\underline{\psi})\nonumber\\&=&-\sum_{p=0}^{\infty}\frac{\lambda_{+}^{-(p+1)}}{(\lambda_{+}+1)^{\frac{1}{2}}(\lambda_{\alpha}+1)^{\frac{1}{2}}}\cos(S^{p}\underline{\psi}\cdot  e_{1})h_{\alpha}^{(1)}(S^{p}\underline{\psi})\\
h_{-}^{(2)}(\underline{\psi})&=&\sum_{p=1}^{\infty}\lambda_{+}^{p-1}\partial_{\alpha}f_{-}(S^{-p}\underline{\psi})h_{\alpha}^{(1)}(S^{-p}\underline{\psi})\nonumber\\&=&\sum_{p=1}^{\infty}\frac{\lambda_{-}^{p-1}}{(\lambda_{-}+1)^{\frac{1}{2}}(\lambda_{\alpha}+1)^{\frac{1}{2}}}\cos(S^{-p}\underline{\psi}\cdot  e_{1})h_{\alpha}^{(1)}(S^{-p}\underline{\psi})
\end{eqnarray}

\begin{eqnarray}
A_{u}^{(2)}(\underline{\psi})&=&\lambda_{+}^{-1}\partial_{\alpha}\partial_{+}f_{+}(\underline{\psi})h_{\alpha}^{(1)}+\lambda_{+}^{-1}\partial_{-}f_{+}(\underline{\psi})k_{-}^{(1)}(\underline{\psi})\nonumber\\&&-\frac{\lambda_{+}^{-2}}{2}\partial_{+}f_{+}(\underline{\psi})\nonumber\\&=&-\frac{\lambda_{+}^{-1}}{(\lambda_{+}+1)(\lambda_{\alpha}+1)^{\frac{1}{2}}}\sin(\psi_{1})h_{\alpha}^{(1)}(\underline{\psi})\nonumber\\&&+\sum_{p=0}^{\infty}\frac{\lambda_{+}^{-2(p+1)}}{(\lambda_{+}+1)(\lambda_{-}+1)}\cos(\psi_{1})\cos(S^{-(p+1)}\underline{\psi}\cdot  e_{1})\nonumber\\&&-\frac{\lambda_{+}^{-2}}{2(\lambda_{+}+1)}\cos(\psi_{1})
\end{eqnarray}

\begin{equation}
\sigma^{(3)}(\underline{\psi})=-\frac{8}{3}\cos^{3}(\psi_{1}).
\end{equation}

Analogamente al caso precedente, i valori di aspettazione presenti nelle $(\ref{eq:20})$, $(\ref{eq:21})$, $(\ref{eq:22})$ consistono essenzialmente in integrali come

\begin{eqnarray}
\int_{0}^{2\pi}\frac{d\varphi_{1}}{2\pi}\frac{d\varphi_{2}}{2\pi}\cos(S^{k}\underline{\varphi}\cdot e_{1})\cos(S^{j}\underline{\varphi}\cdot e_{1})\cos(\varphi_{1})\label{eq:24}\\=\int_{0}^{2\pi}\frac{d\varphi_{1}}{2\pi}\frac{d\varphi_{2}}{2\pi}\cos(S^{k}_{11}\varphi_{1}+S^{k}_{12}\varphi_{2})\cos(S^{j}_{11}\varphi_{1}+S^{j}_{12}\varphi_{2})\cos(\varphi_{1})\nonumber,
\end{eqnarray}

che danno un contributo diverso da zero solo per coppie di indici $k$, $j$ tali che

\begin{eqnarray}
a_{1}S^{k}_{11} + a_{2}S^{j}_{11} + a_{3} &=& 0 \label{eq:25}\\
a_{1}S^{k}_{12} + a_{2}S^{j}_{12} &=& 0, \label{eq:26}
\end{eqnarray}

con $a_{i} = \pm 1$; affinch\' e l'integrale (\ref{eq:24}) dia un contributo diverso da zero almeno una coppia di queste equazioni deve essere soddisfatta. Come dimostreremo facilmente tra poco, ci\` o sar\` a impossibile.

L'equazione (\ref{eq:26}) con $a_{1} = -a_{2}$ ha una soluzione ovvia per $k=j$, mentre se $a_{1} = a_{2}$ la soluzione sar\` a $k=-j$, dal momento che

\begin{equation}
S^{k}_{\alpha\beta}=-S^{-k}_{\alpha\beta} \qquad\mbox{per $\alpha\neq\beta$}; \label{eq:26b}
\end{equation}

 le equazioni del tipo (\ref{eq:26}) non possono essere soddisfatte da nessun'altra coppia di indici,  perch\'e $|S^{k}_{\alpha\beta}|$ con $\alpha\neq\beta$ \` e una funzione strettamente crescente di $|k|$\footnote{Per $k>0$ l'affermazione \` e ovvia, vedere nota (\ref{nota:1}); per $k<0$ ci si riconduce al caso $k>0$ usando la propriet\` a (\ref{eq:26b}).}. Per $j=k$ la (\ref{eq:25}) diventa $\pm1=0$, mentre per $j=-k$ l'equazione da studiare \` e

\begin{equation}
S^{k}_{11}+S^{-k}_{11} + a_{3} =0, \label{eq:27}
\end{equation}

la quale non ha soluzione perch\'e $S^{k}_{11}\geq1$ e $S^{-k}_{11}=S^{k}_{22}\geq1$.

Quindi $C_{3}^{(3)}=0$, e la relazione $\left\langle\sigma\right\rangle_{+}^{(3)}=\frac{C_{2}^{(3)}}{2}$ diventa $0=0$; anche al terzo ordine il teorema di fluttuazione \` e (banalmente) verificato.

\subsubsection{Quarto ordine}

Proseguiamo nella nostra ricerca di una violazione del teorema di fluttuazione provando a verificare le relazioni

\begin{eqnarray}
C_{3}^{(4)} &=& \frac{C_{4}^{(4)}}{2} \label{eq:28}\\
\left\langle\sigma\right\rangle_{+}^{(4)} &=& \frac{C_{2}^{(4)}}{2}-\frac{C_{3}^{(4)}}{6}+\frac{C_{4}^{(4)}}{24}; \label{eq:29}
\end{eqnarray} 

per dimostrare che a questo ordine il teorema di fluttuazione \` e violato \` e sufficiente trovare che $C_{3}^{(4)}\neq\frac{C_{4}^{(4)}}{2}$. Le quantit\` a che appaiono nella (\ref{eq:28}) sono date da

\begin{eqnarray}
C_{3}^{(4)} &=& \sum_{k,j=-\infty}^{\infty}\left\langle(\sigma\circ S^{k}-\left\langle\sigma\right\rangle_{+})(\sigma\circ S^{j}-\left\langle\sigma\right\rangle_{+})(\sigma-\left\langle\sigma\right\rangle_{+})\right\rangle_{+}^{(4)}\nonumber\\&=&\sum_{k,j=-\infty}^{\infty}3\left\langle(\sigma^{(1)}\circ S^{k}\sigma^{(1)}\circ S^{j})(\sigma^{(2)}+\partial_{\alpha}\sigma^{(1)}h_{\alpha}^{(1)}-\left\langle\sigma\right\rangle_{+}^{(2)})\right\rangle_{0}\nonumber\\&&+\sum_{k,j=-\infty}^{\infty}\left\langle\sigma^{(1)}\circ S^{k}\sigma^{(1)}\circ S^{j}\sigma^{(1)}U^{(1)}\right\rangle_{0} \label{eq:30}
\end{eqnarray}

\begin{eqnarray}
C_{4}^{(4)} &=& \sum_{k,j,l=-\infty}^{\infty}\left(\left\langle(\sigma\circ S^{k}-\left\langle\sigma\right\rangle_{+})(\sigma\circ S^{j}-\left\langle\sigma\right\rangle_{+})(\sigma\circ S^{l}-\left\langle\sigma\right\rangle_{+})(\sigma-\left\langle\sigma\right\rangle_{+})\right\rangle_{+}\right.\nonumber\\&&\left.-3\left\langle(\sigma\circ S^{k}-\left\langle\sigma\right\rangle_{+})(\sigma\circ S^{j}-\left\langle\sigma\right\rangle_{+})\right\rangle_{+}\left\langle(\sigma\circ S^{l}-\left\langle\sigma\right\rangle_{+})(\sigma-\left\langle\sigma\right\rangle_{+})\right\rangle_{+}\right)^{(4)}\nonumber\\&=&\sum_{k,j,l=-\infty}^{\infty}\left\langle\sigma^{(1)}\circ S^{k}\sigma^{(1)}\circ S^{j}\sigma^{(1)}\circ S^{l}\sigma^{(1)}\right\rangle_{0}\label{eq:31}\nonumber\\&&-\sum_{k,j,l=-\infty}^{\infty}3\left\langle\sigma^{(1)}\circ S^{k}\sigma^{(1)}\circ S^{j}\right\rangle_{0}\left\langle\sigma^{(1)}\circ S^{l}\sigma^{(1)}\right\rangle_{0};\nonumber\\&&
\end{eqnarray}

discutiamo le (\ref{eq:30}), (\ref{eq:31}) separatamente. Nella (\ref{eq:31}) compaiono sia degli integrali tipo

\begin{equation}
\int_{0}^{2\pi}\frac{d\varphi_{1}}{2\pi}\frac{d\varphi_{2}}{2\pi}\cos(S^{k}\underline{\varphi}\cdot e_{1})\cos(S^{j}\underline{\varphi}\cdot e_{1})\cos(S^{l}\underline{\varphi}\cdot e_{1})\cos(\underline{\varphi}\cdot  e_{1}), \label{eq:32}
\end{equation} 

che dei prodotti di integrali come

\begin{equation}
\int_{0}^{2\pi}\frac{d\varphi_{1}}{2\pi}\frac{d\varphi_{2}}{2\pi}\cos(S^{k}\underline{\varphi}\cdot e_{1})\cos(S^{j}\underline{\varphi}\cdot e_{1})\int_{0}^{2\pi}\frac{d\varphi_{1}}{2\pi}\frac{d\varphi_{2}}{2\pi}\cos(S^{l}\underline{\varphi}\cdot e_{1})\cos(\underline{\varphi}\cdot e_{1}), \label{eq:33}
\end{equation}

e dobbiamo capire per quali terne $(k,\,j,\,l)$ le (\ref{eq:32}), (\ref{eq:33}) danno un contributo diverso da zero. Ripetendo l'analisi fatta per il secondo ordine, \` e chiaro che la (\ref{eq:33}) \` e diversa da zero solo per $k=j$, $l=0$, mentre per la (\ref{eq:32}) il discorso \` e pi\` u complesso; le equazioni da studiare sono

\begin{eqnarray}
a_{1}S^{k}_{11} + a_{2}S^{j}_{11} + a_{3}S^{l}_{11} + a_4 &=& 0 \label{eq:34}\\
a_{1}S^{k}_{12} + a_{2}S^{j}_{12} + a_{3}S^{l}_{12} &=& 0, \label{eq:35}
\end{eqnarray}

dove $a_{i}=\pm 1$. Le equazioni $(\ref{eq:34})$, $(\ref{eq:35})$, scegliendo opportunamente gli $a_{i}$, sono sicuramente vere per gli insiemi di indici $\{k,j,l\in\mathbb{Z}: k=j,\, l=0\}$, $\{k,j,l\in\mathbb{Z}: k=l,\, j=0\}$, $\{k,j,l\in\mathbb{Z}: j=l,\, k=0\}$; vogliamo dimostrare che queste scelte sono le uniche possibili. Per fare ci\` o basta notare che, se $S$ \` e data dalla (\ref{eq:10b}),

\begin{eqnarray}
S^{k+1}_{12}&=&2S^{k}_{12}+S^{k}_{11}>2S^{k}_{12}\qquad\mbox{se $k\geq0$} \label{eq:34b}\\
S^{k}_{12}&=&-S^{-k}_{12}<-2S^{-(k+1)}_{12}=2S^{k+1}_{12}\qquad\mbox{se $k<0$}, \label{eq:34c}
\end{eqnarray}

quindi, per due generici indici $m,n\in\mathbb{Z}$ tali che $|n|>|m|$

\begin{equation}
|S^{n}_{12}|>2|S^{m}_{12}|. \label{eq:36}
\end{equation} 

Grazie alla (\ref{eq:36}) \` e facile rendersi conto che l'unico modo per soddisfare la (\ref{eq:35}) \` e con una combinazione tipo $k=\pm j,\,l=0$; tra queste, il caso $k=-j$,  $l=0$ \` e da escludere perch\'e la (\ref{eq:34}) corrispondente diventerebbe

\begin{equation}
S^{k}_{11}+S^{-k}_{11} + a_{3} + a_{4} = 0, \label{eq:37}
\end{equation}

che, poich\'e $S^{k}_{11},S^{-k}_{11}>1$ per ogni $k\in\mathbb{Z}$, pu\` o essere risolta solo imponendo $k=0$, e questa scelta \` e gi\` a compresa nell'insieme $\{k,j,l\in\mathbb{Z}:k=j,\,l=0\}$. Dunque, l'insieme delle terne $(k,\,j,\, l)$ per le quali l'integrale (\ref{eq:32}) \` e diverso da zero \` e determinato dalle condizioni $k=j$,  $l=0$, e da quelle analoghe ottenute permutando $k$, $j$, $l$.

Passiamo adesso alla (\ref{eq:30}); dal momento che

\begin{equation}
\left\langle\sigma\right\rangle_{+}^{(2)}=\left\langle\sigma^{(1)}U^{(1)}\right\rangle_{0}+\left\langle\sigma^{(2)}+\partial_{\alpha}\sigma^{(1)}h_{\alpha}^{(1)}\right\rangle_{0},
\end{equation} 

indicando $\sigma\circ S^{k}$ con $\sigma_{k}$ e sottintendendo le somme su $k,\,j$, la (\ref{eq:30}) pu\` o essere riscritta come la somma delle seguenti tre quantit\` a:

\begin{eqnarray}
3\left\langle\sigma^{(1)}_{k}\sigma^{(1)}_{j}\partial_{\alpha}\sigma^{(1)}h_{\alpha}^{(1)}\right\rangle_{0}&-&3\left\langle\sigma^{(1)}_{k}\sigma^{(1)}_{j}\right\rangle_{0}\left\langle\partial_{\alpha}\sigma^{(1)}h_{\alpha}^{(1)}\right\rangle_{0}\label{eq:37b}\\3\left\langle\sigma^{(1)}_{k}\sigma^{(1)}_{j}\sigma^{(2)}\right\rangle_{0}&-&3\left\langle\sigma^{(1)}_{k}\sigma^{(1)}_{j}\right\rangle_{0}\left\langle\sigma^{(2)}\right\rangle_{0}\label{eq:37c}\\\left\langle\sigma^{(1)}_{k}\sigma^{(1)}_{j}\sigma^{(1)}U^{(1)}\right\rangle_{0}&-&3\left\langle\sigma^{(1)}_{k}\sigma^{(1)}_{j}\right\rangle_{0}\left\langle\sigma^{(1)}U^{(1)}\right\rangle_{0}.\label{eq:37d}
\end{eqnarray}

Per quanto appena discusso, la parte destra e la parte sinistra della (\ref{eq:37c}) sono diverse da zero solo per $k=j$, mentre la (\ref{eq:37d}), poich\'e $U^{(1)}=-\sum_{l=-\infty}^{\infty}A_{u}^{(1)}\circ S^{l}=-\frac{\lambda_{-}}{(\lambda_{+}+1)}\sum_{l=-\infty}^{\infty}\cos\circ S^{l}$,  d\` a un contributo non nullo considerando gli indici $\{k,j,l\in\mathbb{Z}: k=j,\,l=0\}$, $\{k,j,l\in\mathbb{Z}: k=l,\,j=0\}$, $\{k,j,l\in\mathbb{Z}: j=l,\,k=0\}$ per la parte sinistra e $\{k,j,l\in\mathbb{Z}: k=j,\,l=0\}$ per la parte destra. 

Per quanto riguarda la (\ref{eq:37b}), invece, scrivendo esplicitamente le quantit\` a contenute nei valori di aspettazione si ottiene

\begin{eqnarray}
\sum_{k,j=-\infty}^{\infty}\sum_{l=0}^{\infty}-\frac{24\lambda_{-}^{l+1}}{(\lambda_{+}+1)}\left(\left\langle\cos(S^{k}\cdot)\cos(S^{j}\cdot)\sin(\cdot)\sin(S^{l}\cdot)\right\rangle_{0}\right.\nonumber\\\left.-\left\langle\cos(S^{k}\cdot)\cos(S^{j}\cdot)\right\rangle_{0}\left\langle\sin(\cdot)\sin(S^{l}\cdot)\right\rangle_{0}\right)\nonumber\\+\sum_{k,j=-\infty}^{\infty}\sum_{l=1}^{\infty}\frac{24\lambda_{-}^{l-1}}{(\lambda_{-}+1)}\left(\left\langle\cos(S^{k}\cdot)\cos(S^{j}\cdot)\sin(\cdot)\sin(S^{-l}\cdot)\right\rangle_{0}\right.\nonumber\\\left.-\left\langle\cos(S^{k}\cdot)\cos(S^{j}\cdot)\right\rangle_{0}\left\langle\sin(\cdot)\sin(S^{-l}\cdot)\right\rangle_{0}\right); \label{eq:37e}
\end{eqnarray}

i prodotti di valori di aspettazione presenti nella (\ref{eq:37e}) sono diversi da zero solo per combinazioni di indici che verificano le condizioni $k=j,\,l=0$, ed \` e possibile verificare che queste condizioni sono le uniche che rendono non nulli gli altri termini presenti nella (\ref{eq:37e}). Infatti, per quanto discusso in precedenza riguardo alla (\ref{eq:31}), le scelte di indici che rendono il valore di aspettazione del prodotto di quattro funzioni trigonometriche\footnote{Noi l'abbiamo dimostrato solo per quattro coseni, ad ogni modo il discorso pu\` o essere ripetuto, ottenendo gli stessi risultati, considerando in generale combinazioni di seni e coseni.} (calcolate rispettivamente in $S^{k}\underline{\varphi}\cdot e_{1}$, $S^{j}\underline{\varphi}\cdot e_{1}$, $S^{l}\underline{\varphi}\cdot e_{1}$, $\varphi_{1}$) diverso da zero devono avere due indici uguali e l'indice rimanente uguale a zero; vogliamo dimostrare che nel caso dei valori di aspettazione del prodotto di due seni con due coseni presenti nella (\ref{eq:37e}) le combinazioni utili sono quelle che rispettano $k=j,\,l=0$.

Supponiamo che non sia cos\` i e consideriamo ad esempio una terna $(k,\,j,\,l)$ con $k=0,\, j=l$, per la quale l'integrale da calcolare \` e

\begin{equation}
\left\langle\cos(\cdot)\cos(S^{l}\cdot)\sin(\cdot)\sin(S^{l}\cdot)\right\rangle_{0}; \label{eq:37f}
\end{equation}

la (\ref{eq:37f}) pu\` o essere riscritta come

\begin{equation}
\frac{1}{4}\left\langle\sin(2S^{l}\cdot)\sin(2\cdot)\right\rangle_{0}, \label{eq:37g}
\end{equation}

e la (\ref{eq:37g}) \` e diversa da zero solo per $l=0$. Dunque, l'insieme di indici caratterizzato dalle condizioni $k=0,\,j=l$ contiene una sola combinazione utile, ovvero $k=j=l=0$, la quale per\` o \` e gi\` a compresa tra quelle che verificano $k=j,\,l=0$.

In conclusione, grazie a tutte queste considerazioni \` e facile verificare che 

\begin{eqnarray}
C_{3}^{(4)}&=&\frac{6\lambda_{-}\varepsilon^{4}}{(\lambda_{+}+1)}+\frac{3}{2}\varepsilon^{4}\\
C_{4}^{(4)}&=&3\varepsilon^{4},
\end{eqnarray}

quindi nel caso di una perturbazione composta da una singola armonica il teorema di fluttuazione \` e violato all'ordine $\varepsilon^{4}$.

\subsection{Secondo caso: perturbazione composta da due armoniche}

Perturbiamo adesso il sistema con

\begin{equation}
\underline{f}(\underline{\psi})=\left(
\begin{array}{c}
\sin(\psi_{1})+\sin(2\psi_{1}) \\ 0
\end{array}
\right), \label{eq:38}
\end{equation}

e con questa perturbazione il tasso di contrazione dello spazio delle fasi \` e dato da

\begin{equation}
\sigma(\underline{\psi})=-\log(1+2\varepsilon(\cos(\psi_{1})+2\cos(2\psi_{2}))); \label{eq:39}
\end{equation}

\` e facile vedere che anche in questo caso il teorema di fluttuazione \` e vero all'ordine $\varepsilon^{2}$. All'ordine $\varepsilon^{3}$, invece, le equazioni da studiare sono

\begin{eqnarray}
b_1 S^{k}_{11} + b_2 S^{j}_{11} + b_3 &=& 0 \label{eq:41}\\
b_1 S^{k}_{12} + b_2 S^{j}_{12} &=& 0, \label{eq:42}
\end{eqnarray}

 dove $b_i$ pu\` o essere uguale a $\pm 1,\pm 2$. Per quanto detto in precedenza, le equazioni del tipo (\ref{eq:42}) possono avere soluzione solo per $k=\pm j$; tra queste scelte, la (\ref{eq:41}) \` e soddisfatta solo con $k=j=0$, nella forma 

\begin{equation}
S^{k}_{11}+S^{j}_{11}-2=0
\end{equation}

oppure

\begin{equation}
-S^{k}_{11}-S^{j}_{11}+2=0.
\end{equation}

Quindi, al contrario del caso caratterizzato da una perturbazione composta da una singola armonica, esiste una combinazione di indici che risolve le equazioni (\ref{eq:41}), (\ref{eq:42}) e dunque

\begin{equation}
C_{3}^{(3)}=\sum_{k,j=-\infty}^{\infty}\left\langle\sigma^{(1)}\circ S^{k}\sigma^{(1)}\circ S^{j}\sigma^{(1)}\right\rangle_{0}=-12\varepsilon^{3}\neq0; \label{eq:40}
\end{equation}

questo risultato implica che il teorema di fluttuazione sia violato gi\` a a partire dal terzo ordine in $\varepsilon$.

\subsection{Conferme numeriche {\em a posteriori}}

Possiamo provare a verificare numericamente i risultati ottenuti calcolando esplicitamente il {\em funzionale di grandi deviazioni} $\zeta(p)$ nei casi relativi alle due perturbazioni, e confrontando la differenza $-\zeta(p)+\zeta(-p)$ con le dispersioni sperimentali ottenute per vari valori di $\varepsilon$. Il risultato \` e che la funzione $\zeta(p)$ \` e data da, al quarto ordine in $\varepsilon$,

\begin{equation}
\zeta(p)=\frac{(p-1)^{2}}{2}\left[\left\langle\sigma\right\rangle_{+}-\frac{C_{2}}{4}\right]-\frac{(p-1)^{3}}{48}C_{3}-\frac{(p-1)^{4}}{384}C_{4}+O\left(\varepsilon^{5}\right); \label{eq:funz}
\end{equation}

rimandiamo all'appendice \ref{app:C} per i dettagli di questo calcolo. Notare che la presenza di potenze di $(p-1)^{k}$ con $k>2$ nella (\ref{eq:funz}) implica che le fluttuazioni del tasso di produzione di entropia {\em non sono governate da una legge gaussiana}. 

Grazie alla (\ref{eq:funz}), poich\'e $\left\langle\sigma\right\rangle_{+}=O\left(\varepsilon^{2}\right)$ ({\em cf.} nota \ref{nota:appC1}, appendice \ref{app:C}),

\begin{eqnarray}
-\zeta(p)+\zeta(-p)&=&p\left\langle\sigma\right\rangle_{+}\left[1+A\right]+Bp^{3}+O\left(\varepsilon^{5}\right) \label{eq:funz1}\\
A&=&\frac{1}{\left\langle\sigma\right\rangle_{+}}\left[\left\langle\sigma\right\rangle_{+}-\frac{C_{2}}{2}+\frac{C_{3}}{8}-\frac{C_{4}}{48}\right]+O\left(\varepsilon^{3}\right) \label{eq:funz2}\\
B&=&\frac{1}{24}\left[C_{3}-\frac{C_{4}}{2}\right]+\left(\varepsilon^{5}\right); \label{eq:funz3}
\end{eqnarray}

per verificare la violazione del teorema di fluttuazione studieremo l'andamento di $A$ in funzione di $\varepsilon$. In un andamento tipo quelli riportati nelle figure \ref{fig:sim1},\ref{fig:sim2}, la quantit\` a $A$ pu\` o essere misurata a partire dalla pendenza intorno a $p=0$ della retta che interpola la dispersione di punti sperimentali, e calcolando la deviazione da $1$.

In generale, gli andamenti sperimentali saranno affetti da un ulteriore errore dovuto al fatto che il tempo $\tau$ di $\varepsilon_{\tau}$ \` e finito (nel nostro caso pari a $100$); la correzione pu\` o essere stimata, \cite{GZG}, e indicando con $\zeta_{\tau}(p)$ il funzionale di grandi deviazioni della probabilit\` a $Prob(\varepsilon_{\tau}=p)$ si ottiene che

\begin{equation}
\zeta_{\tau}(p)=\zeta(p)+O\left(\frac{\varepsilon}{\tau}\right).
\end{equation}

Nel caso di una perturbazione composta da una singola armonica ci aspettiamo che, a meno di errori dovuti alla finitezza di $\tau$, la quantit\` a $A$ definita nelle (\ref{eq:funz1}), (\ref{eq:funz2}) sia $O\left(\varepsilon^{2}\right)$, mentre se la perturbazione contiene due armoniche $A=O\left(\varepsilon\right)$. Ci\` o \` e confermato dagli andamenti riportati nelle figure \ref{fig:A1}, \ref{fig:A2}, in cui le dispersioni di punti sperimentali sono interpolate rispettivamente dalle curve

\begin{eqnarray}
f_{1}(\varepsilon)&=&a_{1}\varepsilon^{2}+\frac{b_{1}}{\tau\varepsilon}\\
f_{2}(\varepsilon)&=&a_{2}\varepsilon + b_{2}\varepsilon^{2} +\frac{c_{2}}{\tau\varepsilon},
\end{eqnarray}

dove i parametri sono determinati minimizzando la somma dei quadrati dei residui.
 
\begin{figure}[htbp]
\centering
\includegraphics[width=0.7\textwidth]{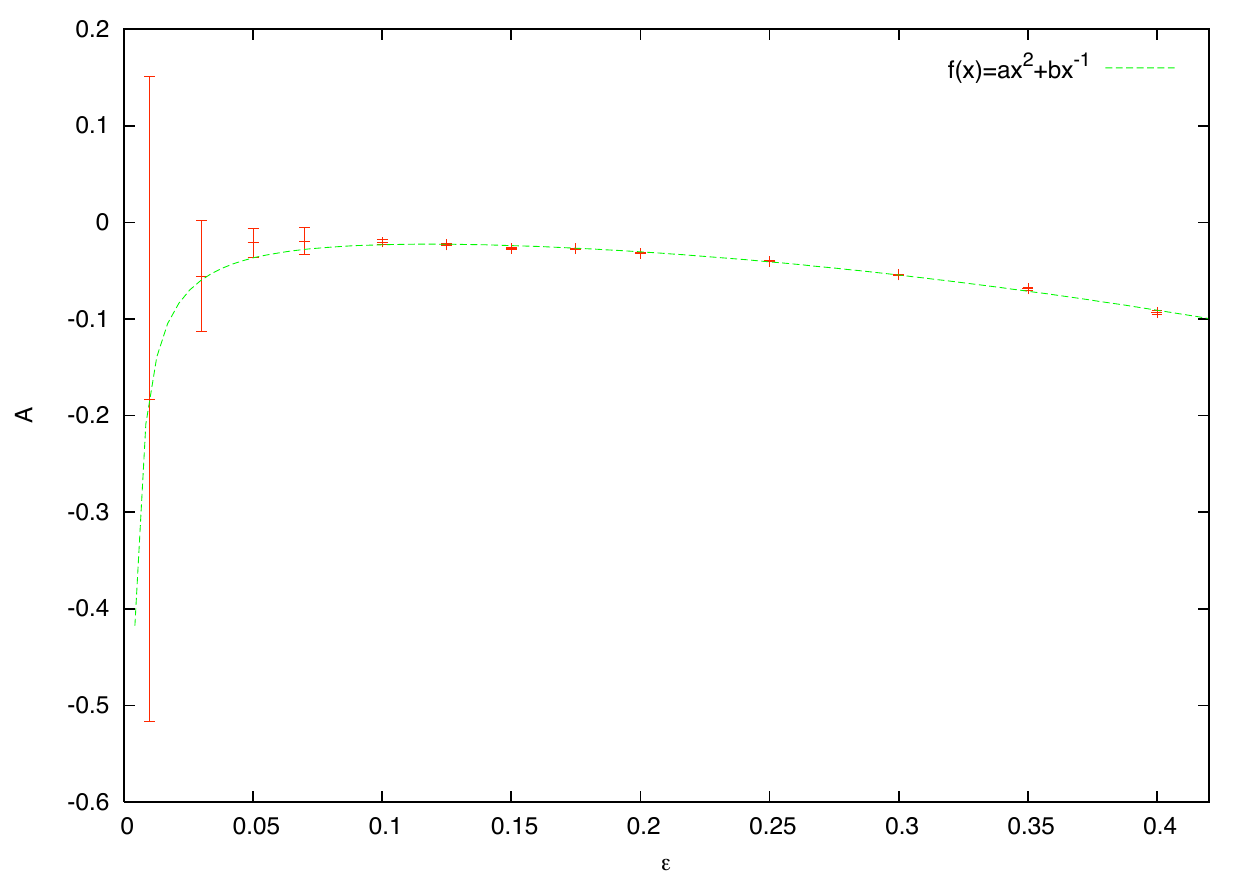}
\caption{Andamento di $A$ in funzione di $\varepsilon$ nel caso di una perturbazione composta da una singola armonica; la dispersione di punti \` e interpolata dalla curva $f_{1}(\varepsilon)$.} \label{fig:A1}
\end{figure}

\begin{figure}[htbp]
\centering
\includegraphics[width=0.7\textwidth]{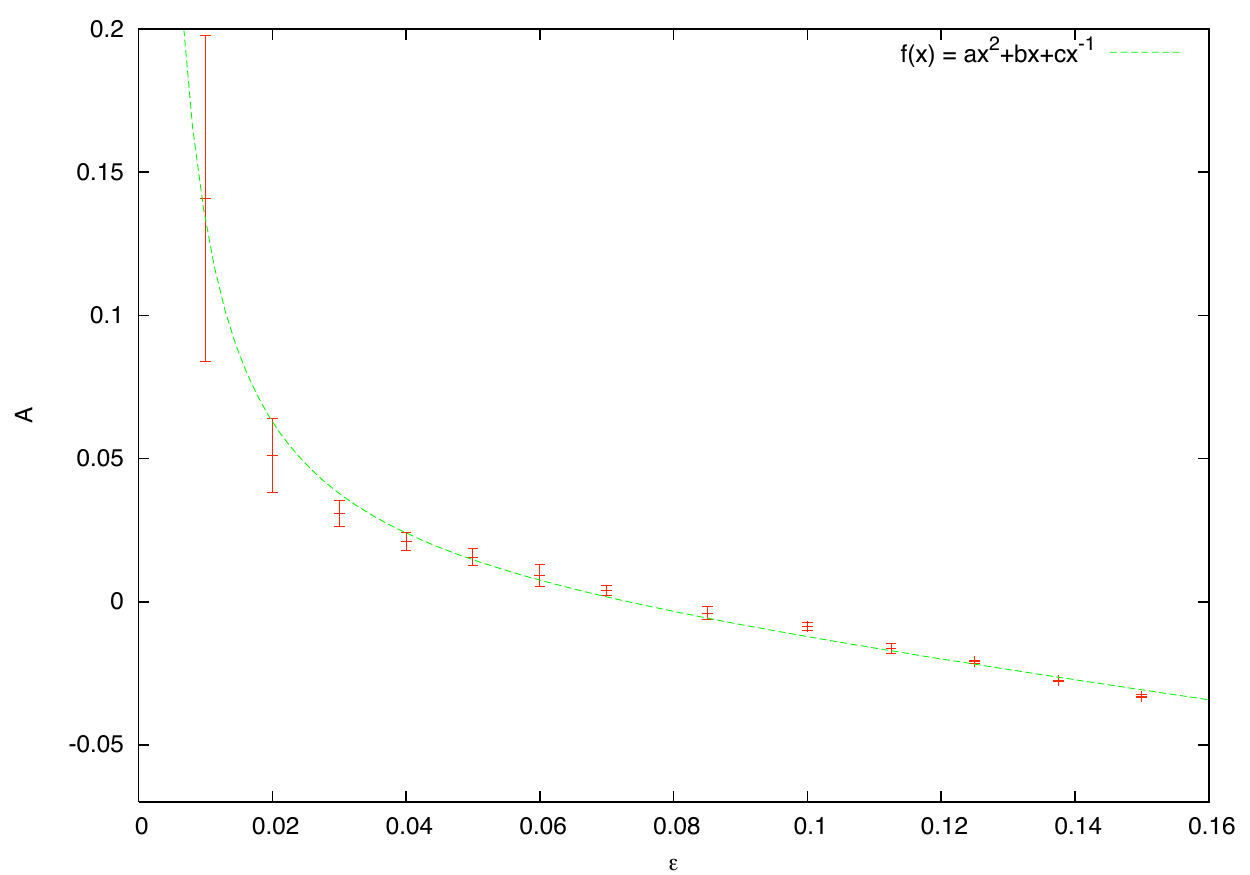}
\caption{Andamento di $A$ in funzione di $\varepsilon$ nel caso di una perturbazione composta da due armoniche; la dispersione di punti \` e interpolata dalla curva $f_{2}(\varepsilon)$.} \label{fig:A2}
\end{figure}

\chapter*{Conclusioni}

In questa tesi abbiamo discusso il problema della descrizione statistica dei sistemi fuori dall'equilibrio; in particolare, abbiamo studiato gli aspetti fisici e matematici di una possibile teoria microscopica degli stati stazionari di non equilibrio proposta da G. Gallavotti e da E. G. D. Cohen, \cite{GC95}, \cite{GC95a}, \cite{trattgal}, \cite{conv}. Questa teoria si basa su un'{\em assunzione}, motivata da risultati sperimentali, che consiste nella seguente affermazione:

\begin{quote}
Ipotesi caotica (IC): \em{Allo scopo di misurare propriet\` a macroscopiche, un sistema di particelle in uno stato stazionario si comporta come un sistema di Anosov reversibile.}
\end{quote}

La conseguenza principale di (IC) \` e l'esistenza di una misura di probabilit\` a invariante, la {\em misura SRB}, che determina la statistica di dati iniziali scelti a caso con la misura di volume, {\em cf.} capitolo \ref{cap:ergo}, sezione \ref{sez:srb}; in particolare, se il sistema dinamico \` e sottoposto a forze non conservative la misura invariante \` e {\em singolare} rispetto alla misura di Lebesgue, ovvero \` e concentrata in un insieme di misura di volume nulla.

Una caratteristica molto importante della misura SRB \` e che pu\` o essere rappresentata come lo stato di Gibbs di un {\em modello di Fisher}, ovvero di un modello di spin unidimensionale con potenziali a molti corpi invarianti per traslazione che decadono esponenzialmente; questi sistemi possono essere studiati in dettaglio, e in particolare \` e possibile dimostrare che non presentano transizioni di fase, nel senso che se i potenziali sono analitici in un parametro $\beta$ allora anche la {\em pressione} lo sar\` a, indipendentemente dalla taglia dei potenziali, {\em cf.} appendici \ref{app:D}, \ref{app:B}. 

Il contatto tra la teoria matematica dei sistemi di Anosov e la Fisica \` e contenuto nel {\em teorema di fluttuazione di Gallavotti - Cohen}; questo risultato \` e stato ispirato da una simulazione numerica in \cite{ECM93}, ed in seguito \` e stato verificato in molti altri esperimenti, {\em cf.} ad es. \cite{BGGtest}, \cite{GZG}, \cite{BCG}. Consideriamo un sistema di Anosov $(\Omega,S)$ che evolve ad intervalli di tempo discreti\footnote{In generale, possiamo pensare che $(\Omega,S)$ sia il sistema dinamico che si ottiene osservando ad intervalli di tempo ``opportuni'' un sistema dinamico che evolve in tempo continuo (un {\em flusso}, {\em cf.} definizione \ref{def:flux}); questa idea \` e alla base della tecnica della {\em sezione di Poincar\'e}. Il sistema dinamico risultante dalle ``osservazioni'' ha tutte le caratteristiche dei sistemi di Anosov, \cite{Ge96}, e il teorema di fluttuazione pu\` o essere dimostrato in un modo completamente analogo a quello che abbiamo visto nel capitolo \ref{cap:noneq}, sezione \ref{sez:dimFT}.}, e indichiamo con $\langle\sigma\rangle_{+}$ il valor medio SRB nel futuro del tasso di produzione di entropia\footnote{In realt\` a, il teorema \` e una relazione sulla fluttuazioni del {\em tasso di contrazione dello spazio delle fasi}; ad ogni modo, come abbiamo discusso nella sezione \ref{sez:problentr} del capitolo \ref{cap:noneq}, questa quantit\` a pu\` o essere interpretata come tasso di produzione di entropia.} $\sigma$ e con $\varepsilon_{\tau}$ la sua media adimensionale su un tempo $\tau$, ovvero

\begin{equation}
\varepsilon_{\tau}\equiv \frac{1}{\tau\langle\sigma\rangle_{+}}\sum_{j=-\frac{\tau}{2}}^{\frac{\tau}{2}-1}\sigma\circ S^{j};
\end{equation}

il teorema afferma che, ponendo $I_{p,\delta}\equiv[p-\delta,p+\delta]$ con $\delta>0$ arbitrariamente piccolo, e chiamando $\pi(\varepsilon_{\tau}\in I_{p,\delta})$ la probabilit\` a dell'evento $\{x:\varepsilon_{\tau}(x)\in I_{p,\delta}\}$,

\begin{equation}
\frac{\pi(\varepsilon_{\tau}\in I_{p,\delta})}{\pi(\varepsilon_{\tau}\in I_{-p,\delta})}=e^{\tau p \langle\sigma\rangle_{+}+O(1)}\qquad |p|\leq p^{*}\geq 1. \label{eq:conflu}
\end{equation}

Inoltre, se $\{F_{1},...,F_{n}\}$ \` e un insieme di osservabili con parit\` a definita sotto l'azione dell'inversione temporale $I$, il risultato (\ref{eq:conflu}) pu\` o essere generalizzato nel seguente modo, {\em cf.} sezione \ref{sez:revcond} del capitolo \ref{cap:noneq}:

\begin{equation}
 \frac{\pi(F_{j}(S^{k}x)\sim \varphi_{j}(k),\varepsilon_{\tau}\in I_{p,\delta}, j=1,...,n, k\in [-\frac{\tau}{2},\frac{\tau}{2}-1])}{\pi(F_{j}(S^{k}x)\sim I\varphi_{j}(k),\varepsilon_{\tau}\in I_{-p,\delta}, j=1,...,n, t\in [-\frac{\tau}{2},\frac{\tau}{2}-1])}=e^{\tau p \langle\sigma\rangle_{+}+O(1)} \quad|p|\leq p^{*},
\end{equation}

dove $I\varphi(k)\equiv \eta_{j}\varphi(-k)$ se $\eta_{j}$ \` e la parit\` a di $F_{j}$ sotto $I$, e con la notazione $F_{j}(S^{k}x)\sim \varphi_{j}(k)$ indichiamo che $F_{j}(S^{k}x)$ \` e ``vicina'' a $\varphi(k)$, nel senso che l'evoluto di $F_{j}$ in un intervallo temporale $[-\frac{\tau}{2},\frac{\tau}{2}-1]$ \` e contenuto in un ``tubo'' di larghezza $\gamma$ attorno a $\varphi_{j}$.

Nella parte originale della tesi abbiamo dimostrato che questi teoremi equivalgono ad un insieme di infinite relazioni tra funzioni di correlazione, che includono come casi particolari ``vicino all'equilibrio'' la formula di Green - Kubo e le relazioni di reciprocit\` a di Onsager, {\em cf.} capitolo \ref{cap:appl}. In particolare, il valor medio di una qualunque osservabile $\mathcal{O}$ dispari sotto inversione temporale pu\` o essere scritto come

\begin{equation}
\langle\mathcal{O}\rangle_{+}=\sum_{k\geq 2}\frac{(-1)^{k}}{2^{k-1}}\sum_{l=0}^{\left\|\frac{k-1}{2}\right\|}\frac{C_{\underline{1}_{k-(2l+1)}\underline{2}_{(2l+1)}}}{(2l+1)!(k-(2l+1))!},\label{eq:GK016bis}
\end{equation}

dove $\underline{1}_{k}=\underbrace{1,1,1,...}_{k\; volte}$, $\|a\|$ \` e la parte intera di $a$ e\footnote{Come dimostriamo nell'appendice \ref{app:B}, nella (\ref{eq:conasp}) \` e possibile portare le derivate in $\beta_{1}$, $\beta_{2}$ dentro al valor medio SRB.}

\begin{eqnarray}
C_{\underline{1}_{n}\underline{2}_{m}}&\equiv&\lim_{\tau\rightarrow\infty}\partial_{\beta_{1}}^{n}\partial_{\beta_{2}}^{m}\lambda_{\tau}(\beta_{1},\beta_{2})  \label{eq:conasp}\\
\lambda_{\tau}(\beta_{1},\beta_{2})&\equiv&\frac{1}{\tau}\left\langle e^{\tau\left(\beta_{1}\left(\varepsilon_{\tau}-1\right)\left\langle\sigma\right\rangle_{+}+\beta_{2}\left(\eta_{\tau}-1\right)\left\langle\mathcal{O}\right\rangle_{+}\right)}\right\rangle_{+}\\
\eta_{\tau}&\equiv&\frac{1}{\tau\langle\mathcal{O}\rangle_{+}}\sum_{j=-\frac{\tau}{2}}^{\frac{\tau}{2}-1}\mathcal{O}\circ S^{j};
\end{eqnarray}

in modo analogo \` e possibile dimostrare anche che

\begin{equation}
C_{\underbrace{\scriptstyle{1,1,1,....}}_{l\: \scriptstyle{volte}},\underbrace{\scriptstyle{2,2,2,....}}_{n-l\:\scriptstyle{volte}}}=\sum_{k\geq 0}\frac{(-1)^{n+k}}{k!}C_{\underbrace{\scriptstyle{1,1,1....}}_{k+l\:\scriptstyle{volte}},\underbrace{\scriptstyle{2,2,2....}}_{n-l\:\scriptstyle{volte}}},\qquad n\geq 2. \label{eq:conrel}
\end{equation}

Inoltre, {\em cf.} appendice \ref{app:C}, ammettendo che le forze non conservative siano caratterizzate da ``piccole'' intensit\` a $\{G_{i}\}$, abbiamo trovato un'espressione esplicita in serie di potenze di $\{G_{i}\}$ del funzionale di grandi deviazioni $\zeta(p)$, definito come

\begin{equation}
\zeta(p)=-\lim_{\tau\rightarrow\infty}\tau^{-1}\pi(\varepsilon_{\tau}\in I_{p,\delta}); \label{eq:conzeta}
\end{equation}

il nostro risultato \` e che $\zeta(p)=\sum_{k\geq 2}a_{k}(p-1)^{k}$, dove i coefficienti $a_{k}$ consistono in opportune combinazioni di funzioni di correlazione  $\sigma - \sigma$. Ad esempio, fino al quarto ordine in $\{G_{i}\}$ si ottiene che, se $C_{k}=C_{\underline{1}_{k},0}$ e $C_{\underline{1}_{k},0}$ \` e definita nella (\ref{eq:conasp}),

\begin{equation}
\zeta(p)=\frac{(p-1)^{2}}{2}\left[\left\langle\sigma\right\rangle_{+}-\frac{C_{2}}{4}\right]-\frac{(p-1)^{3}}{48}C_{3}-\frac{(p-1)^{4}}{384}C_{4}+O\left(G^{5}\right), \label{eq:appC8bis}
\end{equation}

indicando con la notazione $O(G^{k})$ un termine che all'ordine dominante \` e proporzionale al prodotto di $k$ variabili appartenenti a $\{G_{i}\}$. Dunque, poich\'e in generale $a_{k}\neq 0$ se $\underline{G}\neq \underline{0}$, le fluttuazioni asintotiche di $\varepsilon_{\tau}$ {\em non} seguono una distribuzione gaussiana\footnote{Ci\` o non \` e sorprendente, dal momento che la media $\varepsilon_{\tau}$ \` e normalizzata con $\tau$ e non con $\sqrt{\tau}$.}. 

Il teorema di fluttuazione (\ref{eq:conflu}), nella forma

\begin{equation}
-\zeta(p) + \zeta(-p)=p\langle\sigma\rangle_{+}\qquad |p|\leq p^{*},
\end{equation}

implica che i coefficienti $a_{k}$ siano legati da delle precise relazioni; com'era prevedibile, abbiamo verificato fino al quarto ordine in $\{G_{i}\}$ che le relazioni che si ottengono sono esattamente le (\ref{eq:GK016bis}), (\ref{eq:conrel}), sostituendo $\mathcal{O}$ con $\sigma$. In particolare, queste relazioni non modificano il carattere non gaussiano di $\zeta(p)$, che al quarto ordine diventa

\begin{equation}
\zeta(p)=\frac{(p-1)^{2}}{8}C_{2}-\frac{(p-1)^{2}}{48}C_{4}\left(1+\frac{(p-1)}{2}+\frac{(p-1)^{2}}{8}\right)+O\left(G^{5}\right).
\end{equation}

Infine, abbiamo sfruttato i risultati ottenuti per verificare il teorema di fluttuazione ordine per ordine in un caso molto semplice, {\em cf.} sezione \ref{sez:test} del capitolo \ref{cap:appl}, in cui da alcune simulazioni numeriche il teorema sembrava essere verificato sebbene non ne fossero rispettate le ipotesi; la conclusione \` e che la violazione c'\` e, ma era troppo piccola per essere apprezzata nelle simulazioni. {\em A posteriori} abbiamo verificato numericamente questa violazione con delle nuove simulazioni, finalizzate allo studio della pendenza vicino a $p=0$ di $p^{-1}\langle\sigma\rangle_{+}^{-1}\left[-\zeta(p)+\zeta(-p)\right]$; le deviazioni osservate di questa quantit\` a da $1$, il valore atteso nel caso in cui il teorema di fluttuazione sia vero, sono decisamente consistenti con le nostre previsioni teoriche.

\begin{appendix}

\appendix

\chapter{Derivazione euristica della misura SRB}\label{app:eur}

In questa appendice vogliamo presentare un semplice argomento euristico per dimostrare che la misura SRB pu\` o essere rappresentata come uno stato di Gibbs le cui probabilit\` a condizionate verificano le equazioni

\begin{equation}
\frac{m(\underline{\sigma}'_{\Lambda}|\underline{\sigma}_{\Lambda^{c}})}{m(\underline{\sigma}''_{\Lambda}|\underline{\sigma}_{\Lambda^{c}})}=\exp{\left(\sum_{k=-\infty}^{\infty}A_{u}(\tau^{k}\underline{\sigma}'')-A_{u}(\tau^{k}\underline{\sigma}')\right)}, \label{eq:appA1}
\end{equation}

se $\underline{\sigma}'=\left(\underline{\sigma}'_{\Lambda}\underline{\sigma}'_{\Lambda^{c}}\right)$, $\underline{\sigma}''=\left(\underline{\sigma}''_{\Lambda}\underline{\sigma}''_{\Lambda^{c}}\right)$; seguiremo l'analisi svolta in \cite{trattgal}.

Consideriamo un sistema di Anosov $(\Omega,S)$, ed un suo {\em punto fisso}\footnote{Ad esempio, per il gatto di Arnold il punto $\underline{\psi}=(0,0)$ \` e banalmente un punto fisso. L'analisi euristica che stiamo per presentare rimane immutata se al posto di un punto fisso $O$ consideriamo un {\em punto periodico} $O'$, ovvero un punto che \` e fisso per una certa iterata $S^{k}$; si pu\` o dimostrare, {\em cf.} \cite{ergo}, che per un sistema di Anosov $(\Omega,S)$ i punti fissi sono {\em densi} in $\Omega$.}\index{punto fisso} $O$; un punto fisso $O\in\Omega$ \` e caratterizzato dalla propriet\` a $SO=O$. Al tempo $t=0$ consideriamo una sfera $B$ centrata in $O$ contenente una massa unitaria distribuita con densit\` a $\rho$; dopo un tempo $T$ ``molto grande'', l'insieme $S^{T}B$ sar\` a concentrato lungo la variet\` a instabile $W^{u}_{O}$ passante per $O$. Un segmento $\delta$ (consideriamo il caso bidimensionale per semplicit\` a) giacente lungo la variet\` a instabile $W^{u}_{O}\cap S^{T}B$ centrato nel punto $x$ corrisponder\` a ad un segmento $S^{-T}\delta\in B$ centrato in $S^{-T}x$, pi\`u ``corto'' di $\delta$ di un fattore $\prod_{k=0}^{T-1}\lambda_{u}^{-1}(S^{-k}x)$, dunque la massa che riveste $\delta$ sar\`a data da\footnote{Se $n$ \` e la dimensione della variet\` a stabile, dobbiamo moltiplicare  $\rho S^{-T}\delta$ per $h^{n}$ in modo da avere un risultato adimensionale.}

\begin{equation}
\mu(\delta)\sim\rho h\prod_{k=0}^{T-1}\lambda_{u}^{-1}(S^{-k}x)|\delta|, \label{eq:appA2}
\end{equation}

dove $h$ \` e il diametro della sfera $\Sigma_{\gamma}$ nella direzione stabile. Per grandi $T$ l'evoluto $S^{T}B$ di $B$ coprir\` a una porzione sempre pi\` u lunga di variet\` a instabile, quindi possiamo dire che la misura SRB, che \` e il limite per tempi infiniti della misura di volume, \` e {\em concentrata} sulla variet\` a instabile. Se consideriamo due segmenti $\delta$, $\delta'$ appartenenti alla stessa porzione di variet\` a instabile $W^{u}_{\gamma'}(x)$ e centrati in due punti $x$, $x'$ giacenti sulla stessa porzione di variet\` a stabile $W^{s}_{\gamma}(x)$, in modo che i segmenti formino le basi di un ``rettangolo'' come illustrato in figura \ref{fig:app1}, il rapporto delle masse che li ricopre \` e dato da

\begin{figure}[htbp]
\centering
\includegraphics[width=0.7\textwidth]{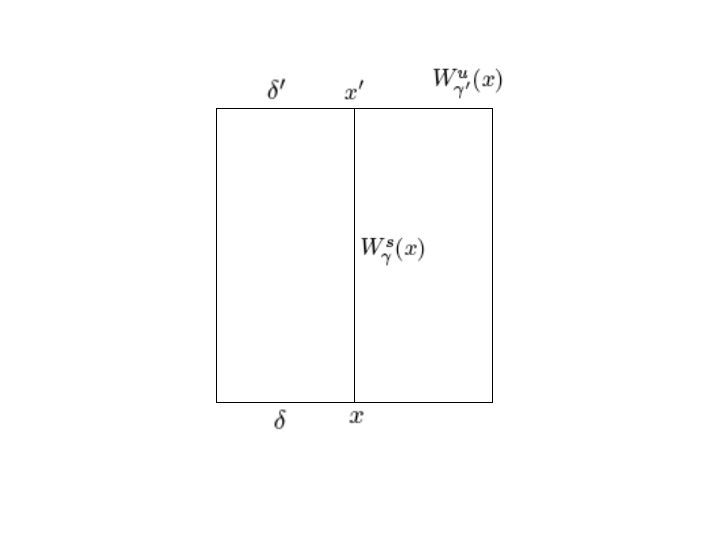}
\caption{I punti $x$ e $x'$ si trovano sulla stessa variet\` a stabile, e i segmenti $\delta$, $\delta'$ sulla stessa variet\` a instabile.} \label{fig:app1}
\end{figure}

\begin{eqnarray}
\frac{\mu(\delta)}{\mu(\delta')}&=&\frac{\prod_{k=0}^{T-1}\lambda_{u}^{-1}(S^{-k}x)|\delta|}{\prod_{k=0}^{T-1}\lambda_{u}^{-1}(S^{-k}x)|\delta'|}=\frac{\prod_{k=0}^{T-1}\lambda_{u}^{-1}(S^{-k}x)|S^{-T}\circ S^{T}\delta|}{\prod_{k=0}^{T-1}\lambda_{u}^{-1}(S^{-k}x')|S^{-T}\circ S^{T}\delta'|}\\
&=&\frac{\prod_{k=0}^{T-1}\lambda_{u}^{-1}(S^{-k}x)\prod_{k=1}^{T}\lambda_{u}^{-1}(S^{k}x)|S^{T}\delta|}{\prod_{k=0}^{T-1}\lambda_{u}^{-1}(S^{-k}x')\prod_{k=1}^{T}\lambda_{u}^{-1}(S^{k}x)|S^{T}\delta'|}; \label{eq:appA3} 
\end{eqnarray}

la (\ref{eq:appA3}) diventa, nel limite $T\rightarrow\infty$,

\begin{equation}
\lim_{T\rightarrow\infty}\frac{\mu(\delta)}{\mu(\delta')}=\frac{\prod_{k=-\infty}^{\infty}\lambda_{u}^{-1}(S^{k}x)}{\prod_{k=-\infty}^{\infty}\lambda_{u}^{-1}(S^{k}x')}, \label{eq:appA4}
\end{equation}

dal momento che $S^{T}\delta$ e $S^{T}\delta'$ tendono a coincidere. Il prodotto (\ref{eq:appA4}) converge perch\'e i punti $x$ e $x'$ appartengono alla stessa variet\` a stabile e alla stessa variet\` a instabile, e le funzioni $\log\lambda_{u}(x)$ sono H\"older continue in $x$. Quindi, codificando i punti $x$ e $x'$ con le sequenze $\underline{\sigma}'=(\underline{\sigma}'_{\Lambda}\underline{\sigma}_{\Lambda^{c}})$, $\underline{\sigma}''=(\underline{\sigma}''_{\Lambda}\underline{\sigma}_{\Lambda^{c}})$, con stessa continuazione in $\Lambda^{c}$ dal momento che i punti $x$, $x'$ si trovano sulla stessa variet\` a stabile e instabile, la (\ref{eq:appA4}) diventa la (\ref{eq:appA1}).

\chapter{Calcolo perturbativo di $\zeta(p)$} \label{app:C}

In questa appendice vogliamo discutere un calcolo perturbativo che ci permetter\` a di calcolare esplicitamente, ad ogni ordine nei parametri perturbatori $\left\{G_{i}\right\}$, il {\em funzionale di grandi deviazioni}\index{funzionale di grandi deviazioni} $\zeta(p)$ definito come

\begin{equation}
\zeta(p)=-\lim_{\tau\rightarrow\infty}\frac{1}{\tau}\log\mu_{+}(\varepsilon_{\tau}=p); \label{eq:appC1}
\end{equation}

nella (\ref{eq:appC1}), $\mu_{+}(\varepsilon_{\tau}=p)$ \` e la probabilit\` a SRB dell'evento $\{x:\varepsilon_{\tau}(x)=p\}$ e, indicando con $\sigma$ il tasso di contrazione dello spazio delle fasi\footnote{Come al solito, ci limitiamo a considerare il caso di un sistema dinamico discreto $(\Omega,S)$.} $\sigma=-\log\left|DS\right|$, la quantit\` a $\varepsilon_{\tau}$ \` e data da

\begin{equation}
\varepsilon_{\tau}=\frac{1}{\tau\left\langle\sigma\right\rangle_{+}}\sum_{j=-\tau/2}^{\tau/2-1}\sigma\circ S^{j}.
\end{equation}

Alternativamente, $\zeta(p)$ pu\` o essere introdotta come la trasformata di Legendre\index{trasformata di Legendre} di

\begin{eqnarray}
\lambda(\beta)=\lim_{\tau\rightarrow\infty}\frac{1}{\tau}\log\left\langle e^{\tau\beta\left\langle\sigma\right\rangle_{+}\left(\varepsilon_{\tau}-1\right)}\right\rangle_{+}&\equiv&\lim_{\tau\rightarrow\infty}\frac{1}{\tau}\log\int e^{\tau\beta\left\langle\sigma\right\rangle_{+}\left(p-1\right)}\pi_{\tau}(p)dp\nonumber\\&&\qquad\qquad\mbox{se $\pi_{\tau}(p)\equiv \mu_{+}(\varepsilon_{\tau}=p)$},\nonumber
\end{eqnarray}

ovvero

\begin{equation}
\zeta(p)=\max_{\beta}\left[\beta\left\langle\sigma\right\rangle_{+}(p-1)-\lambda(\beta)\right]. \label{eq:appC2}
\end{equation}

Dal momento che $\lambda(\beta)$ \` e strettamente convessa\footnote{$\lambda(\beta)$ pu\` o essere vista come la pressione di un modello di spin unidimensionale; \` e noto che, in assenza di transizioni di fase, l'energia libera \` e una funzione strettamente convessa nei suoi argomenti, \cite{RuStat}, \cite{thomp}.} e {\em asintoticamente lineare}, ovvero esiste\footnote{$p^{*}$ \` e definito come $p^{*}=\sup_{x}\limsup_{\tau}\varepsilon_{\tau}(x)$.} $p^{*}>1$ tale che

\begin{equation}
\lim_{\beta\rightarrow\pm\infty}\beta^{-1}\lambda(\beta)=\left\langle\sigma\right\rangle_{+}\left(\pm p^{*}-1\right), \label{eq:appC3.0}
\end{equation}

allora, per $|p|<p^{*}$,

\begin{eqnarray}
\beta\left\langle\sigma\right\rangle_{+}(p-1)-\lambda(\beta)&\rightarrow_{\beta\rightarrow\infty}&\beta\left\langle\sigma\right\rangle_{+}(p-p^{*})=-\infty\\
\beta\left\langle\sigma\right\rangle_{+}(p-1)-\lambda(\beta)&\rightarrow_{\beta\rightarrow-\infty}&\beta\left\langle\sigma\right\rangle_{+}(p+p^{*})=-\infty;
\end{eqnarray}

quindi se $|p|<p^{*}$ il massimo $\beta=\beta_{*}$ dell'argomento della trasformata di Legendre sar\` a raggiunto in un punto a derivata nulla, ovvero risolver\` a l'equazione

\begin{equation}
\left\langle\sigma\right\rangle_{+}(p-1)-\left.\frac{\partial}{\partial\beta}\lambda(\beta)\right|_{\beta=\beta_{*}}=0,
\end{equation}

che implica, sviluppando $\lambda(\beta)$ in aspettazioni troncate,

\begin{equation}
\beta_{*}=\frac{\left\langle\sigma\right\rangle_{+}(p-1)}{C_{2}}-\sum_{k\geq3}\beta_{*}^{k-1}\frac{C_{k}}{(k-1)!C_{2}}.\label{eq:appC3}
\end{equation}

Poich\'e in un sistema di Anosov ``poco'' perturbato i valori medi SRB di funzioni analitiche in $\left\{G_{i}\right\}$ sono a loro volta analitici\footnote{\label{nota:gener}Questo fatto pu\` o essere dimostrato con una generalizzazione della tecnica che abbiamo usato per calcolare coniugazione e potenziali della misura SRB nel caso del gatto di Arnold, {\em cf.} \cite{BFG03}, e grazie a quanto illustrato nelle appendici \ref{app:B}, \ref{app:D}.} in $\left\{G_{i}\right\}$, e dal momento che\footnote{{\em Cf.} nota \ref{nota:gener}.} la derivata in $\beta$ di $\lambda(\beta)$ \` e analitica in $\{G_{i}\}$, {\em cf.} appendici \ref{app:D}, \ref{app:B} e \cite{BFG03}, sotto la condizione\footnote{$C_{2}=0$ implica, per la formula di Green - Kubo, che perturbando il sistema non si genera una risposta lineare; evitiamo di considerare questo tipo di ``patologie''. Ad ogni modo, questa condizione non \` e restrittiva perch\' e se $C_{2}=0$ allora al denominatore avremmo $C_{3}$, e la somma nella (\ref{eq:appC3}) partirebbe da $4$ ecc.} $C_{2}\neq 0$ la quantit\` a $\beta_{*}$ \` e esprimibile come serie di potenze, e l'ordine $n$-esimo dello sviluppo pu\` o essere calcolato per iterazione sapendo che\footnote{\label{nota:appC1}Per dimostrare che $\left\langle\sigma\right\rangle_{+}=O\left(G^{2}\right)$ basta notare che se $\mu_{0}(dx)$ \` e la misura di volume normalizzata, introducendo il cambio di variabile $x=Sy$ e ricordando che il tasso di produzione di entropia \` e identificato con il tasso di contrazione dello spazio delle fasi, 

\begin{equation}
1=\int_{\Omega}\mu_{0}(dx)=\int_{S^{-1}\Omega}\mu_{0}(dy)\left|\frac{\partial S(y)}{\partial y}\right|=\int_{\Omega}\mu_{0}(dy)e^{-\sigma(y)},\qquad\mbox{poich\'e $S^{-1}\Omega=\Omega$; }
\end{equation}

dunque, dal momento che per ipotesi $\left.\sigma(x)\right|_{\underline{G}=\underline{0}}=0$, 

\begin{equation}
0=\sum_{i=1}^{s}G_{i}\partial_{G_{i}}\left.\left\langle e^{-\sigma(\cdot)}\right\rangle_{0}\right|_{\underline{G}=\underline{0}}=\sum_{i=1}^{s}G_{i}\partial_{G_{i}}\left\langle-\sigma(\cdot)\right\rangle_{0}|_{\underline{G}=\underline{0}}=\sum_{i=1}^{s}G_{i}\partial_{G_{i}}\left\langle-\sigma(\cdot)\right\rangle_{+}|_{\underline{G}=\underline{0}}.
\end{equation}

}

\begin{eqnarray}
\beta_{*}^{(0)} &=& \frac{\left\langle\sigma\right\rangle_{+}^{(2)}(p-1)}{C_{2}^{(2)}},\label{eq:appC4}
\end{eqnarray}

\begin{eqnarray}
\beta_{*}^{(n)} &=& \left[\frac{\left\langle\sigma\right\rangle_{+}(p-1)}{C_{2}}\right]^{(n)}-\left[\sum_{k\geq3}\beta_{*}^{k-1}\frac{C_{k}}{(k-1)!C_{2}}\right]^{(n)}\nonumber\\
&=& \left[\frac{\left\langle\sigma\right\rangle_{+}(p-1)}{C_{2}}\right]^{(n)}\nonumber\\&&-\sum_{k=3}^{n+2}\frac{1}{(k-1)!}\sum_{m=0}^{n-k+2}\left(\beta_{*}^{k-1}\right)^{(m)}\left(\frac{C_{k}}{C_{2}}\right)^{(n-m)}\nonumber\\
&=&  \left[\frac{\left\langle\sigma\right\rangle_{+}(p-1)}{C_{2}}\right]^{(n)}\nonumber\\&&-\sum_{k=3}^{n+2}\frac{1}{(k-1)!}\sum_{m=0}^{n-k+2}\sum_{\{n_{i}\}: \sum n_{i}=m}\prod_{i=1}^{k-1}\beta_{*}^{(n_{i})}\left(\frac{C_{k}}{C_{2}}\right)^{(n-m)}. \label{eq:appC5}
\end{eqnarray}

Grazie alle (\ref{eq:appC4}), (\ref{eq:appC5}) il punto $\beta_{*}$ \` e calcolabile esplicitamente a tutti gli ordini e con esso anche $\zeta(p)$, essendo definita come

\begin{equation}
\zeta(p)=\beta_{*}\left\langle\sigma\right\rangle_{+}(p-1)-\lambda(\beta_{*}); \label{eq:appC6}
\end{equation} 

in particolare, all'ordine $n$ la (\ref{eq:appC6}) diventa

\begin{eqnarray}
\zeta^{(n)}(p) &=& \left[\beta_{*}\left\langle\sigma\right\rangle_{+}(p-1)\right]^{(n)}-\lambda(\beta_{*})^{(n)}\nonumber\\
&=&  \sum_{m=0}^{n-2}\beta_{*}^{(m)}\left\langle\sigma\right\rangle_{+}^{(n-m)}(p-1)-\sum_{k=2}^{n}\frac{1}{k!}\left(\beta_{*}^{k}C_{k}\right)^{(n)}\nonumber\\
&=& \sum_{m=0}^{n-2}\beta_{*}^{(m)}\left\langle\sigma\right\rangle_{+}^{(n-m)}(p-1)\nonumber\\&&-\sum_{k=2}^{n}\frac{1}{k!}\sum_{m=0}^{n-k}\sum_{n_{1},n_{2}...,n_{k}:\sum n_{i}=m}\prod_{i=1}^{k}\beta_{*}^{(n_{i})}C_{k}^{(n-m)}. \label{eq:appC7}
\end{eqnarray}

Con la (\ref{eq:appC7}) possiamo, ad esempio, calcolare il funzionale di grandi deviazioni $\zeta(p)$ fino al quarto ordine in $\{G_{i}\}$, ottenendo che

\begin{equation}
\zeta(p)=\frac{(p-1)^{2}}{2}\left[\left\langle\sigma\right\rangle_{+}-\frac{C_{2}}{4}\right]-\frac{(p-1)^{3}}{48}C_{3}-\frac{(p-1)^{4}}{384}C_{4}+O\left(G^{5}\right). \label{eq:appC8}
\end{equation}

Quindi, in generale la differenza $-\zeta(p)+\zeta(-p)$ sar\` a una serie di potenze tipo

\begin{equation}
-\zeta(p)+\zeta(-p)=\sum_{k\geq 0}a_{k}p^{2k+1},
\end{equation}

dove i coefficienti $\{a_{k}\}$ possono essere espressi in termini di $C_{i}$, $\left\langle\sigma\right\rangle_{+}$; il teorema di fluttuazione
\begin{equation}
-\zeta(p)+\zeta(-p)=p\left\langle\sigma\right\rangle_{+} \label{eq:appC9}
\end{equation}

implica che $a_{k}=0$ per $k\neq 1$ e $a_{1}=\left\langle\sigma\right\rangle_{+}$. Ci\` o sar\` a reso possibile dalle relazioni tra funzioni di correlazione che abbiamo ricavato con il teorema di fluttuazione per il funzionale generatore $\lambda(\beta)$, {\em cf.} capitolo \ref{cap:appl}, sezione \ref{sez:GK}.  Infatti, l'identit\` a (\ref{eq:appC9}) al quarto ordine equivale alle relazioni

\begin{eqnarray}
\left\langle\sigma\right\rangle_{+}&=&\frac{C_{2}}{2}-\frac{C_{3}}{6}+\frac{C_{4}}{24}+O\left(G^{5}\right) \label{eq:appC10}\\
C_{3}&=&\frac{C_{4}}{2}+O\left(G^{5}\right), \label{eq:appC11}
\end{eqnarray}

grazie alle quali la (\ref{eq:appC8}) diventa:

\begin{equation}
\zeta(p)=\frac{(p-1)^{2}}{8}C_{2}-\frac{(p-1)^{2}}{48}C_{4}\left(1+\frac{(p-1)}{2}+\frac{(p-1)^{2}}{8}\right)+O\left(G^{5}\right).
\end{equation}

Dunque, anche imponendo le (\ref{eq:appC10}), (\ref{eq:appC11}) la distribuzione di probabilit\` a asintotica delle fluttuazioni di $\varepsilon_{\tau}$ preserva il suo carattere {\em non gaussiano}.

\chapter{Analiticit\` a di $\lambda(\beta)$}\label{app:D}

In questa appendice vogliamo discutere le propriet\` a di analiticit\` a della funzione $\lambda(\beta)$ definita come

\begin{equation}
\lambda(\beta)=\lim_{\tau\rightarrow\infty}\tau^{-1}\log \left\langle e^{\tau\beta\langle\sigma\rangle_{+}\left(\varepsilon_{\tau}-1\right)}\right\rangle_{+}=\lim_{\tau\rightarrow\infty}\tau^{-1}\log \left\langle e^{\tau\beta\langle\sigma\rangle_{+}\varepsilon_{\tau}}\right\rangle_{+}-\beta\langle\sigma\rangle_{+},
\end{equation}

dove $\varepsilon_{\tau}$ \` e data da

\begin{equation}
\varepsilon_{\tau}(X(\underline{\sigma}))=\frac{1}{\langle\sigma\rangle_{+}\tau}\sum_{j=-\frac{\tau}{2}}^{\frac{\tau}{2}-1}\sigma\left(S^{j}X(\underline{\sigma})\right); \label{eq:appD1}
\end{equation}

per dimostrare che $\lambda(\beta)$ \` e analitica in $\beta\in(-\infty,\infty)$ seguiremo l'analisi svolta in \cite{GMM}, \cite{CO}, \cite{ergo}, \cite{BFG03}. Come abbiamo visto ad esempio nella sezione \ref{sez:srbappr} del capitolo \ref{cap:noneq}, il valor medio SRB  di un'osservabile $\mathcal{O}$ \` e definito dal limite termodinamico

\begin{equation}
\left\langle\mathcal{O}\right\rangle_{+}=\lim_{\Lambda\rightarrow\infty}\left\langle\mathcal{O}\right\rangle_{\Lambda},
\end{equation}

se $\langle...\rangle_{\Lambda}$ \` e la misura SRB ``approssimata'',  {\em i.e.} uno stato di Gibbs contenuto in un volume finito $\Lambda$; dunque,

\begin{equation}
\lambda(\beta)+\beta\left\langle\sigma\right\rangle_{+}=\lim_{\tau\rightarrow\infty}\lim_{\Lambda\rightarrow\infty}\tau^{-1}\log\left\langle e^{\tau\beta\langle\sigma\rangle_{+}\varepsilon_{\tau}(x)}\right\rangle_{\Lambda}.\label{eq:appD001}
\end{equation}

Il valor medio argomento del logaritmo nella (\ref{eq:appD001}) pu\` o essere scritto esplicitamente come, se $\left\{\Phi_{X}\right\}$ sono i potenziali della misura SRB,

\begin{equation}
\left\langle e^{\tau\beta\langle\sigma\rangle_{+}\varepsilon_{\tau}(x)}\right\rangle_{\Lambda} = \frac{\sum_{\underline{\sigma}_{\Lambda}}e^{\tau\beta\langle\sigma\rangle_{+}\varepsilon_{\tau}\left(X(\underline{\sigma}_{\Lambda})\right)-\sum_{X\subset\Lambda}\Phi_{X}\left(\underline{\sigma}_{X}\right)}}{\sum_{\underline{\sigma}_{\Lambda}}e^{-\sum_{X\subset\Lambda}\Phi_{X}\left(\underline{\sigma}_{X}\right)}},
\end{equation}

e poich\'e $\varepsilon_{\tau}$ \` e  sviluppabile in potenziali a decadimento esponenziale, {\em cf.} capitolo \ref{cap:noneq}, sezione \ref{sez:dimoFT}, possiamo generalizzare il problema studiando il rapporto\footnote{In realt\` a, nella (\ref{eq:appD010}) la somma sui potenziali $\Psi_{X}$ dovrebbe coinvolgere tutti gli insiemi $X$ che {\em intersecano} $\left[-\frac{\tau}{2},\frac{\tau}{2}\right]$; ma la differenza tra questa somma e quella che appare nella (\ref{eq:appD010}) \` e, se $\beta$ \` e finito, stimata da una costante che pu\` o essere inglobata nel $\log B$ definito in (\ref{eq:appDstima1}).}

\begin{equation}
\frac{\sum_{\underline{\sigma}_{\left[-\frac{\tau}{2},\frac{\tau}{2}\right]}}e^{-\sum_{X\subset\left[-\frac{\tau}{2},\frac{\tau}{2}\right]}\left[\beta\Psi_{X}\left(\underline{\sigma}_{X}\right)+\Phi_{X}\left(\underline{\sigma}_{X}\right)\right]}\sum_{\underline{\sigma}_{\left[-\frac{\tau}{2},\frac{\tau}{2}\right]^{c}}}^{*}e^{-\sum_{X\not\subset\left[-\frac{\tau}{2},\frac{\tau}{2}\right]}\Phi_{X}\left(\underline{\sigma}_{X}\right)}}{\sum_{\underline{\sigma}_{\Lambda}}e^{-\sum_{X\subset\Lambda}\Phi_{X}}\left(\underline{\sigma}_{X}\right)}, \label{eq:appD010}
\end{equation}

dove con l'asterisco ricordiamo che le sequenze devono essere compatibili con $\underline{\sigma}_{\left[-\frac{\tau}{2},\frac{\tau}{2}\right]}$; inoltre, se $a$ \` e il tempo di mescolamento della matrice di compatibilit\` a ({\em cf.} definizione \ref{def:matrcomp}),

\begin{eqnarray}
\sum_{X\not\subset\left[-\frac{\tau}{2},\frac{\tau}{2}\right]}\Phi_{X}\left(\underline{\sigma}_{X}\right)&=&\sum_{X\not\subset\left[-\frac{\tau}{2},\frac{\tau}{2}\right],X\cap\left[-a-\frac{\tau}{2},\frac{\tau}{2}+a\right]=\emptyset}\Phi_{X}\left(\underline{\sigma}_{X}\right)\nonumber\\&&+\sum_{X\not\subset\left[-\frac{\tau}{2},\frac{\tau}{2}\right],X\cap\left[-a-\frac{\tau}{2},\frac{\tau}{2}+a\right]\neq\emptyset}\Phi_{X}\left(\underline{\sigma}_{X}\right), \label{eq:appD0001}
\end{eqnarray}

e il secondo addendo nella parte destra della (\ref{eq:appD0001}) pu\` o essere stimato come, grazie al decadimento esponenziale dei potenziali, 

\begin{equation}
\left|\sum_{X\not\subset\left[-\frac{\tau}{2},\frac{\tau}{2}\neq\emptyset\right],X\cap\left[-a-\frac{\tau}{2},\frac{\tau}{2}+a\right]\neq\emptyset}\Phi_{X}\left(\underline{\sigma}_{X}\right)\right|\leq 2a\sup_{\xi}\sum_{X\ni\xi}\left|\Phi_{X}\left(\underline{\sigma}_{X}\right)\right|\leq \log B>0. \label{eq:appDstima1}
\end{equation}

Quindi, dal momento che le configurazioni di simboli $\underline{\sigma}_{X}$ con $X\cap \left[-a-\frac{\tau}{2},\frac{\tau}{2}+a\right]=\emptyset$ sono {\em indipendenti} da quelle relative agli intervalli contenuti in $\left[-\frac{\tau}{2},\frac{\tau}{2}\right]$, la (\ref{eq:appD010}) \` e maggiorata da, se $\left\{1,...,n\right\}$ \` e lo spazio degli stati di un singolo spin,

\begin{equation}
(\ref{eq:appD010})\leq \frac{\sum_{\underline{\sigma}_{\left[-\frac{\tau}{2},\frac{\tau}{2}\right]}}e^{-\sum_{X\subset\left[-\frac{\tau}{2},\frac{\tau}{2}\right]}\left[\beta\Psi_{X}\left(\underline{\sigma}_{X}\right)+\Phi_{X}\left(\underline{\sigma}_{X}\right)\right]}}{\sum_{\underline{\sigma}_{\left[-\frac{\tau}{2},\frac{\tau}{2}\right]}}e^{-\sum_{X\subset\left[-\frac{\tau}{2},\frac{\tau}{2}\right]}\Phi_{X}\left(\underline{\sigma}_{X}\right)}}n^{2a}B^{2}. \label{eq:appD0100}
\end{equation}

La minorazione pu\` o essere fatta in modo completamente analogo, e il risultato \` e che

\begin{equation}
(\ref{eq:appD010})\left\{\begin{array}{c} \leq \frac{\sum_{\underline{\sigma}_{\left[-\frac{\tau}{2},\frac{\tau}{2}\right]}}e^{-\sum_{X\subset\left[-\frac{\tau}{2},\frac{\tau}{2}\right]}\left[\beta\Psi_{X}\left(\underline{\sigma}_{X}\right)+\Phi_{X}\left(\underline{\sigma}_{X}\right)\right]}}{\sum_{\underline{\sigma}_{\left[-\frac{\tau}{2},\frac{\tau}{2}\right]}}e^{-\sum_{X\subset\left[-\frac{\tau}{2},\frac{\tau}{2}\right]}\Phi_{X}\left(\underline{\sigma}_{X}\right)}}n^{2a}B^{2} \\ \geq \frac{\sum_{\underline{\sigma}_{\left[-\frac{\tau}{2},\frac{\tau}{2}\right]}}e^{-\sum_{X\subset\left[-\frac{\tau}{2},\frac{\tau}{2}\right]}\left[\beta\Psi_{X}\left(\underline{\sigma}_{X}\right)+\Phi_{X}\left(\underline{\sigma}_{X}\right)\right]}}{\sum_{\underline{\sigma}_{\left[-\frac{\tau}{2},\frac{\tau}{2}\right]}}e^{-\sum_{X\subset\left[-\frac{\tau}{2},\frac{\tau}{2}\right]}\Phi_{X}\left(\underline{\sigma}_{X}\right)}}n^{-2a}B^{-2} \end{array}\right.,
\end{equation}

che implica, poich\'e $B<\infty$, ponendo $U^{\Psi}_{\Lambda}\equiv\sum_{X\subset \Lambda}\Psi_{X}(\underline{\sigma}_{X})$,

\begin{equation}
\lim_{\tau\rightarrow\infty}\lim_{\Lambda\rightarrow\infty}\tau^{-1}\log\left\langle e^{-\beta U^{\Psi}_{\tau}}\right\rangle_{\Lambda}=\lim_{\Lambda\rightarrow\infty}\Lambda^{-1}\log\left\langle e^{-\beta U^{\Psi}_{\Lambda}}\right\rangle_{\Lambda};
\end{equation}

in conclusione, il problema dell'analiticit\` a di $\lambda(\beta)$ pu\` o essere generalizzato studiando (il segno di $\beta$ \` e irrilevante)

\begin{equation}
\Delta P(\beta)=\lim_{\Lambda\rightarrow\infty}\Lambda^{-1}\log \left\langle e^{\beta U^{\Psi}_{\Lambda}}\right\rangle_{\Lambda}.\label{eq:appD01}
\end{equation}

La (\ref{eq:appD01}) pu\` o essere riscritta come

\begin{equation}
\Delta P(\beta)=\lim_{\Lambda\rightarrow\infty}\left[\Lambda^{-1}\log Z_{\Lambda}^{\beta}-\Lambda^{-1}\log Z_{\Lambda}^{0}\right],
\end{equation}

dove

\begin{equation}
Z_{\Lambda}^{\beta}=\sum_{\underline{\sigma}_{\Lambda}}e^{-\sum_{X\subset\Lambda}\left[\beta\Psi_{X}(\underline{\sigma}_{X})+\Phi_{X}\right]},
\end{equation}

se $\Phi_{X}$ sono i potenziali della misura SRB. Poich\' e sia $\Psi_{X}$ che $\Phi_{X}$ decadono esponenzialmente, $Z_{\Lambda}^{\beta}$ \` e la funzione di partizione di {\em un modello di Fisher}\index{modello di Fisher} con potenziali a molti corpi 

\begin{equation}
\alpha_{X}(\underline{\sigma}_{X})=\beta\Psi_{X}(\underline{\sigma}_{X})+\Phi_{X}(\underline{\sigma}_{X}); \label{eq:appD2}
\end{equation}

per {\em modello di Fisher} intendiamo un sistema unidimensionale di spin interagenti attraverso potenziali che decadono esponenzialmente, e che sono diversi da zero solo per intervalli $X$ connessi\footnote{Nel nostro caso gli intervalli saranno formati da un numero {\em dispari} di siti, {\em cf.} proposizione \ref{prop:svilpot}.}. Dunque, il problema iniziale \` e ricondotto allo studio delle propriet\` a di analiticit\` a della {\em pressione} di questo modello.

In generale, questo tipo di problema pu\` o essere affrontato grazie alla tecnica della {\em cluster expansion}, \cite{GMM}, \cite{ergo}, \cite{BFG03}. Questa tecnica per\` o richiede che l'interazione iniziale sia {\em debole}; nel nostro caso invece $\beta$ \` e arbitrario, e inoltre non tutte le sequenze $\underline{\sigma}$ sono ammissibili: ci\` o equivale a considerare dei potenziali con un {\em cuore duro}. Risolveremo il problema grazie ad una tecnica di {\em gruppo di rinormalizzazione}, \cite{CO}.

\section{Decimazione: prima parte} \label{sez:decim}

Consideriamo un volume $\Lambda\subset\mathbb{Z}$ suddivisibile in intervalli $\left\{B_{i}\right\}_{i=0}^{l}$, $\left\{H_{i}\right\}_{i=0}^{l-1}$ nel seguente modo:

\begin{equation}
\Lambda = B_{0},H_{0},B_{1}, H_{1},...,B_{l-1}, H_{l-1}, B_{l}, \label{eq:appD3}
\end{equation}

dove $\delta\left(B_{i}\right)=\tau\in\mathbb{N}$ e $\delta\left(H_{i}\right)=h\in\mathbb{N}$. Quindi, stiamo assumendo che $\Lambda$ contenga $l(h+\tau+2)+\tau+1$ elementi; con la suddivisione (\ref{eq:appD3}), una generica sequenza di simboli $\underline{\sigma}_{\Lambda}\in\{1,...,n\}^{\Lambda}$ pu\` o essere rappresentata come

\begin{equation}
\underline{\sigma}_{\Lambda}=\left(\beta_{0},\eta_{0},\beta_{1},\eta_{1},...,\beta_{l-1},\eta_{l-1},\beta_{l}\right), \label{eq:appD4}
\end{equation}

dove ogni $\beta_{i}$, $\eta_{i}$ contiene rispettivamente $\tau+1$, $h+1$ simboli. Poich\'e lo stato di Gibbs \` e {\em unico}, \cite{RuStat}, \cite{ergo}, siamo liberi di rappresentarlo chiamando spin $\beta_{i}$, $\eta_{i}$. 

Come vedremo, integrando le variabili $\eta_{i}$ otterremo un sistema di spin $\beta_{i}$ che vive su un {\em reticolo rinormalizzato}, dove diremo che le distanze tra $B_{i}$ e $B_{i+1}$ \` e pari ad uno, interagente attraverso un {\em potenziale rinormalizzato} $\alpha^{ren}$; il nostro scopo \`e verificare che questo nuovo potenziale abbia tutte le caratteristiche necessarie per applicare la cluster expansion. Per risolvere il problema della compatibilit\` a tra gli spin $\beta_{i}$ sottintenderemo sempre che, se $a_{0}$ \` e il tempo di mescolamento della matrice di compatibilit\` a\footnote{{\em cf.} definizione \ref{def:matrcomp}, capitolo \ref{cap:ergo}.} $T$, $h=na_{0}$ con $\mathbb{N}\ni n\geq 1$. Una volta fissate le configurazioni $\beta_{i}$, le configurazioni $\eta_{i}$ dovranno rispettare dei vincoli di compatibilit\` a con le prime; ad ogni modo, come sar\` a chiaro nel seguito, ci\` o sar\` a dal tutto irrilevante ai fini delle nostre stime. Possiamo scrivere l'energia totale del sistema come

\begin{eqnarray}
U^{\Lambda}&\equiv&\sum_{X\subset\Lambda}\alpha_{X}\left(\underline{\sigma}_{X}\right)\nonumber\\&=&\sum_{i=0}^{l}U^{B}\left(\beta_{i}\right)+\sum_{i=0}^{l-1}\left[W(\beta_{i},\eta_{i},\beta_{i+1})\right]\nonumber\\&&+ \sum_{X\not\subset B_{i}\cup H_{i}\cup B_{i+1}\,\forall i}\alpha_{X}(\underline{\sigma}_{X}), \label{eq:appD5}
\end{eqnarray}

dove

\begin{eqnarray}
U^{B}(\beta_{i})&\equiv&\sum_{X\subset B_{i}}\alpha_{X}(\underline{\sigma}_{X})\\
W(\beta_{i},\eta_{i},\beta_{i+1})&\equiv& \sum_{X\subset B_{i}\cup H_{i}\cup B_{i+1}: X\cap H_{i}\neq\emptyset}\alpha_{X}(\underline{\sigma}_{X});
\end{eqnarray}

dunque, l'ultimo termine nella (\ref{eq:appD5}) coinvolge potenziali che intersecano pi\` u di un $H_{i}$, {\em i.e.} sono i potenziali a molti corpi nell'$H$-reticolo. Quindi, possiamo riscrivere la funzione di partizione come

\begin{eqnarray}
Z_{\Lambda} &=& \sum_{\beta_{0},...,\beta_{l}}\sum_{\eta_{0},...,\eta_{l-1}}\left(\prod_{i=0}^{l}e^{-U^{B}\left(\beta_{i}\right)}\right)\left(\prod_{i=0}^{l-1}e^{W(\beta_{i},\eta_{i},\beta_{i+1})}\right)e^{-\sum^{*}\alpha_{X}\left(\underline{\sigma}_{X}\right)}, \label{eq:appD6}
\end{eqnarray}

dove con l'asterisco indichiamo la condizione $X\not\subset B_{i}\cup H_{i}\cup B_{i+1}\,\forall i$. Definendo

\begin{equation}
Z_{h}\left(\beta_{i},\beta_{i+1}\right)\equiv\sum_{\eta}e^{W\left(\beta_{i},\eta_{i},\beta_{i+1}\right)}, \label{eq:appD7}
\end{equation}

moltiplicando e dividendo per $\prod_{i=0}^{l-1}Z_{h}\left(\beta_{i},\beta_{i+1}\right)$ la (\ref{eq:appD6}) diventa

\begin{eqnarray}
Z_{\Lambda} &=& \sum_{\beta_{0},...,\beta_{l}}\left(\prod_{i=0}^{l}e^{-U^{B}\left(\beta_{i}\right)}\right)\left(\prod_{i=0}^{l-1}Z_{h}\left(\beta_{i},\beta_{i+1}\right)\right)\nonumber\\&&\times\frac{\sum_{\eta_{0},...,\eta_{l-1}}\left(\prod_{i=0}^{l-1}e^{-W\left(\beta_{i},\eta_{i},\beta_{i+1}\right)}\right)e^{-\sum^{*}\alpha_{X}\left(\underline{\sigma}_{X}\right)}}{\prod_{i=0}^{l-1}Z_{h}\left(\beta_{i},\beta_{i+1}\right)}; \label{eq:appD8} 
\end{eqnarray}

per ricondurci ad un modello di spin $\beta_{i}$ interagenti attraverso un potenziale effettivo dobbiamo essere in grado di stimare il rapporto nella (\ref{eq:appD8}). I blocchi di spin $H_{i}$ interagiscono attraverso i potenziali a molti corpi

\begin{equation}
\widetilde{\alpha}_{H_{i},...,H_{i+q}}\left(\eta_{i},...,\eta_{i+q}\right)=\sum_{X\cap H_{k}\neq \emptyset, k=i,...,i+q}\alpha_{X}\left(\underline{\sigma}_{X}\right); \label{eq:appD9}
\end{equation}

in particolare, ci interessa stimare la quantit\` a

\begin{equation}
\sup_{H}\sum_{\mathcal{H}\ni H:\left|\mathcal{H}\right|>1}\sup_{\underline{\eta}_{\mathcal{H}}}\left|\widetilde{\alpha}_{H_{1},...,H_{k}}\left(\eta_{H_{1}},...,\eta_{H_{k}}\right)\right|e^{\frac{1}{2}(k-1)\tau \kappa}, \label{eq:appD10}
\end{equation}

dove $\mathcal{H}=\left(H_{1},...,H_{k}\right)$ \` e un ``intervallo'' nell' $H$-reticolo e $\underline{\eta}_{\mathcal{H}}=\left(\eta_{H_{1}},...,\eta_{H_{k}}\right)$. Possiamo maggiorare la (\ref{eq:appD10}) nel modo seguente:

\begin{eqnarray}
(\ref{eq:appD10})&=& \sup_{H}\sum_{\mathcal{H}\ni H: \left|\mathcal{H}\right|>1}\sup_{\underline{\eta}_{\mathcal{H}}}\left|\sum_{X\cap H_{i}\neq \emptyset,i=1,...,\left|\mathcal{H}\right|}\alpha_{X}\left(\underline{\sigma}_{X}\right)\right|e^{\frac{\kappa\tau}{2}\left|k-1\right|}\nonumber\\&\leq&\sup_{H}\sum_{\mathcal{H}\ni H: \left|\mathcal{H}\right|>1}\sum_{X\cap H_{i}\neq \emptyset, i=1,...,\left|\mathcal{H}\right|}\sup_{\underline{\sigma}_{X}}\left|\alpha_{X}\left(\underline{\sigma}_{X}\right)\right|e^{\kappa\delta\left(X\right)}e^{-\frac{\kappa}{2}\delta\left(X\right)}\nonumber\\&\leq&\left\|\overline{\alpha}\right\|_{\kappa}\sum_{k\geq 2}e^{-\frac{\kappa\tau}{2}(k-1)} = m\left\|\overline{\alpha}\right\|_{\kappa}e^{-\frac{\kappa\tau}{2}},\qquad m =\frac{1}{1-e^{-\frac{\kappa\tau}{2}}}, \label{eq:appD11} 
\end{eqnarray}

dove $\left\|\overline{\alpha}\right\|_{\kappa}\equiv \sup_{\xi}\sum_{X\ni \xi,\left|X\right|>1}\left\|\alpha_{X}\right\|e^{\kappa\delta(X)}$, se $\left\|\alpha_{X}\right\|\equiv \max_{\underline{\sigma}_{X}}\left|\alpha_{X}\left(\underline{\sigma}_{X}\right)\right|$. La stima (\ref{eq:appD11}) \` e ci\` o di cui avremo bisogno per poter usare le tecniche della cluster expansion, che introdurremo nella prossima sezione.

\section{Cluster expansion}

La {\em cluster expansion} \` e una tecnica che permette di studiare somme tipo

\begin{equation}
Z_{\Lambda}^{\Phi}=\sum_{\underline{\sigma}_{\Lambda}}e^{-\sum_{Y\subset\Lambda}\Phi_{Y}\left(\underline{\sigma}_{Y}\right)}, \label{eq:appD12}
\end{equation}

dove $\left\{\Phi_{Y}\right\}$ sono potenziali che in generale possono non essere invarianti per traslazioni, e $\Lambda$ \` e un sottoinsieme finito $\mathbb{Z}^{d}$; come abbiamo visto pi\` u volte, le {\em funzioni di partizione} dei modelli di spin rientrano nella classe di somme (\ref{eq:appD12}). La nostra tecnica si basa su  un risultato fondamentale dovuto a Ruelle, \cite{RuStat}, che enuncieremo e dimosteremo nella prossima sezione.

\subsection{Polimerizzazione}\label{sez:pol}

Chiamiamo {\em polimero} qualunque insieme connesso finito $\emptyset\neq \gamma\subset \Lambda$, e indichiamo con $\Gamma$ una {\em configurazione di polimeri} $\Gamma=\left\{\gamma_{1},...,\gamma_{n}\right\}$; non imporremo nessuna restrizione sui polimeri che formano $\Gamma$, {\em i.e.} permetteremo che i polimeri $\gamma_{i}$ si sovrappongano e anche che coincidano.

Definiamo {\em attivit\` a di un polimero} una qualsiasi funzione $\zeta(\gamma)$ definita sull'insieme dei polimeri $\gamma\subset\Lambda$, e l'{\em attivit\` a di una configurazione di polimeri} $\zeta\left(\Gamma\right)$ come

\begin{equation}
\zeta\left(\Gamma\right)=\prod_{\gamma\in\widetilde{\Gamma}}\zeta(\gamma)^{\Gamma(\gamma)},
\end{equation}

dove $\Gamma(\gamma)$ \` e la {\em molteplicit\` a del polimero $\gamma\in\Gamma$}; quindi $\zeta\left(\Gamma\right)$ \` e una funzione moltiplicativa,

\begin{equation}
\zeta\left(\Gamma_{1}+\Gamma_{2}\right)=\zeta\left(\Gamma_{1}\right)\zeta\left(\Gamma_{2}\right), \label{eq:appD012}
\end{equation}

e imporremo che $\zeta(\emptyset)\equiv 1$. Diremo che una configurazione $\Gamma$ \` e {\em compatibile} se non ci sono sovrapposizioni tra polimeri, {\em i.e.} se $\gamma_{i}\cap\gamma_{j}=\emptyset$ per ogni $\gamma_{i},\gamma_{j}\in\Gamma$. La {\em funzione di partizione delle configurazioni compatibili}  \` e data da

\begin{equation}
Z(\Lambda)\equiv \sum_{\Gamma\subset\Lambda, \gamma_{i}\cap\gamma_{j}=\emptyset}\zeta(\Gamma); \label{eq:appD13}
\end{equation}

definendo $\varphi(\Gamma)$, $\chi_{V}(\Gamma)$ tramite le propriet\` a

\begin{eqnarray}
\varphi(\Gamma)&=&\left\{\begin{array}{cc} 1&\mbox{se $\Gamma$ \`e compatibile} \\  0 &\mbox{altrimenti} \end{array}\right.\\
\chi_{V}(\Gamma)&=&\left\{\begin{array}{cc} 1&\mbox{se ogni $\gamma\in\Gamma$ \` e contenuto in $V$} \\ 0 & \mbox{altrimenti} \end{array}\right., \label{eq:appD013}
\end{eqnarray}

la (\ref{eq:appD13}) diventa

\begin{equation}
Z(\Lambda)=\sum_{\Gamma}\varphi(\Gamma)\zeta(\Gamma)\chi_{\Lambda}(\Gamma). \label{eq:appD14}
\end{equation}

Introduciamo $\mathbf{1}(\Gamma)$ come

\begin{equation}
\mathbf{1}(\Gamma)=\left\{\begin{array}{cc} 1(\Gamma)&\mbox{se $\Gamma=\emptyset$}\\ 0&\mbox{altrimenti} \end{array}\right.,
\end{equation}

e sia $\Psi$ una funzione sullo spazio dei polimeri del tipo

\begin{equation}
\Psi(\Gamma)=\mathbf{1}+\Psi_{0}(\Gamma)\qquad\mbox{con $\Psi_{0}(\emptyset)=0$};
\end{equation}

possiamo definire il {\em logaritmo} e l'{\em esponenziale} di $\Psi(\Gamma)$ come

\begin{eqnarray}
\left(\mbox{Exp}\,\Psi\right)(\Gamma)&=&\mathbf{1}(\Gamma)+\sum_{n\geq 1}\frac{1}{n!}\sum_{\Gamma_{1}+...+\Gamma_{n}=\Gamma}\Psi(\Gamma_{1})...\Psi(\Gamma_{n})\\
\left(\mbox{Log}\,\Psi\right)(\Gamma)&=&\sum_{n\geq 1}\frac{(-1)^{n+1}}{n}\sum_{\Gamma_{1}+...+\Gamma_{n}=\Gamma}\Psi_{0}(\Gamma_{1})...\Psi_{0}(\Gamma_{n}). \label{eq:deflog}
\end{eqnarray}

Quindi, chiamando $\varphi^{T}(\Gamma)$ il logaritmo di $\varphi(\Gamma)$ si ha che, poich\' e $\chi_{\Lambda}(\Gamma)$, $\zeta(\Gamma)$ sono due funzioni moltiplicative\footnote{Notare che nella (\ref{eq:appD15}) il termine $\mathbf{1}(\Gamma)\zeta(\Gamma)\chi_{\Lambda}(\Gamma)$ \` e diverso da zero solo per $\Gamma=\emptyset$, e per questo valore \` e uguale a $1$.},

\begin{eqnarray}
Z_{\Lambda}=\sum_{\Gamma}\left(\mbox{Exp}\,\varphi^{T}\right)(\Gamma)\zeta(\Gamma)\chi_{\Lambda}(\Gamma)=\sum_{\Gamma}\left[\mathbf{1}(\Gamma)\zeta(\Gamma)\chi_{\Lambda}(\Gamma)\right.\\\left.+\sum_{n\geq 1}\frac{1}{n!}\sum_{\Gamma_{1}+...+\Gamma_{n}=\Gamma}\varphi^{T}(\Gamma_{1})\zeta(\Gamma_{1})\chi_{\Lambda}(\Gamma_{1})+...+\varphi^{T}(\Gamma_{1})\zeta(\Gamma_{1})\chi_{\Lambda}(\Gamma_{1})\right]\nonumber\\=e^{\sum_{\Gamma}\varphi^{T}(\Gamma)\zeta(\Gamma)\chi_{\Lambda}(\Gamma)}. \label{eq:appD15}
\end{eqnarray}

Siamo pronti per enunciare e discutere il seguente importante risultato.

\begin{prop}\label{prop:appD1}{{\em (Polimeri e cluster expansion)}}

Sia $\Lambda\subset\mathbb{Z}^{d}$ un insieme finito e sia $\zeta(\gamma)$ l'attivit\` a dei polimeri che verifica la condizione

\begin{equation}
\left|\zeta(\gamma)\right|\leq b_{0}^{2}\nu_{0}^{2|\gamma|}e^{-2\kappa_{0}\delta(\gamma)}\equiv z(\gamma)^{2}, \label{eq:appD014}
\end{equation}

per opportune costanti $b_{0},\nu_{0},\kappa_{0}>0$. Introducendo

\begin{equation}
k\left(d,\kappa_{0}\right)\equiv 2e\sum_{0\neq \xi\in\mathbb{Z}^{d}}e^{-\kappa_{0}|\xi|}, \label{eq:appD015}
\end{equation}

supponiamo che si abbia

\begin{equation}
\left[(1+B_{0})e^{B_{0}\nu_{0}}+k\left(d,\kappa_{0}\right)\right]<1 \label{eq:dimpol00}
\end{equation}

con $B_{0}=2b_{0}$. Allora la somma $Z_{\Lambda}$ nella (\ref{eq:appD13}) pu\` o essere scritta come

\begin{equation}
Z_{\Lambda}=\exp\left(\sum_{\Gamma\subset\Lambda}\varphi^{T}(\Gamma)\zeta(\Gamma)\right), \label{eq:appD16}
\end{equation}

con

\begin{equation}
\sup_{\xi\in\Lambda}\sum_{\Gamma\ni\xi, \delta(\Gamma)\geq r}\left|\varphi^{T}(\Gamma)\right|\left|\zeta(\gamma)\right|<B_{0}\nu_{0}e^{-\kappa_{0}r}.
\end{equation}

\end{prop}

\begin{proof}

Chiamando $\widetilde{\Gamma}$ l'insieme dei polimeri distinti appartenenti a $\Gamma=\left\{\gamma_{1},...,\gamma_{n}\right\}$, ponendo $\Gamma!\equiv \prod_{\gamma\in \widetilde{\Gamma}}\Gamma(\gamma)!$ e chiamando $G$ il grafo formato dai nodi $\{1,...,n\}$ e dagli archi $(i,j)$, $i,j\in\{1,...,n\}$, posti in corrispondenza di una sovrapposizione tra due polimeri $\gamma_{i},\gamma_{j}\in\Gamma$, la funzione $\varphi^{T}(\Gamma)$ che compare nella (\ref{eq:appD16}) pu\` o essere scritta come

\begin{equation}
\varphi^{T}(\Gamma)=\frac{1}{\Gamma!}\sum^{*}_{C\subseteq G}(-1)^{\mbox{(numero di archi in $C$)}}. \label{eq:appD17}
\end{equation}

La (\ref{eq:appD17}) prende il nome di {\em formula dei coefficienti di Mayer}, e l'asterisco indica che la somma coinvolge tutti i sottografi connessi\footnote{Un {\em grafo connesso} \` e un grafo nel quale \` e possibile arrivare in qualunque nodo partendo da qualsiasi altro spostandosi sugli gli archi.} di $G$ {che visitano tutti i nodi}; se $G$ non \` e connesso la somma non ha nessun addendo, dunque in questo caso $\varphi^{T}(\Gamma)\equiv 0$.

Per convincerci della (\ref{eq:appD17}) dimostriamo che vale l'implicazione $(\ref{eq:appD17})\Rightarrow \varphi^{T}(\Gamma)=\left(\mbox{Log}\,\varphi\right)(\Gamma)$. Per fare ci\` o, introduciamo una generica configurazione di polimeri $\Gamma=\left\{\gamma_{1},...,\gamma_{n}\right\}$, dove ogni polimero ha una molteplicit\` a $\Gamma(\gamma_{i})$, $i=1,...,n$; consideriamo distinte le $\Gamma(\gamma_{i})$ copie del polimero $i$-esimo aggiungendo ad ogni copia un nuovo indice $1,2...,\Gamma(\gamma_{i})$. Introducendo la funzione $g_{ij}$ come

\begin{equation}
g_{ij}=\left\{\begin{array}{cc}-1 & \mbox{se $\gamma_{i}\cap\gamma_{j}\neq \emptyset$} \\ 0 & \mbox{altrimenti}\end{array}\right.,
\end{equation}

poich\' e $\varphi(\Gamma)$ \` e definita dalla propriet\` a (\ref{eq:appD013}), possiamo porre

\begin{equation}
\varphi(\Gamma)\equiv \prod_{i<j}\left(1+g_{ij}\right); \label{eq:dimpol0}
\end{equation}

\` e chiaro che $\varphi(\Gamma)\equiv \frac{\varphi(\Gamma)}{\Gamma!}$, dal momento che se $\Gamma!>1$ allora $\varphi(\Gamma)=0$.

Quindi,

\begin{eqnarray}
\varphi(\Gamma)&=&\prod_{i<j}\left(1+g_{ij}\right)=\sum_{k}\sum_{C_{1}...C_{k}\;non\;ordinati,\;C_{1}+...+C_{k}=\Gamma}\prod_{j}C_{j}!\varphi^{T}(\Gamma)\nonumber\\&=&\sum_{k}\frac{1}{k!}\sum_{C_{1},...,C_{k}\;ordinati,\; C_{1}+...+C_{k}=\Gamma}^{**}\frac{\Gamma!}{\prod_{m}C_{m}!}\prod_{j}C_{j}!\varphi^{T}\left(C_{j}\right)\nonumber\\&=&\Gamma!\left(\mbox{Exp}\,\varphi^{T}\right)(\Gamma)\Rightarrow \left(\mbox{Log}\,\varphi\right)(\Gamma)=\varphi^{T}(\Gamma);
\end{eqnarray}

nella prima somma l'ordine non d\` a luogo a nuovi addendi, {\em i.e.} le permutazioni di $\left\{C_{1},...,C_{k}\right\}$ non vengono sommate, mentre nella seconda l'ordine conta (fattore $\frac{1}{k!}$) e con il doppio asterisco indichiamo che smettiamo di considerare distinti polimeri uguali (fattore $\frac{\Gamma!}{\prod_{m}C_{m}!}$). Definiamo la {\em derivata} di una funzione definita sui polimeri come

\begin{equation}
(D_{\Gamma}\Psi)(\Gamma)\equiv \psi(\Gamma+H)\frac{(\Gamma+H)!}{H!},\label{eq:dimpol1}
\end{equation}

se $\Gamma!$ \`e la molteplicit\` a della configurazione $\Gamma$, {\em i.e.} il prodotto delle singole molteplicit\` a che formano la configurazione di polimeri; indicando con $\ast$ l'operazione di {\em convoluzione} tra funzioni di polimeri, {\em i.e.}

\begin{equation}
(\Psi_{1}\ast\Psi_{2})(\Gamma)=\sum_{\Gamma_{1}+\Gamma_{2}=\Gamma}\Psi_{1}(\Gamma_{1})\Psi_{2}(\Gamma_{2}),
\end{equation}

 \` e possibile dimostrare che valgono le propriet\` a

\begin{eqnarray}
D_{\gamma}(\Psi_{1}\ast\Psi_{2})&=&(D_{\gamma}\Psi_{1})\ast\Psi_{2}+\Psi_{1}\ast(D_{\gamma}\Psi_{2})\\
D_{\gamma}\mbox{Exp}\,\Psi&=&D_{\gamma}\Psi\ast\mbox{Exp}\,\Psi.\label{eq:dimpol11}
\end{eqnarray}

Ponendo

\begin{equation}
\overline{\varphi}(\Gamma)\equiv\zeta(\Gamma)\varphi(\Gamma)
\end{equation}

che implica\footnote{Questo fatto \` e una conseguenza della definizione di logaritmo, {\em cf.} formula (\ref{eq:deflog}), e della propriet\` a delle funzioni moltiplicative (\ref{eq:appD012}).} $\mbox{Log}\,\overline{\varphi}(\Gamma)=\overline{\varphi}^{T}(\Gamma)=\zeta(\Gamma)\varphi^{T}(\Gamma)$, grazie alla definizione (\ref{eq:dimpol1}) introduciamo la funzione $\Delta_{\Gamma}$ come, indicando con $\Psi^{-1}$ l'{\em inversa} di $\Psi$, {\em i.e.} la funzione tale che $\Psi^{-1}\ast\Psi(\Gamma)=\textbf{1}(\Gamma)$,

\begin{equation}
\Delta_{\Gamma}(H)=\left(\overline{\varphi}^{-1}\ast D_{\Gamma}\overline{\varphi}\right)(H)=\sum_{H_{1}+H_{2}=H}\overline{\varphi}^{-1}(H_{1})\overline{\varphi}(\Gamma+H_{2}); \label{eq:dimpol2}
\end{equation}

\` e importante notare che $D_{\Gamma}\overline{\varphi}(H_{2})=\overline{\varphi}(\Gamma+H_{2})$, dal momento che $\overline{\varphi}(\Gamma+H_{2})$ \` e identicamente uguale a zero se $(\Gamma+H_{2})!>1$, quindi quando $\overline{\varphi}$ \` e diversa da zero abbiamo che $\frac{(\Gamma+H_{2})!}{H_{2}!}=1$. Quindi, poich\' e $\overline{\varphi}$ si annulla se $\Gamma!$ non \` e compatibile, {\em i.e.} se due polimeri che appartengono a $\Gamma$ si sovrappongono, la (\ref{eq:dimpol2}) implica che

\begin{equation}
\Delta_{\Gamma}(H)=0\qquad\mbox{se $\Gamma$ non \` e compatibile.}
\end{equation}

Dalla (\ref{eq:dimpol0}) sappiamo che $\varphi(\gamma_{1},...,\gamma_{n})=\prod_{i<j}\left(1+g(\gamma_{i},\gamma_{j})\right)$, dunque

\begin{eqnarray}
\overline{\varphi}(\gamma+\Gamma+H_{2})&=&\zeta(\gamma)\overline{\varphi}(\Gamma+H_{2})\prod_{\gamma'\in H_{2}}\left(1+g(\gamma,\gamma')\right)\nonumber\\&=&\zeta(\gamma)\overline{\varphi}(\Gamma+H_{2})\sum_{S\subseteq H_{2}}^{*}(-1)^{N(S)}, \label{eq:dimpol3}
\end{eqnarray}

se con l'asterisco indichiamo che stiamo sommando su tutti gli insiemi di polimeri $S$ incompatibili con $\gamma$, e dove la somma di polimeri $\gamma+H_{2}+\Gamma$ deve essere compatibile (altrimenti $\overline{\varphi}=0$). Se $\gamma$ \` e compatibile con $\Gamma$ dalla definizione (\ref{eq:dimpol2}) otteniamo che, ponendo $H_{2}=S+H_{3}$ e grazie alla propriet\` a (\ref{eq:dimpol3}),

\begin{eqnarray}
\Delta_{\gamma+\Gamma}(H)&=&\zeta(\gamma)\sum_{S\subseteq H}^{*}(-1)^{N(S)}\sum_{H_{1}+H_{3}=H-S}\overline{\varphi}^{-1}(H_{1})\overline{\varphi}(\Gamma+S+H_{3})\nonumber\\&=&\zeta(\gamma)\sum_{S\subseteq H}^{*}(-1)^{N(S)}\Delta_{S+\Gamma}(H-S).\label{eq:dimpol4}
\end{eqnarray} 

\` E importante notare che grazie all'espressione (\ref{eq:dimpol4}) possiamo determinare una funzione di $N(H)+N(\Gamma)+\gamma=m+1$ polimeri attraverso quantit\` a che dipendono da $N(H)+N(\Gamma)=m$  polimeri; grazie a questa propriet\` a saremo in grado di fare la seguente stima. Definiamo $I_{m}$ come

\begin{equation}
I_{m}\equiv \sup_{\gamma_{1},...,\gamma_{q}, m\geq q\geq 1}\sum_{H:N(H)=m-q}\left|\Delta_{\{\gamma_{1},...,\gamma_{q}\}}(H)\right|\prod_{i=1}^{q}z\left(\gamma_{i}\right)^{-1}; \label{eq:dimpol04}
\end{equation}

quindi,

\begin{eqnarray}
I_{m+1}&=&\sup_{\Gamma, 1\leq|\Gamma|\leq m}\sum_{H:N(\Gamma)+N(H)=m}\left|\Delta_{\gamma+\Gamma}(H)\right|z(\gamma+\Gamma)^{-1}\nonumber\\
&\leq&\sup_{\Gamma, 1\leq |\Gamma|\leq m}\sum_{H:N(\Gamma)+N(H)=m}\sum_{S\subseteq H}^{*}\left|\Delta_{\Gamma+S}(H-S)\right|\left|z(\gamma+\Gamma)\right|^{-1}\left|\zeta(\gamma)\right|\nonumber\\
&\leq& \sup_{\Gamma, 1\leq |\Gamma|\leq m}z(\gamma)\sum_{S:N(S)\leq m-N(\Gamma)}^{*}I_{m}z(S)\leq z(\gamma)I_{m}\sum_{n\geq 0}\frac{1}{n!}\left(\sum_{\gamma'\cap\gamma\neq\emptyset}z(\gamma')\right)^{n}\nonumber\\
&\leq& z(\gamma)\exp{\sum_{\xi\in\gamma}\sum_{\gamma'\ni \xi}z(\gamma')} \leq z(\gamma)I_{m}\exp{|\gamma|\sup_{\xi\in\mathbb{Z}^{d}}\sum_{\gamma'\ni \xi}z(\gamma')}\nonumber\\&\leq&z_{1}(\gamma)I_{m},\label{eq:dimpol9}
\end{eqnarray}

dove

\begin{equation}
z_{1}(\gamma)=b_{0}\nu_{0}^{|\gamma|}e^{-\kappa_{0}\delta(\gamma)}e^{|\gamma|\sup_{\xi\in\mathbb{Z}^{d}}\sum_{\gamma'\ni \xi}z(\gamma')}. \label{eq:dimpol5}
\end{equation}

La quantit\` a all'argomento dell'esponenziale nella (\ref{eq:dimpol5}) pu\` o essere scritta esplicitamente come, {\em cf.} formula (\ref{eq:appD014}),

\begin{equation}
\sup_{\xi\in\mathbb{Z}^{d}}\sum_{\gamma'\ni \xi} z(\gamma')=\sup_{\xi\in\mathbb{Z}^{d}}\sum_{\gamma'\ni \xi}b_{0}\nu_{0}^{|\gamma|}e^{-\kappa_{0}\delta(\gamma)}; \label{eq:dimpol6}
\end{equation}

\` e facile dimostrare che

\begin{equation}
\sup_{\xi\in\mathbb{Z}^{d}}\sum_{\gamma'\ni \xi}b_{0}\nu_{0}^{|\gamma|}e^{-\kappa_{0}\delta(\gamma)}\leq B_{0}\nu_{0}.
\end{equation}

Infatti, consideriamo in generale

\begin{equation}
\sup_{\xi\in\mathbb{Z}^{d}}\sum_{\gamma\ni\xi,\delta(\gamma)\geq r}b_{0}\nu_{0}^{|\gamma|}e^{-\kappa_{0}\delta(\gamma)}; \label{eq:dimpol7}
\end{equation}

dal momento che il numero di alberi con $n+1$ nodi di cui uno \` e una {\em radice}, {\em i.e.} \` e fisso, \` e maggiorato da\footnote{Ci\` o pu\` o essere dimostrato nel modo seguente. Consideriamo $n+1$ segmenti orientati indicizzati $1,2...,n+1$, chiamiamo {\em origine} $v$ di un segmento il suo punto iniziale e {\em punti finali} gli altri estremi. Tra questi, scegliamo un segmento e chiamiamo il suo punto finale {\em radice}; questo segmento \` e {\em fisso}, {\em i.e.} non parteciper\` a alla costruzione del grafo. In seguito, ``attacchiamo'' uno dopo l'altro i segmenti rimanenti ai punti finali dei segmenti gi\` a disposti. Chiamando ``nodi'' i punti iniziali di ogni segmento tranne quello associato alla radice, ogni nodo $v$ avr\` a $d_{v}$ segmenti entranti; quindi, il numero $\mathcal{N}(n)$ di possibili grafi distinti pu\` o essere maggiorato con

\begin{equation}
\mathcal{N}(n)=n!\prod_{i=1}^{n}\sum_{d_{i}=1}^{n,*}\frac{1}{d_{i}!}\leq e^{n}n!,
\end{equation}

dove con l'asterisco imponiamo la condizione $\sum_{i=1}^{n}d_{i}=n$, con il fattore $n!$ teniamo conto della permutazione dei vertici e i rapporti $\frac{1}{d_{i}!}$ servono a non contare le permutazioni equivalenti di rami che insistono sullo stesso nodo.

} $e^{n}n!$, chiamando $\delta_{\vartheta}(\xi,\xi_{1},...,\xi_{p})$ la lunghezza dell'albero $\vartheta$ che ha per nodi i punti $\xi,\xi_{1},...,\xi_{p}$ e $\xi$ come radice, possiamo stimare la (\ref{eq:dimpol7}) come, indicando con l'asterisco la condizione\footnote{Ricordiamo che $\delta(\gamma)=\min_{\vartheta}\delta_{\vartheta}(\gamma)$.} $\delta(\xi,\xi_{1},...,\xi_{q})\geq r$:

\begin{eqnarray}
\sup_{\xi\in\mathbb{Z}^{d}}\sum_{\gamma\ni\xi,\delta(\gamma)\geq r}b_{0}\nu_{0}^{|\gamma|}e^{-\kappa_{0}\delta(\gamma)}&=&\sup_{\xi\in\mathbb{Z}^{d}}\sum_{q\geq 0}\frac{1}{q!}\sum_{0\neq \xi_{1}\in\mathbb{Z}^{d}}...\sum_{0\neq \xi_{q}\in\mathbb{Z}^{d}}^{*}\sum_{\vartheta}e^{-\kappa_{0}\delta_{\vartheta}(\xi,\xi_{1},...,\xi_{q})}\nonumber\\&=&\sup_{\xi\in\mathbb{Z}^{d}}\sum_{q\geq 0}\frac{1}{q!}\sum_{\vartheta}\sum_{0\neq \xi_{1}\in\mathbb{Z}^{d}}...\sum_{0\neq \xi_{q}\in\mathbb{Z}^{d}}^{*}e^{-\kappa_{0}\delta_{\vartheta}(\xi,\xi_{1},...,\xi_{q})}\nonumber\\&\leq&\nu_{0}b_{0}e^{-\kappa_{0}\delta(r)}\sum_{q\geq 0}\left(e\sum_{0\neq \xi\in\mathbb{Z}^{d}}\nu_{0}e^{-\kappa_{0}|\xi|}\right)\nonumber\\&\leq& \nu_{0}b_{0}e^{-\kappa_{0}\delta(r)}\sum_{q\geq 0}2^{-q}=\nu_{0}B_{0}e^{-\kappa_{0}\delta(r)}, \label{eq:dimpol13}
\end{eqnarray}

dove nella ultima disuguaglianza abbiamo usato che $k(d,\kappa_{0})\nu_{0}<1$, {\em cf.} formula (\ref{eq:dimpol00}), e abbiamo posto $B_{0}=2b_{0}$. Inoltre, grazie alla (\ref{eq:dimpol00}) abbiamo che $\nu_{0}e^{B_{0}\nu_{0}}<1$, $2b_{0}\nu_{0}e^{B_{0}\nu_{0}}<1$ dunque

\begin{equation}
z_{1}(\gamma)<\frac{1}{2};
\end{equation}

inoltre, sempre grazie alla (\ref{eq:dimpol00}), nel caso $m=1$ la (\ref{eq:dimpol04}) diventa

\begin{equation}
I_{1}=\sup_{\gamma}\left|\Delta_{\gamma}(\emptyset)\right|z(\gamma)^{-1}=\sup_{\gamma}\frac{|\zeta(\gamma)|}{z(\gamma)}\leq b_{0}\nu_{0}<\frac{1}{2},
\end{equation}

quindi, in conclusione la (\ref{eq:dimpol9}) implica che

\begin{equation}
I_{m}<2^{-m}\qquad\mbox{per ogni $m\geq 1$}. \label{eq:dimpol10}
\end{equation}

Grazie alla propriet\` a (\ref{eq:dimpol11}) abbiamo che

\begin{equation}
\Delta_{\gamma}(\Gamma)=\left(\overline{\varphi}^{-1}\ast D_{\gamma}\overline{\varphi}\right)(\Gamma)=D_{\gamma}\overline{\varphi}^{T}(\Gamma)=\overline{\varphi}^{T}(\gamma+\Gamma)\frac{(\gamma+\Gamma)!}{\Gamma!},
\end{equation}

quindi,

\begin{equation}
\sum_{\Gamma}\left|\overline{\varphi}^{T}(\gamma+\Gamma)\right|\leq \sum_{m\geq 1}\sum_{\Gamma:N(\Gamma)=m-1}\left|\Delta_{\gamma}(\Gamma)\right|\leq z(\gamma)\sum_{m\geq 1}I_{m}=z(\gamma) \label{eq:dimpol12}
\end{equation}

e dunque, per le (\ref{eq:dimpol13}), (\ref{eq:dimpol12}), se $Q\in\mathbb{Z}^{d}$ \` e un insieme che dista $r$ dal punto $p$:

\begin{equation}
\sum_{\Gamma\ni p:\Gamma\cap Q\neq\emptyset}\left|\overline{\varphi}^{T}(\Gamma)\right|\leq \sum_{\gamma\ni p:\gamma\cap Q\neq\emptyset}\sum_{\Gamma}\left|\overline{\varphi}^{T}(\gamma+\Gamma)\right|\leq \sum_{\gamma\ni p:\gamma\cap Q\neq\emptyset}z(\gamma)\leq B_{0}\nu_{0}e^{-\kappa_{0} r}. \label{eq:dimpol14}
\end{equation}

\end{proof}

Come vedremo nel prossimo corollario, le propriet\` a che abbiamo appena dimostrato con la proposizione \ref{prop:appD1} implicheranno l'esistenza del {\em limite termodinamico} per un gas di polimeri con attivit\` a invarianti per traslazioni.

\begin{cor}{{\em (Analiticit\` a e espansione in polimeri)}}\label{cor:appD1}

Sotto le ipotesi della proposizione \ref{prop:appD1}, se $\zeta(\gamma)$ \`e invariante per traslazioni e analitica in un parametro $\beta$ allora il limite

\begin{equation}
P\equiv \lim_{\Lambda\rightarrow\infty}P_{\Lambda}\equiv \lim_{\Lambda\rightarrow\infty}\Lambda^{-1}\log Z_{\Lambda}
\end{equation}

esiste ed \` e analitico in $\beta$.
\end{cor}

\begin{proof}

La {\em pressione} $P_{\Lambda}$ dei polimeri con attivit\` a $\gamma\rightarrow\zeta(\gamma)$ pu\` o essere scritta come

\begin{equation}
P_{\Lambda}=\Lambda^{-1}\log Z_{\Lambda}=\Lambda^{-1}\sum_{\Gamma\subset\Lambda}\varphi^{T}(\Gamma)\zeta(\Gamma),
\end{equation}

e sotto le ipotesi della proposizione \ref{prop:appD1} sappiamo che

\begin{equation}
\sup_{\xi\in\Lambda}\sum_{\Gamma\ni\xi,\delta(\Gamma)\geq r}\left|\varphi^{T}(\Gamma)\right|\left|\zeta(\Gamma)\right|<B_{0}\nu_{0}e^{-\kappa_{0}r}. \label{eq:appD19}
\end{equation}

Se $\zeta$ \` e invariante per traslazioni,

\begin{equation}
\lim_{\Lambda\rightarrow\infty}\Lambda^{-1}\sum_{\Gamma\subset\Lambda}\varphi^{T}(\Gamma)\zeta(\Gamma)=\sum_{\Gamma\ni 0}\frac{\varphi^{T}(\Gamma)\zeta(\Gamma)}{n(\Gamma)},
\end{equation}

dove $n(\Gamma)$ \` e il numero di punti contenuti in $\Gamma$; inoltre, chiamando $\Delta$ la costante {\em indipendente da $\beta$} necessaria per stimare i termini al bordo\footnote{\label{nota:beta}La costante $\Delta$ pu\` o essere stimata come

\begin{equation}
\Delta\leq\sup_{\xi}\sum_{\Gamma\ni\xi}\frac{\left|\varphi^{T}(\Gamma)\zeta(\Gamma)\right|}{n(\Gamma)}<B_{0}\nu_{0};
\end{equation}

nelle applicazioni alla meccanica statistica, in generale la costante $\nu_{0}$ sar\` a proporzionale a $e^{\left\|\Phi\right\|_{\kappa}}$, e $\Phi$ dipende da $\beta$. Ma questo non rappresenta un problema, perch\' e nei casi in cui sapremo ricondurci alla proposizione \ref{prop:appD1} la quantit\` a $\left\|\Phi\right\|_{\kappa}$ sar\` a limitata, in particolare {\em piccola}.

} \` e facile vedere che

\begin{eqnarray}
\left|P-P_{\Lambda}\right|&=&\left|\sum_{\Gamma\ni 0}\frac{\varphi^{T}(\Gamma)\zeta(\Gamma)}{n(\Gamma)}-\sum_{\Gamma\ni 0, |\Gamma|<\Lambda}\frac{\varphi^{T}(\Gamma)\zeta(\Gamma)}{n(\Gamma)}\right|+\frac{\Delta}{\Lambda}\nonumber\\
&\leq&\sum_{\Gamma\ni\xi,|\Gamma|\geq \Lambda}\left|\frac{\varphi^{T}(\Gamma)\zeta(\Gamma)}{n(\Gamma)}\right|+\frac{\Delta}{\Lambda}\nonumber\\&\leq&\sum_{p=0}^{\infty}\sum_{\Gamma\ni\xi,\delta(\Gamma)\geq \Lambda+p}\left|\varphi^{T}(\Gamma)\right|\left|\zeta(\Gamma)\right|+\frac{\Delta}{\Lambda}\nonumber\\&\leq& B_{0}\nu_{0}\sum_{p=0}^{\infty}e^{-\frac{\kappa_{0}}{2}(p+\Lambda)}+\frac{\Delta}{\Lambda}\rightarrow_{\Lambda\rightarrow\infty}0
\end{eqnarray}

uniformemente in $\beta$, {\em cf.} nota \ref{nota:beta}. Infine, se $\zeta(\gamma)$ \` e invariante per traslazioni

\begin{eqnarray}
\left|P_{\Lambda}\right|&\leq& \sum_{\Gamma\ni 0}\left|\varphi^{T}(\Gamma)\right|\left|\zeta(\Gamma)\right|+\frac{\Delta}{\Lambda}\leq \sum_{r\geq 0}\sum_{\Gamma\ni 0, \delta(\Gamma)\geq r} \left|\varphi^{T}(\Gamma)\right|\left|\zeta(\Gamma)\right|+\frac{\Delta}{\Lambda}\nonumber\\&\leq& B_{0}\nu_{0}\sum_{r\geq 0}e^{-\frac{\kappa_{0}}{2}r}+\Delta=M<\infty;
\end{eqnarray}

allora, poich\'e $P_{\Lambda}$ \`e analitica in $\beta$ per il teorema della convergenza di Vitali, {\em cf.} \cite{tit}, anche il limite $P=\lim_{\Lambda\rightarrow\infty}P_{\Lambda}$ lo sar\` a.

\end{proof}

\subsection{Cluster expansion}

In questa sezione discutiamo come ``polimerizzare'' lo studio delle propriet\` a di analiticit\` a di un sistema di spin in una dimensione, in modo da rendere applicabili i risultati della sezione precedente; l'analisi che stiamo per fare pu\` o essere resa molto pi\` u generale, \cite{ergo}, ma al prezzo di imporre condizioni pi\` u restrittive sulla forma dei potenziali d'interazione.
 
 \begin{prop}{{\em (Cluster expansion per un modello di Fisher)}} \label{prop:appD3}
 
 Consideriamo un sistema di $\frac{n}{2}$ spin su $\mathbb{Z}$ senza interazioni di tipo ``cuore duro'', ovvero senza vincoli di compatibilit\` a. Sia $\Phi$ un potenziale non necessariamente invariante per traslazioni; ponendo $\left\|\Phi_{X}\right\|\equiv\max_{\underline{\sigma}_{X}}\left|\Phi_{X}\left(\underline{\sigma}_{X}\right)\right|$ definiamo, per $\kappa>0$,
 
 \begin{equation}
 \left\|\overline{\Phi}\right\|_{\kappa}\equiv\sup_{\xi}\sum_{X\ni,|X|>1}\left\|\Phi_{X}\right\|e^{\kappa\delta(X)}, \label{eq:appD019}
 \end{equation}
 
 dove $\delta(X)$ \` e la lunghezza dell'intervallo $X$, {\em i.e.} $\delta(X)=|X|-1$.
 
 (i) Indipendentemente dal fatto che $\Phi$ sia invariante per traslazioni, la funzione di partizione $Z^{\Phi}_{\Lambda}=\sum_{\underline{\sigma}_{\Lambda}}e^{-\sum_{X\subset\Lambda}\Phi_{X}(\underline{\sigma}_{X})}$ pu\` o essere scritta formalmente come
 
 \begin{equation}
 Z^{\Phi}_{\Lambda}=\left[\prod_{\xi\in\Lambda}z(\xi)\right]\exp\left(\sum_{X\subset\Lambda}\varphi^{T}(X)\zeta(X)\right), \label{eq:appD20}
 \end{equation}
 
 dove $z(\xi)=\sum_{\sigma}e^{-\Phi_{\xi}(\sigma)}$, $X=\left\{\gamma_{1},...,\gamma_{n}\right\}$ \` e una configurazione di polimeri, $\varphi^{T}(X)$ \` e la funzione combinatoria definita nella (\ref{eq:appD17}), e $\zeta(X)=\prod_{\gamma\in X}\zeta(\gamma)$ per un certo $\zeta(\gamma)$ che verifica
 
 \begin{equation}
  \left|\zeta(\gamma)\right|\leq \left(e^{2\left(\left\|\overline{\Phi}\right\|_{\kappa}+1\right)}e^{-\frac{\kappa}{2}}\right)^{|\gamma|}e^{-\frac{\kappa}{2}\delta(\gamma)}. \label{eq:appD020}
 \end{equation}
 
 La (\ref{eq:appD20}) diventa un'identit\` a se la somma all'esponente \` e assolutamente convergente.
 
 (ii) Se i potenziali $\Phi$ verificano la disuguaglianza
 
 \begin{equation}
 e^{2\left\|\overline{\Phi}\right\|_{\kappa}}e^{-\frac{1}{4}\kappa}\kappa^{-1}\leq\varepsilon\qquad\mbox{per $0<\varepsilon<1$}, \label{eq:appD021}
 \end{equation}
 
 allora la (\ref{eq:appD020}) \` e sostituita da
 
 \begin{equation}
 \left|\zeta(\gamma)\right|\leq e^{-L_{0}\left|\gamma\right|}e^{-\frac{1}{2}\kappa\delta(\gamma)}, \label{eq:stimaL.0}
 \end{equation}
 
 con $e^{-L_{0}}\equiv \varepsilon c$ per un'opportuna costante $c$; inoltre, la (\ref{eq:appD021}) implica che
  
 \begin{equation}
 \sup_{\xi\in\mathbb{Z}^{d}}\sum_{\widetilde{X}\ni\xi,\delta(X)\geq r}\left|\varphi^{T}(X)\right|\left|\zeta(X)|\right|<B\varepsilon e^{-\frac{\kappa}{2}r}.
 \end{equation}

(iii)  Infine, se la (\ref{eq:appD021}) \`e verificata e se le funzioni $\zeta(\gamma)$ sono analitiche in $\Phi$, nel caso invariante per traslazioni la misura di Gibbs $\mu_{\Phi}$ \` e analitica in $\Phi$\footnote{Nel senso che la pressione $\lim_{\Lambda\rightarrow\infty}\Lambda^{-1}\log Z_{\Lambda}^{\Phi}$ \` e analitica in $\Phi$}.
 
 \end{prop}
 
 \begin{oss}
 Come vedremo, grazie alla tecnica di decimazione che abbiamo introdotto nella sezione \ref{sez:decim} di questa appendice, in un modello di Fisher sar\` a sempre possibile soddisfare la la (\ref{eq:appD021}).  
 \end{oss}
 
 \begin{proof}
 Innanzitutto, dobbiamo precisare cosa intendiamo per {\em polimero} in questo caso. Chiamiamo polimero qualunque insieme $\gamma\subset\Lambda\subset\mathbb{Z}^{d}$ tale che per ogni coppia di punti $\xi,\eta\in\gamma$ esiste una ``catena'' $Y_{1},...,Y_{m}$ di sottoinsiemi di $\gamma$ con $\xi\in Y_{1}$, $\eta\in Y_{m}$, $Y_{i}\cap Y_{i+1}\neq \emptyset$, e per ogni $Y_{i}$ esiste almeno una configurazione $\underline{\sigma}_{Y_{i}}$ con $\Phi_{Y_{i}}\left(\underline{\sigma}_{Y_{i}}\right)\neq 0$.\footnote{Questa definizione in una dimensione \` e banale; lo \` e di meno per $d>1$.}
 
 Detto ci\` o, la funzione di partizione $Z_{\Lambda}^{\Phi}$ pu\` o essere scritta come
 
 \begin{eqnarray}
 Z_{\Lambda}^{\Phi}&=&\sum_{\underline{\sigma}_{\Lambda}\in\{1,...,n\}^{\Lambda}}e^{-\sum_{Y\subset\Lambda}\Phi_{X}\left(\underline{\sigma}_{X}\right)}\nonumber\\&=&\sum_{\underline{\sigma}_{\Lambda}\in\{1,...,n\}^{\Lambda}}\prod_{\xi\in\Lambda}e^{-\Phi_{\xi}\left(\sigma_{\xi}\right)}\prod_{Y\subset\Lambda,|Y|>1}\left[\left(e^{-\Phi_{Y}\left(\underline{\sigma}_{Y}\right)}-1\right)+1\right],\label{eq:appD21}
 \end{eqnarray}
  
e possiamo sviluppare il fattore binomiale $\prod_{Y\subset\Lambda,|Y|>1}\left[\left(e^{-\Phi_{Y}\left(\underline{\sigma}_{Y}\right)}-1\right)+1\right]$ raggruppando i termini relativi a polimeri compatibili, ovvero tali che $\gamma_{i}\cap\gamma_{j}=\emptyset$; quindi la (\ref{eq:appD21}) diventa

\begin{eqnarray}
Z_{\Lambda}^{\Phi}&=&z\left(\Lambda\right)\sum_{\gamma_{1},...,\gamma_{p}}^{*}\prod_{i=1}^{p}\left[z\left(\gamma_{i}\right)^{-1}\sum_{\underline{\sigma}_{\gamma_{i}}}\prod_{\xi\in\gamma_{i}}e^{-\Phi_{\xi}\left(\sigma_{\xi}\right)}\right.\nonumber\\&&\times \left.\sum_{q_{i}}\sum^{**}_{Y_{1}^{(i)},...,Y_{q_{i}}^{(i)}\subset\gamma_{i},\; \bigcup_{j}Y_{j}^{(j)}=\gamma_{i}}\prod_{j=1}^{q_{i}}\left(e^{-\Phi_{Y_{j}^{(i)}}\left(\underline{\sigma}_{Y_{j}^{(i)}}\right)}-1\right)\right], \label{eq:appD22}
\end{eqnarray}
  
  dove $z(\Lambda)=\prod_{\xi\in \Lambda}\sum_{\sigma}e^{-\Phi_{\xi}(\sigma)}$, l'asterisco e il doppio asterisco indicano rispettivamente che i polimeri non si intersecano e che gli insiemi $\left\{Y_{j}^{(i)}\right\}$ devono contenere una catena che connette due qualsiasi punti di $\gamma_{i}$. Grazie alla (\ref{eq:appD22}) \` e naturale introdurre le {\em attivit\` a} dei polimeri $\gamma$ e delle configurazioni di polimeri $X$ come
  
  \begin{eqnarray}
  \zeta(\gamma)&\equiv&z^{-1}(\gamma)\sum_{\underline{\sigma}_{\gamma}}\prod_{\xi\in\gamma}e^{-\Phi_{\xi}\left(\sigma_{\xi}\right)}\nonumber\\&&\times\sum_{q}\sum^{**}_{\begin{array}{c} \scriptstyle{Y_{1},...,Y_{q}\subset\gamma} \\ \scriptstyle{\bigcup_{j}Y_{j}=\gamma} \end{array}}\prod_{j=1}^{q}\left(e^{-\Phi_{Y_{j}}\left(\underline{\sigma}_{Y_{j}}\right)}-1\right)\nonumber\\&\equiv&\left\langle\sum_{q}\sum^{**}_{\begin{array}{c} \scriptstyle{Y_{1},...,Y_{q}\subset\gamma} \\ \scriptstyle{\bigcup_{j}Y_{j}=\gamma} \end{array}}\prod_{j=1}^{q}\left(e^{-\Phi_{Y_{j}}\left(\underline{\sigma}_{Y_{j}}\right)}-1\right)\right\rangle\label{eq:att}\\
  \zeta(X)&\equiv& \prod_{\gamma\in X}\zeta(\gamma)^{X(\gamma)},\qquad\mbox{$\zeta(\emptyset)\equiv 1$}, 
  \end{eqnarray}
  
  e l'ultima cosa che ci rimane da fare \` e verificare che l'attivit\` a $\zeta(\gamma)$ cos\` i introdotta verifichi la stima (\ref{eq:appD014}) necessaria per applicare i risultati della proposizione \ref{prop:appD1} e del corollario \ref{cor:appD1}. Ci\` o pu\` o essere fatto nel modo seguente:
  
  \begin{eqnarray}
  \left|\zeta(\gamma)\right| &\leq& \left|\sum_{q = 1}^{|\gamma|}\sum_{\left\{Y_{i}\right\}_{i=1}^{q}}\prod_{j=1}^{q}\left(e^{-\Phi_{Y_{j}}\left(\underline{\sigma}_{Y_{j}}\right)}-1\right)\right|\nonumber\nonumber\\&\leq&\sum_{q= 1}^{|\gamma|}\sum_{\left\{Y_{i}\right\}_{i=1}^{q}}\prod_{j=1}^{q}\left|\left(e^{-\Phi_{Y_{j}}\left(\underline{\sigma}_{Y_{j}}\right)}-1\right)\right|\nonumber\\&\leq& \sup_{\xi}\sum_{q= 1}^{|\gamma|}\prod_{i=1}^{q}\sum_{Y_{i}\ni\xi}\left|\Phi_{Y_{i}}(\underline{\sigma}_{Y_{i}})\right|e^{\left|\Phi_{Y_{i}}\left(\underline{\sigma}_{Y_{i}}\right)\right|}\label{eq:appD24}\qquad\mbox{(da $\left|e^{x}-1\right|\leq|x|e^{|x|}$)}\quad\qquad\nonumber\\&\leq&e^{|\gamma|\left\|\Phi\right\|_{\kappa}}\sum_{q = 1}^{|\gamma|}\left\|\Phi\right\|^{|\gamma|}_{\kappa}e^{-\kappa\delta(\gamma)}\label{eq:appD25.0}\qquad\quad\qquad\mbox{({\em cf.} (\ref{eq:appD019}))}\nonumber\\&\leq&e^{|\gamma|\left\|\Phi\right\|_{\kappa}}\left|\gamma\right|\left\|\Phi\right\|_{\kappa}^{|\gamma|}e^{-\kappa\delta(\gamma)}\nonumber\\&\leq&\left(e^{2\left\|\Phi\right\|_{c_{0}}+1}\right)^{|\gamma|}e^{-c_{0}\delta(\gamma)}\quad\qquad\qquad\quad\mbox{($|x|<e^{|x|}$)},\nonumber 
  \end{eqnarray}
 
  e, poich\' e per $d=1$ $\delta(\gamma)=|\gamma|-1$, quindi per polimeri con $\delta(\gamma)\geq 1$ $\delta(\gamma)\geq \frac{1}{2}|\gamma|$,
  
  \begin{equation}
  \left|\zeta(\gamma)\right|\leq \left(e^{2\left(\left\|\overline{\Phi}\right\|_{\kappa}+1\right)}e^{-\frac{\kappa}{2}}\right)^{|\gamma|}e^{-\frac{\kappa}{2}\delta(\gamma)}.
  \end{equation}
  
  Ponendo $b_{0}^{2}=1$, $\nu_{0}^{2}=e^{2\left\|\Phi\right\|+1-\frac{\kappa}{2}}$, il risultato finale segue per la proposizione \ref{prop:appD1} e per il corollario \ref{cor:appD1} sotto la condizione
  
  \begin{equation}
  \left[3e^{B_{0}\nu_{0}}+k\left(1,\frac{\kappa}{2}\right)\right]\nu_{0}<1 \label{eq:appD25}
  \end{equation}
  
  con $B_{0}=2$ e ({\em cf.} proposizione \ref{prop:appD1}, formula (\ref{eq:appD015}))
  
  \begin{equation}
  k(1,\kappa_{0})=2e\sum_{\xi\in\mathbb{Z},|\xi|\geq 1}e^{-\kappa_{0}|\xi|}\leq 4e^{-\kappa_{0}}\left(1+\kappa_{0}^{-1}\right)<8\frac{e^{-\frac{\kappa_{0}}{2}}}{\kappa_{0}}; \label{eq:appD26}
  \end{equation}
  
  si pu\` o dimostrare che la condizione (\ref{eq:appD25}) \` e implicata dalla (\ref{eq:appD021}). Infatti, definiamo $e^{-L_{0}}$ come
  
  \begin{equation}
  e^{-L_{0}}\equiv e^{2\left\|\Phi\right\|_{\kappa}+1}\frac{e^{-\frac{\kappa}{4}}}{\kappa}, \label{eq:appD0015}
  \end{equation}
  
  e grazie alla (\ref{eq:appD021}) sappiamo che esiste una costante $c>0$ tale che
  
  \begin{equation}
  e^{-L_{0}}=c\varepsilon; \label{eq:appD00151}
  \end{equation}
  
  inoltre, \` e facile vedere che $\nu_{0}^{2}\leq e^{L_{0}}$, dunque la (\ref{eq:appD021}) implica che
  
  \begin{equation}
  \nu_{0}^{2}\leq c\varepsilon.
  \end{equation}
  
Infine, poich\' e per le (\ref{eq:appD0015}), (\ref{eq:appD00151})
  
  \begin{equation}
  \frac{e^{-\frac{\kappa}{2}}}{\kappa}\leq \frac{e^{-\frac{\kappa}{4}}}{\kappa}= c\varepsilon e^{-2\left\|\Phi\right\|_{k}-1}\leq c\varepsilon,
  \end{equation}
  
\` e possibile trovare un $\varepsilon>0$ tale che
  
  \begin{equation}
  \left[3e^{2\nu_{0}}+16\frac{e^{-\frac{\kappa}{2}}}{\kappa}\right]\nu_{0}\leq c^{\frac{1}{2}}\varepsilon^{\frac{1}{2}}\left[3e^{2c^{\frac{1}{2}}\varepsilon^{\frac{1}{2}}}+16c\varepsilon\right]<1.
  \end{equation}
  
 \end{proof}
 
 \section{Decimazione: seconda parte}
 
In questa sezione vogliamo mostrare che il nostro modello decimato, {\em cf.} sezione \ref{sez:decim}, soddisfa le ipotesi della proposizione \ref{prop:appD3}. Avevamo riscritto la funzione di partizione come

\begin{eqnarray}
Z_{\Lambda} &=& \sum_{\beta_{0},...,\beta_{l}}\left(\prod_{i=0}^{l}e^{-U^{B}_{\tau}\left(\beta_{i}\right)}\right)\left(\prod_{i=0}^{l-1}Z_{h}\left(\beta_{i},\beta_{i+1}\right)\right)\nonumber\\&&\times\frac{\sum_{\eta_{0},...,\eta_{l-1}}\left(\prod_{i=0}^{l-1}e^{-W\left(\beta_{i},\eta_{i},\beta_{i+1}\right)}\right)e^{-\sum^{*}\alpha_{X}\left(\underline{\sigma}_{X}\right)}}{\prod_{i=0}^{l-1}Z_{h}\left(\beta_{i},\beta_{i+1}\right)}, \label{eq:appD025}
\end{eqnarray}

dove l'asterisco ricorda la condizione $X\not\subset B_{i}\cup H_{i}\cup B_{i+1}\,\forall i$, {\em i.e.} coinvolge intervalli che intersecano {\em almeno due} elementi dell'$H$-reticolo, e abbiamo dimostrato che ({\em cf.} equazione (\ref{eq:appD12}))

\begin{equation}
\sup_{H}\sum_{\mathcal{H}\ni H:\left|\mathcal{H}\right|>1}\sup_{\underline{\eta}_{\mathcal{H}}}\left|\widetilde{\alpha}_{H_{1},...,H_{k}}\left(\eta_{H_{1},...,\eta_{H_{k}}}\right)\right|\leq m\left\|\overline{\alpha}\right\|_{\kappa}e^{-\frac{\kappa}{2}\tau}. \label{eq:appD26a}
\end{equation}

Adesso, se il potenziale iniziale aveva un raggio d'interazione $\kappa^{-1}$, il {\em potenziale rinormalizzato} tra i blocchi $H_{i}$ che si ottiene fissando le variabili $\beta_{i}$ avr\` a un raggio d'interazione ordine $\left(\tau\kappa\right)^{-1}$; dunque, possiamo riformulare la condizione (\ref{eq:appD021}) come, sostituendo $\left\|\overline{\alpha}\right\|$ con $m\left\|\overline{\alpha}\right\|_{\kappa}e^{-\frac{\kappa}{2}\tau}$,

\begin{equation}
e^{2m\left\|\overline{\alpha}\right\|_{\kappa}e^{-\frac{\kappa}{2}\tau}-\frac{1}{4}\tau\kappa}(\tau\kappa)^{-1}<\varepsilon.\label{eq:appD27}
\end{equation}

Senza perdita di generalit\` a, supponiamo che $\tau=cost\cdot h$; ci\` o implica che per $h$ opportunamente grandi la (\ref{eq:appD27}) sar\` a sicuramente verificata. Il rapporto nella (\ref{eq:appD025}) ha la familiare forma di una {\em media} sugli spin $\eta_{i}$, con pesi $e^{-W(\beta_{i},\eta_{i},\beta_{i+1})}$; quindi, possiamo ripetere l'argomento esposto nella dimostrazione della proposizione \ref{prop:appD3} e dunque, se $\Gamma$ \` e una configurazione di polimeri appartenenti all'$H$-reticolo, il rapporto nella (\ref{eq:appD025}) pu\` o essere riscritto come

\begin{equation}
\exp\sum_{\Gamma}\varphi^{T}(\Gamma)\zeta(\Gamma), \label{eq:appD027}
\end{equation}

dove la somma all'esponente \` e assolutamente convergente per la propriet\` a (\ref{eq:appD27}), {\em cf.} proposizione \ref{prop:appD3}, grazie alla quale le attivit\` a $\zeta(\gamma)$ verificano la stima\footnote{Anche qui, consideriamo $H$ abbastanza grande da poter tenere conto degli altri fattori all'esponente semplicemente scegliendo opportunamente i coefficienti, in questo caso pari a $\frac{1}{4}$, $\frac{1}{8}$.}

\begin{equation}
\left|\zeta(\gamma)\right|\leq e^{-\frac{1}{8}\kappa\tau|\gamma|}e^{-\frac{1}{4}\kappa\tau\delta(\gamma)}\leq e^{-\frac{1}{8}\kappa\tau|\gamma|-\frac{1}{8}\kappa\tau\delta(\gamma)} \label{eq:appD28b}
\end{equation}

se le distanze $\delta(\gamma)$ sono misurate sull' $H$-reticolo (quindi essenzialmente in unit\` a $\tau+h+2$). Definiamo i nuovi {\em potenziali efficaci}

\begin{equation}
\Psi_{B_{j},...,B_{j+p}}\left(\beta_{j},...,\beta_{j+p}\right)=\sum_{\Gamma\cap B_{i}\neq \emptyset, i=0,...,p}\varphi^{T}(\Gamma)\zeta(\Gamma), \label{eq:appD0028}
\end{equation}

che, lontano dai bordi del $B$-reticolo, riflettono l'invarianza per traslazioni di $\alpha_{X}$. Grazie alla (\ref{eq:appD27}), dalla proposizione \ref{prop:appD3} sappiamo che

\begin{equation}
\sup_{\xi\in\Lambda}\sum_{\Gamma\ni\xi,\delta(\Gamma)\geq r}\left|\varphi^{T}(\Gamma)\right|\left|\zeta(\Gamma)\right|<B_{0}\nu_{0}e^{-\kappa_{0}r},
\end{equation}

con $\kappa_{0}=\frac{\kappa\tau}{8}$, {\em cf.} stima (\ref{eq:appD28b}); dunque

\begin{eqnarray}
\max_{\underline{\beta}_{p}}\left|\Psi_{B_{j},...,B_{j+p}}\left(\beta_{j},...,\beta_{j+p}\right)\right| &\leq&  \sum_{q=0}^{\infty}(q+1)\sup_{\xi}\sum_{\Gamma\ni\xi, \delta(\Gamma)\geq h(p-1)+q}\left|\varphi^{T}(\gamma)\right|\left|\zeta(\Gamma)\right|\nonumber\\&\leq& \sum_{q=0}^{\infty}(q+1)B_{0}\nu_{0}e^{-\frac{\kappa\tau}{8}\left(h(p-1)+q\right)}\nonumber\\&=& A_{0}e^{-\frac{\kappa\tau}{16}}e^{-\frac{\kappa\tau}{8}h(p-1)}, \qquad A_{0}=B_{0}\sum_{q=0}^{\infty}(q+1)e^{-\frac{\kappa\tau}{8}q}, \label{eq:appD28}
\end{eqnarray}

dal momento che $\left|\zeta(\gamma)\right|\leq \nu_{0}^{2|\gamma|}e^{-\kappa_{0}\delta(\gamma)}$ con $\nu_{0}=e^{-\frac{\kappa\tau}{16}}$, $\kappa_{0}=\frac{\kappa\tau}{8}$; {\em cf.} proposizione \ref{prop:appD1}. Quindi, il potenziale efficace $\Psi$ pu\` o essere reso arbitrariamente piccolo. Chiamando $X$ un intervallo nel $B$-reticolo, grazie alla (\ref{eq:appD28}) abbiamo che

\begin{eqnarray}
\left\|\Psi\right\|_{\frac{\kappa\tau}{32}}&\equiv&\sup_{B}\sum_{X\ni B, |X|>1}\left\|\Psi_{X}\right\|e^{\frac{\kappa}{32}\delta(X)}\nonumber\\&\leq& \sum_{p\geq 1}pA_{0}e^{-\frac{\kappa\tau}{16}}e^{-\frac{\kappa h}{8}h(p-1)}e^{\frac{\kappa}{32}\left(h(p-1)+\tau\right)}\nonumber\\&=&C_{0}e^{-\frac{\kappa\tau}{32}}, 
\end{eqnarray}

ovvero anche il raggio d'interazione pu\` o essere reso arbitrariamente piccolo. Dunque, la funzione di partizione (\ref{eq:appD025}) pu\` o essere riscritta come:

\begin{equation}
Z_{\Lambda}=\sum_{\underline{\beta}_{[0,...,l]}}e^{-\sum_{X,|X|>1}\Psi_{X}\left(\underline{\beta}_{X}\right)}e^{-\sum_{i=0}^{l}U^{B}(\beta_{i})}e^{-\sum_{i=0}^{l}\Psi_{\xi_{i}}(\beta_{i})}\prod_{i=0}^{l-1}Z_{h}\left(\beta_{i},\beta_{i+1}\right), \label{eq:appD028}
\end{equation}

dove $\Psi_{X}$ \` e arbitrariamente piccolo se $|X|>1$ e $U^{B}(\beta_{i})$, $\Psi_{\xi_{i}}(\beta_{i})$ rappresentano dei potenziali a singolo corpo; tutto \` e invariante per traslazioni lontano dai bordi del $B$-reticolo. Infine, se $\underline{\omega}_{1}$, $\underline{\omega}_{2}$ sono due generiche sequenze compatibili possiamo tener conto del prodotto di  funzioni di partizione nella (\ref{eq:appD028}) sostituendo $U^{B}(\beta_{i})$, $\Psi_{X}(\underline{\beta}_{X})$ con i nuovi potenziali 

\begin{eqnarray}
\widetilde{U}^{B}(\beta_{i})&\equiv& U^{B}(\beta_{i})-\log\frac{Z_{h}(\underline{\omega}_{1},\beta_{i})Z_{h}(\beta_{i},\underline{\omega}_{2})}{Z_{h}(\underline{\omega}_{1},\underline{\omega}_{2})Z_{h}(\underline{\omega}_{1},\underline{\omega}_{2})}\\
\widetilde{\Psi}_{i,i+1}(\underline{\beta}_{[i,i+1]})&\equiv& \Psi_{i,i+1}(\underline{\beta}_{[i,i+1]})-\log\frac{Z_{h}(\beta_{i},\beta_{i+1})Z_{h}(\underline{\omega}_{1},\underline{\omega}_{2})}{Z_{h}(\underline{\omega}_{1},\beta_{i+1})Z_{h}(\beta_{i},\underline{\omega}_{2})}\\
\widetilde{\Psi}_{X}(\underline{\beta}_{X})&\equiv& \Psi_{X}(\underline{\beta}_{X}),\qquad\mbox{se $|X|>2$}; 
\end{eqnarray}

con questa definizione, la (\ref{eq:appD028}) diventa

\begin{equation}
Z_{\Lambda}=\sum_{\underline{\beta}_{[0,...,l]}}e^{-\sum_{X,|X|>1}\widetilde{\Psi}_{X}\left(\underline{\beta}_{X}\right)}e^{-\sum_{i=0}^{l}\widetilde{U}^{B}(\beta_{i})}e^{-\sum_{i=0}^{l}\Psi_{\xi_{i}}(\beta_{i})}. \label{eq:newpart}
\end{equation}

Abbiamo fatto questa scelta perch\'e, come vedremo nella prossima sezione, 

\begin{equation}
\lim_{h\rightarrow\infty}\sup_{\beta_{i},\beta_{i+1}}\log\frac{Z_{h}(\beta_{i},\beta_{i+1})Z_{h}(\underline{\omega}_{1},\underline{\omega}_{2})}{Z_{h}(\beta_{i},\underline{\omega}_{2})Z_{h}(\underline{\omega}_{1},\beta_{i+1})}=0; \label{eq:appD029}
\end{equation}

dunque, scegliendo $h$ abbastanza grande i potenziali $\widetilde{\Psi}_{X}$ verificheranno le stesse stime che abbiamo ricavato per $\Psi_{X}$, eventualmente modificando in modo opportuno le costanti che vi appaiono. Con ci\` o, il problema \` e definitivamente ricondotto allo studio di un modello unidimensionale caratterizzato da potenziali a pi\` u di un corpo con taglia e raggio d'interazione piccoli a piacere, quindi siamo nelle condizioni della proposizione \ref{prop:appD3}; in particolare, grazie al corollario \ref{cor:appD1} la pressione del nostro modello di Fisher iniziale $P=\lim_{l\rightarrow\infty}\Lambda(l)\log Z_{\Lambda(l)}^{\beta}$ \` e analitica in $\beta$. Per concludere, dimostriamo la propriet\` a (\ref{eq:appD029}).

\section{``Fattorizzazione'' di una funzione di partizione in una dimensione}

Per dimostrare la propriet\` a (\ref{eq:appD029}) ci serviremo di un risultato di Ruelle ben noto in teoria dei sistemi di spin unidimensionali, \cite{pfr}, \cite{abs}, \cite{ergo}. Se $A_{\Phi}(\underline{\sigma})$ \`e  l'energia potenziale per sito di una misura di Gibbs con potenziali $\left\{\Phi_{X}\right\}$, {\em cf.} definizione \ref{def:potsimb}, introduciamo la sua restrizione $\widehat{A}_{\Phi}(\underline{\sigma})$ su $\mathbb{Z}^{+}$ ({\em i.e.} sugli interi $\geq 0$) come

\begin{equation}
\widehat{A}_{\Phi}(\underline{\sigma})\equiv \sum_{X\ni 0, X\subset\mathbb{Z}^{+}}\Phi_{X}(\underline{\sigma}_{X}), \label{eq:dimpro}
\end{equation}

dove $\Phi_{X}\in B$ sono i potenziali della misura di Gibbs su $\mathbb{Z}$; diremo che $m$ \` e una misura di Gibbs\index{misura di Gibbs} su $\mathbb{Z}^{+}$ se le probabilit\` a condizionate ({\em cf.} definizione \ref{def:probcond}) verificano\footnote{Poich\'e $\widehat{A}$ \` e definita in $\{1,...,n\}_{T}^{\mathbb{Z}^{+}}$, la somma argomento dell'esponenziale nella (\ref{eq:dimpro1}) coinvolge al massimo $a+1$ termini.}, ponendo $\underline{\sigma}'=(\underline{\sigma}'_{a}\underline{\sigma}_{(a,\infty]})$, $\underline{\sigma}''=(\underline{\sigma}''_{a}\underline{\sigma}_{(a,\infty]})$,

\begin{equation}
\frac{m(\sigma'_{0}...\sigma'_{a}|\sigma_{a+1}...)}{m(\sigma''_{0},...,\sigma''_{a}|\sigma_{a+1}...)}=\exp\left(-\sum_{k=0}^{\infty}\left[\widehat{A}(\tau^{k}\underline{\sigma}')-\widehat{A}(\tau^{k}\underline{\sigma}'')\right]\right).  \label{eq:dimpro1}
\end{equation}

\begin{prop}\label{prop:pfr}

Introduciamo l'{\em operatore di trasferimento}\index{operatore di trasferimento} $\mathcal{L}_{\Phi}$ come

\begin{equation}
\left(\mathcal{L}_{\Phi}f\right)(\sigma_{0}\sigma_{1}...)=\sum_{\sigma=0}^{n}e^{-\widehat{A}(\sigma\sigma_{0}\sigma_{1}...)}f(\sigma\sigma_{0}\sigma_{1}...);
\end{equation}

 se $f$ \` e una funzione continua nello nello spazio delle sequenze in $\{1,...,n\}^{\mathbb{Z}^{+}}_{T}$ allora, indicando con $\mathcal{L}_{\Phi}^{*}$ l'aggiunto di $\mathcal{L}_{\Phi}$, esiste $\lambda_{\Phi}>0$ tale che le equazioni

\begin{equation}
\mathcal{L}_{\Phi}h_{\Phi}=\lambda_{\Phi} h_{\Phi},\qquad \mathcal{L}^{*}_{\Phi}\widetilde{m}^{+}_{\Phi}=\lambda_{\Phi}\widetilde{m}^{+}_{\Phi}, \label{eq:pfr01}
\end{equation}

ammettono un'unica soluzione $h_{\Phi}$, $\widetilde{m}^{+}_{\Phi}$, con $h_{\Phi}$ continua in $\{1,...,n\}^{\mathbb{Z}}_{T}$ e $\widetilde{m}^{+}_{\Phi}$ definita in $\{1,...n\}^{\mathbb{Z}}_{T}$. Inoltre, indicando con $\widetilde{m}^{+}(f)$ il valore di aspettazione di $f$ con questa misura,

\begin{equation}
\lim_{k\rightarrow\infty}\sup_{\Phi\in B}\left|\lambda_{\Phi}^{-k}\mathcal{L}_{\Phi}^{k}f-\widetilde{m}_{\Phi}(f)h\right|=0.\label{eq:dimpro2}
\end{equation}

\end{prop}

\begin{proof}
Rimandiamo a \cite{pfr}, \cite{ergo}.
\end{proof}

Consideriamo la funzione $f_{\beta_{i}}$ definita nel seguente modo:

\begin{equation}
f_{\Phi,\beta_{i}}^{h}(\underline{\sigma})=\left\{\begin{array}{cc} e^{-\sum_{X\ni 0,X\not\subset[0,...,h-1],X\subset\left\{\mathbb{Z}^{-}\cup [0,...,h-1]\right\}}\Phi_{X}(\beta_{i}\underline{\sigma})_{X}} &\mbox{se $\beta_{i}$, $\underline{\sigma}$ sono compatibili} \\ 0 &\mbox{altrimenti} \end{array}\right.;
\end{equation}

allora, possiamo scrivere la funzione di partizione $Z_{h}(\beta_{i},\beta_{i+1})$, {\em cf.} formula (\ref{eq:appD7}), come, se $\underline{\omega}(\beta_{i+1})$ \` e una sequenza semi-infinita compatibile con $\beta_{i+1}$,

\begin{equation}
Z_{h}(\beta_{i},\beta_{i+1})=\left(\mathcal{L}_{\alpha}^{h}f_{\alpha,\beta_{i}}^{h}\right)(\beta_{i+1}\underline{\omega}(\beta_{i+1})).
\end{equation}

Quindi, se $\omega_{1}$, $\omega_{2}$ sono due generiche sequenze compatibili,

\begin{equation}
\frac{Z_{h}(\beta_{i}\beta_{i+1})Z_{h}(\underline{\omega}_{1},\underline{\omega}_{2})}{Z_{h}(\beta_{i},\underline{\omega}_{2})Z_{h}(\underline{\omega}_{1},\beta_{i+1})}=\frac{\mathcal{L}_{\alpha}^{h}f_{\alpha,\beta_{i}}^{h}(\beta_{i+1},\underline{\omega}(\beta_{i+1}))\mathcal{L}_{\alpha}^{h}f_{\alpha,\underline{\omega}_{1}}^{h}(\underline{\omega}_{2})}{\mathcal{L}_{\alpha}^{h}f_{\alpha,\beta_{i}}^{h}(\underline{\omega}_{2})\mathcal{L}_{\alpha}^{h}f_{\alpha,\underline{\omega}_{1}}^{h}(\beta_{i+1},\underline{\omega}(\beta_{i+1}))}, \label{eq:pfr2}
\end{equation}

e grazie alla propriet\` a (\ref{eq:dimpro2})

\begin{equation}
\lim_{h\rightarrow\infty}\sup_{\beta_{i},\beta_{i+1}}\frac{\mathcal{L}_{\alpha}^{h}f_{\alpha,\beta_{i}}^{h}(\beta_{i+1},\underline{\omega}(\beta_{i+1}))\mathcal{L}_{\alpha}^{h}f_{\alpha,\underline{\omega}_{1}}^{h}(\underline{\omega}_{2})}{\mathcal{L}_{\alpha}^{h}f_{\alpha,\beta_{i}}^{h}(\underline{\omega}_{2})\mathcal{L}_{\alpha}^{h}f_{\alpha,\underline{\omega}_{1}}^{h}(\beta_{i+1},\underline{\omega}(\beta_{i+1}))}=1. \label{eq:pfr0}
\end{equation}

La propriet\` a (\ref{eq:dimpro2}) \` e l'analogo del teorema di Perron - Frobenius\index{teorema di Perron - Frobenius} per l'operatore di trasferimento $\mathcal{L}$; se $h_{\alpha}$ \` e l'autofunzione corrispondente all'autovalore massimo (l'unico positivo secondo il teorema (\ref{prop:pfr})), possiamo decomporre una qualsiasi funzione continua $g$ definita in $\left\{1,...,n\right\}^{\mathbb{Z}^{+}}_{T}$ nella base delle autofunzioni di $\mathcal{L}_{\alpha}$ come

\begin{equation}
g = c_{1}h_{\alpha}+\widetilde{g}, \label{eq:pfr}
\end{equation} 

dove $c_{1}$ \` e la proiezione di $g$ su $h_{\alpha}$, ovvero per la (\ref{eq:pfr01}) $c_{1}=\widetilde{m}_{\alpha}^{+}(g)$; applicando $\lambda^{-h}\mathcal{L}_{\alpha}^{h}$ alla (\ref{eq:pfr}) otteniamo che

\begin{equation}
\lambda^{-h}\mathcal{L}_{\alpha}^{h}g=\widetilde{m}^{+}(g)h_{\alpha}\left(1+O\left(\frac{\Lambda}{\lambda}\right)^{h}\right),
\end{equation}

se $\Lambda$ \` e il secondo autovalore in modulo pi\` u grande appartenente allo spettro di $\mathcal{L}$. Quindi, poich\'e il coefficiente $\widetilde{m}^{+}(g)$ non dipende dal punto in cui era inizialmente calcolata $g$ (al contrario di $h_{\alpha}$), la (\ref{eq:pfr0}) \` e immediata.

\chapter{Analiticit\` a di un valor medio SRB}\label{app:B}

Considerando una generica osservabile $\mathcal{O}$ analitica in un parametro $\varepsilon$ e in $x$, ci chiediamo quale sia il comportamento del valor medio SRB nel caso in cui anche i potenziali della misura dipendano analiticamente da $\varepsilon$. In particolare, se $\partial_{\varepsilon}\left\langle\mathcal{O}\right\rangle_{\Lambda}$ \` e la derivata di un valor medio calcolato con la misura SRB approssimata\footnote{{\em Cf.} capitolo \ref{cap:noneq}, sezione \ref{sez:srbappr}.}, ovvero

\begin{eqnarray}
\partial_{\varepsilon}\left\langle\mathcal{O}\right\rangle_{\Lambda}&=&\left\langle\partial_{\varepsilon}\mathcal{O}\right\rangle_{\Lambda}\nonumber\\&&-\sum_{k=-\Lambda/2}^{\Lambda/2-1}\left(\left\langle\mathcal{O}\partial_{\varepsilon}A_{u}\circ S^{k}\right\rangle_{\Lambda}-\left\langle\mathcal{O}\right\rangle_{\Lambda}\left\langle\partial_{\varepsilon}A_{u}\circ S^{k}\right\rangle_{\Lambda}\right),\label{eq:appB01}
\end{eqnarray}

vorremmo dimostrare che per $\Lambda\rightarrow\infty$ il limite esiste ed \` e uguale a

\begin{equation}
\partial_{\varepsilon}\left\langle\mathcal{O}\right\rangle_{+}=\lim_{\Lambda\rightarrow\infty}\partial_{\varepsilon}\left\langle\mathcal{O}\right\rangle_{\Lambda}. \label{eq:exlim}
\end{equation}

La strategia che seguiremo sar\` a quella di sfruttare le tecniche illustrate nell'appendice \ref{app:D}. Grazie all'invarianza per traslazioni temporali della misura SRB possiamo riscrivere il valor medio di $\mathcal{O}$ come

\begin{equation}
\left\langle\mathcal{O}\right\rangle_{+}=\frac{1}{T}\left\langle\sum_{j=-\frac{T}{2}}^{\frac{T}{2}-1}\mathcal{O}\circ S^{j}\right\rangle_{+}; \label{eq:appB02}
\end{equation}

quindi, per la (\ref{eq:appB02}), indicando con $\left\{\Phi_{X}\right\}$ i potenziali a decadimento esponenziale invarianti per traslazioni della misura SRB, con un argomento del tutto analogo a quello esposto all'inizio dell'appendice \ref{app:D} o nella dimostrazione del teorema di fluttuazione si ottiene che

\begin{eqnarray}
\left\langle\mathcal{O}\right\rangle_{+}&=&\lim_{T\rightarrow\infty}\lim_{\Lambda\rightarrow\infty}\frac{1}{T}\left.\partial_{\beta}\log\frac{\sum_{\underline{\sigma}_{\Lambda}}e^{-\sum_{X\subset\Lambda}\Phi_{X}(\underline{\sigma}_{X})+\beta\sum_{X\subset\left[-\frac{T}{2},\frac{T}{2}\right]}\mathcal{O}_{X}(\underline{\sigma}_{X})}}{\sum_{\underline{\sigma}_{\Lambda}}e^{-\sum_{X\subset\Lambda}\Phi_{X}(\underline{\sigma}_{X})}}\right|_{\beta=0}\nonumber\\&=&\lim_{T\rightarrow\infty}\frac{1}{T}\left.\partial_{\beta}\log\frac{\sum_{\underline{\sigma}_{\left[-\frac{T}{2},\frac{T}{2}\right]}}e^{-\sum_{X\subset\left[-\frac{T}{2},\frac{T}{2}\right]}\Phi_{X}(\underline{\sigma}_{X})+\beta\sum_{X\subset\left[-\frac{T}{2},\frac{T}{2}\right]}\mathcal{O}_{X}(\underline{\sigma}_{X})}}{\sum_{\underline{\sigma}_{\left[-\frac{T}{2},\frac{T}{2}\right]}}e^{-\sum_{X\subset\left[-\frac{T}{2},\frac{T}{2}\right]}\Phi_{X}(\underline{\sigma}_{X})}}\right|_{\beta=0}.\label{eq:appB03}\nonumber\\
\end{eqnarray}

Nella (\ref{eq:appB03}), $\left\{\mathcal{O}_{X}(\underline{\sigma}_{X})\right\}$ sono i potenziali che si ottengono sviluppando la funzione $\mathcal{O}$, {\em cf.} proposizione \ref{prop:svilpot}, anch'essi invarianti per traslazioni ed esponenzialmente decrescenti. A questo punto usariamo le tecniche esposte nell'appendice \ref{app:D} (decimazione e cluster expansion) per scrivere la (\ref{eq:appB03}) come

\begin{eqnarray}
(\ref{eq:appB03})&=&\lim_{T\rightarrow\infty}\frac{1}{T}\left.\partial_{\beta}\sum_{\Gamma\subset\left[-\frac{T}{2},\frac{T}{2}\right]}\varphi^{T}(\Gamma)\left[\zeta_{\beta}(\Gamma)-\zeta(\Gamma)\right]\right|_{\beta=0}\nonumber\\&=&\lim_{T\rightarrow\infty}\frac{1}{T}\sum_{\Gamma\subset\left[-\frac{T}{2},\frac{T}{2}\right]}\varphi^{T}(\Gamma)\left.\partial_{\beta}\zeta_{\beta}(\Gamma)\right|_{\beta=0}, \label{eq:appB030}
\end{eqnarray}

ed \` e facile dimostrare che la {\em nuova attivit\` a} di una configurazione di polimeri $\Gamma$ $\left.\partial_{\beta}\zeta_{\beta}(\Gamma)\right|_{\beta=0}$ pu\` o essere stimata in modo analogo a $\zeta(\gamma)$. Infatti, dopo la procedura di decimazione avremo che

\begin{equation}
\left|\zeta_{\beta}(\gamma)\right|\leq e^{-L_{0}\left|\gamma\right|}e^{-\frac{1}{2}\kappa\delta(\gamma)},
\end{equation}

per $L_{0},\kappa>0$ opportuni, {\em cf.} proposizione \ref{prop:appD3}; quindi, poich\'e ({\em cf.} formula (\ref{eq:att}))  

\begin{equation}
\zeta_{\beta}(\gamma)=\left\langle\sum_{q}\sum^{**}_{\begin{array}{c} \scriptstyle{Y_{1},...,Y_{q}\subset\gamma} \\ \scriptstyle{\bigcup_{j}Y_{j}=\gamma} \end{array}}\prod_{j=1}^{q}\left(e^{-\widetilde{\Phi}^{\beta}_{X}\left(\underline{\sigma}_{Y_{j}}\right)}-1\right)\right\rangle, \label{eq:att2}
\end{equation}

dove $\widetilde{\Phi}_{X}^{\beta}$ \` e il potenziale analitico in $\beta$ che si ottiene dopo aver decimato il modello, usando una stima di Cauchy\footnote{Come abbiamo visto nella dimostrazione del teorema di Anosov per il gatto di Arnold, {\em cf.} sezione \ref{sez:cat} del capitolo \ref{cap:ergo}, se $\mathcal{D}=\left\{\beta:\mbox{Im}\,\beta\leq \beta_{0}\right\}$ con $\beta_{0}>0$, dal teorema di Cauchy sappiamo che possiamo stimare la derivata in $\beta$ di una funzione analitica $f(\beta)$ come

\begin{equation}
\partial_{\beta}f(\beta)\leq \frac{\max_{\beta\in\mathcal{D}}f(\beta)}{\beta_{0}}.
\end{equation}

} per maggiorare le derivate in $\beta$ di  dei potenziali con i loro moduli possiamo trovare\footnote{Il fattore $C|\gamma|$ tiene conto anche delle derivate sulla ``misura'' con la quale \` e calcolato il ``valor medio'' nella (\ref{eq:att2}), {\em cf.} proposizione \ref{prop:appD3}.} $C>0$ tale che

\begin{equation}
\left|\left.\partial_{\beta}\zeta_{\beta}(\gamma)\right|_{\beta=0}\right|\leq C|\gamma|e^{-L_{0}\left|\gamma\right|}e^{-\frac{1}{2}\kappa\delta(\gamma)}, \label{eq:stimaL}
 \end{equation}
 
dunque il fattore $|\gamma|$ non pone alcun problema, dal momento che pu\` o essere controllato con $e^{-\frac{L_{0}}{2}|\gamma|}$. Grazie all'invarianza per traslazioni dei potenziali e alla loro analiticit\` a in $\varepsilon$, dal corollario \ref{cor:appD1} sappiamo che il limite $T\rightarrow\infty$, ovvero il valor medio SRB di $\mathcal{O}$, \` e analitico in $\varepsilon$.

Infine, la derivata del valore di aspettazione con la misura SRB approssimata $\partial_{\varepsilon}\left\langle\mathcal{O}\right\rangle_{T}$ pu\` o essere scritta come

\begin{equation}
\frac{1}{T}\sum_{\Gamma\subset\left[-\frac{T}{2},\frac{T}{2}\right]}\varphi^{T}(\Gamma)\left.\partial_{\varepsilon\beta}\zeta_{\beta}(\Gamma)\right|_{\beta=0}, \label{eq:app05}
\end{equation}

e per le (\ref{eq:att2}), (\ref{eq:stimaL}), in un modo completamente analogo a quello usato per ottenere la stima (\ref{eq:stimaL}), abbiamo che, se $C'$ \` e un'opportuna costante positiva,

\begin{equation}
\left|\left.\partial_{\varepsilon}\partial_{\beta}\zeta(\gamma)\right|_{\beta=0}\right|\leq C'|\gamma|\left|\left.\partial_{\beta}\zeta(\gamma)\right|_{\beta=0}\right|\leq C'C|\gamma|^{2}e^{-L_{0}\left|\gamma\right|}e^{-\frac{1}{2}\kappa\delta(\gamma)}; \label{eq:stimaL2}
\end{equation}

quindi, grazie al risultato (\ref{eq:stimaL2}) e all'invarianza per traslazioni dei potenziali la (\ref{eq:app05}) converge uniformemente in $\varepsilon$ per $T\rightarrow\infty$, e il limite \` e analitico. Questo risultato ci permette di scambiare la derivata rispetto a $\varepsilon$ con il limite $T\rightarrow\infty$ che definisce la SRB, ovvero

\begin{equation}
\partial_{\varepsilon}\lim_{T\rightarrow\infty}\left\langle\mathcal{O}\right\rangle_{T}=\lim_{T\rightarrow\infty}\partial_{\varepsilon}\left\langle\mathcal{O}\right\rangle_{T},
\end{equation}

e dunque

\begin{equation}
\partial_{\varepsilon}\left\langle\mathcal{O}\right\rangle_{+}=\left\langle\partial_{\varepsilon}\mathcal{O}\right\rangle_{+}-\sum_{k=-\infty}^{\infty}\left[\left\langle\mathcal{O}\partial_{\varepsilon}A_{u}\circ S^{k}\right\rangle_{+}-\left\langle\mathcal{O}\right\rangle_{+}\left\langle \partial_{\varepsilon}A_{u}\circ S^{k}\right\rangle_{+}\right].
\end{equation}

\end{appendix}

\chapter*{Ringraziamenti}

Vorrei ringraziare il Prof. Giovanni Gallavotti per avermi seguito durante il lavoro di tesi, e il Dr. Pierluigi Falco per innumerevoli discussioni. Ringrazio anche il Dr. Alessandro Giuliani, Carlo Barbieri e il Prof. Marzio Cassandro, per un chiarimento dell'ultimo minuto. 

\bibliographystyle{unsrtit}


 \end{document}